\documentclass[11pt]{article}

\usepackage{latexsym} 
\input xy
 \xyoption{all}
 \xyoption{arc}
\input{logicmacros.sty} 
\begin{document}

\title{Logics for   Epistemic Actions:\\ Completeness,
Decidability, Expressivity
}

    \author{
           $\begin{array}{c}\mbox{Alexandru Baltag}\\
\mbox{ILLC, University of Amsterdam}\\
\mbox{1090 GE Amsterdam} \\
                     \mbox{Netherlands}\\
           \mbox{\tt  TheAlexandruBaltag@gmail.com}
\\
           \end{array}$
\and
            $\begin{array}{c}\mbox{Lawrence S.~Moss}\\
\mbox{Mathematics Department} \\ 
\mbox{Indiana University} \\ \mbox{Bloomington, IN 47405-7106 USA}\\
           \mbox{\tt lmoss@indiana.edu}\\
           \end{array}$
\and
    $\begin{array}{c}\mbox{S{\l}awomir Solecki}\\
\mbox{Mathematics Department} \\ 
\mbox{Cornell University} \\
%\mbox{1409 W.~Green Street (MC-382)} \\
\mbox{Ithaca, New York 14853 USA} \\
 \mbox{\tt ss3777@cornell.edu}\\
           \end{array}$      }

%\addtocounter{footnote}{1}
%\footnotetext{ }

%\addtocounter{footnote}{1}
%\footnotetext{}

\date{ }

\maketitle

\begin{center}
\fbox{
\begin{minipage}{5.25in}
\medskip

This paper was intended to be the ``journal'' version of our 1998 paper 
``The logic  of common knowledge, public announcements,
and private suspicions.''  It was mainly written in 2004, with a few bibliographic additions
coming a few years later.
\medskip
\end{minipage}
}
\end{center}

\begin{abstract}
We consider dynamic versions of epistemic logic
as formulated in~\cite{kra}. 
That paper proposed a family
of logical languages $\lang(\bSigma)$ parameterized by
   {\em action signatures}.  In addition to 
$\lang(\bSigma)$, we consider two fragments
$\lang_0(\bSigma)$  and $\lang_1(\bSigma)$ of  it.
We review the syntax and semantics of these languages
$\lang_0(\bSigma), \lang_1(\bSigma)$, and $\lang(\bSigma)$,
as well as the  sound
proof systems for the validities in them.
It was shown in~\cite{miller} that validity in $\lang(\bSigma)$
is $\Pi^1_1$-complete, so there are no recursively axiomatized
complete logical systems  for  it.  On the positive  side,
this paper
proves the strong completeness of the axiomatization 
of $\lang_0(\bSigma)$ and the weak completeness
of $\lang_1(\bSigma)$.  The work involve a detour
into term rewriting theory.   And since at the heart of the argument
is modal filtration, it gives the finite model
property and hence decidability.  We also give a translation
of $\lang_1(\bSigma)$ into PDL, hence we obtain a second proof
of the
decidability of $\lang_1(\bSigma)$.  The paper closes with some 
results on expressive power.  These are mostly concerned with 
comparing $\lang_1(\bSigma)$ with modal logic together with
transitive closure operators.  We also answer a natural question
about the 
languages we get by  varying the action signature.
In particular, we prove that a logical language with operators
for private announcements is more expressive  than one for  
public announcements.
\end{abstract}

\pagebreak

\tableofcontents

\pagebreak

\section{Introduction}

One of the goals of  {\em dynamic epistemic  logic\/} is  to
 construct logical languages which allow one to represent 
a variety of possible types of changes affecting the
information states
of agents in a multi-agent setting.   
One wants (formal) logical systems with primitive operations corresponding
to (informal) notions  such as {\em public announcement},
{\em completely private announcement}, 
{\em  private announcement to one agent with suspicion by another},
etc.  And then after the logics are formulated, one would ideally
want technical tools to use in their study and application.

The full formulation of such logics is somewhat of a complicated
story.  In the first place, one needs a reasonable syntax.
On the semantic side, there is an unusual feature in that the
truth of a sentence at a point in one model often depends on
the truth of a related sentence in a {\em different\/} model.
In effect, the kinds of  actions we are interested in give
rise to functions, or relations more generally, on the class
of all possible models.  

There have been some proposals on logical systems for dynamic  epistemic
logic beginning  with the  work of 
Plaza~\cite{plaza},
Gerbrandy~\cite{gerbrandy98,gerbrandyphd}, and
Gerbrandy and Groeneveld~\cite{gerbrandygroeneveld}.
These papers formulated logical systems for the informal notions
of public announcement and completely private announcement.  They
left open the matter of axiomatizing the logics in the presence
of common knowledge operators.  (Without the
common knowledge operators, the systems are seen to be
variants on standard multi-model logic.)  And they also left open the 
question of the decidability of these systems.  
Our work began with complete axiomatizations for  these systems
and for more general systems.  Our results were presented
in~\cite{bms}.  As it happens, the logics which we constructed
in~\cite{bms} did not have the friendliest syntax.  
The first two authors pursued this matter for  some time and eventually
came to the proposals in~\cite{kra}.  
The logical systems for the validities of
the ultimate languages are  certainly related to those in~\cite{bms}.
But due to the different formulation of the overall syntax
and  semantics, all of the work on soundness and completeness had
to be completely reworked.  The main purpose of the present paper
is to present these results.  Secondarily, we present some technical
results  on the systems such as results  on expressive power.

\paragraph{Contents of this paper: a high level view}
This paper is a long technical development, and for this reason
we did not  include
a full-scale motivation of the logical systems themselves.
For a great deal of  motivational material,
 one should see~\cite{bms}.
Section~\ref{section-bigpicture} formulates all of the definitions
needed in the paper.   The presentation is based closely on
the work in~\cite{kra}, so readers familiar with that paper
may use Section~\ref{section-bigpicture} as a review.
Other readers of this paper could take the logical systems here to
simply be extensions  of
propositional 
dynamic logic which allow one to  make 
``transitions from model to model''
in addition to  transitions ``inside a given model.''
In addition, Section~\ref{section-bisimulation} is new in this paper.

The logical languages studied in this paper are presented in 
Section~\ref{section-language}.
Section~\ref{section-Cn} presents some 
examples of the semantics,   chosen to foreshadow
work in Section~\ref{section-expressivepower} on expressive power.
The logical system is presented in Section~\ref{section-langzero},
along with a soundness result for most of the system.    One especially significant
inference rule, called the \emph{Action Rule} is studied in Section~\ref{section-biggersystem}.
The completeness theorem for the logic comes in Section~\ref{section-completeness}.
The final section deals with questions of expressive power, and it may be read
after Section~\ref{section-langzero}.

Two aspects of our overall machinery are worth pointing out.
The first is the use of the \emph{canonical action model}.   This is a semantic
object built from syntactic objects (sequences of \emph{simple actions},  terms  
in our language $\lang(\bSigma)$).    We would like to 
think that    the use of
the canonical model makes for a more elegant presentation than one would
otherwise have.    

The second feature is a \emph{term rewriting system}
for dealing with assertions in the language.  The 
semantic equivalences in dynamic epistemic logic are sufficiently complicated that
the `subsentence' relation is not the most natural
or useful one for many purposes.   Instead,
one needs to handcraft various ordering relations for use in inductive proofs,
or in defining translations.   For just one example, in the logic of public announcements
one has an \emph{announcement-knowledge axiom} in a form such as
\begin{equation}
[!\phi]\necc \psi 
\iiff
(\phi\iif \necc [!\phi]\psi).
\label{law}
\end{equation}
This is for a logic with just one agent, and in this discussion we are forgetting
about common knowledge.   To prove that the logic has a translation $t$
back to ordinary modal logic, one wants to use the equivalence 
in (\ref{law}) above
as the key step in the translation, defining $([\phi]\necc]\psi)^t$ to
be $\phi^t\iif (\necc[!\phi]\psi)^t$.   But then one needs to have some reason
to say that both $\phi$ and (more critically) $\necc[!\phi]\psi$
are of \emph{lower complexity} than the original 
formula $[\phi]\necc\psi$.     There are several ways to make this precise.
One is to directly assign an element of some well-founded set to each 
formula and then use this as a measure of complexity.   This is done in~\cite{dhk},
and the well-founded set is the natural numbers.     Our treatment is 
different,  mainly because it goes via term rewriting.  
In effect, one takes the well-founded set to be the sentences themselves,
with the order given by substitution using laws such as (\ref{law}), but oriented
in a specified direction (left-to-right, for example).   
Then the fact that we have a well-founded relation is a result one proves
about oriented sets of laws (term rewriting systems).   
In our case, we use an interpretation constructed by hand.

Our term rewriting system is presented and studied in Section~\ref{section-proofs}.
It actually deals not with $\lang(\bSigma)$ but with a different language called $\lang_1(\bSigma)$.
 An important by-product of our completeness proof via rewriting
 is a Normal Form Theorem for $\lang_1(\bSigma)$.
 This is an interesting result in itself, and it does not follow from other completeness
 proofs such as those in~\cite{dhk}, or in the paper~\cite{dhk}, which we discuss next.
  In order to use the rewriting system.
  one needs to know that substitution of ``equivalent'' objects preserves
``equivalence''  and that the proof system itself is strong enough to 
reduce sentences to normal form.
 Unfortunately, this ``obvious''  point takes a great deal of work.
 The details are all in this paper, and to our knowledge no other source presents
 complete proofs on these matters.

\paragraph{Comparison with some other work}
The first version of this paper is our  1998 conference publication~\cite{bms}.
We ourselves had several versions of this paper issued as technical
reports or posted on web sites, and so in some sense the results 
here are \emph{public} but not \emph{published}.   
In the past ten years, the subject  of dynamic epistemic logic
has taken off in a serious way.    It sees several papers a year.
But it seems fair to say that the overall topics of investigation are not 
the logical systems presented in the any of the original papers but rather
are adaptations of the logics to settings involving probability, belief revision,
quantum information, and the like.  
Still, the work reported here has been the subject of several
existing  publications.
And so it makes sense
for us to make the case that the results in this paper are still relevant
and do not follow from previously published results.

The book \emph{Dynamic Epistemic Logic} by 
Hans van Ditmarsch, Wiebe van der Hoek, and Barteld Kooi~\cite{dhk}
is a textbook presentation containing
 some of the content of this paper and~\cite{kra},
with additional material on model checking, belief revision, and other topics.
Section 6.4 presents the syntax of what it calls \emph{action model logic}
and writes as $\lang^{\rm act}_{KC\otimes}(A,P)$.
At first glance, the syntax of $\lang^{\rm act}_{KC\otimes}(A,P)$
seems fairly close to the language $\lang(\bSigma)$ presented in~\cite{kra}
and reviewed
in Section~\ref{section-language} of this paper.  
A relatively small difference concerns the \emph{action models} in 
$\lang^{\rm act}_{KC\otimes}(A,P)$.
(In general, action models are like Kripke models together with
a ``precondition'' function mapping worlds to sentences.)
These are required to have the property that
each agent's accessibility relation is an equivalence relation; our treatment
is more general and hence can present logics for epistemic actions such as
``cheating'' in games.   
But the main difference is that our language $\lang(\bSigma)$
uses what everyone would take to be  a bona fide
 syntax: sentences and program expressions are
linearly ordered strings of symbols.   
In contrast, the syntax of $\lang^{\rm act}_{KC\otimes}(A,P)$
employs action models \emph{directly}.     Structured but unordered 
objects occur inside of sentences.  In a different context, it would be 
like studying formal language and their connection to automata by studying 
something that was like a regular expression but allowed 
 finite automata to directly occur inside some syntactic object.
  
  We have no objection ourselves to 
this move.    In fact, this was the presentation we chose in
our first version~\cite{bms} of this paper.
But over the years we have found quite a lot of resistance to
this presentation
on the grounds that one is ``mixing syntax and semantics.''   Again, we do not
assert this objection, but we \emph{address} it by employing the large technical
machinery that we see in~\cite{kra} and also Sections~\ref{section-bigpicture}
and~\ref{section-language}.   (We are keenly aware of the irony of our being
criticized for ``mixing syntax and semantics'' in the ten-year-old publication~\cite{bms},
and after we ``un-mixed'' them in~\cite{kra}, we find others making the same natural
move.)

 We would also like to mention  the paper ``Logics of Communication and Change'',
  van Benthem,   van Eijck and  Kooi~\cite{bek}.  
This  paper   presents completeness theorem
for a logical system in the same family as ours.   
But as it happens, the paper studies a different system.  
To see the differences, it is worthwhile to note that the source of much 
of the work in this area concerns assertions of common knowledge following
an action of some sort.   In our setting, these would be written $[\alpha]\necc^*\phi$.
It turns out that there is no equivalence of the form (\ref{law}) for assertions like
this.   This leads us to our Action Rule, an inference rule allowing the derivation
of sentences of the kind under discussion.   
The idea in~\cite{bek}
 is to start with propositional dynamic logic, called E-PDL
in the paper (`E' for `epistemic').   Then one adds
modalities corresponding to update models
$[{\sf U},e]$; this modality would be like our $[\alpha]$, but in~\cite{bek}
as in~\cite{bms} and~\cite{dhk}, this {\sf U} is an action model
rather than a bona fide syntactic object.   Then one proves that 
the addition can be translated away.   This very interesting and suggestive
result is not available in our setting.      The paper~\cite{bek} notes
that ``Indeed, reduction axioms are not available, as the logic with
epistemic updates is more expressive than the logic without them.
We think this is an infelicity of design.''    This is somewhat debatable,
since the additional expressive power in E-PDL is not something in wide use.
``We have found no practical use for these [complex combinations of agent
acessibility relations] at present, but they are the price that we cheerfully pay
for having a language living in expressive harmony with its dynamic superstructure
$\ldots$.''   It also holds that ``it's completeness theorem [from~\cite{bms}] is correspondingly
mess.''   This extra ``mess'' is
due  partly to the strictly ``syntactic syntax'', partly to the explicit rewriting work
(taken for granted in~\cite{bek}), and partly to the noted ``disharmony''.
It is not known whether E-PDL is actually stronger than the logical
systems in this paper.   If it is, then this would show that the completeness result here
is definitely not a special case of that in~\cite{bek}.   But again, the matter is open at this time:
one cannot derive our result from~\cite{bek} since the systems are
prima facie different.

\Section{Definitions}

This section provides all of our definitions.  It is short on 
motivation, and we have situated the 
examples in  Section~\ref{section-Cn},
 following all the definitions.
See  also Sections 1 and 2 of~\cite{kra} for the motivation from
epistemic  logic and for a more leisurely presentation.

\label{section-big-picture}
\label{section-bigpicture}

\subsection{State models and propositions}
\label{section-state-models}

We fix a set $\AtSen$ of {\em atomic sentences\/} and also
a set $\Agents$ of {\em agents}.  All of our definitions are
relative to these sets.
% and in this paper
%we almost never have occasion to vary
%our choices.

A  {\em state model \/}   is a triple
 ${\bS}=(S,\arrowAgents_{\!\!\bS},\|\cdot\|_{\bS})$
consisting of a set $S$ of ``states''; a family 
$\arrowA_{\!\!\bS}$ of
 binary accessibility
relations $\arrowA_{\!\!\bS}\subseteq S\times S$, one for each agent $A\in\Agents$;
and a ``valuation'' (or a ``truth'' map)
$\map{{\|.\|}_{\bS}}{\AtSen}{\pow(S)}$, assigning to
each   atomic sentence $p$ a set $\|p\|_{\bS}$ of states.
These are exactly Kripke models generalized by having one
accessibility relation for each agent.   We use the terminology of
state models because there are other kinds of  models, 
{\em action models\/}
and {\em  program models}, in our  study.
When dealing with a single fixed state model $\bS$, we often drop
the subscript  $\bS$ from
all the notation.

\begin{definition}
Let $\StateModels$ be the collection of all state models.
 An {\em epistemic proposition\/} is an operation 
${\propositionphi}$ defined on $\StateModels$ such that for all
$\bS\in \StateModels$, $\propositionphi_{\bS}\subseteq S$.
% The standard notation for $s\in \semantics{\propositionphi}{\bS}$
\end{definition}
\label{section-epistemic-prop-definition}

\noindent
The collection of epistemic propositions is closed in
various ways.

\begin{enumerate}
\item For each atomic sentence $p$ we have an atomic proposition $\propositionp$ 
with $\propositionp_{\bS} = \|p\|_{\bS}$.
\item If
$\propositionphi$ is an epistemic proposition,  
then so is
$\nott\propositionphi$, where 
$(\nott\propositionphi)_{\bS} = S\setminus\propositionphi_{\bS}$.
\item  \label{partc}
If  $C$ is a set or class of epistemic propositions, then $\bigwedge
C$ is an epistemic proposition, 
where $$\semanticsoff{\bigwedge C}{\bS} \quadeq
\bigcap \set{\propositionphi_{\bS} : 
\propositionphi\in C}.$$
\item Taking $C$ above to be empty, we
have an {\em ``always true''} epistemic proposition $\trueproposition$, with
$ \trueproposition_{\bS} =
S$.
\item  We also may take $C$ in part \ref{partc} to be a two-element set
$\set{\propositionphi,\propositionpsi}$; here we write 
$\propositionphi\andd\propositionpsi$ instead of $\bigwedge\set{\propositionphi,\propositionpsi}$.
We see that  if $\propositionphi$ and $\propositionpsi$ are
epistemic propositions, then so is
$\propositionphi\andd\propositionpsi$, with
 $\semanticsoff{\propositionphi\andd\propositionpsi}{\bS}
= 
\propositionphi_{\bS}\cap \propositionpsi_{\bS}$.
\item \label{partnecc}
If $\propositionphi$ is an epistemic proposition and $A\in \Agents$, then
$\necc_A\propositionphi$ is an epistemic proposition, with
\begin{equation}
\semanticsoff{\necc_A\propositionphi}{\bS} \quadeq \set{s\in S: \mbox{if $s\arrowA
t$, then
$t\in\propositionphi_{\bS}$}}.
\label{eq-semanticsbox}
\end{equation}
\item 
If $\propositionphi$ is an epistemic proposition and $\BB\subseteq \Agents$, then
$\necc^*_{\BB}\propositionphi$ is an epistemic proposition, with
$$\semanticsoff{\necc^*_{\BB}\propositionphi}{\bS} \quadeq \set{s\in S: \mbox{if $s\arrowstarBB t$,
then
$t\in\propositionphi_{\bS}$}}.
$$
Here $s\arrowstarBB t$ iff there is a sequence
$$ s = u_0 \quad 
\stackrel{A_1}{\rightarrow} \quad u_1 \quad
 \stackrel{A_2}{\rightarrow}  \quad 
\cdots \stackrel{A_n}{\rightarrow}  \quad u_{n+1} = t$$
where $A_1,\ldots, A_n\in \BB$.
In other words, there is a 
sequence of arrows
labelled with agents from the set $\BB$ taking $s$ to $t$.
We allow $n = 0$ here, so $\arrowstarBB$ includes
the identity relation on $S$.
%To see that  $\necc^*_{\BB}\propositionphi$ is indeed an
%epistemic proposition, we use parts~\ref{partc} and~\ref{partnecc} above;
%we may also argue directly, of course.
\end{enumerate}

\paragraph{Syntactic and semantic notions}
It will be important for us to make a sharp distinction between
syntactic and semantic notions.  We have already begun to do
this, speaking of
atomic {\em sentences\/} and atomic {\em propositions}.
The difference for us is that atomic sentences are entirely
syntactic objects: we won't treat an atomic sentence $p$ as 
anything except an unanalyzed mathematical object.  On the
other hand, this atomic sentence $p$ also has associated
with it the atomic proposition $\propositionp$.  
As defined in point 1,
$\propositionp$ will be a function whose domain is the
(proper class of) state models, and it is defined by
\begin{equation}
\propositionp_{\bS} \quadeq
\set{s\in S : s \in \|p\|_{\bS}}.
\label{eq-pedantic}
\end{equation}
This difference may seem pedantic at first, and surely
there are times when it is sensible to blur it.  But
for various reasons that will hopefully become clear,
we need to insist on it.    

\medskip

Up until now, the only syntactic objects
have been the atomic sentences $p\in\AtSen$.  But we can build
the collections of {\em finitary\/} and {\em infinitary
sentences\/} by the same definitions that we have seen,
and then the work of the past section is the {\em semantics\/}
of our logical languages.   For example, we have sentences
$p\andd q$, $\necc_A\nott p$, and $\necc^*_{\BB} q$.  These 
then  have corresponding epistemic propositions as their semantics:
$\propositionp\andd \propositionq$, $\necc_A\nott \propositionp$,
and $\necc^*_{\BB} \propositionq$,
respectively.   Note that the latter is a properly infinitary proposition
(and so  $\necc^*_{\BB} q$ is a properly infinitary sentence);
it abbreviates an infinite conjunction.

\subsection{Updates}
\label{section-epistemic-actions}
\label{section-updates}

A {\em transition relation\/} between state models $\bS$ and $\bT$ is 
a relation 
 between the sets $S$ and $T$;
i.e., a subset of $S\times T$.
%% We write $\map{r}{\bS}{\bT}$ for this.
%%%  I DROP THIS POINT BECAUSE OF SLAWEK'S OBJECTION.
%%
 An {\em update\/} $\update$ is a pair
of operations
 $$\actionalpha \quadeq
(\bS\mapsto\actionmodel{\bS}{\actionalpha},\bS\mapsto \actionalpha_{\bS}),$$
where for each $\bS\in\StateModels$,
%% $\dom$ is an epistemic proposition,
%% $\bS \mapsto \actionmodel{\bS}{\actionalpha}$ on $\StateModels$, 
$\map{\actionalpha_{\bS}}{\bS}{\actionmodel{\bS}{\actionalpha}}$
 is a transition relation. 
We call $\bS\mapsto\actionmodel{\bS}{\actionalpha}$ the {\em update map},
and $\bS\mapsto \actionalpha_{\bS}$ the {\em update relation}.

We continue our general discussion by noting that 
 the collection of updates is closed in various ways.
\begin{enumerate}
\item {Skip}:
there is an update called ``Skip'' and denoted
$\skippaction$, with $\modelaction{\bS}{\skippaction} =
\bS$, and
${\skippaction}_\bS$ is the identity relation on $\bS$. 
\item {\em Sequential Composition}: if $\actionalpha$ and $\actionbeta$ are epistemic updates, then 
their composition $\actionalpha\then \actionbeta$ is again an epistemic update, 
where $\bS(\actionalpha\then\actionbeta) = 
\modelaction{\modelaction{\bS}{\actionalpha}}{\actionbeta}$, and
$(\actionalpha\then\actionbeta)_\bS = 
\actionalpha_{\bS}
 \then {\actionbeta}_{\modelaction{\bS}{\actionalpha}}$. Here, we use on the
right side the usual composition $\then$ of relations.\footnote{We 
are writing relational composition in left-to-right order in this
paper.}
\item \label{pointx} {\em Union} (or {\em Non-deterministic choice}):
If $X$ is any set of epistemic updates,
then the union $\Union X$ is an epistemic update,
defined as follows. 
For each $\bS$, the set of states of the model 
$\bS(\Union X)$
is the {\em disjoint union} of all the sets of states in each model $\bS(\actionalpha)$
for $\actionalpha\in X$: 
$$\set{(s,\actionalpha) : \actionalpha\in X \mbox{ and } s\in \bS(\actionalpha)}.$$
Similarly, each accessibility relation $\arrowA$ is defined as
the disjoint union of the corresponding accessibility relations in
each model: 
$$ (t,\actionalpha) \arrowA (u,\actionbeta)
\quadiff \mbox{ if $\actionalpha=\actionbeta$ and $t\arrowA u$ in 
$\bS(\actionalpha)$}.$$
The valuation $\|p\|_{\bS(\Union X)}$ in $\bS(\Union X)$ is the 
disjoint union of the valuations in each
state model:
$$\|p\|_{\bS(\Union X)} \quadeq \set{(s,\actionalpha): \actionalpha\in X \mbox{ and }
s \in \|p\|_{\bS(\actionalpha)}}.$$
Finally, the 
update relation $(\Union X)_{\bS}$ 
between $\bS$ and $\bS(\Union X)$ is 
the union of all the update relations $\actionalpha_{\bS}$:
$$  t \ (\Union X)_{\bS}\ (u,\actionalpha)
\quadiff 
t \ \actionalpha_{\bS}\ u.$$
\item Special case: {\em binary union}.
The (disjoint) union of two epistemic updates $\actionalpha$ and
$\actionbeta$ is an update
$\actionalpha\union  \actionbeta$, given by
$\actionalpha\union\actionbeta =\Union\{\actionalpha,\actionbeta\}$.
\item Another special case: {\em Kleene star} ({\em iteration}).
We have the operation of Kleene star on updates:
$$\actionalpha^{*} \quadeq
\Union\set{\skippaction,\actionalpha,\actionalpha\cdot\actionalpha, \ldots,
\actionalpha^n,\ldots}$$
where $\actionalpha^n$ is recursively defined by $\actionalpha^0=
\skippaction$,
$\actionalpha^{n+1}=\actionalpha^n \then \actionalpha$.
\item {Crash}: We can also take $X = \emptyset$ in
part~\ref{pointx}.  This gives an update called ``Crash'' and denoted
 $\actionzero$ such that
$\bS(\actionzero)$ is the empty model for each $\bS$, and ${\actionzero}_{\bS}$ is the empty relation.
\end{enumerate}

The operations $\actionalpha \then \actionbeta$, $\actionalpha \union\actionbeta$
 and $\actionalpha^*$ are
the natural analogues of the operations of union of relations,
relational composition and iteration, and of
the regular operations on programs in PDL. 
The intended meanings
are: for $\actionalpha \then \actionbeta$,
 sequential composition ({\em do $\actionalpha$, then do $\actionbeta$});
for $\actionalpha \union\actionbeta$, non-deterministic choice ({\em do 
either $\actionalpha$ or
$\actionbeta$}); for
$\actionalpha^*$, iteration ({\em repeat $\actionalpha$ some finite number of times}).

\paragraph{Standard updates}
An update $\actiona$ is {\em standard\/} if for
all state models $\bS$, $(\actiona_{\bS})^{-1}$ is a partial function.
(The relation  $(\actiona_{\bS})^{-1}$
 the inverse of the update
relation; it is a subset of $\bS(\actiona)\times \bS$.)

\begin{proposition} The update
$\skippaction$ is
 standard.   The composition of standard updates
is standard, as is any union of standard updates.
\label{proposition-standard}
\end{proposition}

\paragraph{Updates Determine Dynamic
Modalities}\label{section-newop} 
 If $\propositionphi$ is an epistemic proposition and
$\actionalpha$ an update,
then $[\actionalpha]\propositionphi$ is an epistemic proposition defined by
\begin{equation}
 \semanticsoff{[\actionalpha]\propositionphi}{\bS}  \quadeq \set{s\in S: 
\mbox{for all $t$ such that $s\ \actionalpha_{\bS}\ t$,  $t\in
\propositionphi_{\actionmodel{\bS}{\actionalpha}}$}}.
\label{eq-alpha-phi}
\end{equation}
We should compare (\ref{eq-alpha-phi}) and
(\ref{eq-semanticsbox}).  The point is that we may treat updates
in a similar manner to other box-like modalities; the 
structure given by an update allows us to do this.  

We also define the dual proposition  $\pair{\actionalpha}\propositionphi$
by
$$ \semanticsoff{\pair{\actionalpha}\propositionphi}{\bS}  \quadeq \set{s\in S: 
\mbox{for some $t$ such that $s\ \actionalpha_{\bS}\ t$,  $t\in
\propositionphi_{\actionmodel{\bS}{\actionalpha}}$}}.
$$

\subsection{Action models and program models}
\label{section-simple-action-structures}

 Let $\Phi$ be the collection of
all epistemic propositions.
An   {\em (epistemic)  action model\/}  is a triple
$\bSigma=(\Sigma, \arrowAgents, \Presemantic)$, 
where $\Sigma$ is a set of {\em simple actions},
$\arrowAgents$ is an $\Agents$-indexed family of binary
relations on $\Sigma$, and $\map{\Presemantic}{\Sigma}{\Phi}$.

\rem{
An epistemic action model 
is similar to  a  state
model. But we call the members
of the set $\Sigma$  ``simple actions'' (instead
of states). 
We use different notation
and terminology because
of a technical difference and a bigger conceptual point.
The technical difference is that $\map{\Presemantic}{\Sigma}{\Phi}$
(that is, the codomain is the collection
of all epistemic propositions). 
The conceptual point is that we
think of ``simple'' actions as being {\em deterministic} actions
whose epistemic impact is {\em uniform on states} (in the sense
explained in our Introduction). 
So we think of ``simple'' actions as
particularly simple kinds of deterministic actions, whose appearance
to agents is uniform: the agents' uncertainties concerning the
current action are independent of their uncertainties concerning the
current
state. This allows us to abstract away the action uncertainties
and represent them as a Kripke structure of actions, in effect
forgetting the state uncertainties. 

As announced in the Introduction, this uniformity of appearance is
restricted only to the action's domain of applicability, defined
by its preconditions. Thus, 
for a simple action $\sigma\in\Sigma$,
we interpret $\Presemantic(\sigma)$ as giving the {\em
 precondition  of   $\sigma$}; this is
the fact that needs to hold at a state
(in a state model) in order for action $\sigma$ to
be ``accepted'' in that state. So   $\sigma$
will be executable in  $s$ iff  its precondition 
$\Presemantic(\sigma)$ holds at $s$.

At this point we have mentioned the ways in which 
action models and state models differ.  What they
have in common is that they use
accessibility
relations to express each agent's uncertainty concerning
something.  For state models, the uncertainty has to do with 
which state is the real one; for action models,
it has to do with which action is taking place.
}

\rem{
\begin{ex}
{\rm
Here is an action model:
We take $\Sigma=\set{\sigma,\tau}$; $\sigma\arrowA\sigma$, $\sigma\arrowB\tau$,
$\tau\arrowA\tau$, $\tau\arrowB\tau$; $\Presemantic(\sigma) = \propositionp$,
and $\Presemantic(\tau)= \trueproposition$; recall that
$\trueproposition$ is the {\em ``always true''} proposition. 

This action model may be used in the modeling
 a {\em completely private announcement to $A$ of the proposition 
$\propositionp$}.
\label{ex-uio}
}\end{ex}
}

%\subsection{Program models}
\label{section-epistemic-program-models}
\label{section-epistemic-action-models}

To model non-deterministic actions and non-simple actions (whose appearances
to agents
are not uniform
on states), we define {\em epistemic program models}. In effect, this means
that we decompose complex actions (`programs') into ``simple'' ones: 
they correspond to sets of simple, deterministic actions from a given
action
model.

A  {\em program model\/}  is defined as
a pair $\pi = (\bSigma,\Gamma)$ consisting
of an action model
$\bSigma$ and a set $\Gamma\subseteq \Sigma$ of
\emph{designated simple actions}. Each of the simple actions $\gamma\in\Gamma$
may be thought of as a possible ``deterministic
resolution'' of the non-deterministic action $\pi$.
As announced above, the intuition about the map called $\Presemantic$ is that
an action is executable in a given state only if all its preconditions
hold at that state.
We often spell out an epistemic program model
 as $(\Sigma, \arrowAgents, \Presemantic,\Gamma)$
rather than $((\Sigma, \arrowAgents, \Presemantic),\Gamma)$.
Also, we usually drop the word ``epistemic'' and just refer
to these as {\em program models}.

\subsection{The update product}

\label{section-update-product-1}

Given a state model $\bS=(S,\arrowAgents_{\!\!\bS},\|\cdot\|_{\bS})$ 
and an action model
$\bSigma=(\Sigma,\arrowAgents,\Presemantic)$,
we define their {\em update product\/}
to be the  state model  
$$\bS\otimes\bSigma \quadeq (S\otimes\Sigma,\arrowAgents,\|.\|_{\bS\otimes\bSigma}),$$
given by
the following: the new states are pairs of old states $s$
and simple actions $\sigma$
which  are ``consistent'',  in the sense that all  preconditions
of the action $\sigma$ ``hold'' at the state $s$
\begin{equation}
S\otimes\Sigma \quadeq \set{(s,\sigma)\in S\times \Sigma:
s\in\Presemantic(\sigma)_{\bS}}.
\label{eq-consistent}
\end{equation}
The new accessibility relations are taken to be
the ``products'' of the corresponding accessibility relations
in the two frames; i.e., for $(s,\sigma),(s',\sigma')\in \bS\otimes\bSigma$ we put
\begin{equation}
(s,\sigma)\arrowA (s',\sigma') \quadiff s\arrowA s' \mbox{ and
}\sigma\arrowA\sigma',
\label{eq-product}
\end{equation}
and the new valuation map
$\map{{\|.\|}_{\bS\otimes\bSigma}}{\AtSen}{\pow(S\otimes\Sigma)}$
is essentially given by the old valuation:
\begin{equation}
\|p\|_{\bS\otimes\bSigma} \quadeq \set{(s,\sigma)\in
S\otimes\Sigma:s\in\|p\|_{\bS}}.
\label{eq-new-valuation}
\end{equation}

\rem{
\begin{remark}
The update operation is
defined whenever $\bS$ is a  state model and whenever 
  the codomain of $\Presemantic$ is the 
collection of all epistemic propositions. 
A different point concerns the interpretations
$\semantics{\propositionphi}{\bS\otimes\bSigma}$ in the  
product $\bS\otimes\bSigma$.  
By (\ref{eq-new-valuation}) and 
(\ref{eq-pedantic}) from Section~\ref{section-guide},
we see that 
 $$\propositionp_{\bS\otimes\bSigma} \quadeq \set{(s,\sigma)\in
S\otimes\Sigma:s\in\propositionp_{\bS}}.$$ 
We should stress
that this last fact holds only
for {\em atomic\/} propositions $\propositionp$.  It
will not hold in general for 
arbitrary epistemic propositions $\propositionphi$.
\end{remark} }

\rem{
The {\it update $\bS\otimes\alpha$ of a state model $\bS$ with an
update $\alpha$} is defined as just the model
$\bS\otimes\alpha=:\bS\otimes\bSigma(\alpha)$. In $\bS\otimes\alpha$,
the {\it update of a state $s\in\bS$ with the update
$\actionalpha$} is the set of all possible outputs
$(s,\gamma)\in\bS\otimes\alpha$, with $\sigma\in |\alpha|$; so we put:
$$s\otimes\alpha=:\{(s,\sigma)\in\bS\otimes\alpha:\sigma\in |\alpha|\}.$$
We write $s\arrowalpha s'$ iff $s'\in s\otimes\alpha$, and we call
this {\it the $\alpha$-transition relation}: it is the binary
accessibility relation\footnote{Observe that unlike in 
$PDL$, our transition relations are between states {\it living in
different epistemic
models}. This corresponds that to the fact that such actions affect
the global epistemic structure of the model: for each agent, the
set of ``epistemically possible'' states may change after such an
action.
Theoretically, it is possible to regain the one-model picture by
taking the sum of all the models obtainable by any iterated 
application of any update to the initial model, and consider this as one big
model. However, this is computationally unfeasible: the model
obtained in this way is huge (usually infinite).}
describing the input-output behavior of the
action $\alpha$. 
Finally, we define the {\it update $(\bS,s)\otimes\alpha$ of an epistemic state $(\bS,s)$
with an update $\actionalpha$}  by:
$$(\bS,s)\otimes\alpha=:(\bS\otimes\alpha,s\otimes\alpha).$$
}

\rem{
\paragraph{Intended interpretation} The update
product restricts the
full Cartesian product $S\times\Sigma$ to the smaller set
$S\otimes\Sigma$ in order to insure that  {\em states survive actions\/}
in the appropriate sense.

For each agent $A$, the product arrows $\arrowA$ on the output frame represent
agent  $A$'s epistemic uncertainty about the output state. The
intuition is that the components of our action models are ``simple
actions'', 
so the uncertainty regarding the action is
assumed to be independent of the uncertainty regarding the current
(input) state. This independence allows us to ``multiply'' these
two uncertainties in order to  compute  the
uncertainty regarding the output state: if whenever the input state is
$s$, agent $A$ thinks the input might be
some other state $s'$, and if
whenever the current action happening is $\sigma$, agent $A$ thinks
the current action might be some other action $\sigma'$, then
whenever the output state $(s,\sigma)$ is reached, agent $A$ thinks
the alternative output state $(s',\sigma')$ might have been reached.
Moreover, these are the only output states that $A$ considers possible.

As for the valuation, we essentially take the same valuation as
the one in the input model. If a state $s$ survives an action,
then the same facts $p$ hold at the output state
$(s,\sigma)$
as at the input state $s$. This means that {\it our actions, if
successful,
do not change the facts}. This condition can of course be relaxed
in various ways, to allow for fact-changing actions. But
in this paper we are primarily concerned with  purely epistemic 
actions, such as the earlier examples in this section.
}

\subsection{Updates induced by program models}
\label{section-epistemic-program-models-give-updates}

Recall that we defined updates
 in Section~\ref{section-epistemic-actions}.
And above, in Section~\ref{section-epistemic-program-models},
we defined epistemic  {\em program models}.  
Note that there is a big difference: 
the updates are pairs of  operations on the class of all state models,
and the program models are typically finite structures. 
 We think of {\it program models\/} as capturing specific mechanisms,
 or algorithms, for inducing updates. This connection is
made precise in the following definition.

\begin{definition}
Let $(\bSigma, \Gamma)$ 
be a  program model.  We define 
an update which we also denote $(\bSigma, \Gamma)$ as follows:
\begin{enumerate}
\item $\bS(\bSigma,\Gamma) = \bS \otimes \bSigma$.
\item $s\ (\bSigma,\Gamma)_{\bS}\ (t, \sigma)$ 
iff $s = t$ and $\sigma\in \Gamma$.
\end{enumerate}
We call this {\em the update induced by $(\bSigma, \Gamma)$.}
\end{definition}

Note that   updates
of the form $(\bSigma, \Gamma)$ 
have the property that for all state models $\bS$,
the inverse of the update relation
$(\bSigma, \Gamma)_{\bS}$ 
is a partial function.  That is, these updates are standard.

\subsection{Operations on program models}

\label{section-operations-on-program-models}

\paragraph{$\skippaction$ and  $\actionzero$} 
We define program models $\skippaction$ and $\actionzero$ as follows:
$\skippaction$ is a one-action set $\set{\sigma}$ with $\sigma\arrowA\sigma$
for all $A$, $\pre(\sigma) = \trueproposition$, and with distinguished
set $\set{\sigma}$.   The point here is that the update induced by this
program model is exactly the update $\skippaction$  from
Section~\ref{section-updates}. 
 We
purposely use the same notation.  Similarly, we let 
$\actionzero$ be the empty program model.  Then its induced 
update is what we called $\actionzero$ in 
Section~\ref{section-updates}.

\paragraph{Sequential Composition}
In all settings involving ``actions'' in some sense or other, sequential 
composition is a natural operation.   In our setting, we would like
to define a composition operation on program models, corresponding to
the sequential composition of updates. 
Here is the relevant definition.

Let $\bSigma = (\Sigma, \arrowAgents,\presemantic_\Sigma,\Gamma_{\bSigma})$ and
$\bDelta = (\Delta, \arrowAgents,\presemantic_\Delta,\Gamma_{\bDelta})$
 be  program models. We define the {\em composition\/} 
$$\bSigma\then \bDelta
\quadeq (\Sigma\times\Delta,
\arrowAgents,\presemantic_{\Sigma\then\Delta},\Gamma_{\Sigma\then\Delta})
$$ to be the
following   program model:
\begin{enumerate}
\item $\Sigma\times \Delta$ is the cartesian product
of the sets $\Sigma$ and $\Delta$.
\item $\arrowAgents$ in the 
composition $\bSigma\then \bDelta$
is the family
of product relations, in the natural way:
$$(\sigma, \delta) \arrowA (\sigma', \delta') \quadiff \sigma\arrowA \sigma'\mbox{ and  }
\delta\arrowA\delta'.$$

\item  $\presemantic_{\Sigma\then\Delta}(\sigma,\delta) =
\pair{(\bSigma,\sigma)}\presemantic_\Delta(\delta)$.
%(We usually write this as
%$\pair{(\bSigma,\sigma)}\presemantic_\Delta(\delta)$.)
\item $\Gamma_{\Sigma\then\Delta} = \Gamma_{\bSigma} \times \Gamma_{\bDelta}$.
\end{enumerate}
In the definition of $\presemantic$,  
$(\bSigma,\sigma)$ is an abbreviation for the induced update
$(\bSigma,\set{\sigma})$ as defined in
Section~\ref{section-epistemic-program-models-give-updates}.

\paragraph{Unions}
 If
 $\bSigma = (\Sigma, \arrowAgents,\presemantic_\Sigma,\Gamma_{\bSigma})$ and
$\bDelta = (\Delta, \arrowAgents,\presemantic_\Delta,\Gamma_{\bDelta})$,
we take $\bSigma \union \bDelta$ to be the disjoint union of the models,
with union of the  distinguished actions. The intended meaning is the {\em 
non-deterministic choice} between the programs represented by
$\bSigma$ and $\bDelta$. 
Here is the definition in more  detail, generalized to arbitrary
(possibly infinite) disjoint unions:
 let $\{\bSigma_i\}_{i\in I}$ be
a
family of program models, with
$\bSigma_i =(\Sigma_i, \arrowAgents, \presemantic_i, \Gamma_i)$; we
define their {\em (disjoint) union\/}
$$\Union_{i\in I} {\bSigma_i} 
\quadeq  \biggl(\Union_{i\in I}\Sigma_i, \arrowAgents,
\presemantic, \Gamma \biggr)$$
to be the model given by:
\begin{enumerate}
\item $\Union_{i\in I}\Sigma_i$ is $\bigcup_{i\in I} (\Sigma_i\times\set{i})$,
 the
disjoint union of the sets $\Sigma_i$.
\item $(\sigma,i)\arrowA (\tau, j)$ 
iff $i = j$ and $\sigma \arrowA_i\tau$.
\item $\presemantic(\sigma,i) = \presemantic_i(\sigma)$.
\item $\Gamma = \bigcup_{i\in I} (\Gamma_i\times\set{i})$.
\end{enumerate}  

\paragraph{Iteration}
Finally,   
we define an iteration
operation
by
$\bSigma^* = \Union\set{\bSigma^n: n\in N}$.
Here $\bSigma^0 = \skippaction$, and $\bSigma^{n+1} = \bSigma^n \then  \bSigma$.

\rem{
\paragraph{Summary}
We have seen enough examples of the update product
to make the case that complicated representations of 
epistemic situations, including ones which are not easy to construct
by hand, may be obtained from simpler representations
of initial situations using the update product with
program models.
The notion of a  program model is closely related to the
 familiar
   multi-agent
Kripke modes.  Each program model
induces an update in a natural way.  And the collection of
program models has a natural composition operation.
}

\rem{
Two immediate examples.  First, let $\bSigma$ be a
singleton $\set{\sigma}$ with $\sigma\arrowA\sigma$
for all $A$,   
 $\Presemantic(\sigma)= \trueproposition$,
and $\Gamma = \set{\sigma}$.  Then it is not hard
to check that the induced update of $(\Sigma,\Gamma)$
 is what we wrote as  $\skippaction$
in Section~\ref{section-update-composition}.
So it makes sense to write $\skippaction$
for the program model $(\Sigma,\Gamma)$.
Second, if $\Sigma = \emptyset = \Gamma$, then
the induced update is $\actionzero$; we therefore
write $(\bSigma,\Gamma)$ as $\actionzero$ in this case.
}

\medskip

Our definition of the operations
on program models
are   faithful to the
corresponding operations on updates from
Section~\ref{section-updates}.

\begin{proposition} [\cite{kra}]
The
 update induced by a
composition of program models is isomorphic
to the composition of the induced updates.
Similarly for sums and iteration, mutatis mutandis.
\label{proposition-composition-verified}
\end{proposition}

Since we shall not use this result, we omit the proof.

\rem{
\begin{proof}
Let $(\bSigma,\Gamma_{\bSigma})$ and $(\bDelta,\Gamma_{\bDelta})$ be
program models. We denote by $\actionalpha$ the update induced by
$(\bSigma,\Gamma_{\bSigma})$, by
$\actionbeta$ the update induced by $(\bDelta,\Gamma_{\bDelta})$, and by
$\actiongamma$ the update induced by $(\bSigma,\Gamma_{\bSigma}) \then
 (\bDelta,\Gamma_{\bDelta})$. We need to prove that
$\actionalpha\then\actionbeta=\actiongamma$. 
Let $\bS = (S,\arrowAgents_{\!\!\bS},\semantics{.}{\bS})$ be a state
model.  Recall that
 $$\bS(\actionalpha\then\actionbeta) \quadeq
\modelaction{\modelaction{\bS}{\actionalpha}}{\actionbeta} 
\quadeq 
 (\bS
\otimes(\bSigma,\Gamma_{\bSigma}))\otimes(\bDelta,\Gamma_{\bDelta}).$$
We claim that this is isomorphic to 
$\bS\otimes (\bSigma\then\bDelta,\Gamma_{\bSigma\then\bDelta})$,
and indeed the isomorphism is 
$(s,(\sigma,\delta)) \mapsto ((s,\sigma),\delta)$.
We   check that 
$(s,(\sigma,\delta)) \in \bS\otimes (\bSigma\then\bDelta) $
iff  $((s,\sigma),\delta) \in (\bS \otimes\bSigma)\otimes\bDelta$.
And the following are equivalent:
\begin{enumerate}
\item $(s,(\sigma,\delta)) \in \bS\otimes (\bSigma\then\bDelta) $.
\item $s\in \|\presemantic_{\Sigma\then\Delta}(\sigma,\delta)\|_{\bS}$.
\item $s\in \|\pair{(\bSigma,\sigma)}\presemantic_\Delta(\delta)\|_{\bS}$.
\item $(s,\sigma)\in \bS \otimes\bSigma$  and 
$(s,\sigma)\in \|\presemantic_{\bDelta} (\delta)\|_{\bS \otimes\bSigma}$.
\item  $((s,\sigma),\delta) \in (\bS \otimes\bSigma)\otimes\bDelta$.
\end{enumerate}
The rest of  verification 
of  isomorphism is fairly direct.

We also need to check that $\actiongamma_\bS$ and 
$(\actionalpha\then\actionbeta)_\bS$ are related by the isomorphism.
Now 
$$\actiongamma_{\bS}\quadeq
\set{(s, (s,(\sigma,\delta))) \in  S\otimes(\Sigma\then\Delta) : \sigma\in \Gamma_{\bSigma},
\delta\in \Gamma_{\bDelta}}.$$
Recall that  $(\actionalpha\then\actionbeta)_\bS =
\actionalpha_{\bS}
 \then {\actionbeta}_{\modelaction{\bS}{\actionalpha}}$
and that this is a  relational composition in left-to-right order.
And indeed,
$$\begin{array}{lcl}
 \actionalpha_{\bS} & \quadeq  & \set{(s,(s,\sigma)) : (s,\sigma)\in S\otimes\Sigma,
\sigma\in 
 \Gamma_{\bSigma}}\\
{\actionbeta}_{\modelaction{\bS}{\actionalpha}} & \quadeq & 
\set{((s,\sigma), ((s,\sigma),\delta)) \in \bS(\actionalpha)\otimes \Delta: 
\delta\in \Gamma_{\bDelta}}.\\
\end{array}$$
This completes the proof for composition.  We omit the proofs for 
sums and iteration.
\end{proof}
}

\subsection{Action signatures}
\label{section-action-signatures}

\begin{definition}
An {\em action signature\/} is a structure 
$$\bSigma \quadeq (\Sigma,\arrowAA,(\sigma_1,\sigma_2,\ldots,\sigma_n))$$
where  $\bSigma = (\Sigma,\arrowAA)$ is a finite
Kripke frame, and $\sigma_1,\sigma_2,\ldots,\sigma_n$ is an 
  enumeration of $\Sigma$
in a list without repetitions.  
We call the elements of $\Sigma$ {\em action types}.
When we deal with an action signature, our notation for the
action types usually includes a subscript (even though this is
occasionally redundant); thus the action types come with a
number that indicates their poisition in the fixed enumeration
of the action signature.
\end{definition}

An action signature $\bSigma$ together
with an assignment of  epistemic propositions to the action types
in $\Sigma$ gives us a full-fledged action model.
And this is the exact sense in which an \emph{action signature} is
an abstraction of the notion of \emph{action model}.
We shall  use  action signatures in constructing logical languages.

\rem{
\begin{examples}
{\rm

Here is a very simple
 action signature  which we call $\ActionSig_{\skipp}$.
$\Sigma$ is a singleton $\set{\skipp}$,
$\Presemantic(\skipp) = \trueproposition$,
%%%$\ar(\skipp) = 0$,
 and  $\skipp\arrowA \skipp$ for all agents $A$.
In a sense which we shall make clear later,
this is an action in which 
``nothing happens'', and moreover it
is common knowledge that this is the case.

The next  simplest action signature
 is the ``test'' signature $\bSigma_{?}$.
We take $\bSigma_{?}=\set{?,\skipp}$, with the enumeration $?,\skipp$.
 We also take  $?\arrowA\skipp$, and
$\skipp\arrowA\skipp$ for all $A$.  
This turns out to be a totally opaque form of test: $\varphi$
is tested on the real world, but nobody knows this is happening. 
Its function 
will be to generate tests $?\varphi$, which affect the states precisely in the way
dynamic logic tests do.

For each set $\BB\subseteq\Agents$ of agents, we define
 the action signature $\Pri^{\BB}$ of 
{\em completely private announcements to the group $\BB$}.  It has
 $\Sigma=\{\Pri^{\BB},\skipp\}$; 
%$\Presemantic(\Pri^{\BB})=1$, 
%and $\Presemantic(\skipp) = 0$;
%%%%$\ar$  is the constant function $1$, 
$\Pri^{\BB}\arrowB \Pri^{\BB}$ for all $B\in\BB$, $\Pri^{\BB}\arrowC \skipp$
for $C\not\in\BB$, and $\skipp\arrowA \skipp$ for all agents $A$.

Next, we consider the action signature $\Prss_k^{\BB}$ of  
 {\em private announcements to  the group $\BB$ with secure suspicion\/}
of $k$ possible announcements by the outsiders.
It has the following components:
$\Sigma=\set{1,2,\ldots,k\}\cup\{1',2',\ldots,k'\}\cup\{\skipp}$
(we assume these sets to be disjoint); 
%$\Presemantic(i)=\Presemantic(i')=i$ for all $i$ , while $\Presemantic(\skipp) = 0$;
%%$\ar$ is the constant function $k$;
$i\arrowB i'$ for all $B\in\BB$ and all $i\leq k$; $i\arrowC j$ for all $i,j\leq k$ and $C\not\in\BB$;
$i'\arrowB i'$ for $i\leq k$ and $B\in\BB$; $i'\arrowC \skipp$ for all $i\leq
k$ and $C\not\in\BB$; and finally $\skipp\arrowA \skipp$ for all agents $A$. 
%We define $\Prss_k^{\BB}$ as the type $1$ in
%this action signature. 

The action signature $\Cka_k^{\BB}$ is given by: $\Sigma=\{1,\ldots,k\}$;
%\Presemantic(i)=i$ for all $i\leq k$;
%%%$\ar$ is the constant function $k$;
 $i\arrowB i$ for $i\leq k$ and $B\in\BB$; and
finally $i\arrowC j$ for $i,j\leq k$ and $C\not\in\BB$. 
%As before, we define $\Cka_k^{\BB}$ as the type 
%$1$ in this action signature.
This action signature is called the signature of
{\em common knowledge of  alternatives  for an announcement to the group $\BB$}.
}
\label{ex-program-signatures}
\end{examples}
}

\paragraph{Signature-based program models}
\label{section-signature-based-epistemic-action-models}
\label{section-signature-based-epistemic-program-models}
\label{section-signature-based-program-models}

Let $\bSigma$ be an action signature,
let $n$ be the number of action types in $\Sigma$, 
let $\Gamma\subseteq \Sigma$, and let
$\vec{\propositionpsi} = \propositionpsi_1, 
\ldots, \propositionpsi_n$  be a list of
epistemic propositions.
We obtain a  program model 
$(\bSigma, \Gamma,\vec{\propositionpsi})$
in the following way:
\begin{enumerate}
\item The set of simple actions is $\Sigma$, and the 
accessibility relations are those given by the action signature.
\item For $j = 1, \ldots, n$, $\Presemantic(\sigma_j) = \propositionpsi_j$.
\item  The set of distinguished actions is $\Gamma$.
\end{enumerate}

In the special case that $\Gamma$ is the singleton set $\set{\sigma_i}$,
we write the resulting  signature-based program model as 
$(\bSigma,\sigma_i,\vec{\propositionpsi})$.

Finally, recall from Section~\ref{section-epistemic-program-models-give-updates}
 that every
signature-based program model induces an update.

To summarize:
{\em every action signature, set of distinguished action types in it, and tuple of 
epistemic propositions gives a program model in a canonical way.
Every program model induces a standard update}.

\rem{
\begin{proposition}
Let $\actiona$ be the update induced by $(\bSigma,\sigma_i,\vec{\psi})$.
Let $\bS$ be a state model.
For $s\in S$ and   $t'\in \bS(\actiona)$ the  following are equivalent:
\begin{enumerate}
\item There is some $t$ such that $s\ \actiona_{\bS}\ t$
and $t \arrowA t'$.
\item  There are $s'$ and $\sigma_j\in \Sigma$ such that
$s\arrowA s'$, $\alpha\arrowA \alpha'$, and $s'\ \actionb_{\bS}\
t' $, where $\actionb$ 
is the update induced by $(\bSigma,\sigma_j,\vec{\psi})$.
\end{enumerate} 
$$
\xymatrix{ s \ar[r]^{{\actiona}_{\bS}}  \ar[d]_{A} & t \ar[d]^{A} \\
s' \ar[r]_{\actionb_{\bS}} & t'   \\
}
$$
\label{proposition-square-semantic}
\end{proposition}

\begin{proof}
Both of the assertions above are equivalent to the following:
$s\in (\propositionpsi_i)_{\bS}$ and $(s,\sigma_i)\arrowA t'$.
\end{proof}

\begin{proposition}
Let $\actiona$ be the update induced by 
the signature-based program model
$(\bSigma,\sigma_i,\vec{\propositionpsi})$.
Then $\dom(\actiona) = \propositionpsi_i$. 
That is, for all state models $\bS$,
the  domain of the relation $\actiona_{\bS}$ is $(\propositionpsi_i)_{\bS}$.
\label{prop-new46}
\end{proposition}

\begin{proof}
Our definitions imply that
 $s\in \dom(\actiona_{\bS})$
iff $s\in (\Presemantic(\sigma_i))_{\bS} = (\propositionpsi_i)_{\bS}$.
\end{proof}
}

\subsection{Bisimulation-based notions of equivalence}
\label{section-bisimulation}
In this section, we discuss natural notions of equivalence
for some of the definitions which we have already seen.
We begin by recalling  the most important 
notion of equivalence for state models, {\em bisimulation}.

\begin{definition}
Let $\bS$ and $\bT$ be state models.  A {\em bisimulation\/} between
$\bS$ and $\bT$ is a relation $R\subseteq S\times T$ such that
whenever $s\ R\ t$, the 
following three properties hold:
\begin{enumerate}
\item $s\in \| p\|_{\bS}$ iff $t\in \|p\|_{\bT}$ for all atomic sentences  $p$.
%That is, $s$ and $t$ agree on all atomic information.
\item  For  $A\in \Agents$ and    $s'$ such that
$s\arrowA s'$, there is some   $t'$ such that $t\arrowA t'$ and  
$s'\ R\ t'$.
\item
  For  $A\in \Agents$ and    $t'$ such that
$t\arrowA t'$, there is some   $s'$ such that $s\arrowA s'$ and  
$s'\ R\ t'$.
\end{enumerate}
$R$ is a {\em total bisimulation\/} if it is a bisimulation 
and in addition: for  all $s\in S$ there is some $t\in T$ such 
that $s\ R\ t$; and vice-versa.
\end{definition}

\begin{proposition} If there is a bisimulation $R$
such that 
 $s\ R\ t$,  then $s$ and $t$ agree on all sentences $\phi$ in
%%% the language $\lang^\infty$ of 
{\em infinitary\/} modal logic:
 $s\in \semantics{\phi}{\bS}$ iff
$t\in\semantics{\phi}{\bT}$.
\label{proposition-bisim-infinitary}
\end{proposition}

\label{section-bisim-preserve}

Recall that  $\StateModels$   is the class of all state models.  
We have spoken of states as the elements of state models, but we
also use the term {\em states\/} to refer to the pairs $(\bS,s)$ with
$\bS\in \StateModels$ and
$s\in S$.  The class of all states is itself a state model, except
that its collection of states is a proper class rather than a set.
Still, it makes sense to talk about bisimulation relations on the class
of all states.  The largest such is given by 
$$(\bS,s) \equiv (\bT, t) \quadiff
\mbox{there is a bisimulation $R$ between $\bS$ and $\bT$ such that $s\ R\ t$}.
$$
This relation $\equiv$ is indeed an equivalence relation.   

When $\bS$ and $\bT$ are clear from the context, we write $s\equiv t$
instead of  $(\bS,s) \equiv (\bT, t)$.

\paragraph{Equivalence and preservation by bisimulation}
Since we are discussing notions of equivalence here, it makes sense
to also think about epistemic propositions.
Two propositions $\propositionphi$ and $\propositionpsi$
are equal if they are the same  operation on $\StateModels$.
That is, for all $\bS$, $\propositionphi_{\bS} = \propositionpsi_{\bS}$.
Later we shall introduce syntactically defined languages 
and also proof systems to go with them; in due course we
shall see other interesting notions of equivalence.  

Moving towards a  connection of propositions with bisimulation.
we say that a proposition $\propositionphi$ 
  is {\em  preserved by bisimulations\/} if whenever 
$(\bS,s) \equiv (\bT,t)$, then $s\in \propositionphi_{\bS}$
iff $t\in \propositionphi_{\bT}$.

\begin{proposition}  [\cite{kra}]
The  propositions which are preserved by bisimulation
include  $\trueproposition$ 
and the atomic propositions $\propositionp$,
and they are closed under all of the 
(infinitary) operations on propositions 
from Section~\ref{section-state-models}.
\label{prop-bisimulation-preservation}
\end{proposition}

\paragraph{Equivalence of updates}
We continue our discussion of equivalence
with the relevant notion for updates.

\begin{definition}
Updates $\actiona$ and $\actionb$ are \emph{equivalent} if for
all state models $\bS$ and $\bT$
and all total bisimulations $R$ between 
$\bS$ and $\bT$, $(\actiona_{\bS})^{-1};R;\actionb_{\bT}$ is a
total bisimulation between $\bS(\actiona)$ and $\bT(\actionb)$.
We write $\actiona\sim\actionb$ for this relation.
\end{definition}

\begin{proposition}
Let $\actiona\sim\actionb$  and  $\actionc\sim\actiond$.
Then $\actiona;\actionc\sim\actionb;\actiond$,
and also
$\actiona\union\actionc\sim\actionb\union\actiond$.
\label{prop-update-composition-equivalence}
\label{prop-update-union-equivalence}
\end{proposition}

\begin{proof}
Fix   $\bS$, $\bT$ and $R$.
Then 
$$\begin{array}{lcl}
(\actiona;\actionc)_{\bS}^{-1}; R;
(\actionb;\actiond)_{\bT}  & \quadeq & 
(\actiona_{\bS};\actionc_{\bS(\actiona)})^{-1}; R;
(\actionb_{\bT};\actiond_{\bT(\actionb)})
\\
& \quadeq & 
(\actionc_{\bS(\actiona)})^{-1} ;
(\actiona_{\bS})^{-1}; R ; \actionb_{\bT}; \actiond_{\bT(\actionb)}
\end{array}
$$
Since $(\actiona_{\bS})^{-1}; R ; \actionb_{\bT}$ is a 
total bisimulation between $\bS(\actiona)$ and $\bT(\actionb)$,
we have the first assertion.

For the second, note that
 $(\actiona\union\actionc)_{\bS}^{-1}; R;
(\actionb\union\actiond)_{\bT}$
is the disjoint union of the total bisimulations
$(\actiona_{\bS})^{-1};R;\actionb_{\bT}$
and 
$(\actionc_{\bS})^{-1};R;\actiond_{\bT}$.
\end{proof}

\begin{proposition} Let $\actiona$ and $\actionb$ be standard
updates.
If  $\propositionphi$ is preserved by bisimulations
and $\actiona\sim\actionb$, then
 $[\actiona]\propositionphi = [\actionb]\propositionphi$,
and also 
 $\pair{\actiona}\propositionphi = \pair{\actionb}\propositionphi$.
\label{proposition-preserve-sim}
\end{proposition}

\begin{proof}
We check the first assertion only.
Fix $\bS$, and let $s\in S$
be such that that $s\in ([\actiona]\propositionphi)_{\bS}$.
Let $t\in \bS(\actionb)$ be such that 
$s\ \actionb_{\bS}\ t$.  We must show that $t\in \propositionphi_{\bS(\actionb)}$.

The identity relation $I_{\bS}$ is a total 
bisimulation on $\bS$.   Let $R = (\actiona_{\bS})^{-1};I_{\bS};\actionb_{\bS}$.
 The definition
of action equivalence $\sim$ implies that 
$R$ is a 
total bisimulation between
$\bS(\actiona)$ and $\bS(\actionb)$.   Let $u\in\bS(\actiona)$ 
be such that $u\ R\ t$.  Thus there is some $s'\in S$ such that
$u\  (\actiona_{\bS})^{-1}\ s'\ \actionb_{\bS}\ t$.
Since $(\actiona_{\bS})^{-1}$ is a partial function,
we have $s' = s$.
Thus $s\ \actiona_{\bS}\ u$.  Since  $s\in ([\actiona]\propositionphi)_{\bS}$,
we have $u\in \propositionphi_{\bS(\actiona)}$.  And as
 $\propositionphi$ is preserved by the bisimulation $R$, we have the desired conclusion:
$t\in \propositionphi_{\bS(\actionb)}$. 
\end{proof}

\begin{proposition}
 Let $\actiona$ and $\actionb$ be standard.
If   $\actiona\sim\actionb$, then $\actiona$ and $\actionb$ have
the same domain.  That is, for all state models $\bS$,
$\dom(\actiona_{\bS}) = \dom(\actionb_{\bS})$.
 \label{proposition-preserve-sim-dom}
\end{proposition}

\begin{proof}
The domains are $(\pair{\actiona}\trueproposition)_{\bS}$ and 
$(\pair{\actionb}\trueproposition)_{\bS}$.
So the result follows from
Proposition~\ref{proposition-preserve-sim}
and the fact that $\trueproposition$ is preserved by bisimulations.
\end{proof}

\begin{definition}
An
update $\actionalpha$ {\em preserves bisimulations\/}
 if the following two conditions hold:
\begin{enumerate}
%\item Whenever $R$ is a bisimulation 
%%between $\bS$ and $\bT$, 
%then $R(\actionalpha)$ is a bisimulation between $\bS(\actionalpha)$ and
%$\bT(\actionalpha)$, where
%$$\mbox{$s'\ R(\actionalpha)\ t'$ iff there are $s\in S$, $t\in T$ so
%%that $s\ \actionalpha_{\bS}\ s'$,  $t\ \actionalpha_{\bS}\ t'$, %
%and $s\ R\ t$.}
%$$
\item  If    $s\ \actionalpha_{\bS}\ s'$ and $(\bS,s) \equiv (\bT,t)$,
then there is some $t'$ such that
 $t\ \actionalpha_{\bT}\ t'$ and $(\bS(\actionalpha),s') \equiv  (\bT(\actionalpha),t')$.
\item  If   $t\ \actionalpha_{\bT}\ t'$ and $(\bS,s) \equiv (\bT,t)$,
then there is some $s'$
such that  $s\ \actionalpha_{\bS}\ s'$
 and $(\bS(\actionalpha),s') \equiv (\bT(\actionalpha),t')$.
\end{enumerate}

An action model $\bSigma$
\emph{preserves bisimulations} if 
$\Presemantic(\sigma)$ is preserved under bisimulations for
all $\sigma\in \Sigma$.
\end{definition}

\begin{proposition} [\cite{kra}]
Concerning bisimulation preservation:
\begin{enumerate}
\item   The bisimulation preserving updates are closed under composition and 
(infinitary) unions.
\item  If $\propositionphi$ is preserved by
bisimulations and $\actionalpha$ preserves bisimulations,
then  $[\actionalpha]\propositionphi$ is 
preserved by
bisimulations.
\end{enumerate}
\label{prop-bisimulation-preservation-again}
\end{proposition}

\begin{proposition} [\cite{kra}] Let $\bSigma$ be a
bisimulation-preserving action model.  Let $\Gamma\subseteq \Sigma$
be arbitrary. 
Then the update induced by $(\bSigma, \Gamma)$
 preserves bisimulation.
\label{proposition-bisim-preserve-induced}
\end{proposition}

\begin{definition}
Let 
$\bSigma=(\Sigma, \arrowAgents, \Presemantic)$
and $\bDelta=(\Delta, \arrowAgents, \Presemantic)$
be action models.  As this notation indicates,
we shall not introduce additional notation to 
differentiate the arrows and the $\Presemantic$ functions
on these action models.
  A {\em bisimulation\/} between
$\bSigma$ and $\bDelta$ is a relation $R\subseteq \Sigma\times \Delta$ such that
whenever $\sigma\ R\ \delta$, the 
following three properties hold:
\begin{enumerate}
\item $\Presemantic(\sigma)=\Presemantic(\delta)$.
\item  For  $A\in \Agents$ and    $\sigma'$ such that
$\sigma\arrowA \sigma'$, there is some   $\delta'$ such that
 $\delta\arrowA \delta'$ and  
$\sigma'\ R\ \delta'$.
\item  For  $A\in \Agents$ and    $\delta'$ such that
$\delta\arrowA \delta'$, there is some   $\sigma'$ such that
 $\sigma\arrowA \sigma'$ and  
$\sigma'\ R\ \delta'$.
\end{enumerate}
We write $\sigma\sim\delta$ if there is a bisimulation $R$
between $\bSigma$ and $\bDelta$ such that $\sigma\ R\ \delta$.

Let $\pi = (\bSigma,\Gamma)$ and $\rho = (\bDelta, B)$ be 
program models.  We write $\pi\sim\rho$ if there is  a bisimulation
between $\bSigma$ and $\bDelta$ such that
the following hold:
\begin{enumerate}
\item Each
$\sigma\in \Gamma$ belongs to the domain of $R$,
and each $\delta\in B$ belongs to the image of $R$.
\item If $\sigma\in\Gamma$ and $\sigma\ R \ \delta$, then $\delta\in B$.
\item If $\delta\in B$ and $\sigma\ R \ \delta$, then  $\sigma\in\Gamma$.
\end{enumerate}
We write $\pi \sim \rho$ in this case.
\end{definition}

\rem{
\paragraph{Action models}
Action models $\bSigma=(\Sigma, \arrowAgents, \Presemantic)$
and $\bDelta=(\Delta, \arrowAgents, \Presemantic)$ are equivalent
if there is a relation $R\subseteq \Sigma\times \Delta$
which is a bisimulation between the Kripke models
$(\Sigma, \arrowAgents)$ and $(\Delta, \arrowAgents)$, and
with the additional property that if $\sigma\ R\ \delta$, then
$\Presemantic(\sigma) =\Presemantic(\delta)$.
We write $\bSigma\sim\bDelta$ in this case.
}

\begin{proposition} Equivalent  bisimulation-preserving 
program models
induce equivalent updates.
\label{proposition-equivalence-action-updates}
\end{proposition}

\begin{proof}
Suppose that $R$ is a bisimulation
showing that $(\bSigma,\Gamma)\sim(\bDelta,B)$.
To check that the induced updates are equivalent, 
let $\bS$ and $\bT$ be state models, and let $R'$
be a total bisimulation between them.
We check that $((\bSigma,\Gamma)_{\bS})^{-1}; R'; (\bDelta,B)_{\bT}$ is a
total bisimulation between $\bS(\bSigma,\Gamma)$ and $\bT(\bDelta,B)$.
To save on some notation, call this relation $Q$.

We verify the
  bisimulation properties of $Q$.
Let $(s,\sigma) \ Q\ (t,\delta)$.  Recall that the updates induced
by program models are standard.
Hence we have
 $$(s,\sigma)\ ((\bSigma,\Gamma)_{\bS})^{-1}\ s\ 
R' \ t \  (\bDelta,B)_{\bT} \ (t,\delta).$$
It is clear from this that $(s,\sigma)$ and $(t,\delta)$ satisfy the same
atomic sentences.
Suppose next that $(s',\sigma')\in   \bS(\bSigma,\Gamma)$ and
 $(s,\sigma) \arrowA (s',\sigma')$.
Note that since $(s',\sigma')\in   \bS(\bSigma,\Gamma)$,
 we also have  $\sigma'\in \Gamma$.
Let $t'$ be such that  $t\arrowA t'$ and $s' \ R'\ t'$.
Let $\delta'$ be such that $\delta\arrowA\delta'$ and 
$\sigma'\ R\ \delta'$.
We have $\delta'\in B$ by one of the conditions in the definition of equivalence.
Thus
 $$(s',\sigma')\ ((\bSigma,\Gamma)_{\bS})^{-1}\ s'\ 
R' \ t' \  (\bDelta,B)_{\bT} \ (t',\delta').$$
And $(t,\delta)\arrowA (t',\delta')$.
This shows the second bisimulation condition, and the third is similar.

For the totality of $Q$, let $(s,\sigma)\in S\otimes \Sigma$.
Let $t\in T$
 and $\delta\in B$ be such that $s\ R'\ t$ and $\sigma\ R\ \delta$.
Then   $(t,\delta)\in T\otimes \Delta$ since 
$\Presemantic(\sigma) = \Presemantic(\delta)$, and since
this proposition is preserved by bisimulations.
We   have  $(s,\sigma)\ Q \ (t,\delta)$ because
 $$(s,\sigma)\ ((\bSigma,\Gamma)_{\bS})^{-1}\ s\ 
R' \ t \  (\bDelta,B)_{\bT} \ (t,\delta).$$
This completes the verification that $Q$ is a 
total bisimulation.
\end{proof}

   %% section 2
\section{The languages $\lang(\bSigma)$}
\label{section-language}

At this point, we have enough general definitions
to present the syntax and semantics of 
the language  $\lang(\bSigma)$.  
Most of the remaining parts of this paper
study these languages.

\subsection{The syntax and semantics of $\lang(\bSigma)$}
\label{section-logics-based-on-program-signatures}

Fix a action signature $\bSigma$.
  We   present
in Figure~\ref{fig-lang} a logical language 
$\lang(\bSigma)$ which we study in the remainder of this paper.
Actually, we have three languages there, the full language
$\lang(\bSigma)$ and the smaller fragments $\lang_0(\bSigma)$ 
and $\lang_1(\bSigma)$.

The number $n$ which figures
into the syntax is 
the   number of action types in $\Sigma$.
In the programs of the form
$\sigma\psi_1\cdots\psi_{n}$ we have 
{\em sentences\/} $\vec{\psi}$ rather than 
{\em epistemic propositions\/} (which we had written
using boldface letters $\propositionpsi_1, 
\dots, \propositionpsi_n$ in
Section~\ref{section-epistemic-actions}).
Also, the signature $\bSigma$ figures into the semantics
exactly in those programs $\sigma\psi_1\cdots\psi_{n}$;
in those we require that $\sigma\in \Sigma$.

The second thing to note is that, as in $PDL$, we have two sorts
of syntactic objects: sentences and programs.
We call programs of the form  $\sigma\psi_1\cdots\psi_{n}$
{\em basic actions}.  Note that they might not be ``atomic''
in the sense that the sentences $\psi_j$ might themselves
contain programs.

\begin{figure}[tbp]
\fbox{
\begin{minipage}{6.0in}
$$
\begin{array}{lcc|c|c|c|c|c|c}
\mbox{\bf sentences $\phi$}  &  & \true & p_i  & \neg\varphi &  
\varphi\wedge\psi
 &\necc_A\varphi  & \necc_{\BB}^*\varphi
 & [\pi]\varphi \\  
\mbox{\bf programs $\pi$}  &      &    \skipp & \crash &  
\sigma_i
\psi_1\cdots\psi_{n}  
& \pi \union \rho  &
\pi \then \rho   & \pi^* 
\end{array}
$$

$$
\begin{array}{lcl}
\semanticsnm{\true} & \quadeq & \trueproposition \\
\semanticsnm{p}  & \quadeq   &  \propositionp
\\
\semanticsnm{\phi\andd\psi}  & \quadeq &  
\semanticsnm{\phi}\andd \semanticsnm{\psi}
\\
\semanticsnm{\nott\phi}  & \quadeq &  \nott\semanticsnm{\phi} \\
\semanticsnm{\necc_A\phi}  & \quadeq &
\necc_A \semanticsnm{\phi}  \\
 \semanticsnm{\necc^*_{\BB}\phi}  & \quadeq &
\necc^*_{\BB}\semanticsnm{\phi}   \\
\semanticsnm{[\pi]\phi}  & \quadeq &  [\semanticsnm{\pi}]\semanticsnm{\phi} 
\\
\end{array}
\qquad
\begin{array}{lcl}
\semanticsnm{\skipp} & \quadeq & \skippaction \\
\semanticsnm{\crash} & \quadeq & \crashaction \\
 \semanticsnm{\sigma_i\psi_1\ldots\psi_n} & \quadeq &
 (\bSigma,\sigma_i,\semanticsnm{\psi_1} \cdots \semanticsnm{\psi_n})
 \\
\semanticsnm{\pi\then\rho} & \quadeq &
\semanticsnm{\pi}\then\semanticsnm{\rho} \\
\semanticsnm{\pi \union \rho} & \quadeq &
\semanticsnm{\pi}\union\semanticsnm{\rho} \\
\semanticsnm{\pi^*}  &\quadeq &   \semanticsnm{\pi}^*  \\
\\
\end{array}
$$

\noindent For $\lang_1(\bSigma)$, we drop the $\pi^*$ construct.
For $\lang_0(\bSigma)$, we drop the $\pi^*$ and
 $\necc^*_{\BB}$ constructs.

%\noindent  We usually refer to the
%programs as {\em actions} and use letters like $\alpha$, $\beta$, etc.
%%%
\end{minipage}
}
\caption{The language $\lang(\bSigma)$ and its semantics,
and the fragments $\lang_0(\bSigma)$, and $\lang_1(\bSigma)$. \label{fig-lang}
\label{fig-semantics}}
\end{figure}

\rem{
\begin{figure}[t]
\fbox{
\begin{minipage}{6.0in}
$$
\begin{diagram}[tight,size=2.84em]  
\begin{array}{llllllll|
\lang_0 & \rInto & \lang_1 & & & & \\
\dInto &  & \dInto &  & & &  \\
\lang_0(\bSigma) & \rInto & \lang_1(\bSigma) & \rInto & \lang(\bSigma)
 & \rDashto & \lang^{\omega}_0  &  \rInto &\lang^{\infty}_0 \\
 & &  & \rdDashto &  & \ruDashto &    \\
&  &  & &  \mbox{PDL}   &  
 &   &   \\
\end{diagram}
%\end{array}
$$
\end{minipage}
}
\caption{Languages in this paper\label{fig-languages}}
\end{figure}
}}

We use the following standard abbreviations:
$\false =\neg\true$, $\phi\vee\psi = \neg(\neg\phi\wedge\neg\psi)$,
 $\phi\iif\psi =\nott(\phi\andd\nott\psi)$,
$\Diamond_A \phi = \neg\Box_A\neg \phi$, 
$\Diamond^*_{\BB} = \neg\Box^*_{\BB}\neg \phi$, and
$\pair{\pi}\phi = \neg [\pi]\neg\phi$.

\paragraph{The semantics} defines two operations by
simultaneous recursion on $\lang(\bSigma)$:
\begin{enumerate}
\item $\phi\mapsto \semanticsnm{\phi}$, taking the
sentences of $\lang(\bSigma)$ into epistemic propositions; and
\item  $\pi\mapsto  \semanticsnm{\pi}$,
taking the programs of $\lang(\bSigma)$ into program models
(and hence into induced updates).
\end{enumerate}
The formal definition is given in Figure~\ref{fig-semantics}.
The first map $\phi\mapsto \semanticsnm{\phi}$ might be called
the {\em truth map\/} for the language. 
When we began our study of state models, we 
started with a 
``valuation'' (or a ``truth'' map)
$\map{{\|.\|}_{\bS}}{\AtSen}{\pow(S)}$, assigning to
each   atomic sentence $p$ a set $\|p\|_{\bS}$ of states.
The truth map here extends this to sentences and actions
of  $\lang(\bSigma)$.   The overall definition is 
by simultaneous recursion on the    $\lang(\bSigma)$.
We employ the standard device of speaking of the
definition in terms of a temporal metaphor.  That is,
we think of the definition of the semantics of 
a sentence or action as coming ``after'' the 
definitions of its subsentences and subactions.

With one key exception,
the operations on   the right-hand sides
are immediate
applications of our general definitions of the closure
conditions on epistemic propositions
from Section~\ref{section-state-models}
 and the operations
on program models from 
Section~\ref{section-operations-on-program-models}.
A good example to explain this is the clause for
the semantics of sentences $[\pi]\phi$.
Assuming that we have a  program model $\semanticsnm{\pi}$,
we get an induced update in Section~\ref{section-epistemic-program-models-give-updates}
which we again denote $\semanticsnm{\pi}$.
We also have  an epistemic  proposition $\semanticsnm{\phi}$.
We can therefore form the epistemic proposition
$[\semanticsnm{\pi}]\semanticsnm{\phi}$
(see equation (\ref{eq-alpha-phi}) in Section~\ref{section-epistemic-actions}).
Note that we have overloaded the square bracket notation;
this is intentional, and we have done the same with other
notation as well.

Similarly,
the semantics of $\skipp$ and $\crash$ are the program models
$\skippaction$ and $\actionzero$ of Section~\ref{section-operations-on-program-models}.

We also discuss the  definition of the semantics
for basic actions $\sigma_i\vec{\psi}$.
For this, recall that 
we have a general definition of
a signature-based program model
$(\bSigma,\Gamma, \propositionpsi_1,\ldots,\propositionpsi_n)$,
where $\Gamma\subseteq\Sigma$ and the $\propositionpsi$'s
are any epistemic propositions.  What we have in the
semantics of $\sigma_i\vec{\psi}$ is the special case of
this where $\Gamma$ is the singleton $\set{\sigma_i}$ and 
$\propositionpsi_i$ is $\semanticsnm{\psi_i}$,  a proposition
which
we already have defined when we come to define 
$\semanticsnm{\sigma_i\vec{\psi}}$.

\rem{
\begin{figure}[tbp]
\fbox{
\begin{minipage}{6.0in}
\label{truth}
$$
\begin{array}{lcl}
\semanticsnm{p}  & \quadeq   &  \propositionp
\\
\semanticsnm{\phi\andd\psi}  & \quadeq &  
\semanticsnm{\phi}\andd \semanticsnm{\psi}
\\
\semanticsnm{\nott\phi}  & \quadeq &  \nott\semanticsnm{\phi} \\
\semanticsnm{\necc_A\phi}  & \quadeq &
\necc_A \semanticsnm{\phi}  \\
 \semanticsnm{\necc^*_{\BB}\phi}  & \quadeq &
\necc^*_{\BB}\semanticsnm{\phi}   \\
\semanticsnm{[\pi]\phi}  & \quadeq &  [\semanticsnm{\pi}]\semanticsnm{\phi} 
\\
\end{array}
\qquad
\begin{array}{lcl}
\semanticsnm{\skipp} & \quadeq & \skippaction \\
\semanticsnm{\crash} & \quadeq & \crashaction \\
 \semanticsnm{\sigma_i\psi_1\ldots\psi_n} & \quadeq &
 (\bSigma,\sigma_i)\semanticsnm{\psi_1},\ldots,\semanticsnm{\psi_n})
 \\
\semanticsnm{\pi\cdot\rho} & \quadeq &
\semanticsnm{\pi}\cdot\semanticsnm{\rho} \\
\semanticsnm{\pi\union\rho} & \quadeq &
\semanticsnm{\pi}\union\semanticsnm{\rho} \\
\semanticsnm{\pi^*}  &\quadeq &   \semanticsnm{\pi}^*  \\
\end{array}
$$
\end{minipage}
}
\caption{The semantics of $\lang(\bSigma)$. \label{fig-semantics}}
\end{figure}
}

\rem{
At this point, it is probably good to go back to our earlier discussions
in Section~\ref{section-big-picture}
of epistemic propositions and updates.  What we have done overall is
to give a fully syntactic presentation of languages of these
epistemic propositions and updates.   The constructions of the language
correspond to the closure properties noted in 
Section~\ref{section-big-picture}.  (To be sure, we have restricted to
the finite case at several points because we are interested in a syntax,
and at the same time we have re-introduced some infinitary aspects
via the Kleene star.)
}

\subsection{Examples}
\label{section-Cn}

In this section, we provide examples of the concepts from 
Sections~\ref{section-state-models}--\ref{section-logics-based-on-program-signatures}.
These examples are chosen with an eye towards one of the inexpressivity results
in Sections~\ref{section-Cn-results} and~\ref{section-privatepower}
near the end of the paper.  
So we know that they may appear
artificial.  Still, we trust that having a suitably detailed example may
help the reader.

In these examples, our set $\AtSen$ of 
atomic sentences is a two-element set $\set{p,q}$,
and the set  $\Agents$ of agents
is the two-element  set  $\set{A,B}$.

First, here is a family of state models $C_n$ for even positive
numbers $n$.
  $C_n$   is a cycle of
$5n$ points 	$a_1, \ldots, a_{5n}$ arranged as follows:
$$\xymatrix{
a_1 \ar@{<->}[r]^{A} & a_2 \ar@{<->}[r]^{B} & a_3
& \cdots & a_{5n-1} \ar@{<->}[r]^{A} & a_{5n} \ar@{<->}[r]^{B} & a_1
}
$$
Since $n$ is even, for $1\leq i\leq 5$, the connection is $a_{in-1}\arrowA a_{in}\arrowB
a_{in+1}$.  (We are taking subscripts modulo $5n$ here.)
 We also specify that 
$p$ is true at all points except $a_1$ and $a_{2n+1}$,
and $q$ is true {\em only\/} at $a_{4n+1}$.

Let $\bSigmapub$ be the   action signature  given by
  $\Sigma = \set{\Pub}$, 
   $\Pub\arrowA \Pub$, and $\Pub\arrowB\Pub$.
This action signature is one of the most important 
in the applications since it is used in the representation of
{\em public announcements}.

Then we have $\lang(\bSigmaPub)$ and  its sublanguages
$\lang_0(\bSigmaPub)$ and $\lang_1(\bSigmaPub)$.
One example of a sentence of  $\lang_1(\bSigmaPub)$ is
 $\pair{\Pub\ p}\poss_{A,B}^* q$. (Without abbreviations,
this would be 
 $\nott[\Pub\ p]\necc_{A,B}^* \nott q$.)
The point of this first example is to 
calculate $\semantics{\pair{\Pub\ p}\poss_{A,B}^* q}{C_n}$.
This takes a few steps.

The definitions tell us that $\semanticsnm{p} =
\propositionp$, and so
$$\semantics{p}{C_n} \quadeq \propositionp_{C_n} \quadeq \|p\|_{C_n}
\quadeq \set{a_i : i \neq 1, 2n+1}.$$

 We have a signature-based program model
$(\bSigmaPub, \Pub,\propositionp)$.
(In more detail,  
$\bSigmaPub$ is an action signature, 
our distinguished set $\Gamma\subseteq \Sigma$ is
$\set{\Pub}$  and $\Pre(\Pub)$ is the
 the proposition $\propositionp$.)
We can take the update product of this with $C_n$
to get the model $C_n\otimes (\bSigmaPub, \Pub,\propositionp)$.
The universe of this model is
$$ \set{(a_i,\Pub) : 1\leq i \leq 5n \ \& \ i\neq 1 \ \& \ i\neq 2n+1}.
$$ 
The accessibilities are given by
$$\xymatrix{
(a_2, \Pub) \ar@{<->}[r]^-{B} & (a_3, \Pub) \ar@{<->}[r]^-{A}
& \cdots & (a_{2n-1}, \Pub) \ar@{<->}[r]^-{A} &
 (a_{2n}, \Pub)\\ (a_{2n+1}, \Pub) \ar@{<->}[r]^-{A} &  (a_{2n+2}, \Pub) 
\ar@{<->}[r]^-{B}  & \cdots & (a_{5n-1}, \Pub) \ar@{<->}[r]^-{A} &
(a_{5n},
\Pub) }
$$
In  this model $p$ is true at all points, and $q$ is true only
at $(a_{4n+1},\Pub)$.
The update relation $(\bSigmaPub, \Pub,\propositionp)_{C_n}$ is
$$\set{(a_i,(a_i,\Pub)) : 1\leq i \leq 5n \ \& \ i\neq 1 \ \& \ i\neq 2n+1}.
$$
The intuition is that
when we relativize the model to $p$
(that is, we update the model with a public announcement
of $p$), $a_1$ and $a_{2n+1}$ 
disappear.  The cycle breaks into two disconnected
components.

To save on notation, let us write $D_n $ for the model 
$C_n\otimes (\bSigmaPub, \Pub,\propositionp)$.
It follows that $$\semantics{q}{D_n} \quadeq 
\|\propositionq\|_{D_n} \quadeq \set{(a_{4n+1},\Pub)},$$
and therefore that
 $$\semantics{\poss^*_{A,B} q}{D_n} \quadeq (\poss^*_{A,B}\propositionq)_{D_n}
\quadeq \set{(a_i,\Pub) : 2n + 1 \leq i \leq 5n}.$$
It now follows that
$$\begin{array}{lcl}
\semantics{\pair{\Pub\ p}\poss_{A,B}^* q}{C_n} & \quadeq &     
\set{a_i : (a_i,\Pub) \in D_n  \ \&\ (a_i,\Pub)\in \semantics{\poss^*_{A,B} q}{D_n}} \\
& \quadeq &  
\set{a_i :  2n +2 \leq i \leq 5n}
\end{array}
$$
Later, we are going to be especially interested in the points
$a_{n+1}$ and $a_{3n+1}$.  The analysis above shows that
$a_{n+1}$ does not belong to 
$\semantics{\pair{\Pub\ p}\poss_{A,B}^* q}{C_n}$, but
$a_{3n+1}$ does belong to it.
The intuition is that 
after we relativize  the original cycle to $p$
by deleting $a_1$ and $a_{2n+1}$,
there is no path from $a_{n+1}$ to $a_{4n+1}$.  But even after
we make the deletion, there is a path from $a_{3n+1}$ to $a_{4n+1}$,
and this is why $a_{3n+1}$ satisfies
the sentence $\pair{\Pub\ p}\poss_{A,B}^* q$.

\paragraph{A second example}
We next consider a different example based on another action signature.
This is the action signature
$\Pri^{A}$ of 
{\em completely private announcements to $A$}.  It
concerns two agents, $A$ and $B$, and it has
 $\Sigma=\{\Pri^{A},\skipp\}$.
The fixed enumeration here is that $\Pri^{A}$ comes first
and $\skipp$ comes second.
We also have  
$\Pri^{A}\arrowA \Pri^{A}$, $\Pri^{A}\arrowB \skipp$,
$\skipp\arrowA \skipp$, and $\skipp\arrowB\skipp$.

 $\lang_1(\bSigmaPri)$  contains
the sentence $\pair{\Pri^A\ p\,\true} \poss_A^*\poss_B \nott p$.
We show in Section~\ref{section-privatepower}
that this sentence is not expressible in 
$\lang_1(\bSigmaPub)$, even by a set of sentences.
We take this to be a formal 
confirmation that a logical system
with 
private announcements is more powerful than a
system with only public announcements.

Let $N^+$ be the set $\set{1,2, \ldots}$ of non-zero
natural numbers.
   We shall construct
models
$\bS_f$, where $J\subseteq N^+$ and $\map{f}{J}{N^+}$.  
(Note that we leave the domain $J$ out of the notation $\bS_{f}$.)
We also construct
models
$\bT_{f,j}$ where $J\subseteq N^+$,
 $\map{f}{J}{N^+}$,
and $j\in J$.  The models 
$\bS_f$ and $\bT_{f,j}$ have the same state set:
$$\set{a,b}\cup
\set{c^{i}_k : i\in J \mbox{ and } 1 \leq k \leq f(i)}.
$$
The atomic proposition $p$ is true at all points except $b$.
The arrows in $\bS_f$ are given by
\begin{enumerate}
\item $x\arrowA x$ for all $x$.
\item  $a\arrowA  b$.
\item $b\arrowA c^i_1$ for all $i\in J$.
%\item $c^i_k \arrowA a$ for all $i\in J$ and $1\leq i \leq f(i)$.
\item $c^i_k \arrowA b$ for all $i\in J$ and $1\leq i \leq f(i)$.
\item $c^i_{f(i)}\arrowB b$ for all $i\in J$. 
\end{enumerate}
$\bT_{f,j}$ has all these arrows and exactly one more:
$$  a\arrowA c^{j}_1.
$$
%In stating the lemma below, we let $D$ be the one-point model
%$\set{d}$
The model $\bS_{f}$ are shown in Figure~\ref{fig-model}.
Note that we only show the case when $J = \set{j,k}$; the general case
is of course similar.  We did not show the reflexive
$\arrowA$ relations.
And we remind the reader that $\bT_{f,j}$ has 
one more arrow.

\begin{figure}[tbp]
\fbox{
\begin{minipage}{6.0in}
$$
\xymatrix{
a \ar[r]^-{A} & b  \ar[rd]
\ar[r]^-{A} & c^{j}_1 \ar[r]^-{A}
   & c^{j}_2 \ar[r]^-{A} & \cdots    
\ar[r]^-{A} & c^{j}_i  \ar[r]^-{A} 
& \cdots &
\ar[r]^-{A} &
\ar@/_4pc/[lllllll]_{B} 
c^{j}_{f(j)}    \\
& & c^{k}_1 \ar[r]^{A} &  & \cdots & 
\ar[r]^-{A} &
\ar@/^1pc/[lllllu]_{B} 
c^{k}_{f(k)}    \\
}  
$$
\end{minipage}
}
\caption{The model $\bS_f$ for $J$ a 
two-element set $\set{j,k}$,
omitting the reflexive 
arrows. \label{fig-model}}
\end{figure}
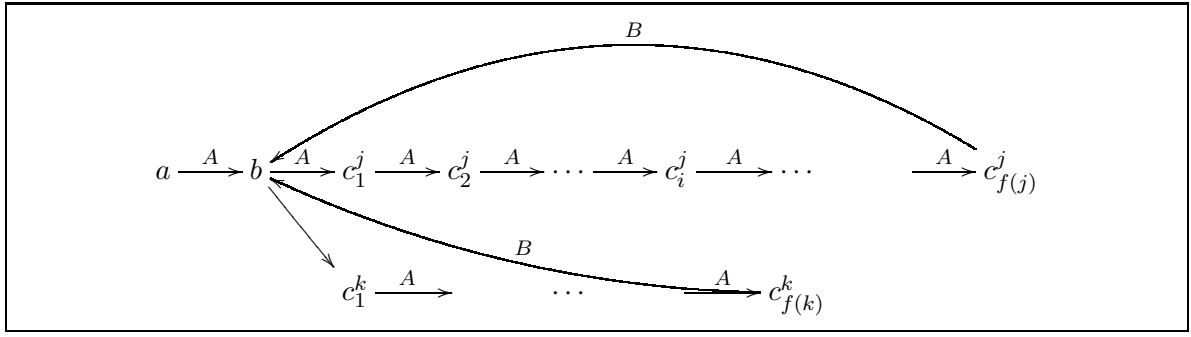

We shall need to calculate
 $\bS_f\otimes (\bSigmaPri,\Pri^A,\propositionp,\trueproposition)$.
We need the same thing for the models $\bT_{f,j}$.
To save on notation, we call the updated models
$\hat{\bS}_f$ and $\hat{\bT}_{f,j}$.
Finally, we need the update relations between $\bS_f$ and $\hat{\bS}_f$
and between $\bT_{f,j}$ and $\hat{\bT}_{f,j}$.

Here are some facts about the structure of  
$\hat{\bS}_f$:
\begin{enumerate}
\item $(a, \Pri^A)\arrowA  (a, \Pri^A)$.
\item  $(a, \Pri^A)\arrowA  (b, \Pri^A)$ does not hold in 
 $\hat{\bS}_f$.
\item  The only $\arrowA$-successor of $(a, \Pri^A)$
is itself. 
\item $(x, \skipp)\arrowA (y, \skipp)$
whenever $x\arrowA x$ in $\bS_f$; similarly for
$(x, \skipp)\arrowB (y, \skipp)$.
\item
$c^i_k\arrowA c^i_{k+1}$ for all $i\in J$ and $1\leq k<f(i)$.
\end{enumerate}

The structure of  $\hat{\bT}_{f,j}$ differs from $\hat{\bS}_f$.
Since $a\arrowA c^j_1$ in $\bT_{f,j}$,  
we   have a path in   $\hat{\bT}_{f,j}$:
$$(a, \Pri^A)\ \arrowA\  (c^j_1, \Pri^A)\ \arrowA \\
\ \cdots\ \arrowA\ (c^j_{f(j)}, \Pri^A)\ \arrowB\ (b, \skipp).
$$
Once again, $\bS_{f}$ contains no path from 
$(a,\Pri^A)$ to any other point, in particular, no
path  to $(b,\skipp)$.
 %%%changed line above to add comma on 7/4/04 %%%

The update relation between $\bS_f$ and $\hat{\bS}_f$ is 
 $\set{(x,(x,\Pri^A)) : x\neq b}.$ 
And from this, we see that 
$$\begin{array}{lcl}
\semantics{\pair{\Pri^A\ p,\true} \poss_A^*\poss_B \nott p}{\bS_{f}} 
& \quadeq & \set{x : (x,\Pri^A)\in \semantics{\poss_A^*\poss_B \nott
p}{\hat{\bS}_f}}\\
& \quadeq & \set{c^{i}_k : i\in J \mbox{ and } 1 \leq k \leq f(i)}
\end{array}
$$
In particular, $a\notin
\semantics{\pair{\Pri^A\ p,\true} \poss_A^*\poss_B \nott p}{\bS_{f}}$.

The update relation between $\bT_{f,j}$ and $\hat{\bT}_{f,j}$
is the same, and this time
$$\begin{array}{lcl}
\semantics{\pair{\Pri^A\ p,\true} \poss_A^*\poss_B \nott p}{\bT_{f,j}} 
& \quadeq & \set{x : (x,\Pri^A)\in \semantics{\poss_A^*\poss_B \nott
p}{\hat{\bT}_{f,j}}}\\
& \quadeq & \set{x : x \neq b}\\
\end{array}
$$
And this time, 
 $a\in
\semantics{\pair{\Pri^A\ p,\true} \poss_A^*\poss_B \nott p}{\bS_{f}}$.

In Section~\ref{section-privatepower}, we shall use these models to prove
that our sentence 
$\pair{\Pri^A\ p,\true} \poss_A^*\poss_B \nott p$ is not expressible by
any sentence or even any set of sentences using public announcements only.
The proof does not use any of our
other results, and the reader may turn to 
it at this point.

\rem{

The proof uses some models $G_n$ for $n\geq 1$ described
as follows:
  We begin with
a cycle in $\arrowA$:
\begin{equation}
b\ 
\arrowA \ a_n\ \arrowA\ a_{n-1}\ \arrowA\ \cdots\ \arrowA\ a_2 \ \arrowA\  a_1\
\arrowA \ a_{\infty}\ \arrowA\ b\
\label{eq:cycle}
\end{equation}
We add edges   $a_i\arrowA b$ for
all $i$  (including $i =1$ and $i=\infty$), and also $a_i\arrowA a_{\infty}$
for all $x$ (again including $i =1$ and $i=\infty$).  
The only $\arrowB$ edge is 
 $a_1\arrowB b$.  (See the figure for this frame.)
The atomic sentence $p$ is true at all points except $b$.

We shall need to calculate $\bS_f\otimes (\bSigmaPri,\Pri^A,\propositionp,\trueproposition)$;
we also shall need the same thing for the models $\bT_{f,j}$.

 $$
\begin{array}{l}
 (a_n,\Pri^A)\ \arrowA\ (a_{n-1},\Pri^A)\ \arrowA\ \cdots\
 \arrowA\ (a_2,\Pri^A) \ \arrowA\ 
(a_1,\Pri^A)\
\arrowA \ (a_{\infty},\Pri^A)    \\
 (a_n,\skipp)\ \arrowA\ (a_{n-1},\skipp)\ \arrowA\ \cdots\
 \arrowA\ (a_2,\skipp) \ \arrowA\ 
(a_1,\skipp)\
\arrowA \ (a_{\infty},\skipp)  \\
(b,\skipp)\ 
\arrowA \   (a_n,\skipp) \\
 (a_{\infty},\skipp) \ \arrowA \   (b,\skipp)  \\
(a_1,\Pri^A)\
\arrowB \  (b,\skipp)\ 
\end{array}
$$
The only point where $p$ is false is $(b,\skipp)$.
Now  $$\semantics{\poss_A^*\poss_B \nott p}{H_n} \quadeq   \set{ (a_n,\Pri^A),
\ldots, (a_1,\Pri^A)}\cup \set{(x,\skipp) : x\in G_n}.$$
 The important point is that there is no path from
$(a_{\infty},\Pri^A)$ to $(b,\skipp)$.
 
The update relation between $G_n$ and $H_n$ is 
$$\set{(a_i,(a_i,\Pri^A)) : 1\leq i \leq \infty}.$$
And from this, we see that 
$$\begin{array}{lcl}
\semantics{\pair{\Pri^A\ p,\true} \poss_A^*\poss_B \nott p}{G_n} 
& \quadeq & \set{a_i : 1\leq i \leq \infty \ \& \ a_i\in \semantics{\poss_A^*\poss_B \nott
p}{H_n}}\\
& \quadeq & \set{a_i : 1 \leq i \leq n}
\end{array}
$$
}

\rem{
The first thing to note is that 
in $G_n(\Pri^A\ p)$, 
$a_i\models\poss_A^*\poss_B \nott p$ for all $i < \infty$.
The relevant path is $a_i \arrowA\cdots\arrowA a_1  \arrowB b$;
the important point is that since the announcement was private, the edge
$a_1 \arrowB b$ survives the update.
On the other  hand, in this same model $G_n(\Pri^A\ p)$,
 $a_{\infty}\models\nott\poss_A^*\poss_B \nott p$.
This is  because the only path in the original model from $a_{\infty}$ to $b$ 
is  $$a_{\infty}\ \arrowA\ b\ \arrowA\ a_n\ 
\arrowA\ \cdots\ \arrowA\ a_1 \ \arrowB\ b,$$ and 
the edge $a_{\infty}\arrowA b$ is lost in the update.
}

\rem{
\begin{figure}[t]
\fbox{
\begin{minipage}{6.0in}
$$
\xymatrix{
b \ar[r]^-{A} & a_n \ar[r]^-{A} & 
% a_{n-1} \ar[r]^-{A}
 & \cdots & \ar[r]^-{A}& a_{i} \ar[r]^-{A} \ar@/_2pc/[lllll]_{A}
\ar@/^2pc/[rrrr]^{A}
 &  \cdots & %a_{i-1} \cdots
% &  a_2  \ar[r]^{A} 
 & a_1  \ar[r]^{A} \ar@/^2pc/[llllllll]^{B} 
 & a_{\infty} % \ar@/^2pc/[lllllllll]^{A} % \ar@(dr,ur)[]_{A}
}
$$
\end{minipage}
}
%\caption{$G_n$.  The picture omits the arrow $a_1\arrowB b$ j.
 \label{figure-induced}
%}
\end{figure}
}

\rem{
\subsection{Fragments}
\label{section-fragments}

In applications, it is often convenient to work with
special subsets of our languages $\lang(\bSigma)$ that
we call {\em fragments}.

\begin{definition}
A {\em fragment\/} of $\lang(\bSigma)$ is 
a set $F$ of sentences and programs such that the following hold:
\begin{enumerate}
\item $F$ is closed under subsentences and subprograms.
\item $F$ is closed under the sentence forming operations of
$\nott$, $\andd$, $\necc_A$, $\necc^*_{\BB}$, and $[\pi]\phi$.
\item $F$ is closed under the program forming operations of $\skipp$, union
$\union$, composition $\then$, and iteration  $\pi^*$.
\item If $\sigma \psi_1,\ldots, \psi_n$ belongs to $F$ and $\sigma\arrowA\tau$ for
some $A$, 
then also $\tau \psi_1,\ldots, \psi_n$ belongs to $F$.
\end{enumerate} 
\end{definition}

For example, consider the language of private announcements to the agent $A$.
The fragment of interest is the one where in basic actions $\sigma\ \psi_1,\psi_2$,
we have the following:
 if $\sigma$ is $\Pri^A$, then   $\psi_2$ is $\true$.
When dealing with this fragment, it is sensible 
abbreviate, for example, 
 $[\Pri^A\  \phi, \true]\psi
 $
as $[\Pri^A\ \phi]\psi$.  
We also would abbreviate any action $\skipp^X \vec{\psi}$ to just $\skipp$.
\footnote{{\bf from larry}: I need to check this point, and also the
related discussion of fragments after the completeness theorem.}
}

\rem{
\subsection{A guide to the concepts in this paper}
\label{section-guide}

In this section, we catalog the main notions that we shall introduce
in due course. After this, we turn to a discussion of the
particular languages that we study, and we situate them with
existing logical languages.

Recall the we insist on a distinction of syntactic and semantic
objects in this paper.
We list in Figure~\ref{fig-notions} the main notions
that we will need.   We do this mostly to help
 readers as they explore the different notions.  We mention
now that the various notions are not developed in the order
listed; indeed, we have tried hard to organize this paper
in a way that will be easiest to read and understand.
For example, one of our main goals is to present a set
of languages (syntactic objects) and their semantics
(utilizing semantic objects). 
}

\rem{
\begin{figure}[t]
\fbox{
\begin{minipage}{6.0in}
$$
\begin{diagram}[tight,size=2.84em]  
\lang_0 & \rInto & \lang_1 & & & & \\
\dInto &  & \dInto &  & & &  \\
\lang_0(\bSigma) & \rInto & \lang_1(\bSigma) & \rInto & \lang(\bSigma)
 & \rDashto & \lang^{\omega}_0  &  \rInto &\lang^{\infty}_0 \\
 & &  & \rdDashto &  & \ruDashto &    \\
&  &  & &  \mbox{PDL}   &  
 &   &   \\
\end{diagram}
$$
\end{minipage}
}
\caption{Languages in this paper\label{fig-languages}}
\end{figure}
}

\rem{
\paragraph{Languages}
This paper studies a number of languages, and to help the reader we
list these in Figure~\ref{fig-languages}.  Actually, what we study
are not individual languages, but rather {\em families\/} of
languages parameterized by different choices of primitives. 
It is standard in modal logic to begin with a set of atomic
propositions, and we do the same.  The difference is that we
shall call these  atomic {\em sentences\/} in order to make
a distinction between these essentially syntactic objects
and the {\em semantic\/} propositions that we study
beginning in Section~\ref{section-epistemic-actions}. This
is our first parameter, 
 a set $\AtSen$ of 
{\em atomic sentences}.  The second is a set $\Agents$ of {\em agents}.

Given these, $\lang_0$ is ordinary modal logic with
the elements of $\AtSen$ as atomic sentences and
with  agent-knowledge (or belief)
modalities $\necc_A$ for $A\in \Agents$.

We add {\em common-knowledge operators\/} $\necc_{\BB}^*$, for sets of 
agents $\BB\subseteq \Agents$,
to get a larger language $\lang_1$.
In Figure~\ref{fig-languages}, we note the fact that $\lang_0$ is 
literally a subset of $\lang_1$ by using the inclusion arrow.
The syntax and semantics of  $\lang_0$ and 
 $\lang_1$ are presented in
 Figure~\ref{fig-modal-lgs}.

Another close neighbor of the system in this paper is 
Propositional Dynamic Logic (PDL).  
PDL was first formulated by 
Fischer and Ladner~\cite{FischerLadner77}, following the introduction
of dynamic logic in
Pratt~\cite{pratt76}.
 The syntax and the main
clauses in the semantics of PDL are shown in Figure~\ref{fig-modal-pdl}.

We may also take $\lang_0$ and close under countable conjunctions
(and hence also disjunctions).  We call this language $\lang_0^{\omega}$.
Note that $\lang_1$ is not literally a subset of $\lang_0^{\omega}$, but
there is a translation of $\lang_1$ into $\lang_0^{\omega}$ that preserves
the semantics.  We would indicate this in a chart with a dashed arrow.

PDL is propositional dynamic logic, formulated with atomic
programs $a$ replacing the agents $A$ that we have fixed.
We might note that we can translate $\lang_1$ into PDL.
The main clauses in the translation are:
$$\begin{array}{lcl}
(\necc_A\phi)^t & \quadeq & [a]\phi^t \\
(\necc^*_{\BB}\phi)^t & \quadeq & [(\bigcup_{b\in B} b)^*]\phi^t \\
\end{array}
$$

Beginning in Section~\ref{section-logics-based-on-program-signatures},
we study some new languages.  These will be based on a third parameter,
an {\em action signature} $\bSigma$.  For each ``type of
epistemic action'' there will be an action signature
$\bSigma$.  For each $\bSigma$, we'll
have languages $\lang_0(\bSigma)$,  $\lang_1(\bSigma)$, and  $\lang(\bSigma)$. 
These will be related to the other languages as indicated in the figure.

So we shall extend modal logic by
adding one or another type of epistemic action.  The idea is then we can
generate a logic from an action signature $\bSigma$ corresponding to an
intuitive action. For example, corresponding
to the notion of a public announcement is a particular action signature
$\bSigmapub$, and then the languages 
 $\lang_0(\bSigmapub)$,   $\lang_1(\bSigmapub)$, and  $\lang(\bSigmapub)$
will have something to do with our notion of a ``logic of public announcements.''

In Section~\ref{section-expressivepower}, we compare the expressive
power of the languages mentioned so far.  
It turns out that all of the arrows in  Figure~\ref{fig-languages}
correspond to proper increases in expressive power.
(There is one exception: $\lang_0$ turns out to equal $\lang_0(\bSigma)$
in expressive power for all $\bSigma$.)
It is even more interesting to compare expressive power as we change
the action signature.  For example, we would like to compare logics
of public announcement to logics of private announcement to groups.
Most of the natural questions in this area are open as of this writing.

\begin{figure}[t]
\fbox{
\begin{minipage}{6.0in}
$$\begin{array}{llll}
  \mbox{\bf Syntax} & \mbox{\bf Sentences $\phi$} & \begin{array}{c|c|c|c|c}
  p_i  & \neg\varphi &  
\varphi\wedge\psi
 &\necc_A\varphi  & \necc_{\BB}^*\varphi
\end{array}\\ \\
 \mbox{\bf Semantics} & \mbox{\bf Main Clauses} & 
\begin{array}{lcl}
\semantics{\necc_A\phi}{\bS}   & \quadeq &
\necc_A \semantics{\phi}{\bS}  \\
 \semantics{\necc^*_{\BB}\phi}{\bS}   & \quadeq &
\necc^*_{\BB}\semantics{\phi}{\bS}    \\
\end{array}
\\ 
\rem{
\mbox{\bf Proof System}
& \mbox{\bf Basic Axioms} \\
  & \mbox{All sentential validities}\\
& \mbox{($\necc_A$-normality)} & \proves\necc_A(\pphi\iif \psi)\iif (\necc_A\pphi\iif\necc_A\psi) &
  \\
{*} &  \mbox{($\necc^*_{\CC}$-normality)}
 &  \proves\necc^*_{\CC}(\pphi\iif \psi)\iif 
       (\necc^*_{\CC}\pphi\iif\necc^*_{\CC}\psi) &
 \\ \\
{*} & \mbox{\bf Mix Axiom}
&  \proves \necc^*_{\CC}\phi\iif \phi\andd
\bigwedge \set{\necc_{A}\necc^*_{\CC}\phi: A\in \CC}
\\
\\
 & \mbox{\bf Modal Rules} \\
& \mbox{(Modus Ponens)}& \mbox{From } \proves\phi \mbox{ and }
   \proves \phi\iif\psi, \mbox{ infer } \proves \psi \\
& \mbox{($\necc_A$-necessitation)}&  \mbox{From }\proves \pphi, \mbox{ infer } \proves \necc_A \pphi &  \\
{*} &  \mbox{($\necc^*_{\CC}$-necessitation)}
&   \mbox{From }\proves \pphi, \mbox{ infer } \proves \necc^*_{\CC} \pphi 
\\
\\
{*} &  \mbox{\bf Induction Rule} &
\mbox{From $\proves  \chi\iif\psi$
and  $\proves  \chi\iif\necc_A\chi$ for $A\in \CC$,}\\
 & & \mbox{infer $\proves \chi\iif
\necc^*_{\CC}\psi$}
\\
} %% proof system gone
\end{array}
$$
%%%
\end{minipage}
}
\caption{The languages $\lang_0$ and $\lang_1$.
 For $\lang_0$, we drop the 
 $\necc^*_{\BB}$
construct.
 \label{fig-modal-lgs}}
\end{figure}

\begin{figure}[tbp]
\fbox{
\begin{minipage}{6.0in}
$$\begin{array}{llll}
  \mbox{\bf Syntax}  
 &  \begin{array}{l} \mbox{\bf Sentences $\phi$} \\
\mbox{\bf Programs $\pi$} \end{array}
 &
\begin{array}{c|c|c|c|c}
    p_i  & \neg\varphi &  
\varphi\wedge\psi
 &[\pi]\phi \\
a &  ?\phi  & \pi;\sigma & \pi \cup \sigma & \pi^* \\
\end{array}\\
 \\
 \mbox{\bf Semantics} & \mbox{\bf Main Clauses} & 
\semantics{[\pi]\phi}{\bS} = \set{s: \mbox{if $s \semantics{\pi}{\bS} 
t$, then
$t\in \semantics{\phi}{\bS}$}} \\
 & & 
\semantics{?\phi}{\bS} = \set{(s,s) : s\in\semantics{\phi}{\bS}} \\
& &
\semantics{\pi;\sigma}{\bS} = \semantics{\pi}{\bS}\then \semantics{\sigma}{\bS}    \\
& &
\semantics{\pi \cup \sigma}{\bS} = \semantics{\pi}{\bS}\cup \semantics{\sigma}{\bS}    \\
& &
\semantics{\pi^*}{\bS} = (\semantics{\pi}{\bS})^*
\\  
\rem{
\mbox{\bf Proof System}
& \mbox{\bf Basic Axioms} \\
 & \mbox{All sentential validities}\\
& \mbox{($[\pi]$-normality)} & \proves [\pi](\pphi\iif \psi)\iif
([\pi]\pphi\iif[\pi]\psi) &
 \\ \\
& \mbox{\bf Program Axioms} \\
  & \mbox{(Test)} & \proves [?\phi]\psi \iiff (\phi\iif \psi)\\
& \mbox{(Skip)} & \proves [\skipp]\phi \iiff \phi \\
& \mbox{(Composition)} & \proves [\pi;\sigma]\pphi \iiff [\pi][\sigma]\pphi \\
& \mbox{(Choice)} & \proves [\pi \union \sigma]\pphi\iiff 
([\pi]\pphi \andd [\sigma]\pphi) 
  &
 \\ \\
  & \mbox{\bf Mix Axiom}
&  \proves [\pi^*]\phi\iif \phi\andd [\pi][\pi^*]\phi
\\
\\
 & \mbox{\bf Modal Rules} \\
& \mbox{(Modus Ponens)}& \mbox{From } \proves\phi \mbox{ and }
   \proves \phi\iif\psi, \mbox{ infer } \proves \psi \\
& \mbox{($[\pi]$-necessitation)}&  \mbox{From }\proves \pphi, \mbox{ infer } \proves
[\pi] \pphi &  
\\
  &  \mbox{\bf Induction Rule} &
\mbox{From $\proves  \chi\iif\psi$
and  $\proves  \chi\iif [\pi]\chi$}\\
 & & \mbox{infer $\proves \chi\iif
[\pi^*]\psi$}
}  %end of rem of the proof system
\end{array}
$$
%%%
\end{minipage}
}
\caption{Propositional Dynamic Logic (PDL)
 \label{fig-modal-pdl}}
\end{figure}

}

\subsection{Basic properties}
 \label{section-bisim-preserve-updates}

\begin{proposition}
Let $\phi$ be a sentence of $\lang(\bSigma)$,
and let $\alpha$ be an action of $\lang(\bSigma)$.
Then
\begin{enumerate}
\item $\semanticsnm{\phi}$
is  preserved by bisimulation.
\item $\semanticsnm{\alpha}$ is standard and preserves bisimulation.
\end{enumerate}
\label{prop-bisim-preserve-again}
\end{proposition}

\begin{proof}
By induction on $\lang(\bSigma)$.
For $\true$ and the atomic sentences $p_i$, we use
Proposition~\ref{prop-bisimulation-preservation}.
The same result takes care of the induction steps
for all the sentential operators except $[\pi]\phi$.
For this, we use 
Proposition~\ref{prop-bisimulation-preservation-again}
part 2,
and also the induction hypothesis.
Turning to the programs, the standardness comes from
Proposition~\ref{proposition-standard} and the 
observation that signature-based program models induce
standard updates.
The assertion that the interpretation of programs of
 $\lang(\bSigma)$ preserve bisimulation comes from
Proposition~\ref{prop-bisimulation-preservation-again}
part 1; for the programs $\sigma_i\vec{\psi}$ we
use Proposition~\ref{proposition-bisim-preserve-induced} 
and the induction hypothesis.
\end{proof}

In what follows,
we shall use Proposition~\ref{prop-bisim-preserve-again}
without mentioning it.  %%% section 3 on the languages
\Section{The logical system for validity of $\lang(\bSigma)$ sentences}
\label{section-langaction}
\label{section-langzero}

We write $\models\phi$ to mean that for all state models $\bS$
and all $s\in S$, $s\in\semantics{\phi}{\bS}$.
In this case, we say that $\phi$ is {\em valid}.

\paragraph{Sublanguages} 
We are of course interested in the full
languages $\lang(\bSigma)$.  As have already been mentioned,
the satisfiability problems for 
these languages are in general not recursively axiomatizable.
(See~\cite{miller} for details on this.)
This is one of the reasons we also consider sublanguages
$\lang_0(\bSigma)$ and $\lang_1(\bSigma)$:
$\lang_1(\bSigma)$ is the fragment without the action
iteration construct $\pi^*$;
$\lang_0(\bSigma)$ is the fragment without $\pi^*$ and $\necc^*_{\BB}$.
It turns out that $\lang_0(\bSigma)$ is the easiest to study:
it is of the same expressive power as ordinary multi-modal logic.
The main completeness result of  the paper is
a {\em sound and complete\/} proof system for the validities
in $\lang_1(\bSigma)$.

In Figure~\ref{figure-logical-system} below
 we present a   logic for $\lang(\bSigma)$.   We write $\proves\phi$
if $\phi$ can be obtained from the axioms
of the system using its inference rules.  We often omit the 
turnstile $\proves$ when it is clear from the context.

%%\subsection{Preservation of bisimulation and atomic propositions}

\rem{The central issue in this paper is the study of the logical
systems presented in Figure~\ref{figure-logical-system}.
In this section, we take some first steps. 
}

\rem{
\begin{proof}
We show the last part.
Suppose that $s\in \semanticsoff{[\actionalpha]\propositionphi}{\bS}$,
and let $(\bS,s) \equiv (\bT,t)$.
To show that $t\in \semanticsoff{[\actionalpha]\propositionphi}{\bT}$,
let $t\ \actionalpha_{\bT}\ t'$. 
 Then by condition (2) above,  there is 
some $s'$ such that  $s\ \actionalpha_{\bS}\ s'$
 and $(\bS(\actionalpha),s') \equiv (\bT(\actionalpha),t')$.   Since 
$s\in \semanticsoff{[\actionalpha]\propositionphi}{\bS}$, we have 
$s'\in {\propositionphi}_{\bS(\actionalpha)}$.
And then
$t'\in {\propositionphi}_{\bT(\actionalpha)}$, since $\propositionphi$ too
is preserved by bisimulation.
\end{proof}
}

\rem{
\begin{proof} 
 Write  $\actionalpha$ for the update induced 
by 
the program model
  $(\bSigma, \Gamma)$. Fix $\bS$ and $\bT$, and suppose that $s\equiv t$
via the relation $R_0$.
Suppose that $s\ {\actionalpha}_{\bS}\ s'$, so $s'\in \bS(\actionalpha)$
 is of the form
$(s,\sigma)$ for some $\sigma\in\Gamma$.  
Then $(t,\sigma)\in \bT(\actionalpha)$, and
 clearly $t\ {\actionalpha}_{\bT}\ (t,\sigma)$.
We need only show that $(s,\sigma)\equiv (t,\sigma)$.  But
the following relation $R$
 is a bisimulation between $\bS(\actionalpha)$ and 
$\bT(\actionalpha)$:
$$ (s',\tau_1)  \ R\ (t',\tau_2) 
\quadiff \mbox{$s' \ R_0\ t'$ and $\tau_1 = \tau_2$}. 
$$
The verification of the bisimulation properties is easy.
And $R$ shows that  $(s,\sigma)\equiv (t,\sigma)$,
as desired.
\end{proof}
}

\paragraph{Logics generated by families of signatures}
 Given a family  $\mathcal{S}$ of action signatures, we would
like to combine all the logics $\{\lang(\bSigma)\}_{\bSigma\in\mathcal{S}}$ into a single
logic. Let us assume the signatures $\bSigma\in\mathcal{S}$ are mutually disjoint
(otherwise, just choose mutually disjoint copies of these signatures).
We define {\em the logic $\lang(\mathcal{S})$ generated by the family $\mathcal{S}$} in the following way: the syntax is
defined by taking the same definition we had in Figure~\ref{fig-semantics}
for the syntax of $\lang(\bSigma)$, but in
which on the side of the programs we take instead as {\em basic actions} all
expressions of the form
$$\sigma_i \psi_1 \cdots\psi_n$$
where $\sigma\in \Sigma$, for some 
 signature $\bSigma\in\mathcal{S}$, and $n$ is the length
of the listing of non-trivial action types of $\bSigma$. 
The semantics is again given by the same definition as in
Figure~\ref{fig-semantics},
but in which the clause about $\sigma\psi_1\cdots\psi_n$ refers to
the appropriate signature: for every $\bSigma\in\mathcal{S}$, every
$\sigma\in\Sigma$, if $n$ is the length of 
the listing of  $\bSigma$,
then 
$$\semanticsnm{\sigma_i\psi_1 \cdots \psi_n} 
\quadeq
(\bSigma,\sigma,\semanticsnm{\psi_1},
\ldots,\semanticsnm{\psi_n}).
$$

\rem{
Most of  the system
will be quite standard from modal logic.  The Action Axioms
are    new, however.   These include the 
 in the Atomic Permanence axiom; note that in this axiom
$p$ is an atomic sentence.
The axiom  says that announcements
do not change the brute fact  of whether or not $p$ holds.
This axiom reflects the fact that our actions do not change
any kind of local state.   
}
%(We discuss an extension of our
%system in Section~\ref{section-two-extensions} where this
%axiom is not sound.)

\rem{
The
 Action-Knowledge Axiom
gives a criterion for knowledge after an action.
It is perhaps easier to appreciate in dual form:
 $$\pair{\sigma_i\vec{\psi}}\poss_A\phi
\iiff   ( \psi_i \andd
\bigvee\set{
 \poss_A [\sigma_j\vec{\psi}]\phi :
\sigma_i\arrowA\sigma_j \mbox{ in $\bSigma$}}).$$
In words, two sentences are equivalent: first, the assertion
that the action $\sigma_i\vec{\psi}$ may be executed and as 
a result, $A$ will consider $\phi$ possible;
and second, 
the precondition $\psi_i$ of the action holds, and 
$A$ considers it possible (before the action) that 
the action is really  $\sigma_j$ and that there is some possible
world in which that action 
$\sigma_j$  may be executed and results in a state satisfying $\phi$.
This axiom should be compared with the Ramsey axiom in conditional
logic. 

The Action Rule gives a necessary criterion for 
{\em common\/} knowledge after a simple action.   Since common knowledge
is formalized by the $\necc^*_{\BB}$ construct, 
this rule is
a kind of induction rule.  (The sentences $\chi_{\beta}$ 
play the role of strong induction hypotheses.) 

}

\rem{
As it happens,  the logic without the $**$ axioms is complete for
$\lang_1(\bSigma)$.  Dropping both the $*$ and $**$ axioms gives
a system complete for $\lang_0(\bSigma)$.  Our completeness results
appear in a separate paper.
}

\begin{figure}[tbp]
\fbox{
\begin{minipage}{6.0in}
\label{logic}
\medskip
$$\begin{array}{rlll}
& \mbox{\bf Basic Axioms} \\
  & \mbox{All sentential validities}\\
& \mbox{($[\pi]$-normality)} &
\proves [\pi]  (\phi \iif  \psi)
   \iif ([\pi] \phi \iif [\pi] \psi)\\
& \mbox{($\necc_A$-normality)} & \proves\necc_A(\pphi\iif \psi)\iif (\necc_A\pphi\iif\necc_A\psi) &
  \\
{*} &  \mbox{($\necc^*_{\CC}$-normality)}
 &  \proves\necc^*_{\CC}(\pphi\iif \psi)\iif 
       (\necc^*_{\CC}\pphi\iif\necc^*_{\CC}\psi) &
 \\ \\
 & \mbox{\bf Action Axioms}\quad \quad \quad &
%%\mbox{\underline{for simple $\alpha$ only}}
\\
 & \mbox{(Atomic Permanence)} & \proves [\sigma_i\vec{\psi}]  p\iiff(\psi_i\iif p)\\
 &  \mbox{(Partial  Functionality}) & \proves [\sigma_i\vec{\psi}]\nott\chi \iiff
(\psi_i\iif  \nott [\sigma_i\vec{\psi}]\chi) &\\
 & \mbox{(Action-Knowledge)} &\proves [\sigma_i\vec{\psi}]\necc_A\phi
\iiff   ( \psi_i \iif
\bigwedge\set{
 \necc_A [\sigma_j\vec{\psi}]\phi :
\sigma_i\arrowA\sigma_j \mbox{ in $\bSigma$}}) &\\
\\
%% The NULL ACTION IS NOT NEEDED
%{*} & \mbox{\bf Null Action}
%& \proves [\Pub\ \true]\phi\iiff \phi
%\\ 
%\\
{**} & \mbox{\bf Action Mix Axiom}
&  \proves [\pi^*]\phi\iif \phi\andd [\pi][\pi^*]\phi
\\ 
{*} & \mbox{\bf Epistemic Mix Axiom}
&  \proves \necc^*_{\CC}\phi\iif \phi\andd
\bigwedge \set{\necc_{A}\necc^*_{\CC}\phi: A\in \CC}
\\
\\
 & \mbox{\bf Skip  Axiom}
  & \proves[\skipp]\phi \iiff  \phi
\\ 
& \mbox{\bf Crash  Axiom}
  & \proves[\crash]\false 
\\  
%% the COMPOSITION AXIOM IS NOT NEEDED
 & \mbox{\bf Composition Axiom}
  & \proves[\pi][\rho]\phi \iiff [\pi\then\rho]\phi
\\
 & \mbox{\bf Choice  Axiom}
  & \proves[\pi \union \rho]\phi \iiff [\pi]\phi \andd [\rho]\phi
\\ \\
\hline \\ 
 & \mbox{\bf Modal Rules} \\
& \mbox{(Modus Ponens)}& \mbox{From } \proves\phi \mbox{ and }
   \proves \phi\iif\psi, \mbox{ infer } \proves \psi \\
& \mbox{($[\pi]$-necessitation)}
&   \mbox{From }\proves \psi, \mbox{ infer } \proves [\pi] \psi   \\
& \mbox{($\necc_A$-necessitation)}&  \mbox{From }\proves \pphi, \mbox{ infer } \proves \necc_A \pphi &  \\
{*} &  \mbox{($\necc^*_{\CC}$-necessitation)}
&   \mbox{From }\proves \pphi, \mbox{ infer } \proves \necc^*_{\CC} \pphi 
%\\
%\\
%{*} &  \mbox{\bf Induction Rule} &
%\mbox{From $\proves  \chi\iif\psi$
%and  $\proves  \chi\iif\necc_A\chi$ for $A\in \CC$,}\\
% & & \mbox{infer $\proves \chi\iif
%\necc^*_{\CC}\psi$}
\\ \\
{**} & \mbox{\bf Program Induction Rule} & \mbox{From }
\proves \chi\iif \psi\andd [\pi]\chi,
\mbox{infer } \proves\chi\iif [\pi^*]\psi
\\
{*} &  \mbox{\bf Action Rule}\\
\end{array}
$$
\begin{narrower}
Let $\psi$ be a  sentence, let $\alpha$ be a simple action,
and let $\CC$ be a set
of agents. Let there be sentences $\chi_{\justplain{\beta}}$ for
all $\justplain{\beta}$ such that
 $\justplain{\alpha}\rightarrow_{\CC}^* \justplain{\beta}$
(including $\justplain{\alpha}$ itself),
and such  that
\begin{enumerate}
\item $\proves\chi_{\justplain{\beta}}\iif [\justplain{\beta}]\psi$;
\item
if  $A\in\CC$ and $\justplain{\beta}\arrowA\justplain{\gamma}$,
then $\proves (\chi_{\justplain{\beta}} \andd\Pre(\beta) )
\iif \necc_A \chi_{\gamma}$.
\end{enumerate}
{} \noindent From these assumptions, infer
 $\proves\chi_{\justplain{\alpha}}\iif [\justplain{\alpha}]\necc_{\CC}^*\psi$.
\end{narrower}
\medskip
\end{minipage}
}
\vskip .25in
\caption{The logical system for $\lang(\bSigma)$.
For $\lang_1(\bSigma)$, we drop the ${**}$ axioms and rule;
for $\langaction$, we also drop the
 $*$ axioms and rules.
The definition of $\Pre$ in the Action Rule is given in 
Section~\ref{section-bOmega}.
\label{figure-logical-system}}
\end{figure}

\rem{
\subsection{Some derivable principles}

  In this section, we prove the soundness of
all the axioms and rules of the system except for the 
Action Rule.
We delay  discussion of the Action Rule
 mostly because it  takes a fair amount or
work to prove the soundness.  Also, delaying things
emphasizes   the significance of  adding   the
infinitary operators $\necc_{\BB}^*$ to 
$\langaction$ and then the program repetition construct $\pi^*$ to
this.

\begin{lemma}
The Action-Knowledge Axiom is provable for all simple actions
$\alpha$:
\begin{equation}  
\proves [\alpha]\necc_A\phi
\iiff   ( \Pre(\alpha) \iif
\bigwedge\set{
 \necc_A [\beta]\phi :
\alpha\arrowA\beta \mbox{ in $\bOmega$}})
\label{eq-aka-general}
\end{equation}
\label{lemma-aka-general}
\end{lemma}

\begin{lemma}
For all $A\in \CC$, all simple $\alpha$, and all
$\beta$
such that
 $\alpha\rightarrow_{A}  \beta$,
\begin{enumerate}
\item $\proves [\alpha]\necc_{\CC}^*\psi
\iif [\alpha]\psi$.
\label{part-ffff}
%\item $\proves [\alpha]\necc_{\CC}^*\psi
%\iif  [\beta]\necc_A \psi$.
%\label{part-gggg}
\item $\proves [\alpha]\necc_{\CC}^*\psi\andd\Pre(\alpha)
\iif \necc_A
[\beta]\necc_{\CC}^*\psi$.
\label{part-hhhh}
\end{enumerate}
\label{lemma-converse-rule}
\end{lemma}

\begin{proof}
 Part (\ref{part-ffff})
follows easily from the Epistemic Mix Axiom and
modal reasoning.
For part (\ref{part-hhhh}), we start with a consequence of
  the Epistemic Mix Axiom:
$\proves  \necc_{\CC}^*\psi\iif
 \necc_A\necc_{\CC}^*\psi$.
Then by modal reasoning,
$\proves  [\alpha] \necc_{\CC}^*\psi\iif
 [\alpha] \necc_A\necc_{\CC}^*\psi$.
By the Action-Knowledge Axiom generalized as we have it 
in Lemma~\ref{lemma-aka-general},  we have
$\proves [\alpha]\necc_{\CC}^*\psi\andd\Pre(\alpha)\iif
  \necc_A
[\beta]\necc_{\CC}^*\psi$.
%%This is what we need to show.
\end{proof}
}

\subsection{Soundness of the axioms}
\label{section-soundness-langzero}

In this section, we check the soundness of the axioms
of the system.  The basic axioms are all routine, and so
we omit the details on them.

\begin{proposition}
The Atomic Permanence Axiom $[\sigma_i\vec{\psi}]  p\iiff(\psi_i\iif p)$ 
 is sound.
\label{proposition-atomic-permanence-soundness}
\label{prop-actions-preserve-atomics}
\end{proposition}

\begin{proof}
Recall from Section~\ref{section-signature-based-epistemic-program-models}
 how
$\semanticsnm{\sigma_i\vec{\psi}}$ works; call this update
$\actionalpha$.
%Let $\alpha = \sigma_i\vec{\psi}$.  Let $\bSigma$ be the  
%signature-based action model
% $\bSigma= \bSigma(\sigma_i\vec{\psi}) $
%(see Section~\ref{section-signature-based-epistemic-action-models}).
 Fix a state model $\bS$.
The following are equivalent:
\begin{enumerate}
\item  $s\in \semantics{[\sigma_i\vec{\psi}]{p}}{\bS}$.
\item  $s\in \semanticsoff{[\actionalpha]{\propositionp}}{\bS}$.
\item If $s\  {\actionalpha}_{\bS} \ t$, then $t\in
\propositionp_{\bS(\actionalpha)}$.
\label{pt3}
\item If $(s,\sigma_i)\in  \bS(\actionalpha)$, then $(s,\sigma_i)\in
\propositionp_{\bS(\actionalpha)}$.
\label{pt4}
\item If $s\in \semantics{\psi_i}{\bS}$, then $s\in \|p\|_{\bS} =
\propositionp_{\bS}$.
\label{pt5}
\item $s\in \semantics{\psi_i\iif {p}}{\bS}$.
\end{enumerate}
Most of the individual equivalences are easy, and we only 
comment on  (\ref{pt3})$\Longleftrightarrow$(\ref{pt4})
and (\ref{pt4})$\Longleftrightarrow$(\ref{pt5}).
the definition of the update
$\actionalpha = (\bSigma,\sigma_i)$ implies that
$\bS(\actionalpha) = \bS \otimes\bSigma$.
And $\actionalpha_{\bS}$ has the property that everything
which $s\in S$
is related to   by $\actionalpha_{\bS}$ is of the form
$(s,\sigma_i)$, where $\sigma_i$ is from the statement of 
this axiom
and $(s,\sigma_i)\in  \bS(\actionalpha)$.
% Moreover, $s\ \actionalpha_{\bS}\ (s,\sigma_i)$ only
%when 
%$s\in \Presemantic(\sigma_i)_{\bS}= \semantics{\psi_i}{\bS}$.
These points imply the equivalence 
(\ref{pt3})$\Longleftrightarrow$(\ref{pt4}).  For
(\ref{pt4})$\Longleftrightarrow$(\ref{pt5})
note that 
   $(s,\sigma_i)\in  \bS(\actionalpha)$ iff 
$s\in \Presemantic(\sigma_i)_{\bS}= \semantics{\psi_i}{\bS}$; 
also  the definition of the update product in 
equation (\ref{eq-new-valuation}) implies that
$(s,\sigma_i)\in
\propositionp_{\bS(\actionalpha)}$ iff  $s\in \|p\|_{\bS} $.
\end{proof}
%%%%%

\begin{proposition}
The Partial Functionality Axiom  $[\sigma_i\vec{\psi}]\nott\chi \iiff
(\psi_i\iif  \nott [\sigma_i\vec{\psi}]\chi)$
 is sound.
\label{proposition-partial-functionality-soundness}
\end{proposition}

\begin{proof}
Again, let $\actionalpha  = \semanticsnm{\sigma_i\vec{\psi}}$.
Also, let $\propositionchi = \semanticsnm{\chi}$.
 Fix a state model $\bS$.
The following are equivalent:
\begin{enumerate}
\item  $s\in \semantics{[\sigma_i\vec{\psi}]\nott\chi}{\bS}$.
\item  $s\in \semanticsoff{[\actionalpha]{\nott\propositionchi}}{\bS}$.
\item If $s\ {\actionalpha}_{\bS} \ t$, then $t\in
\nott\propositionchi_{\bS(\actionalpha)}$.
\label{ppfs-crucial-first}
\item If $s\ {\actionalpha}_{\bS} \ t$, then $t\notin
\propositionchi_{\bS(\actionalpha)}$.
\label{ppfs-crucial}
\item  If $s\in \semantics{\psi_i}{\bS}$, then $(s,\sigma_i)\notin 
([\actionalpha]\propositionchi)_{\bS}$. 
\item  If $s\in \semantics{\psi_i}{\bS}$, then $s\notin
\nott\propositionchi_{\bS(\actionalpha)}$.
\item If $s\in \semantics{\psi_i}{\mathbf S}$, then $s\in
(\nott [\actionalpha]\propositionchi)_{\bS}$.
\item $s\in \semantics{\psi_i\iif  \nott [\actionalpha]\propositionchi}{\bS}$.
\end{enumerate}
The crucial equivalence here is  
(6)$\Longrightarrow$(5). 
The reason this holds
is that in the action model for
$\sigma_i\vec{\psi}$, there is just one distinguished world.
So for each fixed $s\in S$, there is at most one $t$ such that
  $s\ {\actionalpha}_{\bS} \ t$; i.e., $(s,\sigma_i)$.
\end{proof}

\begin{proposition}
The Action-Knowledge Axiom 
$$[\sigma_i\vec{\psi}]\necc_A\phi
\iiff   ( \psi_i \iif
\bigwedge\set{
 \necc_A [\sigma_j\vec{\psi}]\phi :
\sigma_i\arrowA\sigma_j \mbox{ in $\bSigma$}})
$$
is sound.
\label{proposition-AB-soundness}
\end{proposition}

\begin{proof} 
Again let $\actionalpha = \semanticsnm{\sigma_i\vec{\psi}}$.
%%% this is 18^4
 Fix a state model $\bS$.
The following are equivalent:
\begin{enumerate}
\item   $s\in \semantics{[\sigma_i\vec{\psi}]\necc_A\phi}{\bS}$.
\item  $s\in [\actionalpha]\semantics{\necc_A\phi}{\bS}$.
\item   If $s\ {\actionalpha}_{\bS}\ u$, then $u\in
\semantics{\necc_A\phi}{\bS(\actionalpha)}$.
\item If $s\in \semantics{\psi_i}{\bS}$, then 
$(s,\sigma_i)\in \semantics{\necc_A\phi}{\bS(\actionalpha)}$.

\item 
If $(s,\sigma_i)\in \bS(\actionalpha)$, then for all $(t,
\sigma_j)$  such that $(t,\sigma_j)\in \bS(\actionalpha)$,
 $s\arrowA t$ and $\sigma_i\arrowA \sigma_j$, we have
$(t,\sigma_j)\in \semantics{\phi}{\bS(\actionalpha)}$.
\rem{  \item I think this is point 6 -- deleted per Slawek's suggestion
If $(s,\sigma_i)\in \bS(\actionalpha)$, then for all 
$t$ such that $s\arrowA t$,  if $\sigma_i\arrowA \sigma_j$ 
and
$(t,\sigma_j)\in \bS(\actionalpha)$,
  then
$(t,\sigma_j)\in \semantics{\phi}{\bS(\actionalpha)}$.
}\item If $(s,\sigma_i)\in \bS(\actionalpha)$, then for all 
$t$ and all $\sigma_j$
such that $s\arrowA t$ and $\sigma_i\arrowA \sigma_j$: 
if
$(t,\sigma_j)\in \bS(\actionalpha)$,
  then
$(t,\sigma_j)\in \semantics{\phi}{\bS(\actionalpha)}$.
\item If $(s,\sigma_i)\in \bS(\actionalpha)$, then for all 
$t$ and $\sigma_j$
such that $s\arrowA t$ and  $\sigma_i\arrowA \sigma_j$:
$t\in \semantics{[\sigma_j\vec{\psi}]\phi}{\bS}$.
\item
 If $s\in \semantics{\psi_i}{\bS}$, 
then for all 
$\sigma_j$ such that $\sigma_i\arrowA\sigma_j$,
$s\in \semantics{\necc_A [\sigma_j\vec{\psi}]\phi}{\mathbf S}$.
% old formulation below
% $t$ such that $s\arrowA t$,
%$t\in \semantics{\bigwedge\set{[\sigma_j\vec{\psi}]\phi:
%\sigma_i\arrowA\sigma_j }}{\bS}$.
\item    $s\in  \semantics{\psi_i \iif
\bigwedge\set{\necc_A [\sigma_j\vec{\psi}]\phi :
\sigma_i\arrowA\sigma_j }}{\bS}$.
\end{enumerate}
\end{proof}

%$\sigma_j$ such that $\sigma_i \to^A \sigma_j$ 
%$s\in [[ \box_A [\sigma_j\psi^{arrow}]\phi ]]_{\mathbf S}$."

\begin{proposition}
The  Action Mix Axiom  $[\pi^*]\phi\iif \phi\andd [\pi][\pi^*]\phi$
is sound.  The Program Induction Rule
is also sound: 
from $\proves \chi\iif \psi\andd [\pi]\chi$,
infer $\proves\chi\iif [\pi^*]\psi$.
\label{proposition-actionmix-soundness}
\end{proposition}

\begin{proof}
These are standard.
\end{proof}

\begin{proposition}
The Skip Axiom \mbox{$[\mbox{\it skip\/}]\phi\iiff \phi$} is sound.
\label{proposition-skipp-soundness}
\end{proposition}

\begin{proof}  
 Fix a state model $\bS$.  Recall that the semantics of $\skipp$ is the
update $\skippaction$.  So
the following are equivalent:
\begin{enumerate}
\item   $s\in \semantics{[\skipp]\phi}{\bS}$.
\item If $s\ {\skippaction}_{\bS}\ t$, then
$t\in\semantics{\phi}{\actionmodel{\bS}{\skippaction}}$.
\item $s\in \semantics{\phi}{\bS}$.
\end{enumerate}
\end{proof}

\begin{proposition}
The Crash Axiom \mbox{$[\mbox{\it crash\/}]\false $} is sound.
\label{proposition-null-soundness}
\end{proposition}

\begin{proof}
Recall that $\semanticsnm{\crash} = \nullaction$, the update such that
$\bS(\nullaction)  $ is the empty model, and $\nullaction_{\bS}$ is the empty
relation.  Working through the definitions, $[\crash]\false$ is easily seen
to hold vacuously.
\end{proof}

\begin{proposition}
The Composition Axiom 
 $[\pi\then \rho] \phi
\iiff   [\pi][\rho]\phi$
is sound.
\label{proposition-composition-soundness}
\end{proposition}

\begin{proof}  
Write $\actionalpha$ for $\semanticsnm{\pi}$
and   $\actionb$ for $\semanticsnm{\rho}$.
 Fix a state model $\bS$.
The following are equivalent:
\begin{enumerate}
\item   $s\in \semantics{[\pi\then\rho]\phi}{\bS}$.
\item  If $s\ {\semanticsnm{\pi\then \rho}}_{\bS}\ u$, then 
$u\in \semantics{\phi}{\modelaction{\bS}{\actionalpha \then
\actionb}}$.
\item  If $s\ {(\actionalpha\then \actionb)}_{\bS}\ u$, then 
$u\in \semantics{\phi}{\modelaction{\bS}{\actionalpha \then
\actionb}}$.
\label{ppp}
\item  If $s\ {\actionalpha}_{\bS}\ t$
and $t\ {\actionb}_{\modelaction{\bS}{\actionalpha}}\ u$, then 
  $u\in
\semantics{\phi}{\modelaction{\bS}{\actionalpha\then\actionb}}$.
\label{qqq}
\item  If $s\ {\actionalpha}_{\bS}\ t$, then $t\in 
\semantics{[\rho]\phi}{\modelaction{\bS}{\actionalpha}}$.
\label{rrr}
\item   $s\in \semantics{[\pi][\rho]\phi}{\bS}$.
\end{enumerate}
The equivalence of (\ref{ppp}) and (\ref{qqq}) is by the definition
of relational composition.  The remaining equivalences are from the 
semantic definitions. 
The equivalence of (\ref{qqq}) and (\ref{rrr}) is by the fact that
$\bS(\actionalpha\then\actionb) = 
\modelaction{\modelaction{\bS}{\actionalpha}}{\actionb}$.
\end{proof}

\begin{proposition}
The Choice Axiom  $[\pi \union \rho]\phi\iiff [\pi]\phi\andd [\rho]\phi$
is sound.
\label{proposition-plus-soundness}
\end{proposition}

\begin{proof}
 Fix a state model $\bS$; we drop $\bS$ from the notation in the rest of this proof.
Write $\gamma$ for $\pi \union \rho$, 
 $\actionalpha$ for $\semanticsnm{\pi}$,
   $\actionb$ for $\semanticsnm{\rho}$, and $\actionc$ for $\semanticsnm{\gamma}$.
Then the following are equivalent:
\begin{enumerate}
\item   $s\in \semantics{[\pi \union \rho]\phi}{\bS}$.
\item   If $s\ {\actionc}\ u$, then $u\in \semantics{\phi}{\bS(\actionc)}$.
\label{choice2}
\item    If $s\ {\actionalpha}\ t$, 
then $t\in \semantics{\phi}{\bS(\actionc)}$;
 and if   $s\ {\actionb}\ t'$,
then $t'\in \semantics{\phi}{\bS(\actionc)}$.
\label{choice3}
\item   If $s\ {\actionalpha}\ t$, 
then $t\in \semantics{\phi}{\bS(\actionalpha)}$;
 and if   $s\ {\actionb}\ t'$,
then $t'\in \semantics{\phi}{\bS(\actionb)}$.
\label{choice4}
\item $s\in \semantics{[\pi]\phi}{\bS}$ and $s\in\semantics{[\rho]\phi}{\bS}$.
\item $s\in \semantics{[\pi]\phi \andd [\rho]\phi}{\bS}$.
\end{enumerate}
The equivalence (\ref{choice2})$\Longleftrightarrow$(\ref{choice3}) comes
from the fact that each of the  $u$ such that $s\ {\actionc}\ u$
is either either (a) an element $t$ of $\bS(\actionalpha)$ related to $s$ by $ {\actionalpha}$,
or else (b) an element
$t'$ of $\bS(\actionb)$ related to $s$ by ${\actionb}$.
This is by the definition of $\actionalpha \union \actionb$.

We show the equivalence (\ref{choice3})$\Longleftrightarrow$(\ref{choice4}) 
by considering the cases (a) and (b) noted just above.  We use the
fact that the natural  injections of 
 $\bS(\actionalpha)$ and $\bS(\actionb)$ in $\bS(\actionc)$ are bisimulations,
and also the fact that $\semanticsnm{\phi}$ 
is preserved by bisimulations
(see Proposition~\ref{prop-bisim-preserve-again}).
\end{proof}

\rem{OLD PROOF BELOW
\begin{proof}
We remind the reader that the relevant definitions
and notation are found in Section~\ref{section-semantics}.
Let $\alpha$ be the action $\pair{K,k}$.
Fix a pair $\pair{W,w}$.
If
$\pair{W,w}\models\nott\Pre(\alpha)$,
then both sides
of our biconditional hold.  We therefore assume that
$\pair{W,w}\models\Pre(\alpha)$ in the rest of this proof.
 Assume    that
$\pair{W,w}^{\alpha}\models\necc_A\phi$.
Take some $\beta$ such that $\alpha\arrowA \beta$.
This $\beta$ is of the form $\pair{K,k'}$ for some $k'$
such that $k\arrowA  k'$.
Let $w\arrowA w'$.
We have two cases:  $\pair{W,w'}\models \Pre(k')$,
and $\pair{W,w'}\models \nott\Pre(k')$.
In the latter case, $\pair{W,w'}\models [\beta]\phi$ trivially.
We'll show this in the former case, 
so assume $\pair{W,w'}\models \Pre(k')$.
Then $(w',k')$ is a world of $W^K$, and indeed 
$(w,k)\arrowA (w',k')$. 
Now our assumption that 
$\pair{W,w}^{\alpha}\models\necc_A\phi$ implies
that 
$\pair{W^K, (w',k')}\models \phi$.
This means that $\pair{W,w'}^\beta \models\phi$.
Hence $\pair{W,w'}\models [\beta]\phi$.
Since $\beta$
and $w'$ were arbitrary, $\pair{W,w}\models\bigwedge_{\beta}\necc_A[\beta]\phi$.

The other direction is similar.\end{proof}
}

\rem{
\paragraph{A simplified form of the Action Rule}
Our logical system has one additional rule of inference, the Action 
Rule.  We are not going to discuss this rule in full generality 
here (see Section~\ref{section-biggersystem} for this),
but we do want to mention a simplified form of this that the reader
could understand and use at this point.

Let $\bSigma$ be an action signature, and let $\sigma_1,\ldots, \sigma_n$
be our enumeration of it.
Let $\alpha$ be a  simple action.
Let $\psi$ be a  sentence, and let $\CC$ be a set
of agents. Let there be sentences  $\chi_{j}$ for
all $1\leq j \leq n$,
and such that
\begin{enumerate}
\item For all $j$, $\proves\chi_j\iif [\sigma_j\vec{\psi}]\phi$.
\item
If  $A\in\CC$ and  $\sigma_j\arrowA\sigma_k$,
then $\proves (\chi_{j} \andd\psi_j )
\iif \necc_A \chi_{k}$.
\end{enumerate}
 {} \noindent From these assumptions, infer
 $\proves\chi_{i}\iif [\sigma_i\vec{\psi}]\necc_{\CC}^*\phi$.

One can see that the purpose of the Action Rule is to allow us to 
infer sentences of the form $[\sigma_i\vec{\psi}]\necc_{\CC}^*\phi$.
Without this rule, or something like it, we have no way to do this.
For the same reason, we shall need a result permitting us to
infer sentences of more general forms,
such as
$$[\sigma_{i_1}\vec{\psi^1}\then\sigma_{i_2}\vec{\psi^2}\then\sigma_{i_l}\vec{\psi^l}]\necc_{\CC}^*\phi.
$$
The more general Action Rule is needed for this.
In the case of $l= 0$, our statement statement just above simplifies to a 
familiar rule:

\paragraph{The Induction Rule}
 From  $\proves\chi \iif  \psi \andd  \necc_A \chi$ 
for all $A$,
infer $\proves\chi \iif   \necc^*_{\Agents} \psi$.

\medskip

This will  turn out to be derivable from the Action Rule 
and the Skip Axiom.

}

\rem{
\subsection{Logical systems for the target logics}
\label{section-special-cases-logic}

As we know, we have logics $\langaction$ corresponding
to the logics presented in Section~\ref{section-target-logics}.
Each type of announcement gives us an action signature, and
each action signature a logic.
What we want to do here is to spell out what the axioms
of $\langaction$ come to when we specialize to some
of those logics.  In doing this, we find it convenient
to adopt simpler notations tailored for the fragments.

The
 logic of public announcements is shown in Figure~\ref{figure-lpa}.
We only included the axioms and rule of inference that specifically
used the structure of the signature $\bSigmaPub$.  So 
we did not include the sentential validities, the normality axiom
for $\necc_A$,  the composition axiom, modus ponens, etc.

\begin{figure}[htbp]
\fbox{
\begin{minipage}{6.0in}
\medskip
$$\begin{array}{lll}
\mbox{\bf Basic Axioms} \\
%  \mbox{All sentential validities}\\
 \mbox{($[\Pub\ \phi]$-normality)} &
\proves [\Pub\ \phi]  (\psi \iif  \chi)
   \iif ([\Pub\ \phi] \psi \iif [\Pub\ \phi] \chi)\\
%\mbox{($\necc_A$-normality)} & \proves\necc_A(\pphi\iif \psi)\iif (\necc_A\pphi\iif\necc_A\psi) &
%  \\
  \\ 
  \mbox{\bf Action Axioms}\quad \quad \quad\\
  \mbox{(Atomic Permanence)} & \proves [\Pub\ \phi]  p\iiff( \phi\iif p)\\
   \mbox{(Partial  Functionality}) & \proves [\Pub\ \phi]\nott\psi \iiff
(\phi \iif  \nott [\Pub\ \phi]\psi) &\\
  \mbox{(Action-Knowledge)} &\proves [\Pub\ \phi]\necc_A\psi
\iiff   ( \phi \iif
 \necc_A [\Pub\ \phi]\psi) \\ 
\\
  \mbox{\bf Modal Rules} \\
% \mbox{(Modus Ponens)}& \mbox{From } \proves\phi \mbox{ and }
%   \proves \phi\iif\psi, \mbox{ infer } \proves \psi \\
 \mbox{($[\Pub\ \phi]$-necessitation)}
&   \mbox{From }\proves \psi, \mbox{ infer } \proves [\Pub\ \phi] \psi   \\
% \mbox{($\necc_A$-necessitation)}&  \mbox{From }\proves \pphi, \mbox{ infer }
%\proves \necc_A \pphi & \\
\\
  \mbox{\bf Announcement Rule}
& \mbox{From } \proves\chi\iif [\Pub\ \phi]\psi \mbox{ and }
 \proves\chi\andd\phi\iif \necc_A\chi \mbox{ for all $A$,} \\
& \mbox{infer } \proves\chi\iif [\Pub\ \phi]\necc^*_{\Agents}\psi
 \\

\end{array}
$$
\medskip
\end{minipage}
}
%\vskip .2in
\caption{The main points of the logic of public announcements.
\label{figure-lpa}}
\end{figure}

Our next logic is the 
 logic of completely private announcements to groups.
We discussed the formalization of this in Section~\ref{section-target-logics-formalized}.
We have actions $\Pri^B \phi$ and  (of course) $\skipp$. 
The axioms and rules are just as in the logic of public announcements,
with a few changes.
We  must of course  consider the relativized
operators $[\Pri^\BB \phi]$ instead of their simpler
counterparts $[\Pub\ \phi]$.)

The  actions  $\Pri^\BB \phi$ all have $\true$ as their precondition, and 
since $(\true \iif \psi)$ is logically equivalent to $\psi$, we 
certainly may omit these  actions  from the notation in the axioms
and rules.
The most substantive change which we need to make  in
Figure~\ref{figure-lpa} concerns the  Action-Knowledge Axiom.
It splits into two axioms, noted below:
$$\begin{array}{lcl}
[\Pri^\BB \phi]\necc_A\psi \iiff
 ( \phi \iif \necc_A [\Pri^\BB \phi]\psi) & \mbox{for $A\in\BB$}
\\ 
{}
[\Pri^\BB \phi]\necc_A\psi \iiff   ( \phi \iif \necc_A \psi) & \mbox{for $A\notin \BB$}\\
\end{array}
$$
The  last equivalence says:  assuming that $\phi$ is true,
then after a  private announcement of $\phi$ to
the members of $\BB$, an outsider  knows
$\psi$ just in case she knew $\psi$ before the announcement.

\rem{
Next, we consider the logic of private announcements to groups,
 with secure suspicion by outsiders.
Here, 
the actions are $\Prss(\vec{\phi},\BB)$
and $Prss'(\vec{\phi},\BB)$.
The Action-Knowledge Axiom can be written 
out as above.  More insightfully,
we can use the fact that $Prss'(\vec{\phi},\BB)$ is equivalent
to $\Pri(\phi_1,\BB)$ to write the axiom as:
$$\begin{array}{lcl}
[\Prss(\vec{\phi},\BB)]\necc_A\psi \iiff
[\Pri(\phi_1,\BB)]\necc_A\psi  & \mbox{for $A\in\BB$}
\\ 
{}[\Prss(\vec{\phi},\BB)]\necc_A\psi \iiff   ( \phi_1 \iif 
\bigwedge_{0\leq i\leq k}
\necc_A [\Prss(\vec{\phi},\BB)]\psi)& 
\mbox{for $A\notin \BB$}\\
\end{array}
$$
Thus, the insiders know the same things after 
$\Prss(\vec{\phi},\BB)$ as they know  after
the completely private announcement  of $\phi_1$ to $\BB$.
For the outsiders, there is no simplification or reduction to 
the previous logics.
}

Finally, we study the logic of common knowledge
of alternatives.
The Action-Knowledge now becomes
$$\begin{array}{lcl}
[\Cka^{\BB}\vec{\phi}]\necc_A\psi \iiff
(\phi_1\iif \necc_A 
[\Cka^{\BB}\vec{\phi}]\psi)
& \mbox{for $A\in\BB$}
\\ 
{}[\Cka^{\BB}\vec{\phi}]\necc_A\psi \iiff   ( \phi_1 \iif 
\bigwedge_{0\leq i\leq k}
\necc_A [\Cka^{\BB}{\vec{\phi}}^i]\psi)& 
\mbox{for $A\notin \BB$}\\
\end{array}
$$
where in the last clause,
$(\phi_1,\ldots, \phi_n)^i$ is the sequence
$\phi_i, \phi_1, \ldots, \phi_{i-1}, \phi_{i+1}, \ldots, \phi_k$.
(That is, we bring $\phi_i$ to the front of the sequence.)

\rem{
Finally, we study the logic
of common knowledge of suspicion by outsiders.
$$\begin{array}{lcl}
[\phi;\vec{\psi}]^{prss}_{\BB}\necc_A\chi \iiff
(\phi \iif \necc_A [\phi;\vec{\psi}]^{prss}_{\BB} \chi
  & \mbox{for $A\in\BB$}
\\ 
{}
[\phi;\vec{\psi}]^{prss}_{\BB}\necc_A\chi \iiff   ( \phi \iif 
\bigwedge_{0\leq i\leq k}
\necc_A [(\phi;\vec{\psi})_i]^{pr}_{\BB}\chi)& 
\mbox{for $A\notin \BB$}\\
\end{array}
$$
where $(\phi;\vec{\psi})_0 = \phi;\vec{\psi}$,
and $(\phi;\vec{\psi})_i = \psi_i;\phi,\psi_1,\ldots, \psi_{i-1},\psi_{i+1},\ldots, \psi_k$.
}

\subsection{Examples in the target logics}
\label{section-action-rule-in-targets}

This section studies some examples of the logic at work.
We begin with an  application of the Action Rule in the logic of public
announcements.  We show
$$
\proves
\necc^*(\Heads \iiff \nott \Tails) \iif [\Pub\ \Heads]\necc^*\nott\Tails.
$$
That is, on the assumption that 
it is common knowledge that
Heads and Tails are mutually exclusive,
then as a result of a public announcement of Heads it will be common
knowledge that the state is not Tails.

We give this application in detail.  Recall that $\bSigmaPub$ has one
simple action which we call $\Pub$.  We take $\chi_{\Pub}$ to be
$\necc^*(\Heads \iiff \nott \Tails)$.  In addition $\Pub\arrowA \Pub$ for all $A$,
and there are no other arrows in $\bSigmaPub$.   We take $\alpha$ to be 
$\Pub\ \Heads$; note that this is the only action accessible from
itself in the canonical action  model.
To use the Action Rule,
we must show that
\begin{enumerate}
\item $\proves \necc^*(\Heads \iiff \nott \Tails)\iif [\Pub\ \Heads]\nott \Tails$.
\item
  $\proves (\necc^*(\Heads \iiff \nott \Tails) \andd \Heads)
\iif \necc_A \necc^*(\Heads \iiff \nott \Tails)$ for all agents $A$.
\end{enumerate}
 {} \noindent From these assumptions, we may infer
 $\proves\necc^*(\Heads \iiff \nott \Tails)\iif [\Pub\ \Heads]\necc^*\nott \Tails$.

For the first statement,
$$
\begin{array}{llr}
{\rm (a)}  & \Tails \iiff [\Pub\  \Heads]\Tails 
& \mbox{Atomic Permanence} \\
{\rm (b)}  & (\Heads \iiff \nott \Tails) \iif (\Heads \iif
\nott [\Pub\  \Heads]\Tails) & \mbox{(a), propositional
reasoning} \\ 
{\rm (c)}  & [\Pub\ \Heads]\nott \Tails \iiff (\Heads
\iif \nott [\Pub\ \Heads]  \Tails) & \mbox{Partial
Functionality} \\
{\rm (d)}  & \necc^*(\Heads \iiff \nott
\Tails)\iif (\Heads \iiff \nott \Tails)  &\mbox{Epistemic
Mix}
\\
{\rm (e)}  & \necc^*(\Heads \iiff \nott \Tails)\iif[\Pub\ \Heads]\nott \Tails
& \mbox{(d), (b), (c),  propositional reasoning}\\
\end{array}
$$
And the second statement is an easy consequence of the Epistemic  Mix Axiom.

\paragraph{What happens when a publically known fact is announced?}

One intuition about public announcements and common knowledge is that
if $\phi$ is common knowledge, then announcing $\phi$ publically does
not change anything.  Formally, we express this by
a scheme rather than a single equation:
\begin{equation}
\necc^*\phi \iif ([\Pub\ \phi]\psi \iiff  \psi)
\label{irrelevance}
\end{equation}
What we would like to say is
$\necc^*\phi \iif   \bigwedge_{\psi}
([\Pub\  \phi]\necc_A \psi \iiff \psi)$, but of course this cannot be
expressed in our language.  So we consider only the sentences of the
form (\ref{irrelevance}), and we show that all of these are provable.
We argue by induction on $\phi$.

For an atomic sentence $p$, (\ref{irrelevance})  follows from
the Epistemic Mix and Atomic Permanence Axioms.   The induction
steps for $\andd$ and $\nott$ are easy.  Next, assume  (\ref{irrelevance})
for $\psi$.   By necessitation and Epistemic Mix, we have
$$\necc^*\phi \iif (\necc_A[\Pub\ \phi]\psi \iiff  \necc_A\psi)$$
Note also that by the Announcement-Knowledge Axiom
$$\necc^*\phi \iif ([\Pub\ \phi]\necc_A\psi \iiff  \necc_A[\Pub\ \phi]\psi)$$
These two imply (\ref{irrelevance})  for $\necc_A\psi$.

Finally, we assume   (\ref{irrelevance})  for
$\psi$ and prove it for $\necc^*_{\BB}\psi$.
We show first that
$\necc^*\phi\andd \necc^*\psi \iif   [\Pub\ \phi]\necc^*\psi$.
For this we use the Action Rule.  We must show that
\begin{description}
\item{(1)}  $\necc^*\phi\andd \necc^*\psi \iif   [\Pub\ \phi]\psi$.
\item{(2)} $\necc^*\phi\andd \necc^*\psi  \andd \necc_A( \necc^*\phi\andd \necc^*\psi)
\iif \necc_A(\necc^*\phi\andd \necc^*\psi)$.
\end{description} 
(1) is easy from our induction hypothesis, and (2) is an easy
consequence of Epistemic Mix.

To conclude, we show
$\necc^*\phi\andd   [\Pub\ \phi]\necc^*\psi \iif \necc^*\psi$.
For this, we use the Induction Rule; that is, we show
\begin{description}
\item{(3)} $\necc^*\phi\andd   [\Pub\ \phi]\necc^*\psi \iif  \psi$.
\item{(4)}  $\necc^*\phi\andd   [\Pub\ \phi]\necc^*\psi \iif \necc_A(\necc^*\phi\andd   [\Pub\
\phi]\necc^*\psi)$ for all $A$.
\end{description}
For (3), we 
use Lemma~\ref{lemma-converse-rule},  part~(\ref{part-ffff})
to see that $[\Pub\ \phi]\necc^*\psi \iif  [\Pub\ \phi]\psi$;
and now (3) follows from our induction hypothesis.
For (4), it will be sufficient to show that
$$\phi\andd [\Pub\ \phi]\necc^*\psi \iif \necc_A[\Pub\
\phi]\necc^*\psi
$$
This follows from
  Lemma~\ref{lemma-converse-rule}, part~(\ref{part-hhhh}).

\paragraph{A commutativity principle for private announcements}
Suppose that $\BB$ and $\CC$ are disjoint sets of agents.
Let $\phi_1$, $\phi_2$, and $\psi$ be sentences.  Then
we claim that 
$$ \proves [\Pri^{\BB} \phi_1][\Pri^{\CC} \phi_2]\psi
\iiff [\Pri^{\CC}\phi_2][\Pri^{\BB} \phi_1]\psi.$$
That is,   order does not matter
with private announcements to disjoint groups.

\paragraph{Actions do not change common knowledge of simple sentences}
For yet another application, let $\psi$ be any boolean combination of
atomic sentences.
Then for all actions $\alpha$ of any of our logics,
$\proves \psi\iiff [\alpha]\psi$.   The proof is an easy induction on $\psi$.
Even more, we have
$\proves \necc^*_{\CC}\psi\iiff [\alpha]\necc^*_\psi$. 
In one direction, we use the Action Rule, and in the other,
the Induction Rule.

\rem{%% no room for this one
\subsection{Application: syntactic treatment of an earlier
example}
\label{section-syntactic=treat}

In this section, we return to 
the ``cheating'' Scenario~\ref{scenario-A-cheats}.
We further studied this example in
Section~\ref{section-justifying-the-models};
in that section we justified the earlier model.
What we want to do here is to work out explicitly
some of the points which we argued semantically.
It might be noted that our entire discussion in
this section is not really necessary: everything
follows from the completeness of our logics and
some facts on characterizing sentences.  Still,
it might help the reader to see how our logic
works in this setting.  In particular, it might
be useful to see a concrete application of the 
Action Rule.

We consider the logic of private announcements to groups.
Let $\alpha$ be the announcement of $\Heads$ to $A$ 
privately, and let $\beta$ be the announcement of
$\true$ to both $A$ and $B$.  Then we have
$\alpha\arrowA \alpha$, $\alpha\arrowB\beta$,
and $\beta\arrowAB\beta$.

We use
 the notation from Section~\ref{section-justifying-the-models}.
In particular, $\chi_{\Heads}$ denotes the sentence which says that
the coin lies $\Heads$ up but it is common knowledge that
neither $A$ nor $B$ knows this
(and it is common knowledge that $\Heads\iiff \nott\Tails$).
  And $\psi$ is the sentence 
in equation (\ref{eqpsi}) in Section~\ref{section-justifying-the-models}.
 Our main claim 
in this section  is 
$\proves \chi_{\Heads}\iif \pair{\alpha}\psi$.

We now present a list of sentences which are provable in our logic.
Since a complete list would be too long, we only present the main steps.
We have abbreviated $\nobodyknows$ to
\renewcommand{\nobodyknows}{\mbox{\sc NK}}
 $\nobodyknows$.

\begin{enumerate}
\item $[\alpha]\Heads \iiff (\Heads \iiff \Heads)$.  
This is an Action-Knowledge (AK) axiom. \label{aa}
\item $[\alpha]\Heads$.  By (\ref{aa}).\label{bbb}
\item $[\alpha]\necc^*_{A}\Heads$.  By (\ref{bbb}) and the Action Rule,
taking $\chi_{\alpha} = \true$.
\label{goal1}
%%\item $\chi_{\Heads} \iif[\alpha]  \necc_B\necc^*_{A,B}\nobodyknows $.
%\item $\nott\necc_B\Heads\iif [\beta]\nott\necc_B\Heads$.
\item 
   $\nott\necc_B\nott\Heads\iif (\Heads\iif \nott[\beta]\necc_B\nott\Heads)$.  
      By (*) and propositional reasoning.
\label{12165}
\item 
   $\nott\necc_B\nott\Heads\iif (\Heads\iif \nott\necc_B[\beta]\nott\Heads)$.  
By (\ref{12165}) and the AK Axiom $[\beta]\necc_B\nott\Heads \iiff  \necc_B[\beta]\nott\Heads$.
\label{12166}
\item  $\nott\necc_B\nott\Heads\iif (\Heads\iif \nott[\alpha]\necc_B\nott\Heads)$.
By (\ref{12166}) and the AK Axiom $[\alpha]\necc_B\nott\Heads \iiff (\Heads\iif \necc_B[\beta]\nott\Heads)$.
\label{1216}
\item $\nott\necc_B\nott\Heads\iif[\alpha]\nott\necc_B\nott\Heads$.
By (\ref{1216}) and a Partial Functionality (PF) Axiom.
\label{12168}
\item $\poss_B\Tails\iif[\alpha]\poss_B\Tails$.  Similar to (\ref{12165}-\ref{12168})
\label{simppppp}
\item $\necc^*_{A}(\poss_B\Heads\andd\poss_B\Tails) \iif
[\alpha]\necc^*_{A}(\poss_B\Heads\andd\poss_B\Tails)$. 
By the Action Rule, (\ref{12168}), and (\ref{simppppp}).  We also use Mix and $[\alpha]$-normality.
\item  $\necc^*_{A,B}\nobodyknows \iif [\beta]\necc^*_{A,B}\nobodyknows$.
By (*).  \label{fornec}
\item $\necc_B\necc^*_{A,B}\nobodyknows \iif \necc_B[\beta]\necc^*_{A,B}\nobodyknows$.
By (\ref{fornec}) and necessitation.
\label{forfornec}
\item $\necc^*_{A,B}\nobodyknows  \iif \necc_B[\beta]\necc^*_{A,B}\nobodyknows$.
By (\ref{fornec}) and (\ref{forfornec}).
\label{cc}
\item $[\alpha]  \necc_B\necc^*_{A,B}\nobodyknows \iiff (\Heads \iif \necc_B[\beta]
\necc^*_{A,B}\nobodyknows)$.  This is an AK axiom. \label{dd}

\item  $\necc^*_{A,B}\nobodyknows \iif [\alpha]\necc_B\necc^*_{A,B}\nobodyknows$.
By (\ref{cc}) and (\ref{dd}).
\label{zz}
\item  $\necc^*_{A,B}\nobodyknows \iif\necc_A\nobodyknows$.  By Mix.
\label{yy}
\item $\necc^*_{A,B}\nobodyknows \iif [\alpha]\necc_A^*\necc_B\necc^*_{A,B}\nobodyknows$.
From the Action Rule, (\ref{zz}) and (\ref{yy}).
\label{goal2}
\item $[\alpha]\poss_A\true \iiff \poss_A\pair{\alpha}\true$.
Using a PF Axiom and an  AK Axiom.  \label{344}
\item $\pair{\alpha}\true\iiff \Heads$. 
Using a PF Axiom and the necessitation fact   $[\alpha]\true$.
 \label{7888}
\item $\poss_A\Heads \iif [\alpha]\poss_A\true$.
By (\ref{344}) and (\ref{7888}).
\label{goal3}
\rem{
\item $[\alpha]\nott\Tails \iiff (\Heads\iif \nott\Tails)$.
Using a PF Axiom and the Atomic Permanence (AP) Axiom $[\alpha]\Tails \iiff (\Heads\iif \Tails)$.
\label{prennn}
\item  $\necc^*_{A,B}(\Heads\iiff\nott\Tails)\iif
  [\alpha](\Heads\iiff\nott\Tails)$.
By (\ref{prennn}), $[\alpha]$-necessitation, and
the  fact from AP that $[\alpha]\Heads$. 
\label{nnn}
\item  $\necc^*_{A,B}(\Heads\iiff\nott\Tails)\iif
  [\beta](\Heads\iiff\nott\Tails)$.  By Mix and (*).
\label{mmm}
}
\item  $\necc^*_{A,B}(\Heads\iiff\nott\Tails)\iif
  [\alpha] \necc^*_{A,B}(\Heads\iiff\nott\Tails)$.
\label{goal4} 
%By the Action Rule, (\ref{nnn}), and (\ref{mmm}).
By  our work in Section~\ref{section-action-rule-in-targets}.
\item $\chi_{\Heads}\iif  [\alpha]\psi$.  Using 
the definition of $\chi_{\Heads}$, $[\alpha]$-Necessitation,
(\ref{goal1}), (\ref{goal2}), (\ref{goal3}), and (\ref{goal4}).
\label{goalgoal}
%\item $\chi_{\Heads}\iif \pre(\alpha)$.  Since $\Pre(\alpha)$ is $\Heads$.
\item $\chi_{\Heads}\iif \pair{\alpha}\psi$.
Using (\ref{goalgoal})
and the equivalence of $\pair{\alpha}\psi$ and $\Heads\andd [\alpha]\psi$.
\end{enumerate}

}  %end of  the rem of the long application

\rem{%%Qian's theorem deleted for space
\subsection{Qian's Theorem on sentences true after being announced}
\label{section-further-results-on-the-target-logics}

We consider  a natural question about the logic of public
announcements, especially in its one-agent setting as the
logic of modal relativization.  Some sentences $\phi$
have the property that $\models[\Pub\ \phi]\phi$. 
That is, whenever $\phi$ holds and is announced, then $\phi$
holds after the announcement.
 For example,
the atomic sentences $p$ have this property, as do their boolean
combinations.   So do sentences $\necc p$.
The question here is to characterize syntactically
those $\phi$ such that
$[\Pub\ \phi]\phi$ is valid.  The first guess is that 
$[\Pub\ \phi]\phi$ is valid iff $\phi$ is equivalent to 
a sentence in the {\em universal\/} fragment of modal
logic, the fragment built from atomic sentences and their
negations using $\necc$, $\andd$, and $\orr$.  However, this is
not to be.

Suppose that $\phi$ and $\psi$ are non-modal  sentences
(that is, boolean combinations of atomic sentences).  Suppose
that $\models \psi\iif \phi$.  Consider 
$\phi\orr\poss\psi$.  This is clearly not in general equivalent to
a sentence in our fragment.  Yet we claim that 
$$\models [\Pub\ \phi\orr\poss\psi](\phi\orr\poss\psi).
$$
To see this, fix a state model $\bS$ and some $s\in S$.  
If $s\in\semantics{\phi}{\bS}$, then since $\phi$ is   non-modal,
$s\in\semantics{\phi}{\bS(\Pub\ \phi\orr\poss\psi)}$.
On the other hand, suppose that
$s\in\semantics{\poss\psi}{\bS}$.  Let $s\rightarrow t$ with
$t\in\semantics{\psi}{\bS}\subseteq\semantics{\phi}{\bS}$.
Hence $t$ belongs to the model after the announcement,
$\bS(\Pub\ \phi\orr\poss\psi)$.
And $s\in \semantics{\poss\phi}{\bS(\Pub\ \phi\orr\poss\psi)}$
via $t$.

This example is due to Lei Qian.  He also found a hypothesis under
which the ``first guess'' above indeed holds.  Here is his result.
Let $T_0$ be the set of  non-modal sentences.
Let 
$$T_1 \quadeq T_0 \cup \set{\necc\phi : \phi \in T_0} \cup
 \set{\poss\phi : \phi \in T_0}.
$$
Finally, let $T_2$ be the closure of $T_1$ under $\andd$ and $\orr$.

\begin{theorem}[Qian~\cite{qian}]
Let $\phi\in T_2$ have the property that $\models [\Pub\ \phi]\phi$.
Then there is some $\psi$ in the universal fragment of modal logic
such that $\models\phi\iiff\psi$.
\end{theorem}

We also should discuss some natural conjectures.
For example, here is a very strong way to say that
it is impossible to attain common knowledge by
private announcements.

\begin{conjecture}  Let there be just two agents,
$A$ and $B$.
Let $\phi$ 
be  a sentence  in $\lang_1(\bSigmaPri)$,
and let $\alpha$ be an action in $\lang_1(\bSigmaPri)$
in which all of the announcements 
are private.  (That is, $\alpha$ has no occurrences
of $\Pri^{A,B}$, only $\Pri^A$ and $\Pri^B$.)
Suppose that $\models \phi \iif [\alpha]\necc^*_{A,B} p$.
Then  $\models \phi \iif \necc^*_{A,B} p$.

In other words, if $\phi$ implies that running $\alpha$
always leads to common knowledge, then $\phi$ itself
already implies that there is common knowledge of $\phi$
(and presumably running $\alpha$ does not change this.)
\end{conjecture}
} %% end of rem of Qian's theorem

\paragraph{Endnotes}  
Although the general logical systems
in this paper are new,  there are important precursors for 
the target logics.   Plaza~\cite{plaza} constructs what we would call
$\lang_0(\bSigma_{\Pub})$, that is, the logic of public announcements
without common knowledge operators or program iteration.  (He worked
only on   models where each accessibility relation is an equivalence
relation, so his system includes the S5 axioms.)
Gerbrandy~\cite{gerbrandy98,gerbrandyphd},
and also Gerbrandy and Groeneveld~\cite{gerbrandygroeneveld} went a bit further.
They studied the logic of completely private announcements
(generalizing public announcements) and presented
a logical system which  included
the common knowledge operators.   That is, their system included Mix.
They argued that all of the reasoning in the original
Muddy Children scenario can be carried out in their system.
This is important because it 
 shows that in order to get a formal treatment of that problem and 
related ones,
one need not posit models which maintain histories.   Their system
was not complete since it 
did not have anything like the Action Rule; this first appears
in a slightly different form in~\cite{bms}.

\subsection{Conclusion}

We have been concerned with actions in the social world that
affect  the intuitive concepts of knowledge (actually
justified true belief) and common knowledge.
This paper has shown how to define and study logical languages
that contain constructs corresponding to such actions.  
The many examples in this
paper show that the logics ``work''.  Much more can be said
about specific tricky examples, but we hope that the examples
connected to our scenarios make the point that we are developing
valuable tools.

The key steps in the development are the recognition that 
we can associate to a social  action $\alpha$
a mathematical model $\bSigma$.
$\bSigma$ is an action model.  In particular, 
it is a multi-agent Kripke model, so it has features
in common with the state models that underlie  
formal work in the  entire area.  There is a natural operation
of update product at the heart of our work.  This operation is
surely of independent interest because it enables
one to build complex and interesting state models.
  The logical languages 
that we introduce use the update product in their semantics,
but the syntax is a small variation on propositional dynamic
logic.  We feel that presenting the update product first (before
the languages) will make  this paper easier to read, and having
a relatively standard syntax should also help.

Once we have our languages, the next natural step is to study them.
This paper presented logical systems for validities, omitting
many proofs due to the lack of space.  

}
     %%  section 4

\Section{The Action Rule}
\label{section-langactionstar}
\label{section-biggersystem}

Recall that 
our logical system  is listed in
Figure~\ref{figure-logical-system} above.
We   have shown the soundness of all of the axioms
and rules, except for the Action Rule.
  But before we can even define the operation $\Pre$
that figures into  the Action Rule,
we need to discuss a structure that comes up in our work,
the {\em canonical action  model $\bOmega$}.
And even after we define $\bOmega$, there is a
fair amount of work to do before we can prove
the soundness of the Action Rule.

\subsection{The canonical action  model $\bOmega$}
\label{section-bOmega}

Recall that we defined   action models 
and  program models in
Sections~\ref{section-simple-action-structures}.
 At this point, we define 
the {\em canonical action  model\/}
$\syntacticactionmodel$.

\begin{definition}
A {\em basic action\/} is a program
of the form  $\sigma\vec{\psi}$.
A  {\em simple action\/} of $\lang(\bSigma)$ is
 a program 
 in which neither
the program sum operation $\union$ nor 
the iteration operation $\pi^*$  occur.
We use letters like $\alpha$ and $\beta$ to denote simple actions
{\em only}.
%A {\em simple sentence\/} is a sentence of $\lang_1(\bSigma)$ in which
%neither action sum nor action iteration occur.  So all programs
% in simple sentences are simple actions, and vice-versa.
%We define the length  $l(\alpha)$ of a simple action  $\alpha$ as follows:
%$l(\skipp) = 0$,
%$l(\crash) = 0$,
% $l(\sigma_i\vec{\psi}) = 1$.
\end{definition}

\paragraph{The canonical  action model $\syntacticactionmodel$}
We define a program model $\syntacticactionmodel$ in several steps.
The  actions of the model  $\syntacticactionmodel$   
(that is, the elements of its carrier set) 
are the 
simple actions of the language $\lang(\bSigma)$
(as defined just above).
  For all $A$,
the accessibility relation  $\arrowA$ is the smallest relation
such that 
\begin{enumerate}
\item $\skipp\arrowA\skipp$.
\item $\sigma_i\vec{\phi} \arrowA \sigma_j\vec{\psi}$ iff
$\sigma_i\arrowA \sigma_j$ in $\bSigma$ and
$\vec{\phi} = \vec{\psi}$.
\item If $\alpha\arrowA\alpha'$ and $\beta\arrowA \beta'$,
then $\alpha\then \beta \arrowA \alpha'\then\beta'$.
%\item If $\alpha\arrowA\gamma$  or $\beta\arrowA \gamma$,
%then $\alpha + \beta \arrowA \gamma$.
\end{enumerate}

\begin{proposition}
As a frame, $\syntacticactionmodel$   is {\em locally
finite}: for each simple $\alpha$, there are only finitely
many $\beta$ such that $\alpha\arrowAgentsstar \beta$.
\label{proposition-locallyfinite}
\end{proposition}

\begin{proof}
By induction on  $\alpha$; we use
heavily the fact that the accessibility relations
on $\syntacticactionmodel$ are the smallest family
with their defining property.   
For the simple action expressions 
 $\sigma\vec{\psi}$, we use the
assumption that the action model $\bSigma$ underlying
all our definitions is finite.  (Taking it to be
locally finite would also be sufficient.)
\end{proof}

Next, we define $\map{\Pre}{\Omega}{\lang_1(\bSigma)}$
by recursion so that
$$\begin{array}{lcl}
\Pre(\skipp) & \quadeq & \true\\
\Pre(\crash) & \quadeq & \false\\
\Pre(\sigma_i\vec{\psi}) & \quadeq & \psi_i \\
\Pre(\alpha\then\beta) & \quadeq & 
%%%%\Pre(\alpha)
%%% \andd [\semanticsnm{\alpha}]\Pre(\beta)\\
\pair{\alpha}\Pre(\beta)\\
%\Pre(\alpha + \beta) & \quadeq & \Pre(\alpha) \orr \Pre(\beta) \\
\end{array}
$$
 This function $\Pre$ is not 
 the function
 $\Presemantic$  which is part of
the structure of an epistemic action model.
However,  there is a
connection:
We are in the midst of
 defining the epistemic  action model $\bOmega$,  
 and its $\presemantic$ function is defined 
 in terms of $\Pre$.

Another point: the clause in the definition
of $\Pre$ could read
$\Pre(\alpha\then\beta) =
 \Pre(\alpha)
  \andd [{\alpha}]\Pre(\beta)$.  This is equivalent to 
what we have above.  It is sometimes tempting to think that the
definition should be
$\Pre(\alpha\then\beta) =
 \Pre(\alpha)
  \andd \pair{{\alpha}}\Pre(\beta)$.  But in this, the
first conjunct is redundant:  
$\pair{{\alpha}}\Pre(\beta)$ implies 
$ \Pre(\alpha)$.

%%\paragraph{Completing the definition of $\bOmega$}
We set 
$$\presemantic(\sigma) \quadeq \semanticsnm{\Pre(\sigma)}.
$$
This simple action model $\bOmega$ is the {\em canonical 
epistemic action model}; it plays the same role in our
work as the canonical model in modal logic.

\begin{remark}
The structure
$(\Omega,\arrowAA,\Pre)$ is entirely syntactic. 
It is \emph{not} an action model.
 This contrasts
with the canonical action model
$\bOmega = (\Omega,\arrowAA,\Presemantic)$; the last component of this 
structure is semantic.  The syntactic structure 
$(\Omega,\arrowAA,\Pre)$ is the one which is actually used in 
the statement of the Action Rule of the logical system.
We emphasize this point to allay any suspicion that our
logical system is formulated in terms of semantic concepts.

This is also perhaps a good place to remind the
reader that neither $\Pre$ 
nor $\Presemantic$
is a first-class symbol in $\lang(\bSigma)$;
it is only a defined symbol.   In Section~\ref{section-completeness-theorems}
below we shall introduce another language called
$\lang_1^+(\bSigma)$.  In that language,
$\Pre$ will be a first-class function symbol.
\end{remark}

\subsection{The main result about $\bOmega$}

\begin{definition}
For any
$\alpha\in\Omega$, let $\hat{\alpha}$ be the 
program model $(\syntacticactionmodel,\set{\alpha})$.
As in
Section~\ref{section-epistemic-program-models-give-updates},
we use the same notation $\hat{\alpha}$ to denote induced
{\em update}.  (However, it will be important  to distinguish
the two uses, and so we speak of the program model $\hat{\alpha}$
and also the update $\hat{\alpha}$.)
\end{definition}

\begin{theorem}
Let $\alpha\in\Omega$.
Then there  is a
 signature-based program model
$(\bDelta,\delta_i,\vec{\propositionpsi})$
such that 
\begin{enumerate}
\item As  updates,  
$\semanticsnm{\alpha}\sim(\bDelta,\delta_i,\vec{\propositionpsi})$.
\item As propositions, $ \Presemantic(\alpha) = \propositionpsi_i$.
\item As program models,  $(\bDelta,\delta_i,\vec{\propositionpsi})
\sim \hat{\alpha}$.
\end{enumerate}
In particular, as updates, $\semanticsnm{\alpha}\sim  \hat{\alpha}$.
\label{theorem-representation}
\end{theorem}

\label{section-theorem-sam}

\begin{proof}
By induction on $\alpha$.
For $\skipp$, we take
$(\bDelta,\delta_i,\vec{\propositionpsi})$
to be $\skippaction$  from
Section~\ref{section-operations-on-program-models}. Our semantics of
$\skipp$ is exactly the first point above, and the second follows from
the definition  of  $\Pre$ earlier in this section.
The last point comes from the fact that the only 
arrows from $\skipp$ in $\bOmega$ are $\skipp\arrowA\skipp$
for all $A$.

For $\crash$, we define
$(\bDelta,\delta_i,\vec{\propositionpsi})$
to be $\actionzero$, 
the empty program model.  This is not exactly what
we want here.  Instead, consider 
  a one-action set $\set{\sigma}$ with no $A$-arrows 
for any agent $A$, $\pre(\sigma) = \falseproposition$.
This gives a signature-based action model whose
induced update is equivalent to $\crash$.
Our definition of $\Pre$ is that $\Pre(\crash) = \false$.
And $\crash$ is not the source of any arrows in $\Omega$.

We next consider an action $\sigma_i\vec{\psi}$.
We take $\bDelta$ to be $\bSigma$, the action signature
with which we are working in this lemma.
We also take $\delta_i$ to be $\sigma_i$.
And for each $j$ in the fixed enumeration of 
$\Sigma$, we take $\propositionpsi_j = \semanticsnm{\psi_j}$.
Then the result follows from the definitions of
the semantics of $\lang(\bSigma)$ and of $\Pre$.
The bisimulation relating 
$(\bSigma,\sigma_i,\vec{\propositionpsi})$ and
 $(\bOmega,\sigma_i\vec{\psi})$ 
is the relation 
$\set{(\sigma_j,\sigma_j\vec{\psi}): 1\leq j\leq n}$.
Here $n$ is the number of elements in $\Sigma$.

Finally, assume our result for $\alpha$ and $\beta$.
Let $(\Delta,\delta_i,\vec{\propositionpsi})$
and $(\Delta',\delta'_i,\vec{\propositionpsi'})$
satisfy our lemma for $\alpha$ and $\beta$.
Then by Proposition~\ref{prop-update-composition-equivalence},
\begin{equation}
\semanticsnm{\alpha;\beta}
\quad  \sim\quad
(\Delta,\delta_i,\vec{\propositionpsi});
 (\Delta',\delta'_i,\vec{\propositionpsi'}).
\label{eq-updatecomp}
\end{equation}
This is not quite what we want, since on the right we do not  
have a signature-based program model.
Let the enumeration of $\Delta$ be $\delta_1,\ldots,  \delta_n$,
and let  the enumeration of $\Delta'$ be $\delta'_1,\ldots,  \delta'_m$.
Let $$\map{f}{\set{1,\ldots,  nm}}{\set{1,\ldots,n}\times\set{1,\ldots,m}}$$
be a bijection.  
Let $g$ and $h$ be such that $f= g\times  h$.
 Let $\Delta''=\Delta\times\Delta'$,
and enumerate this set by $\delta''_{k} = (\delta_{g(k)},\delta'_{h(k)})$.
For $1\leq k\leq nm$, let
$$\propositionchi_k \quadeq
\pair{(\bDelta,\delta_{g(k)},\vec{\propositionpsi})}
   \Presemantic(\delta'_{h(k)}).$$ 
Let $k^*$ be such
that
$f(k^*) = (i,i')$. For our lemma, we consider
 $(\bDelta'',\delta''_{k^*}, \vec{\propositionchi})$.
There is an obvious bisimulation between this model  
and  the update on the right of
equation (\ref{eq-updatecomp}) above.
For the second assertion of  our lemma,
$$\begin{array}{lcll}
\Presemantic(\alpha;\beta) & \quadeq & 
\semanticsnm{\pair{\alpha}\Pre(\beta)}  \\
& \quadeq &  \pair{\semanticsnm{\alpha}} \semanticsnm{\Pre(\beta)} \\
& \quadeq &   \pair{(\bDelta,\delta_i,\propositionpsi_i)} \propositionpsi'_j 
& \mbox{(*)} \\
& \quadeq &   \propositionchi_{k^*}\\
\end{array}$$
In the equivalence marked $(*)$, we used Proposition~\ref{proposition-preserve-sim}.

We turn to the last part.
Let $R$ and $S$ be  bisimulations showing 
 $(\bDelta,\delta_i,\vec{\propositionpsi})
\sim (\bOmega,\alpha)$
and
 $(\bDelta',\delta_j,\vec{\propositionpsi'})
\sim (\bOmega,\beta)$.
Then the following bisimulation, called $T$,   shows that 
$(\bDelta'',\delta''_{k^*}, \vec{\propositionchi})\sim
(\bOmega,\alpha;\beta)$:
$$(\delta_c,\delta'_d) \ T \ \rho;\eta
\quadiff
\delta_c\ R\ \rho \mbox{ and }
\delta'_d\ S \ \eta\ .
$$

This completes our induction.
In  the last point of this result,
the assertion that
 $\semanticsnm{\alpha}\sim  \hat{\alpha}$, we use
Proposition~\ref{proposition-equivalence-action-updates} and parts (1) and (3)
of the present result.
\end{proof}

\subsection{Soundness of the Action Rule}

As we have mentioned, the Action Rule is the main rule in the
system for $\lang_1(\bSigma)$ that goes beyond the rules of
$\lang_0(\bSigma)$.  It is an induction rule which in the
setting of epistemic  logic  allows us
to prove that common knowledge obtains after various actions.
Here  again is our statement of the Action Rule:

Let $\alpha$ be a  simple action.
Let $\psi$ be a  sentence, and let $\CC$ be a set
of agents. Let there be sentences  $\chi_{\beta}$ for
all $\beta$ such that
 $\alpha\rightarrow_{\CC}^* \beta$
(including $\alpha$ itself),
and such that
\begin{enumerate}
\item $\proves\chi_{\beta}\iif [\beta]\psi$.
\item
If  $A\in\CC$ and  $\beta\arrowA\justplain{\gamma}$,
then $\proves (\chi_{\beta} \andd\Pre(\beta) )
\iif \necc_A \chi_{\gamma}$.
\end{enumerate}
 {} \noindent From these assumptions, infer
 $\proves\chi_{\alpha}\iif [\alpha]\necc_{\CC}^*\psi$.

\begin{remark}
We use $\arrow_{\CC}^*$ as an abbreviation for the 
reflexive and transitive closure
of the relation $\bigcup_{A\in C} \arrowA$.
Recall 
from Proposition~\ref{proposition-locallyfinite} 
that there are
only finitely many
$\beta$ such that
 $\alpha\rightarrow_{\CC}^* \beta$.
So even though the Action Rule might look like it
takes infinitely many premises, it really only takes finitely many.
\end{remark}

\begin{remark} It is possible  to drop the Composition Axiom
in favor of a more involved version of the Action Rule.
The point is we shall later introduce normal forms for
sentences, and using 
 the Composition Axiom will greatly simplify these normal
forms. 
Moreover, adding the
Composition Axiom leads to shorter proofs.
 The Action Rule here is geared towards those normal forms.
So if we were to drop  the Composition Axiom, we would
need a stronger, more complicated,
 formulation of the Action Rule, one  which involved
{\em sequences} of actions.  It is not terribly difficult
to formulate such a rule, and completeness can be obtained
by an elaboration of the work  which we shall do.
\end{remark}

\begin{lemma}
  $s\in\semantics{\pair{\alpha}\poss_{\CC}^*\phi}{\bS}$
iff
there is a  sequence of states from $S$
$$s \ = \ s_0\quad\rightarrow_{A_1}\quad s_1\quad\rightarrow_{A_2}\quad
    \cdots\quad \rightarrow_{A_{k-1}} \quad
s_{k-1} \quad\rightarrow_{A_k} \quad s_k$$
where $k\geq 0$,
and also a sequence of actions of the same  length $k$,
$$\alpha \ = \ \alpha_0\quad\rightarrow_{A_1}\quad
   \alpha_1\quad\rightarrow_{A_2}\quad
    \cdots\quad \rightarrow_{A_{k-1}} \quad
\alpha_{k-1} \quad\rightarrow_{A_k} \quad
\alpha_k$$
such that
 $A_i\in \CC$ and $s_i\in \presemantic(\alpha_i)_{\bS}$
 for all $0\leq i < k$, and
$s_k\in\semantics{\pair{\alpha_k}\phi}{\bS}$.
\label{proposition-reduction-diamond-star}
\end{lemma}

\begin{remark}   The case $k=0$ just says that
$s\in\semantics{\pair{\alpha}\poss_{\CC}^*\phi}{\bS}$
is implied by
$s\in\semantics{\pair{\alpha}\phi}{\bS}$.
\end{remark}

\begin{proof}
The following are equivalent:
\begin{enumerate}
\item $s\in \semantics{\pair{\alpha}\poss_{\CC}^*\phi}{\bS}$.
\label{uuo}
\item $s\in (\pair{\semanticsnm{\alpha}}\semanticsnm{\poss_{\CC}^*\phi})_{\bS}$.
\label{uuv}
\item $s\in (\pair{\hat{\alpha}}\semanticsnm{\poss_{\CC}^*\phi})_{\bS}$.
\label{uuhat}
\item $s\in \presemantic(\alpha)_{\bS}$, and
$(s,\alpha)\in \semantics{\poss_{\CC}^*\phi}{\bS\otimes\bOmega}$.
\label{uuoo}
\item
 $s\in  \presemantic(\alpha)_{\bS}$, and for some $k\geq 0$
there is a sequence in $\bS\otimes\bOmega$,
$$(s,\alpha) \ = \ t_0\quad\rightarrow_{A_1}\quad
t_1\quad\rightarrow_{A_2}\quad
    \cdots\quad \rightarrow_{A_{k-1}} \quad
t_{k-1} \quad\rightarrow_{A_k} \quad t_k$$
such that
$A_i\in \CC$ and $t_k\in \semantics{\phi}{\bS\otimes\bOmega}$.
\label{prds-3}
\item There are sequences of $s=s_0,\ldots, s_k$ and
$\alpha_0,\ldots,\alpha_k$ as in the statement of this lemma.
\label{prds-4}
\end{enumerate}
The first equivalence is by the semantics of $\lang(\bSigma)$.
The equivalence (\ref{uuv})$\Longleftrightarrow$(\ref{uuhat}) uses
the equality $\semanticsnm{\alpha} =\hat{\alpha}$  
which we saw in the concluding
statement of Theorem~\ref{theorem-representation}
in
Section~\ref{section-theorem-sam}, and also 
Proposition~\ref{proposition-preserve-sim}.
Equivalence  (\ref{uuhat})$\Longleftrightarrow$(\ref{uuoo})
uses our overall definitions and the structure  of $\bOmega$
as an action model.
 (\ref{uuoo})$\Longleftrightarrow$(\ref{prds-3}) is just the semantics of $\poss_{\CC}^*$.
Finally, the equivalence
 (\ref{prds-3})$\Longleftrightarrow$ (\ref{prds-4}) again uses 
the conclusion of Theorem~\ref{theorem-representation}.
That is, for all $i$, $(s_i,\alpha)\in \bS\otimes\bOmega$ iff
$s_i\in \presemantic(\alpha_i)_{\bS}$.
 \end{proof}

\begin{proposition} The Action Rule is sound.
\label{lemma-soundness-action}
\end{proposition}

\begin{proof}
Assume that
$s\in\semantics{\chi_{\alpha}}{\bS}$ but also
$s\in\semantics{\pair{\alpha}\poss_{\CC}^*\nott\psi}{\bS}$.
According to Lemma~\ref{proposition-reduction-diamond-star},
there is a labeled  sequence of states from $S$
$$s \ = \ s_0\quad\rightarrow_{A_1}\quad s_1\quad\rightarrow_{A_2}\quad
    \cdots\quad \rightarrow_{A_{k-1}} \quad
s_{k-1} \quad\rightarrow_{A_k} \quad s_k$$
where $k\geq 0$ and each $A_i\in \CC$,
and also a sequence of  actions of length $k$, with the same labels,
$$\alpha \ = \ \alpha_0\quad\rightarrow_{A_1}\quad
   \alpha_1\quad\rightarrow_{A_2}\quad
    \cdots\quad \rightarrow_{A_{k-1}} \quad
\alpha_{k-1} \quad\rightarrow_{A_k} \quad
\alpha_k$$
such that $s_i\in \Presemantic(\alpha_i)_{\bS}$
 for all $0\leq  i < k$, and
 $s_k\in\semantics{\pair{\alpha_k}\nott\psi}{\bS}$.
 If $k=0$,
we would have
 $s\in\semantics{\pair{\alpha}\nott\psi}{\bS}$.
But  one of the assumptions in the statement of the Action Rule is that
$\proves \chi_{\alpha}\iif [\alpha]\psi$.
So we would have
 $s\in\semantics{[\alpha]\psi}{\bS}$.
This would be a contradiction.

Now we argue the case $k>0$.
We show by induction on $1\leq i\leq k$ that
$s_i\in\semantics{\chi_{\alpha_i}}{\bS}$.  The case $i=0$ is the opening 
 assumptions of this proof.
Assume that $s_i\in\semantics{\chi_{\alpha_i}}{\bS}$.
By hypothesis, $s_i\in \presemantic(\alpha_i)_{\bS}$.  
 In view of
the second assumption in the Action Rule, $s_i\in
\semantics{\necc_{A_{i+1}}\chi_{\alpha_{i+1}}}{\bS}$. Hence 
$s_{i+1}\in \semantics{\chi_{\alpha_{i+1}}}{\bS}$.
This completes our induction.

In particular, $s_k\in\semantics{\chi_{\alpha_k}}{\bS}$.
Using again the first assumption  in  the Action Rule, we have
$s_k\in\semantics{[\alpha_k]\psi}{\bS}$.
This is a contradiction.
\end{proof}

\subsection{Syntactic Facts}

At this point, we have presented the semantic facts which we need
concerning $\lang_1(\bSigma)$. 
These include the soundness of the logical system
for validity.  To prove completeness,
we also need some syntactic facts.  
We could have presented this section earlier, but since
it leans on the Action Rule, we have delayed
it  until establishing the soundness of that
rule had been established.  

In this section,
$\alpha$, $\beta$, etc.~denote simple actions in $\lang_1(\bSigma)$.

\paragraph{A stronger form  of the Action-Knowledge Axiom}
At this point, we
 establish a stronger form  of  the Action-Knowledge
Axiom.
In Figure~\ref{figure-logical-system} in Section~\ref{section-langzero},
this axiom was stated only for basic actions.
The strengthenings here is to {\em simple\/} actions,
that is, to compositions of basic actions, $\skipp$,
and $\crash$.  
\rem{We only
prove the stronger forms which we actually will use in
later sections of this paper.
}

\rem{not needed! uphere
\begin{lemma} Let $\alpha$ be a 
simple action.
 Then there is a simple action $\althat{\alpha}$ of
the form $\skipp$, $\crash$, or 
$$  \sigma_1\vec{\psi^1}  \then  (\sigma_2\vec{\psi^2} \then
( 
\cdots \then   \sigma_r\vec{\psi^r}  )) 
$$
such that 
\begin{enumerate}
\item For all $\phi$, $\proves [\alpha]\phi \iiff [\althat{\alpha}]\phi$.
\item
$\proves\Pre(\alpha) \iiff \Pre(\althat{\alpha})$.  
\end{enumerate}
Moreover,
$$\set{\althat{\beta}
 :\alpha\arrowA \beta} \quadeq \set{\gamma: \althat{\alpha}\arrowA\gamma}.$$
\label{lemma-seebelow}
\end{lemma}

{\bf I'm in the middle of thinking about this proof and the issues
Slawek raises about it.}

\begin{proof}
By induction on the length of $\alpha$.  
If $\alpha$ is basic, $\skipp$, or $\crash$, we are done.
Otherwise, write $\alpha$ as $\beta\then\gamma$.
Our first case is when $\beta$ is  basic, $\skipp$, or $\crash$.
By induction hypothesis, let $\althat{\gamma}$ be such that
the assertions in this lemma hold  for $\gamma$ and  $\althat{\gamma}$.
Let $\althat{\alpha}$ be $\beta\then\althat{\gamma}$.
For each $\phi$, $[\gamma]\phi\iiff [\althat{\gamma}]\phi$.
So  $[\beta][\gamma]\phi\iiff [\beta][\althat{\gamma}]\phi$.
Then using the Composition Axiom, we see that 
$[\alpha]\phi \iiff [\althat{\alpha}]\phi$.
Also,
$\Pre(\alpha) \equiv
\Pre(\beta\then\gamma) \equiv
\pair{\beta}\Pre(\beta\then\gamma)
\equiv
\pair{\beta}\Pre(\beta\then\althat{\gamma})
\equiv
\Pre(\althat{\alpha})$.

Our second case is when $\beta$ itself is a composition.
Let $\beta =\beta_1\then\beta_2$.
Then 

\end{proof}

{\bf old proof below}
\begin{proof}
$\althat{\alpha}$ is a re-association of the simple actions in $\alpha$.
We can check  by induction on $\alpha$ that  $\proves [\alpha]\phi \iiff
[\althat{\alpha}]\phi$ for all $\phi$.
The fact that $\proves\Pre(\alpha) \iiff \Pre(\althat{\alpha})$ comes from 
Lemma~\ref{lemma-towards-composition-admissibility}.
\end{proof}
} %% see uphere

\rem{
\begin{lemma}
The Atomic Permanence
Axiom is provable for all simple actions
$\alpha$:
$$
%\begin{equation}
 \proves [\alpha]p \iiff (\Pre(\alpha)\iif p).
%\label{eq-pfn}
%\end{equation}
$$
\label{lemma-AP-general}
\end{lemma}

\begin{lemma}
The Partial Functionality
Axiom is provable for all simple actions
$\alpha$:
$$
%%\begin{equation}
 \proves [\alpha]\nott\chi \iiff (\Pre(\alpha)\iif \nott[\alpha]\chi)
%%\label{eq-pfn}
%%\end{equation}
$$\label{lemma-pf-general}
\end{lemma}

The proof of Lemmas~\ref{lemma-AP-general}
and~\ref{lemma-pf-general} are similar to what we
see in Lemma~\ref{lemma-aka-general}
below.  Since we do not need
 Lemmas~\ref{lemma-AP-general}
and~\ref{lemma-pf-general}, we omit their proofs.
}

\rem{
\begin{proof}
By induction on the length of $\alpha$.
If the length of $\alpha$ is $0$, $\alpha$ is $\skipp$
or $\crash$.  In the first case, $\Pre(\alpha) = \true$, 
and with
 the Skip Axiom,
equation (\ref{eq-pfn}) reduces to the propositional tautology
$\nott\chi \iiff (\true \iif \nott\chi)$.
In the second case,   $\Pre(\alpha) = \false$, 
and using the Crash Axiom, 
 (\ref{eq-pfn}) is the propositional tautology
$\true \iiff (\false \iif \nott\true)$.

If the length is $1$, then we  have the Partial Functionality Axiom
of Figure~\ref{figure-logical-system}.
 So assume our lemma for  actions $\alpha$ of
length $n$.   We may replace $\alpha$ by $\althat{\alpha}$ in Lemma~\ref{lemma-seebelow}
if necessary, to assume that composition in $\alpha$ associates to the right.
(That is,   (\ref{eq-aka-general}) for $\althat{\alpha}$ implies
(\ref{eq-aka-general}) for  $\alpha$, as implied
by the various assertions in Lemma~\ref{lemma-seebelow}.)
 We therefore 
show that for all actions $\sigma_i \vec{\psi}$,
\begin{equation} 
 \proves [\sigma_i \vec{\psi}\then\alpha]\nott\chi
\iiff   ( \Pre(\sigma_i \vec{\psi}\then\alpha) \iif
\nott [\sigma_i\vec{\psi}\then\beta]\chi)
\label{eq-pfn-desired}
\end{equation}
For this, we start with
equation (\ref{eq-pfn}) and apply necessitation
and some modal reasoning to get
$$ \proves  [\sigma_i\vec{\psi}][\alpha]\nott\chi \iiff 
( [\sigma_i\vec{\psi}]\Pre(\alpha)\iif  
[\sigma_i\vec{\psi}]\nott[\alpha]\chi).
$$
Now by the Partial Functionality Axiom for
basic actions, we
have 
$$\proves  [\sigma_i\vec{\psi}]\nott[\alpha]\chi 
\iiff (\pre(\sigma_i\vec{\psi}) \iif 
\nott[\sigma_i\vec{\psi}] [\alpha]\chi).
$$
And so we have
$$ \proves  [\sigma_i\vec{\psi}][\alpha]\nott\chi \iiff 
(( [\sigma_i\vec{\psi}]\Pre(\alpha)\andd
 \pre(\sigma_i\vec{\psi}))\iif\nott [\sigma_i\vec{\psi}][\alpha] 
\chi)
$$
Then we use the Composition Axiom and the definition of 
$\Pre(\sigma \vec{\psi}\then\alpha)$,
followed by the Partial Functionality Axiom for basic
actions, to get
the desired assertion in equation (\ref{eq-pfn-desired}).
\end{proof}
}

\label{section-stronger-forms-axioms}

\begin{lemma}
The Action-Knowledge Axiom is provable for all simple actions
$\alpha$:
\begin{equation}  
\proves [\alpha]\necc_A\phi
\iiff   ( \Pre(\alpha) \iif
\bigwedge\set{
 \necc_A [\beta]\phi :
\alpha\arrowA\beta \mbox{ in $\bOmega$}})
\label{eq-aka-general}
\end{equation}
\label{lemma-aka-general}
\end{lemma}

\begin{proof}
By induction on   $\alpha$.
 If $\alpha$ is $\skipp$, $\Pre(\alpha) = \true$, 
the only $\beta$ with $\alpha\arrowA \beta$ is $\skipp$,
and equivalence (\ref{eq-aka-general}) reads:
$$\proves [\skipp]\necc_A\phi
\iiff   ( \true \iif
 \necc_A [\skipp]\phi)
$$
And this is an easy consequence of propositional
reasoning, the Skip Axiom,
and $\necc_A$-necessitation.
If $\alpha$ is $\crash$, $\Pre(\alpha) = \false$, and there are no $\beta$
such that $\alpha\arrowA \beta$.
Equivalence (\ref{eq-aka-general}) then reads:
$$\proves [\crash]\necc_A\phi
\iiff   ( \false \iif \true)
.
$$  By the 
Crash Axiom and modal reasoning, $\proves  [\crash]\necc_A\phi$;
and so the assertion just above holds.

If $\alpha$ is of the form $\sigma_i\vec{\psi}$,
 then we simply have the Action-Knowledge Axiom in
the form we know it.  

So assume our lemma for   $\alpha'$ and $\alpha$;
we prove it for $\alpha';\alpha$.
 We 
show that  
\begin{equation} 
\proves [\alpha'\then\alpha]\necc_A\phi
\iiff   ( \Pre(\alpha'\then\alpha) \iif
\bigwedge\set{
 \necc_A [\beta'\then\beta]\phi :
\alpha'\then\alpha\arrowA\beta'\then\beta
\mbox{ in $\bOmega$  } })
\label{eq-aka-desired}
\end{equation}
We start with the equivalence in
(\ref{eq-aka-general}), use $[\alpha']$-necessitation
and normality, and get
\begin{equation} 
\proves [\alpha'][\alpha]\necc_A\phi
\iiff   ( [\alpha']\Pre(\alpha) \iif
\bigwedge\set{
[\alpha'] \necc_A [\beta]\phi :
\alpha\arrowA\beta \mbox{ in $\bOmega$}})
\label{aka-andget}
\end{equation}
By the induction hypothesis on $\alpha'$ and several   uses of the
Composition Axiom, we have that for all $\beta$,
$$\proves
[\alpha'] \necc_A [\beta]\phi  \iiff (\Pre(\alpha')\iif 
\bigwedge\set{ \necc_A  [\beta'\then\beta]\phi :
\alpha'\arrowA\beta' \mbox{ in $\bOmega$}}).
$$
The Composition Axiom and (\ref{aka-andget}) lead  to the provable equivalence of 
$[\alpha'\then\alpha]\necc_A\phi$ and
$$  
%\proves [\sigma \vec{\psi}\then\alpha]\necc_A\phi \iiff    
[\alpha']\Pre(\alpha) \iif
( \Pre(\alpha')\iif
\bigwedge\set{
\necc_A  [\beta'\then\beta]\phi :
\alpha\arrowA\beta \mbox{ and } \alpha'\arrowA \beta'
\mbox{ in $\bOmega$ }}).
$$
But $\proves\Pre(\alpha'\then \alpha) \iiff 
\Pre(\alpha')\andd [\alpha']\Pre(\alpha)$.
In addition, we have a general fact 
$\alpha'\then\alpha\arrowA\beta'\then\beta$ iff  
$\alpha'\arrowA\beta'$
and $\alpha\arrowA\beta$.
Using  these observations and some propositional reasoning, we get
(\ref{eq-aka-desired}), as desired. 
\end{proof}

%%{\bf old proof hidden below!}
\rem{\begin{proof}
By induction on the length of $\alpha$.
If the length of $\alpha$ is $0$, $\alpha$ is $\skipp$
or $\crash$.  In the first case, $\Pre(\alpha) = \true$, 
the only $\beta$ with $\alpha\arrowA \beta$ is $\skipp$,
and equivalence (\ref{eq-aka-general}) reads:
$$\proves [\skipp]\necc_A\phi
\iiff   ( \true \iif
 \necc_A [\skipp]\phi)
$$
And this is an easy consequence of propositional
reasoning, the Skip Axiom,
and $\necc_A$-necessitation.
In the second case, $\Pre(\alpha) = \false$, and there are no $\beta$
such that $\alpha\arrowA \beta$.
Equivalence (\ref{eq-aka-general}) then reads:
$$\proves [\crash]\necc_A\phi
\iiff   ( \false \iif \true)
.
$$  By the 
Crash Axiom and modal reasoning, $\proves  [\crash]\necc_A\phi$;
and so the assertion just above holds.
This takes care of the case when the length of $\alpha$ is $0$.

If the length is $1$, then we simply have the Action-Knowledge Axiom in
the form we know it.  So assume our lemma for  actions $\alpha$ of
length $n$.   We may replace $\alpha$ by $\althat{\alpha}$ in Lemma~\ref{lemma-seebelow}
if necessary, to assume that composition in $\alpha$ associates to the right.
(That is,   (\ref{eq-aka-general}) for $\althat{\alpha}$ implies
(\ref{eq-aka-general}) for  $\alpha$, as implied
by the various assertions in Lemma~\ref{lemma-seebelow}.)
 We therefore 
show that for all actions $\sigma \vec{\psi}$,
\begin{equation} 
\proves [\sigma \vec{\psi}\then\alpha]\necc_A\phi
\iiff   ( \Pre(\sigma \vec{\psi}\then\alpha) \iif
\bigwedge\set{
 \necc_A [\tau\vec{\psi}\then\beta]\phi :
\alpha\arrowA\beta \mbox{ in $\bOmega$ and } \sigma\arrowA \tau \mbox{ in $\bSigma$}})
\label{eq-aka-desired}
\end{equation}
We start with the equivalence in
(\ref{eq-aka-general}), use $[\sigma\vec{\psi}]$-necessitation
and normality, and get
\begin{equation} 
\proves [\sigma \vec{\psi}][\alpha]\necc_A\phi
\iiff   ( [\sigma \vec{\psi}]\Pre(\alpha) \iif
\bigwedge\set{
[\sigma \vec{\psi}] \necc_A [\beta]\phi :
\alpha\arrowA\beta \mbox{ in $\bOmega$}})
\label{aka-andget}
\end{equation}
The Composition Axiom tells us that $\proves  [\sigma \vec{\psi}][\alpha]\necc_A\phi \iiff 
 [\sigma \vec{\psi}\then\alpha]\necc_A\phi$.
And the Action-Knowledge Axiom that we already know tells us that for each $\beta$
$$\proves
[\sigma \vec{\psi}] \necc_A [\beta]\phi  \iiff (\Pre(\sigma\vec{\psi})\iif 
\bigwedge\set{ \necc_A  [\tau \vec{\psi}][\beta]\phi :\sigma\arrowA 
\tau \mbox{ in $\bSigma$}}).
$$
So (\ref{aka-andget}) leads to the provable equivalence of 
$[\sigma \vec{\psi}\then\alpha]\necc_A\phi$ and
$$  
%\proves [\sigma \vec{\psi}\then\alpha]\necc_A\phi \iiff    
[\sigma \vec{\psi}]\Pre(\alpha) \iif
( \Pre(\sigma\vec{\psi})\iif
\bigwedge\set{
\necc_A  [\tau \vec{\psi}\then\beta]\phi :
\alpha\arrowA\beta \mbox{ in $\bOmega$ and } \sigma\arrowA \tau \mbox{ in $\bSigma$}})
$$
But $\Pre(\sigma\vec{\psi}\then \alpha) = \Pre(\sigma\vec{\psi})\andd [\sigma\vec{\psi}]\Pre(\alpha)$.
Using  this and some propositional reasoning, we get
(\ref{eq-aka-desired}), as desired.
\end{proof}
}

\begin{remark}
There are no sound versions of the 
 Action-Knowledge Axiom for  actions containing $\union$.
For example  
let $\Agents$ be a singleton, hence omitted from the notation.
 Consider $\bSigma = \set{\sigma,\tau}$ this
as an action signature with the enumeration $\sigma$, $\tau$,
and with  (say) the arrows $\sigma\rightarrow \tau$
and $\tau\rightarrow\sigma$.
Note that by Atomic Permanence
$[\sigma\ p  q]r$ is equivalent to $p\iif r$,
$[\tau\ p  q]r$ is equivalent to $q\iif r$,
by Action-Knowledge and Atomic Permanence
$[\sigma\ p  q]\necc r$ is equivalent to $p\iif \necc(q\iif r)$,
and
$[\tau\ p  q]\necc r$ is equivalent to $q\iif \necc(p\iif r)$,
Consider also
the sentence
 $[(\sigma\ p q)\union(\tau\ p q)]\necc r$.
By the Choice Axiom, it is equivalent to 
\begin{equation}
(p\iif \necc(q\iif r))\andd
(q\iif \necc(p\iif r)).
\label{notequivalent}
\end{equation}
Now we did not define $\Pre( (\sigma\ p q)\union(\tau\ p q))$,
but the most reasonable choice is $p\orr q$.
And we did not define the accessibility structure of
actions containing $\union$, but in this case the most
likely choice is to have 
$(\sigma\ p q)\union(\tau\ p q)$ relate to itself and nothing else.
But then when we write out the right-hand side
of the equivalence (\ref{eq-aka-general}), we get
$$(p\orr q) \iif (\necc(p\iif r) \andd \necc(q\iif r)).$$
Clearly this is not equivalent to (\ref{notequivalent}).
 Even if one
were to change $\Pre( (\sigma\ p q)\union(\tau\ p q))$ to,
say $\true$, or to $p\andd q$,  we still would not have
an equivalence: (\ref{notequivalent}) is stronger.

\paragraph{A stronger form of the Partial Functionality Axiom}

We shall also need the following result, a version of the 
Partial Functionality Axiom:

\begin{lemma}
$\proves \pair{\alpha}\phi \iiff 
\Pre(\alpha) \andd [\alpha]\phi$.
\label{lemma-pfa-general}
\end{lemma}

\begin{proof}
This is an induction, much like the proof of Lemma~\ref{lemma-aka-general}
above.
\end{proof}

\paragraph{Syntactic bisimulation}

Our next set of results 
 pertains to a syntactic notion of 
 action equivalence, one with a  bisimulation-like flavor.

\begin{definition}
A {\em syntactic bisimulation\/} 
is a relation $R$ on the set $\Omega$ of simple actions
of $\lang(\bSigma)$ with the following property:
if $\alpha\ R\ \beta$, then
\begin{enumerate}
\item $\proves \Pre(\alpha) \iiff \Pre(\beta)$.
\item For all $\alpha'$ and $A$ such that 
$\alpha\arrowA \alpha'$, there is some $\beta'$
such that $\beta \arrowA \beta'$ and $\alpha'\ R\ \beta'$.
\item For all  $\beta'$  and $A$
 such that $\beta \arrowA \beta'$,
 there is some $\alpha'$ 
such that $\alpha\arrowA \alpha'$ and $\alpha'\ R\ \beta'$.
\end{enumerate}
We say that $\alpha$ and $\beta$ are {\em provably equivalent\/} if
there is some syntactic bisimulation relating them.
We write $\alpha\equiv\beta$ in this case.

For sentences, we
write $\phi\equiv \psi$ and
say that   $\phi$ and $\psi$ are {\em equivalent\/}
if $\proves\phi\iiff \psi$.
\end{definition}

\rem{%% let's see if we need this
\begin{proposition}
Let $\alpha$ be a simple action, let $\bS$ be a state model.
For $s\in S$ and for $t'$ the  following are equivalent:
\begin{enumerate}
\item There is some $t$ such that $s\ \semanticsnm{\alpha}_{\bS}\ t$
and $t \arrowA t'$.
\item  There are $s'$ and $\alpha'$ such that
$s\arrowA s'$, $\alpha\arrowA \alpha'$, and $s'\ \semanticsnm{\alpha'}_{\bS}\
t' $.
\end{enumerate} 
$$
\xymatrix{ s \ar[r]^{\semanticsnm{\alpha}_{\bS}}  \ar[d]_{A} & t \ar[d]^{A} \\
s' \ar[r]_{\semanticsnm{\alpha'}_{\bS}} & t'   \\
}
$$
\label{proposition-square-biggersystem}
\end{proposition}

\begin{proof}
By Theorem~\ref{theorem-representation}, Proposition~\ref{proposition-preserve-sim-dom},
and Proposition~\ref{proposition-square-semantic}.
\end{proof}
}

\rem{old proof below
\begin{proof}  
By induction on $\alpha$.  When $\alpha = \skipp$, the equivalence is immediate.
We consider the case when $\alpha = \sigma_i\vec{\psi}$.
Write $\actionalpha$ for $\semanticsnm{\alpha}$.
In one direction, when $s\ {\actionalpha}_{\bS}\ t$, $t$ must be $(s,\sigma_i)$;
and $t'$ is of the form $(s',\sigma_j)$ for some $s'$ and $\sigma_j$ 
such that $s\arrowA s'$ and
$\sigma_i\arrowA \sigma_j$.  We take $\alpha'$ to be $\sigma_j\vec{\psi}$.
The other direction of the equivalence is similar.

Finally, assume the result for $\alpha$ and $\beta$, and consider $\alpha\then\beta$.
First, assume that $s\ ({\actionalpha\then\actionbeta})_{\bS}\ t$
and $t \arrowA t'$.  Then there is some $u$ such that 
$s\ {\actionalpha}_{\bS}\ u \ {\actionbeta}_{\bS(\actionalpha)}\ t$.
By induction hypothesis, there are $u'$ and $\beta'$ such that
$u\arrowA u'$, $\beta\arrowA \beta'$, and $u'\ {\actionbeta'}_{\bS(\actionalpha)}\ t'
$, where  $\actionbeta'= \semanticsnm{\beta}$.
By induction hypothesis again, 
there are $s'$ and $\alpha'$ such that
$s\arrowA s'$, $\alpha\arrowA \alpha'$, and $s'\ {\actionalpha'}_{\bS}\ u' $,
where $\actionalpha' = \semanticsnm{\alpha'}$.
And now, $\alpha\then\beta \arrowA \alpha'\then\beta'$, and
$s'\ {(\actionalpha'\then\actionbeta')}_{\bS}\ t' $.
The other direction is similar.
\end{proof}
}

\rem{%% let's see if we need this
\begin{proposition}
For all  simple  actions $\alpha$ and all state models $\bS$,
$\semantics{\Pre(\alpha)}{\bS} = \dom\ \semanticsnm{\alpha}_{\bS}$.
That is, the domain of the relation $ \semanticsnm{\alpha}_{\bS}$ is the 
interpretation in $\bS$ of $\Pre(\alpha)$.
\label{proposition-domains}
\end{proposition}

\begin{proof}
By 
Theorem~\ref{theorem-representation}, Proposition~\ref{proposition-preserve-sim-dom},
and Proposition~\ref{prop-new46}.
\end{proof}
}

\rem{Old proof below
\begin{proof}
By induction on $\alpha$.   The cases of $\skipp$,
$\crash$, and the simple action expressions
are easy.  We check the composition case, so assume the result for $\alpha $ and $\beta$.
Then 
$$\begin{array}{lcl}
\semantics{\Pre(\alpha\then\beta)}{\bS}  & \quadeq &   
 \semantics{\pair{\alpha}\Pre(\beta)}{\bS}    \\
 & \quadeq &    \set{s\in S: 
\mbox{there is $t$ so
that  $s\ \semanticsnm{\alpha}_{\bS}\ t$ and $t\in \semantics{\Pre(\beta)}{\bS(\semanticsnm{\alpha})}$}} 
\\
 & \quadeq &     \set{s\in S: 
\mbox{there is $t$ so
that  $s\ \semanticsnm{\alpha}_{\bS}\ t$ and $t\in \dom\ {\semanticsnm{\beta}}_{\bS(\semanticsnm{\alpha})}$}}  \\
& \quadeq &    \dom\ ({\semanticsnm{\alpha}}_{\bS} \then \
{\semanticsnm{\beta}}_{\bS(\semanticsnm{\alpha})} )
\\ & \quadeq &    \dom\ \semanticsnm{\alpha\then\beta}_{\bS}   \\
\end{array}
$$
\end{proof}
}

\begin{lemma} $\proves\pair{\alpha}\true \iiff \Pre(\alpha)$.
\label{lemma-alpha-true}
\end{lemma}

\begin{proof}
By induction on $\alpha$.
For $\alpha = \skipp$ and
$\alpha=\crash$,  we use the Skip Axiom and the Crash Axiom,
respectively.
Here is the argument for $\alpha$ of the form $\sigma_i\vec{\psi}$.
First,  by necessitation we have $\proves[\sigma_i\vec{\psi}]\true$.
And by this and  Partial Functionality,
$\proves [\sigma_i\vec{\psi}]\nott\true\iiff  \nott\psi_i$.
So $\proves \nott[\sigma_i\vec{\psi}]\nott\true\iiff   \psi_i$.  That is,
 $\proves \pair{\sigma_i\vec{\psi}}\true \iiff
\psi_i$.

Finally, assume the result for $\beta$.
Then by normality and necessitation,
 $\proves\pair{\alpha}\pair{\beta}\true \iiff \pair{\alpha} \Pre(\beta)$.
So 
 $\proves\pair{\alpha}\pair{\beta}\true \iiff \Pre(\alpha\then\beta)$.
We conclude by showing as a general fact that
 $\proves\pair{\alpha}\pair{\beta}\phi \iiff \pair{\alpha\then\beta}\phi$.
For this, the  Composition  Axiom tells us that
$\proves [\alpha][\beta]\nott\phi  \iiff [\alpha\then\beta]\nott\phi$.
So   
$\proves \nott[\alpha]\nott \nott [\beta]\nott\phi  \iiff \nott[\alpha\then\beta]\nott\phi$.
Thus 
$\proves\pair{\alpha}\pair{\beta}\phi \iiff \pair{\alpha\then\beta}\phi$,
as desired.
\end{proof}

\begin{lemma}  The following monoid-type laws hold:
\begin{enumerate}
\item $\proves \Pre(\alpha\then\skipp) \iiff \Pre(\alpha)  \iiff \Pre(\skipp\then\alpha)$.
\item $\proves \Pre(\alpha\then(\beta\then\gamma))\iiff \Pre((\alpha\then\beta)\then\gamma)$.
\end{enumerate}
\label{lemma-towards-composition-admissibility}
\end{lemma}

%\begin{remark}
%Later we shall see that these turn into {\em admissible rules\/} for the
%equivalence of actions.  Lemma~\ref{lemma-towards-composition-admissibility}
%will be used to establish this, and also several other important results.
%\end{remark}

\begin{proof}  We use the definitions to calculate
$$\begin{array}{lcl}
\Pre(\alpha\then \skipp) & \quadeq & \pair{\alpha}\true \\
\Pre(\skipp\then \alpha) & \quadeq & \pair{\skipp}\Pre(\alpha).
\end{array}
$$
We obtain a formal proof using Lemma~\ref{lemma-alpha-true}
and
the Skip Axiom.
Turning to  the second law,
$$\begin{array}{lcl}
\Pre(\alpha\then(\beta\then\gamma)) & \quadeq & 
\pair{\alpha}\Pre(\beta\then\gamma)  \\
& \quadeq & \pair{\alpha}\pair{ \beta}\Pre(\gamma)  
\\
& \quadequiv & \pair{\alpha \then\beta}\Pre(\gamma)   
\\
& \quadeq & 
 \Pre((\alpha\then\beta)\then\gamma).\\
\end{array}
$$
The equivalence above which mentions  $\equiv$ uses
 the   Composition Axiom.
\end{proof}

\rem{
This last result leads to another concept.

\begin{definition} $\approx$ is the smallest relation  on $\bOmega$
containing the identity relation on $\bOmega$ and all instances of associativity,
$$ \alpha\then(\beta\then\gamma) \quad \approx\quad (\alpha\then\beta)\then\gamma,$$
 and  closed under substitution: $\alpha\approx \beta$ 
and $\gamma\approx\delta$ imply
$\alpha\then \gamma \approx\beta\then\delta$.
We say that simple actions $\alpha$ and $\alpha'$ 
{\em differ by associativity\/} if $\alpha\approx\alpha'$.\footnote{We could 
modify $\approx$ to include the instances of
laws such as $\alpha\then\skipp\approx\alpha$
$\skipp\then\alpha\approx\alpha$,
and $\crash\then\alpha\approx\crash$.   All our results below would hold for
these modifications.  However, we have no use for  larger relations
of this type.}
\end{definition}

\begin{proposition}
The relation $\approx$ of ``differing by associativity'' 
is a syntactic bisimulation equivalence on $\bOmega$.
That is, $\approx$ is an equivalence relation.  And
if $\alpha\approx \alpha'$,
then:
\begin{enumerate}
\item $\proves\Pre(\alpha)\iiff\Pre(\alpha')$.
\item
Whenever $\alpha\arrowA \beta$, there is some $\beta'$ 
such that $\alpha'\arrowA \beta'$ and 
$\alpha'\approx\beta'$.
\end{enumerate}
\label{prop-pre-prebisim}
\end{proposition}

\begin{proof}
(1) and (2) are shown by induction on the relation $\approx$.
Lemma~\ref{lemma-towards-composition-admissibility}
is used in (1), and the argument for (2) is direct.
The transitivity of $\approx$ is similar.
\end{proof}
}

%%%\subsection{Stronger forms of the axioms}
\label{section-syntactic-results}

\rem{
We check here that there is no sound version of 
the Action-Knowledge Axiom is  the actions contain $\union$.
The natural way to define $\Pre(\sigma\union\tau)$ is $\Pre(\sigma)\orr \Pre(\tau)$,
but even if we make other choices we'll see that we don't get a sound statement.

  Let $\Agents$ be a singleton, hence omitted from the notation.
 $\bSigma = \set{\sigma,\tau}$ with all  four arrows.  We consider this
as an action signature with the enumeration $\sigma$, $\tau$.
We have sentences like $[\sigma\ p,\nott p]p$, $[\sigma\ p,\nott p]\nott p$,
$[\tau\ p,\nott p]p$,  and $[\tau\ p,\nott p]\nott p$. 
It is not hard to work out the semantics of these:
 $[\sigma\ p,\nott p]p$ and  $[\tau\ p,\nott p]\nott p$
are equivalent to $\true$;  
 $[\sigma\ p,\nott p]\nott p$ and
$[\tau\ p,\nott p]p$ are equivalent to $\false$.

Consider $[\sigma\ p,\nott p +\tau\ p \nott p]\necc p$.

$$[\sigma\ p,\nott p +\tau\ p \nott p]\necc p
$$

with $\alpha\arrow\beta$, $\gamma\arrow\delta$, and no other arrows.
 }
\end{remark}

\rem{deleted!
\subsection{The main result on the canonical action model}

The main fact about the canonical action model is 
stated in Theorem~\ref{theorem-syntactic-action-model-preliminary} below.
At this point, we need only one more technical tool
before turning to
this result.

\begin{definition}
We define a map $\alpha\mapsto\sharp{\alpha}$,
$\phi\mapsto\sharp{\phi}$ from simple actions and simple sentences
to
the set  $N$ 
of natural numbers by the following recursion:
$$\begin{array}{lcl}
\sharp{\true} & \quadeq & 0 \\
\sharp{p} & \quadeq & 0 \\
\sharp{(\nott\phi)} & \quadeq & 1+ \sharp{\phi} \\
\sharp{(\phi\andd\psi)} & \quadeq & 1+ \max(\sharp{\phi},\sharp{\psi})\\
\sharp{(\necc_A\phi)} & \quadeq & 1+ \sharp{\phi} \\
\sharp{(\necc^*_{\BB}\phi)} & \quadeq & 1 + \sharp{\phi} \\  \\
\end{array}
\qquad
\begin{array}{lcl}
\sharp{([\alpha]\phi)} & \quadeq & 1 + \max(\sharp{\alpha},\sharp{\phi})\\
\sharp{\skipp} & \quadeq & 1 \\
\sharp{\crash} & \quadeq & 2 \\
\sharp{(\sigma_i\vec{\psi})} & \quadeq & 1 + \max_j\sharp{\psi_j} \\
\sharp{(\alpha\then\beta)} & \quadeq & 3 + \max(\sharp{\alpha},\sharp{\beta})\\
\end{array}
$$
\end{definition}

\begin{proposition}  Concerning the operation $\#$:
\begin{enumerate}
\item
For all simple actions $\alpha$, $\sharp{\Pre(\alpha)} < \sharp{\alpha}$.
\item  If $\phi$ and $\psi$ are simple sentences and $\phi$ is a
proper subsentence 
of $\psi$, then $\sharp{\phi} < \sharp{\psi}$.
\end{enumerate}
\label{proposition-sharps}
\end{proposition}

\begin{proof}
We prove the first part by induction on $\alpha$.
The cases of $\alpha = \skipp$, $\alpha = \crash$, or $\alpha$ of
the form $\sigma_i\vec{\psi}$ are easy.  
Assuming  that 
 $\sharp{\Pre(\alpha)} < \sharp{\alpha}$
and 
 $\sharp{\Pre(\beta)} < \sharp{\beta}$, we see that
$$\begin{array}{lcl}
\sharp{(\Pre(\alpha\then\beta))}  & \quadeq & 
\sharp{(\pair{\alpha}\Pre(\beta))}  \\
& \quadeq & 1 + \max(\sharp{\alpha}, 1 + \sharp{\Pre(\beta)}) \\
& \quad < \quad &  3+ \max(\sharp{\alpha},\sharp{\beta})\\
& \quadeq & \sharp{(\alpha\then\beta)} \\
\end{array}
$$
The second part is a similar induction.
\end{proof}

At this point, we have all the general results we need.
We take up the main issue of this section.

\begin{theorem}
 For every simple action $\alpha$, the programs $\semanticsnm{\alpha}$
 and $\hat{\alpha}$
are bisimilar. 
\end{theorem}

\noindent
Theorems~\ref{theorem-syntactic-action-model-preliminary} 
and~\ref{theorem-syntactic-action-model} below
are the main results of this section.
They figure in to 
Proposition~\ref{proposition-reduction-diamond-star} and
then into the soundness of the Action Rule, so they are 
crucial for our
work.
In reading the theorems, recall that
$$ \begin{array}{lcl}
([\hat{\alpha}]\semanticsnm{\phi})_{\bS} & \quadeq & 
\set{s\in S : \mbox{ if $(s,\alpha)\in S\otimes\Omega$,
then $(s,\alpha)\in \semantics{\phi}{\bS\otimes\bOmega}$}} \\
& \quadeq & 
\set{s\in S : \mbox{ if $s\in \semantics{\Pre(\alpha)}{\bS}$,
then $(s,\alpha)\in \semantics{\phi}{\bS\otimes\bOmega}$}} \\
\end{array}
$$
and also that
$$ \semantics{[\alpha]\phi}{\bS} \quadeq
 ([\semanticsnm{\alpha}]\semanticsnm{\phi})_{\bS}
\quadeq
 \set{s\in S: 
\mbox{if $s\ {\semanticsnm{\alpha}}_{\bS}\ t$, then $t\in
\semanticsnm{\phi}_{\actionmodel{\bS}{\semanticsnm{\alpha}}}$}}. 
$$

\begin{theorem}  For all simple  actions $\alpha$
and simple sentences $\phi$, 
$$\semanticsnm{[\alpha]\phi}
\quadeq [\hat{\alpha}]\semanticsnm{\phi}.$$
That is, for  all $\phi$ and all state models
$\bS$,
\begin{equation}
%%\semantics{[\alpha]\phi}{\bS} \quadeq
\begin{array}{cl}
 & \set{s\in S: 
\mbox{if $s\ {\semanticsnm{\alpha}}_{\bS}\ t$, then $t\in
\semanticsnm{\phi}_{\actionmodel{\bS}{\semanticsnm{\alpha}}}$}}
\\
\quadeq  & \set{s\in S : \mbox{if $s\in \semantics{\Pre(\alpha)}{\bS}$,
then $(s,\alpha)\in \semantics{\phi}{\bS\otimes\bOmega}$}}.
\end{array}
\label{eq-sam}
\end{equation}
\label{theorem-syntactic-action-model-preliminary}
\end{theorem}

\label{section-theorem-sam}

\begin{proof} Write  $\propositionphi$ for $\semanticsnm{\phi}$.
We need to show that the following conditions are equivalent:
\begin{description}
\item{(a)}  If $s\ {\semanticsnm{\alpha}}_{\bS}\ t$, then
    $t\in\propositionphi_{\bS(\semanticsnm{\alpha})}$.
\item{(b)}  If $s\in \dom\ {\semanticsnm{\alpha}}_{\bS}$, then 
$(s,\alpha)\in \propositionphi_{\bS\otimes\bOmega}$.
\end{description}
We argue
by induction on $\sharp{\phi}$.
To strengthen the induction hypothesis, we also show that another
  equivalence 
(c)$\Longleftrightarrow$(d)
  holds for all
simple actions
 $\beta$, and all $(s,\alpha)\in \bS\otimes \bOmega$.
\begin{description}
\item{(c)} $(s,\alpha)\in \semantics{[\beta]\phi}{\bS\otimes \bOmega}$.
\item{(d)}  $(s,\alpha)\in\semantics{\Pre(\beta)}{\bS\otimes \bOmega}$
implies $(s,\alpha\then\beta)\in \semantics{\phi}{\bS\otimes\bOmega}$.
\end{description}
Finally, we also show at the same time the following assertion:
\begin{description}
\item{(e)} If $\alpha\approx\alpha'$, then 
$(s, \alpha) \in \semantics{\phi}{\bS\otimes\bOmega}$
iff $(s, \alpha') \in \semantics{\phi}{\bS\otimes\bOmega}$.
\end{description}
We argue all of this by induction on $\sharp{\phi}$.

\paragraph{$\sharp{\phi} = 0$}
 When $\phi$ is $\true$,
the verifications of all the equivalences are easy:
$\semanticsnm{\true}_{\bU} = U$ for all $U$, and
$\semanticsnm{[\beta]\true}$ has the same property.
 So we consider the
case of  $\phi$   an atomic sentence
(say) $p$.  Fix $\alpha$, $\bS$ and $s\in S$.
To see (a)$\Longleftrightarrow$(b), the following are equivalent
for all $t$:
\begin{enumerate}
\item If $s\ {\semanticsnm{\alpha}}_{\bS}\ t$, then
    $t\in\propositionp_{\bS(\semanticsnm{\alpha})}$.
\item If $s\in \dom\ {\semanticsnm{\alpha}}_{\bS}$, then
    $s\in\propositionp_{\bS}$.
\item   If $s\in \dom\ {\semanticsnm{\alpha}}_{\bS}$, then 
$(s,\alpha)\in \propositionp_{\bS\otimes\bOmega}$.
\end{enumerate}
The equivalence (1)$\Longleftrightarrow$(2) follows from 
Proposition~\ref{prop-actions-preserve-atomics}.

The equivalence (c)$\Longleftrightarrow$(d) in the atomic case
is checked as follows.
We assume $(s,\alpha)\in \bS\otimes\bOmega$.
We note the following equivalences:
\begin{enumerate}
\item  $(s,\alpha)\in \semantics{[\beta]p}{\bS\otimes \bOmega}$.
\item  If $(s,\alpha)\in\semantics{\Pre(\beta)}{\bS\otimes \bOmega}$,
then $(s,\alpha)\in \semantics{p}{\bS\otimes\bOmega}$.
\item  If $(s,\alpha)\in \semantics{\Pre(\beta)}{\bS\otimes \bOmega}$,
then $s\in \semantics{p}{\bS}$.
\item  If $(s,\alpha)\in\semantics{\Pre(\beta)}{\bS\otimes \bOmega}$,
then $(s,\alpha\then\beta)\in \semantics{p}{\bS\otimes\bOmega}$.
\end{enumerate}
 (1)$\Longleftrightarrow$(2) here is by 
Lemma~\ref{lemma-AP-general}. The other equivalences
are by the definition of 
$\bS\otimes\bOmega$ as an update product; see
equation (\ref{eq-new-valuation})
in Section~\ref{section-update-product-1}.

As for (e) in the atomic case, it follows from 
the definition of $\bS\otimes\bOmega$
and Proposition~\ref{prop-pre-prebisim} (1).
  This  result implies that $(s,\alpha)$ is an action
in $\bS\otimes\bOmega$ iff $(s,\alpha')$ also is one; hence 
our result.

\paragraph{The induction step}
For the rest of the proof, we argue by cases on the shape
of $\phi$.  That is, we consider sentences
$\nott\phi$, $\phi\andd\psi$, $\necc_A\phi$, $\necc^*_{\BB}\phi$,
and $[\gamma]\phi$.   In each case, we assume that for all
sentences $\psi$ whose value under $\#$ is strictly smaller than
the sentence in question, we have 
 (a)$\Longleftrightarrow$(b),  (b)$\Longleftrightarrow$(c),
and (e).

\paragraph{$\nott\phi$}  Recall that $\sharp{\phi} < \sharp{(\nott\phi)}$.  
So we assume our result for $\phi$ and prove it for $\nott\phi$.
 Fix $\alpha$, $\bS$ and $s\in S$.
For  (a)$\Longleftrightarrow$(b), we check that
 the following are equivalent for all $t$:
\begin{enumerate}
\item If $s\ {\semanticsnm{\alpha}}_{\bS}\ t$, then
    $t\in (\nott\propositionphi)_{\bS(\semanticsnm{\alpha})}$.
\item If $s\in \dom\ {\semanticsnm{\alpha}}_{\bS}$,
then it is false that $s\ {\semanticsnm{\alpha}}_{\bS}\ t$
implies 
$t\in\propositionphi_{\bS(\semanticsnm{\alpha})}$.
\item If  $s\in \dom\ {\semanticsnm{\alpha}}_{\bS}$, 
then $(s,\alpha)\notin \propositionphi_{\bS\otimes\bOmega}$.
\item If  $s\in \dom\ {\semanticsnm{\alpha}}_{\bS}$, 
then $(s,\alpha)\in (\nott\propositionphi)_{\bS\otimes\bOmega}$.
\end{enumerate}
Equivalence (2)$\Longleftrightarrow$(3) above is by induction.
To see (c)$\Longleftrightarrow$(d), we use the following
equivalences:
\begin{enumerate}
\item $(s,\alpha)\in \semantics{[\beta]\nott\phi}{\bS\otimes \bOmega}$. 
\item If $(s,\alpha)\in\semantics{\Pre(\beta)}{\bS\otimes \bOmega}$,
then $(s,\alpha)\notin \semantics{[\beta]\phi}{\bS\otimes \bOmega}$. 
\item If $(s,\alpha)\in\semantics{\Pre(\beta)}{\bS\otimes \bOmega}$,
then $(s,\alpha\then\beta)\notin \semantics{\phi}{\bS\otimes\bOmega}$.
\item If $(s,\alpha)\in\semantics{\Pre(\beta)}{\bS\otimes \bOmega}$,
then $(s,\alpha\then\beta)\in \semantics{\nott\phi}{\bS\otimes\bOmega}$.
\end{enumerate}
Equivalence
(1)$\Longleftrightarrow$(2) uses the soundness of the
Partial Functionality Axiom in the strong form of
Lemma~\ref{lemma-pf-general}.
The induction hypothesis
(c)$\Longleftrightarrow$(d)
 is used in (2)$\Longleftrightarrow$(3).

(e) in this case is immediate.

\paragraph{$\phi\andd\psi$}
The induction step for $\andd$ is easy, and so we omit the details.

\paragraph{$\necc_A$}
Next, assume  (a)$\Longleftrightarrow$(b),
 (c)$\Longleftrightarrow$(d), and (e)
 for $\phi$; we check it for $\necc_A\phi$.
 Fix $\alpha$, $\bS$ and $s\in S$.
First, the following establish 
 (a)$\Longleftrightarrow$(b) for $\necc_A\phi$.
\begin{enumerate}
\item If $s\ {\semanticsnm{\alpha}}_{\bS}\ t$, then
    $t\in (\necc_A\propositionphi)_{\bS(\semanticsnm{\alpha})}$.
\item If $s\ {\semanticsnm{\alpha}}_{\bS}\ t$
and $t \arrowA t'$, then $t'\in\propositionphi_{\bS(\semanticsnm{\alpha})}$.
\label{uii}
\item If $s\in \dom\ {\semanticsnm{\alpha}}_{\bS}$,
$\alpha\arrowA\alpha'$, $s\arrowA s'$, and
$s'\ {\semanticsnm{\alpha}}_{\bS}\ t'$, then
    $t'\in  \propositionphi_{\bS(\semanticsnm{\alpha'})}$.
\label{uuiii}
 \item If $s\in \dom\ {\semanticsnm{\alpha}}_{\bS}$,
$\alpha\arrowA\alpha'$ and $s\arrowA s'$, then:
  if $s'\ {\semanticsnm{\alpha}}_{\bS}\ t'$, then
    $t'\in  \propositionphi_{\bS(\semanticsnm{\alpha'})}$.
\label{iuiuu}

\item If $s\in \dom\ {\semanticsnm{\alpha}}_{\bS}$,
$\alpha\arrowA\alpha'$ and $s\arrowA s'$, then:
if $s'\in  \dom\ {\semanticsnm{\alpha'}}_{\bS}$, then
$(s',\alpha')\in\propositionphi_{\bS\otimes\bOmega}$.
\label{iuiuuu}

\item  If $s\in \dom\ {\semanticsnm{\alpha}}_{\bS}$,
 $s'\in \dom\ {\semanticsnm{\alpha'}}_{\bS}$,
 and
$(s,\alpha)\arrowA (s',\alpha')$, then
$(s',\alpha')\in\propositionphi_{\bS\otimes\bOmega}$.

\item  If $s\in \dom\ {\semanticsnm{\alpha}}_{\bS}$, then 
$(s,\alpha)\in (\necc_A\propositionphi)_{\bS\otimes\bOmega}$.
\end{enumerate}
The equivalence (\ref{uii})$\Longleftrightarrow$(\ref{uuiii})
comes from Proposition~\ref{proposition-square-biggersystem};
(4)$\Longleftrightarrow$(\ref{iuiuuu}) is by the induction
hypothesis.
And we check (c)$\Longleftrightarrow$(d) for $\necc_A\phi$
as follows:
\begin{enumerate}
\item $(s,\alpha)\in \semantics{[\beta]\necc_A\phi}{\bS\otimes \bOmega}$.
\item
 $(s,\alpha)\in \semantics{\Pre(\beta)\iif
\displaystyle{\bigwedge_{\beta\arrowA \beta'}}
\necc_A[\beta']\phi}{\bS\otimes\bOmega}$.
 
%% 4 below %%
\item If $(s,\alpha)\in \semantics{\Pre(\beta)}{\bS\otimes\bOmega}$,
then whenever $\beta\arrowA\beta'$,
  $(s,\alpha)\arrowA (s',\alpha')$, 
and  $ s' \in\semantics{\Pre(\alpha')}{\bS}$,
we have
$ (s',\alpha')\in \semantics{[\beta']\phi}{\bS\otimes\bOmega}$.
   
%% 4.5 below %%
\item If $(s,\alpha)\in \semantics{\Pre(\beta)}{\bS\otimes\bOmega}$,
then whenever $\beta\arrowA\beta'$,
  $(s,\alpha)\arrowA (s',\alpha')$, 
  $ s' \in\semantics{\Pre(\alpha')}{\bS}$,
and  $(s',\alpha')\in\semantics{\Pre(\beta')}{\bS\otimes \bOmega}$,
we have
$(s',\alpha'\then\beta')\in \semantics{\phi}{\bS\otimes\bOmega}$.
   
%% 5 below %%
\item If $(s,\alpha)\in \semantics{\Pre(\beta)}{\bS\otimes\bOmega}$,
then whenever $(s,\alpha;\beta)\arrowA (s',\alpha';\beta')$
and  $(s',\alpha';\beta')$ belongs to $\bS\otimes \bOmega$,
we have
$(s',\alpha'\then\beta')\in \semantics{\phi}{\bS\otimes\bOmega}$.

%% 6 below %%
\item $(s,\alpha)\in\semantics{\Pre(\beta)}{\bS\otimes \bOmega}$
implies $(s,\alpha\then\beta)\in \semantics{\necc_A\phi}{\bS\otimes\bOmega}$.
\end{enumerate}
(1)$\Longleftrightarrow$(2) is by the
Action-Knowledge
axiom as strengthened in 
Lemma~\ref{lemma-aka-general}.
 (2)$\Longleftrightarrow$(3) is by the semantics of
$\rightarrow$ and $\necc_A$, and also by the definition of
$\bS\otimes \bOmega$.
 (3)$\Longleftrightarrow$(4) is by the induction hypothesis
(c)$\Longleftrightarrow$(d)  for $\phi$.
(4)$\Longleftrightarrow$(5) is again by
 the definition of
$\bS\otimes \bOmega$.
(5)$\Longleftrightarrow$(6) is by the semantics of $\necc_A$.

To complete our discussion of $\necc_A\phi$, we check
(e) for it.   We use Proposition~\ref{prop-pre-prebisim},
and the argument is straightforward.

\paragraph{$\necc^*_{\BB}\phi$}
We have  (a)$\Longleftrightarrow$(b),
 (c)$\Longleftrightarrow$(d), and (e) for $\phi$.
By what we showed above, we have these equivalences for
$\necc_A\phi$ for all $A$.    We can use this to show
by induction that 
  (a)$\Longleftrightarrow$(b),
 (c)$\Longleftrightarrow$(d), and (e)  hold
for every sentence of the form 
 $\necc_{A_1}\necc_{A_2}\cdots \necc_{A_k}\phi$,
with all $A$'s from the set $\BB$. 
Let $W$ be the set of such sentences; so $\necc^*_{\BB}\phi$ is 
semantically
equivalent to $\bigwedge W$.
This  gives    (a)$\Longleftrightarrow$(b),
 (c)$\Longleftrightarrow$(d),
and 
(e) for $\necc^*_{\BB}\phi$ (for all $\alpha$, $\bS$ and $s$).
The details in all these proofs
are similar,
  and so we'll just prove (c)$\Longleftrightarrow$(d).
Fix $\alpha$, $\bS$, $s$, and $\beta$.  Then the following are
equivalent:
\begin{enumerate}
\item  $(s,\alpha)\in \semantics{[\beta]\necc^*_{\BB}\phi}{\bS\otimes \bOmega}$.
\item For all $\psi\in W$, $(s,\alpha)\in \semantics{[\beta]\psi}{\bS\otimes \bOmega}$.
\item   $(s,\alpha)\in\semantics{\Pre(\beta)}{\bS\otimes \bOmega}$
implies that for $\psi\in W$,
 $(s,\alpha\then\beta)\in \semantics{\psi}{\bS\otimes\bOmega}$.
\item   $(s,\alpha)\in\semantics{\Pre(\beta)}{\bS\otimes \bOmega}$
implies $(s,\alpha\then\beta)\in \semantics{\necc^*_{\BB}\phi}{\bS\otimes\bOmega}$.
\end{enumerate}
 (1)$\Longleftrightarrow$(2) here is by the semantics of $[\beta]\necc^*\phi$.
The other equivalences are easy from (c)$\Longleftrightarrow$(d)
for $\psi\in W$.

\paragraph{$[\gamma]\phi$}
Assume our result for sentences with smaller value than
$[\gamma]\phi$.   (In particular, we have the result
for $\Pre(\gamma)$; this is the main point.)
 We first  check 
(a)$\Longleftrightarrow$(b) for $[\gamma]\phi$.  Fix $\alpha$, $\bS$ and $s\in S$.
Then the following are equivalent:
\begin{enumerate}
\item  If $s\ {\semanticsnm{\alpha}}_{\bS}\ t$, then
    $t\in ([\gamma]\propositionphi)_{\bS(\semanticsnm{\alpha})}$.
\item  If $s\ {\semanticsnm{\alpha}}_{\bS}\ t$ and 
 $t\ {\semanticsnm{\gamma}}_{\bS(\semanticsnm{\alpha})}\ u$, then
    $u\in \propositionphi_{\bS(\semanticsnm{\alpha})(\semanticsnm{\gamma})}$.
\item  If $s\ {(\semanticsnm{\alpha}\then\semanticsnm{\gamma})}_{\bS}\ u$, then
    $u\in \propositionphi_{\bS(\semanticsnm{\alpha}\then\semanticsnm{\gamma})}$.
\item  If $s\in \dom\ {(\semanticsnm{\alpha\then\gamma})}_{\bS}$, then 
$(s,\alpha\then\gamma)\in \propositionphi_{\bS\otimes\bOmega}$.
\item 
If $s\in \dom\ {\semanticsnm{\alpha}}_{\bS}$
and $(s,\alpha)\in \semantics{\Pre(\gamma)}{\bS\otimes\bOmega}$, then
 $(s,\alpha\then\gamma)\in
\propositionphi_{\bS\otimes\bOmega}$.
\item 
If $s\in \dom\ {\semanticsnm{\alpha}}_{\bS}$, then $(s,\alpha)\in
([\gamma]\propositionphi)_{\bS\otimes\bOmega}$.
\end{enumerate}
The equivalence
 (3)$\Longleftrightarrow$(4) uses the  equivalence
of  (a)$\Longleftrightarrow$(b)  for $\phi$,
and the fact that 
 $\semanticsnm{\alpha}\then\semanticsnm{\gamma} = 
\semanticsnm{\alpha \then\gamma}$.
Equivalence (4)$\Longleftrightarrow$(5)
 uses (a)$\Longleftrightarrow$(b)
for $\PRE(\gamma)$ which is allowed by our inductive assumption.
The equivalence 
 (5)$\Longleftrightarrow$(6) uses the  equivalence
 (c)$\Longleftrightarrow$(d) for $\phi$.
Finally, we consider 
(c)$\Longleftrightarrow$(d) for $[\gamma]\phi$.   Fix $\beta$, $\bS$, and $s$.
The following are then equivalent:
\begin{enumerate}
\item $(s,\alpha)\in \semantics{[\beta][\gamma]\phi}{\bS\otimes \bOmega}$.
\item   $(s,\alpha)\in \semantics{[\beta\then\gamma]\phi}{\bS\otimes \bOmega}$.
\item   $(s,\alpha)\in\semantics{\Pre(\beta\then\gamma)}{\bS\otimes \bOmega}$
implies $(s,\alpha\then (\beta\then\gamma))\in \semantics{\phi}{\bS\otimes\bOmega}$.
\item   $(s,\alpha)\in\semantics{\Pre(\beta\then\gamma)}{\bS\otimes \bOmega}$
implies $(s,(\alpha\then  \beta)\then\gamma)\in \semantics{\phi}{\bS\otimes\bOmega}$.

\item   $(s,\alpha)\in\semantics{\pair{\beta}\Pre(\gamma)}{\bS\otimes \bOmega}$
implies $(s,(\alpha\then  \beta)\then\gamma)\in \semantics{\phi}{\bS\otimes\bOmega}$.

\item   $(s,\alpha)\in\semantics{\Pre(\beta)}{\bS\otimes \bOmega}$
and $(s,\alpha\then\beta)\in \semantics{\Pre(\gamma)}{\bS\otimes \bOmega}$
imply $(s,(\alpha\then\beta)\then\gamma )\in \semantics{\phi}{\bS\otimes\bOmega}$.

\item $(s,\alpha)\in\semantics{\Pre(\beta)}{\bS\otimes \bOmega}$
implies $(s,\alpha\then\beta)\in \semantics{[\gamma]\phi}{\bS\otimes\bOmega}$.
\end{enumerate}
We have used the Composition Axiom
in (1)$\Longleftrightarrow$(2). 
(2)$\Longleftrightarrow$(3) and 
(6)$\Longleftrightarrow$(7) are by the induction hypothesis
(c)$\Longleftrightarrow$(d).
(3)$\Longleftrightarrow$(4)  is by (e) for $\phi$.  (Indeed, this very step
is the main place that (e) is used.)
(4)$\Longleftrightarrow$(5) is by the definition of $\Pre$.
As for (5)$\Longleftrightarrow$(6), it holds by the induction hypothesis
(a)$\Longleftrightarrow$(b).
That is, 
$\sharp{(\Pre(\gamma))}< \sharp{[\gamma]\phi}$.   So we may use
the equivalence of (c) and (d) for $\Pre(\gamma)$.  We also used some
propositional reasoning at this point.   This concludes the verification
of (c)$\Longleftrightarrow$(d) for $[\gamma]\phi$.

Finally, (e)  for $[\gamma]\phi$ follows from 
(c)$\Longleftrightarrow$(d) for $\phi$.
As in showing (e) for  $\phi$ atomic, we use 
Lemma~\ref{lemma-towards-composition-admissibility}.
\end{proof}

\rem{
\begin{proof} Write $\actionalpha$ for  $\semanticsnm{\alpha}$
and $\propositionphi$ for $\semanticsnm{\phi}$.
We need to show that the following conditions are equivalent:
\begin{enumerate}
\item  If $s\ {\actionalpha}_{\bS}\ t$, then
    $t\in\propositionphi_{\bS(\actionalpha)}$.
\item   If $s\in \dom\ {\actionalpha}_{\bS}$, then 
$(s,\alpha)\in \propositionphi_{\bS\otimes\bOmega}$.
\end{enumerate}
We argue
by induction on $\phi$.

First, we check the equivalence when $\phi$ is an atomic sentence
$p$.
In this case, the following are equivalent:
\begin{enumerate}
\item If $s\ {\actionalpha}_{\bS}\ t$, then
    $t\in\propositionp_{\bS(\actionalpha)}$.
\item If $s\in \dom\ {\actionalpha}_{\bS}$, then
    $s\in\propositionp_{\bS}$.
\item   If $s\in \dom\ {\actionalpha}_{\bS}$, then 
$(s,\alpha)\in \propositionp_{\bS\otimes\bOmega}$.
\end{enumerate}
The equivalence (1)$\Longleftrightarrow$(2) follows from 
Proposition~\ref{prop-actions-preserve-atomics}.
Next, assume the result for $\phi$, and consider $\nott\phi$.
Then the following are equivalent:
\begin{enumerate}
\item If $s\ {\actionalpha}_{\bS}\ t$, then
    $t\in (\nott\propositionphi)_{\bS(\actionalpha)}$.
\item If $s\in \dom\ {\actionalpha}_{\bS}$,
then it is false that $s\ {\actionalpha}_{\bS}\ t$
implies 
$t\in\propositionphi_{\bS(\actionalpha)}$.
\item If  $s\in \dom\ {\actionalpha}_{\bS}$, 
then $(s,\alpha)\notin \propositionphi_{\bS\otimes\bOmega}$.
\item If  $s\in \dom\ {\actionalpha}_{\bS}$, 
then $(s,\alpha)\in (\nott\propositionphi)_{\bS\otimes\bOmega}$.
\end{enumerate}
The induction step for $\andd$ is easy, and as before the ideas
extend to infinitary conjunction.

Next, assume the equivalence for $\phi$; we check it for $\necc_A\phi$.
\begin{enumerate}
\item If $s\ {\actionalpha}_{\bS}\ t$, then
    $t\in (\necc_A\propositionphi)_{\bS(\actionalpha)}$.
\item If $s\ {\actionalpha}_{\bS}\ t$
and $t \arrowA t'$, then $t'\in\propositionphi_{\bS(\actionalpha)}$.
\label{uii}
\item If $s\in \dom\ {\actionalpha}_{\bS}$,
$\alpha\arrowA\alpha'$, $s\arrowA s'$, and
$s'\ {\actionalpha}_{\bS}\ t'$, then
    $t'\in  \propositionphi_{\bS(\actionalpha)}$.
\label{uuiii}
 \item If $s\in \dom\ {\actionalpha}_{\bS}$,
$\alpha\arrowA\alpha'$ and $s\arrowA s'$, then:
  if $s'\ {\actionalpha}_{\bS}\ t'$, then
    $t'\in  \propositionphi_{\bS(\actionalpha)}$.
\label{iuiuu}

\item If $s\in \dom\ {\actionalpha}_{\bS}$,
$\alpha\arrowA\alpha'$ and $s\arrowA s'$, then:
if $s'\in  \dom\ {\actionalpha'}_{\bS}$, then
$(s',\alpha')\in\propositionphi_{\bS\otimes\bOmega}$.
\label{iuiuuu}

\item  If $s\in \dom\ {\actionalpha}_{\bS}$,
 $s'\in \dom\ {\actionalpha'}_{\bS}$,
 and
$(s,\alpha)\arrowA (s',\alpha')$, then
$(s',\alpha')\in\propositionphi_{\bS\otimes\bOmega}$.

\item  If $s\in \dom\ {\actionalpha}_{\bS}$, then 
$(s,\alpha)\in (\necc_A\propositionphi)_{\bS\otimes\bOmega}$.
\end{enumerate}
The equivalence (\ref{uii})$\Longleftrightarrow$(\ref{uuiii})
comes from Proposition~\ref{proposition-square-biggersystem}.
(\ref{uuiii})$\Longleftrightarrow$(\ref{iuiuuu}) is by the induction
hypothesis.

Finally, assume the equivalence for $\phi$; we check it for $[\beta]\phi$.
Write $\actionbeta$ for $\semanticsnm{\beta}$.
\begin{enumerate}
\item  If $s\ {\actionalpha}_{\bS}\ t$, then
    $t\in ([\beta]\propositionphi)_{\bS(\actionalpha)}$.
\item  If $s\ {\actionalpha}_{\bS}\ t$ and 
 $t\ {\actionbeta}_{\bS(\actionalpha)}\ u$, then
    $u\in \propositionphi_{\bS(\actionalpha)(\actionbeta)}$.
\item  If $s\ {(\actionalpha\then\actionbeta)}_{\bS}\ u$, then
    $u\in \propositionphi_{\bS(\actionalpha\then\actionbeta)}$.
\item  If $s\in \dom\ {(\actionalpha\then\actionbeta)}_{\bS}$, then 
$(s,\alpha\then\beta)\in \propositionphi_{\bS\otimes\bOmega}$.
\item 
If $s\in \dom\ {\actionalpha}_{\bS}$
and $(s,\alpha)\in \semantics{\Pre(\beta)}{\bS\otimes\bOmega}$, then
 $(s,\alpha\then\beta)\in
\propositionphi_{\bS\otimes\bOmega}$.
\item 
If $s\in \dom\ {\actionalpha}_{\bS}$, then $(s,\alpha)\in
([\beta]\propositionphi)_{\bS\otimes\bOmega}$.
\end{enumerate}
We used Proposition~\ref{proposition-domains} at several points.
\end{proof}
}%% end of rem

\begin{theorem}  For all simple  actions $\alpha$
and  {\em  all\/} sentences $\phi$ of $\lang_1(\bSigma)$, 
 $\semanticsnm{[\alpha]\phi}
= [\hat{\alpha}]\semanticsnm{\phi}.$ 
That is, equation (\ref{eq-sam}) 
in Theorem~\ref{theorem-syntactic-action-model-preliminary}
holds even when 
the program union operation $\union$ occurs in $\phi$.
\label{theorem-syntactic-action-model}
\end{theorem}

\begin{proof}
By induction on the number of occurrences of $\union$ in $\phi$.
The base case is when this number is $0$, and this is
Theorem~\ref{theorem-syntactic-action-model-preliminary}.
The induction step is then again by cases.
The most interesting step is for sentences of the form
$[\alpha\union \beta]\phi$, and for this we use the
soundness of the  Choice Axiom.
\end{proof}

} %%%%end of rem of a whole subsection

%\subsection{The Substitution Lemma for $\lang_1(\bSigma)$}
\label{section-equivalence-of-actions}

\rem{
\paragraph{Syntactic bisimulation}
There are two natural notions of equivalence for action models:
bisimilarity, and
{\em  update equivalence}. 
Recall that action models are Kripke frames with an extra
function $\Presemantic$.  
So 
 bisimilarity here is just the
standard definition.
Concretely,   we have the following definition:
}

\begin{lemma}
Concerning the syntactic equivalence $\equiv$ on $\Omega$:  
\begin{enumerate}
\item  The relation $\equiv$  is an
equivalence relation.
\label{part-equivreln}
\item  If $ \phi_1\equiv \psi_1$, $\ldots$, 
$\phi_n\equiv \psi_n$, then 
$\sigma_i \vec{\phi} \equiv \sigma_i \vec{\psi}$.
\label{lemma-equivalence-actions-part-sigma}
\item $\alpha \then  skip  \equiv   \alpha\equiv skip\then\alpha$.
\label{used-lemma-towards-composition-admissibility-1}
\item  $\alpha\then (\beta\then \gamma) \equiv
 (\alpha\then  \beta)\then \gamma$.
\label{used-lemma-towards-composition-admissibility-2}
\item  If $\alpha \equiv \alpha'$ and $\beta\equiv\beta'$,
then   $\alpha \then \beta\equiv \alpha'\then \beta'$.
\label{part-equivreln3}
\end{enumerate}
\label{lemma-equivalence-actions}
\end{lemma}

\begin{proof}
The first part is routine.  
Part (\ref{lemma-equivalence-actions-part-sigma}) uses the bisimulation
consisting of all pairs $(\sigma_j\vec{\phi},\sigma_j\vec{\psi})$
such that $1\leq j\leq n$ and for all $i$, $\phi_i\equiv\psi_i$.
Lemma~\ref{lemma-towards-composition-admissibility} is used in 
parts~(\ref{used-lemma-towards-composition-admissibility-1})
and~(\ref{used-lemma-towards-composition-admissibility-2}).
The rest of  argument for  
part~(\ref{used-lemma-towards-composition-admissibility-1}) is easy and we omit it.
For part~(\ref{used-lemma-towards-composition-admissibility-2}), we
take $R$ to be the set of pairs
$$(\alpha'\then(\beta'\then \gamma'), (\alpha'\then \beta')\then \gamma')
$$
such that for some
sequence $w$ of
$\arrow_A$ with $A\in \Agents^*$, $\alpha'$ is reachable from
$\alpha$ via $w$, and similarly for $\beta'$ and $\gamma'$ (via the same $w$).
Part~(\ref{part-equivreln3}) is similar. 
% The remaining parts are easy bisimulation arguments.
\end{proof}

\rem{
\begin{lemma} If $\alpha\equiv\beta$, then $\alpha\sim\beta$.
\label{lemma-syntactic-semantic}
\end{lemma}

\begin{proof}
{\bf  I don't have this one yet.  We don't need this result,
but it would be nice to have.  My hunch is that every bisimulation
on $\bOmega$ extends to a total bisimulation, and  this would
imply this result.}

  Fix a syntactic bisimulation
 $R_0$ relating $\alpha$ and $\beta$.
We check that $\hat{\alpha}\sim\hat{\beta}$ as updates.
Fix $\bS$,  $\bT$, and a total bisimulation $Q$ between them.

Consider the following relation $R$ between
$\bS(\hat{\alpha})$ and $\bT(\hat{\beta})$: 
\begin{equation}
\mbox{
$(t,\gamma)\ R \ (u,\delta)$ iff $t \ Q\ u$ and $\gamma\ R_0\ \delta$
}
\label{eq-notyet}
\end{equation}
Then
\end{proof}
}

\begin{lemma}
For all $A\in \CC$ and all
$\beta$
such that
 $\alpha\rightarrow_{A}  \beta$,
\begin{enumerate}
\item $\proves [\alpha]\necc_{\CC}^*\psi
\iif [\alpha]\psi$.
\label{part-ffff}
%\item $\proves [\alpha]\necc_{\CC}^*\psi
%\iif  [\beta]\necc_A \psi$.
%\label{part-gggg}
\item $\proves [\alpha]\necc_{\CC}^*\psi\andd\Pre(\alpha)
\iif \necc_A
[\beta]\necc_{\CC}^*\psi$.
\label{part-hhhh}
\end{enumerate}
\label{lemma-converse-rule}
\end{lemma}

\begin{proof}
 Part (\ref{part-ffff})
follows  from the Epistemic Mix Axiom and
modal reasoning.
For part (\ref{part-hhhh}), we start with a consequence of
  the Epistemic Mix Axiom:
$\proves  \necc_{\CC}^*\psi\iif
 \necc_A\necc_{\CC}^*\psi$.
Then by modal reasoning,
$\proves  [\alpha] \necc_{\CC}^*\psi\iif
 [\alpha] \necc_A\necc_{\CC}^*\psi$.
By the Action-Knowledge Axiom in the generalized form of 
Lemma~\ref{lemma-aka-general},  we have
$\proves [\alpha]\necc_{\CC}^*\psi\andd\Pre(\alpha)\iif
  \necc_A
[\beta]\necc_{\CC}^*\psi$.
%%This is what we need to show.
\end{proof}

\rem{
\begin{proposition}   Let $\tau$ be the action from
Example~\ref{example-tau}.  Then for  all $\phi$,
$\proves[\tau]\phi\iiff \phi$.
\label{proposition-tau}
\end{proposition}

\begin{proof} By induction on $\phi$.
\end{proof}
}

We present next our main result  on   syntactic 
 equivalence of  actions.  It is a syntactic version
of Proposition~\ref{proposition-preserve-sim}.

\begin{lemma}   Let $\alpha$ and $\beta$ be simple actions.
If $\alpha\equiv\beta$, then 
for all $\phi$, $\proves [\alpha]\phi\iiff [\beta]\phi$.
\label{lemma-equivalence-actions-two}
\end{lemma}

\begin{proof}
By induction on $\phi$.    Fix a syntactic bisimulation
 $R$ relating $\alpha$ and $\beta$.
For $\phi=\true$ or an atomic sentence $p_i$, our result is easy.
The induction steps for $\nott$ and $\andd$ are trivial.
The step for $\necc_A$ is not hard, and so we omit it.

We next check the result for sentences $[\gamma] \phi$.
We need to see that
$$ \proves [\alpha][\gamma]\phi \iiff  [\alpha'][\gamma]\phi.
$$
For this, it is sufficient by the Composition Axiom to show that 
$ \proves [\alpha\then\gamma]\phi \iiff  [\alpha'\then\gamma]\phi$.
By Lemma~\ref{lemma-equivalence-actions}, parts (\ref{part-equivreln})
and~(\ref{part-equivreln3}),
$\alpha\then\gamma \equiv \alpha'\then\gamma$.  So we are done
by the induction hypothesis.

This leaves the step for
sentences  of the form
$\necc_{\BB}^*\phi$, assuming the result for $\phi$.
We  use the Action Rule to show that
$\proves[\alpha]\necc_{\CC}^*\phi\iif
[\beta]\necc_{\CC}^*\phi$.
We need a functional witness to $R$; that is, a
map $\beta'\mapsto \alpha'$ for all  $\beta'$ such that $\beta\arrowAgentsstar\beta'$, 
such that $\alpha' R \beta'$ for all these actions.
Further,
let $\chi_{\beta'}$ be
 $[\alpha']\necc_{\CC}^*\phi$.
We need to show that
for all $A\in C$ and all 
  $\beta'\arrowA\beta''$,
%whenever
%$\alpha'\rightarrow_{\CC}^{*,A}\beta'$,
\begin{description}
\item[{\rm a.}]$\proves [\alpha']\necc_{\CC}^*\phi
\iif [\beta']\phi$; and
\item[{\rm b.}]
If $\beta'\arrowA\beta''$, then
 $\proves [\alpha']\necc_{\CC}^*\phi\andd\Pre(\beta')
\iif \necc_A
[\alpha'']\necc_{\CC}^*\phi$.
\end{description}
For (a), we know from
Lemma~\ref{lemma-converse-rule}, 
that
$\proves [\alpha']\necc_{\CC}^*\phi\iif [\alpha']\phi$.
By induction hypothesis on $\phi$,
$\proves  [\alpha']\phi\iiff [\beta']\phi$.
And this implies (a).
For (b), Lemma~\ref{lemma-converse-rule}
tells us that under the assumptions,
$$\proves [\alpha']\necc_{\CC}^*\phi\andd\Pre(\alpha')
\iif \necc_A
[\alpha'']\necc_{\CC}^*\phi.$$
The fact that $R$ is a syntactic bisimulation tells us that
$\proves \Pre(\alpha')\iiff \Pre(\beta')$.
This implies (b).

This completes the
induction on $\phi$.
\end{proof}

\rem{
\begin{conjecture}
For basic actions $\alpha$ and $\beta$, the following are equivalent:
\begin{enumerate}
\item  $\alpha$ and $\beta$ are equivalent actions.
\item  $\semanticsnm{\alpha}$ and $\semanticsnm{\beta}$ are equivalent
updates.
\item For all $\phi$ in $\lang_1(\bSigma)$,
 $\proves [\alpha]\phi\iiff \beta[\phi]$.
\end{enumerate}
\end{conjecture}
}

Lemma~\ref{lemma-equivalence-actions-two} 
will be used in several places to come.

\rem{
In order to formulate substitution, it is convenient to  
formulate an extension of  $\lang_1(\bSigma)$ which 
has additional {\em sentence variables\/} and
{\em action variables}.    We also need the notion of
substitution, of course.  
}

\rem{
\begin{lemma} [Substitution Lemma for $\lang_1(\bSigma)$]
Provable equivalence is preserved under substitutions: 
Let $v$ and $w$ be maps from sentence variables to sentences (without
variables)
and action variables to actions
(without variables).  Assume that $v(x)\equiv w(x)$
 for all $x$.  Let $v^*$ and $w^*$ be the extensions
of $v$ and $w$ to  $\lang_1(\bSigma)$.  Then
$v^*(t)\equiv w^*(t)$ for all sentences or actions $t$
of $\lang_1(\bSigma)$.
\label{lemma-substitution-biggersystem}
\end{lemma}
}%% old formulation - Delete!

\rem{
\begin{lemma} [A Substitution Lemma for $\lang_1(\bSigma)$]
Provable equivalence is preserved under substitutions: 
Let $v$ and $w$ be maps from sentence variables to sentences (without
variables)
and action variables to simple actions of $\lang_1(\bSigma)$.
(without variables)  Assume that $v(x)\equiv w(x)$
 for all $x$.  Let $v^*$ and $w^*$ be the evident homomorphic
 extensions
of $v$ and $w$ to  $\lang_1(\bSigma)$.  Then
\begin{enumerate}
\item  For all sentences $\phi$, $v^* \phi \equiv w^* \phi$;
i.e.,   $\proves v^* \phi \iiff w^* \phi$.
\item For all programs $\pi$ (simple or not) and all sentences $\phi$,
$\proves [v^*\pi]\phi \iiff [w^*\pi]\phi$.
\end{enumerate}
\label{lemma-substitution-biggersystem}
\end{lemma}

\begin{proof}
By induction  on  $\lang_1(\bSigma)$.
The cases of atomic sentences and actions
are trivial, and the cases for  the 
sentence variables and action variables are
immediate.   The induction steps
for the boolean connectives
are easy calculations of propositional
equivalence.  The induction steps for sentences
$\necc_A\phi$ and $\necc^*_{\BB}$ are also easy using the
necessitation rules of $\lang_1(\bSigma)$.

The key step is for sentences $[\pi]\phi$.
By induction hypothesis, $v^*\pi \equiv w^*(\alpha)$
and also $v^* \phi \equiv w^*(\phi)$.
So 
$$\begin{array}{lcll}
v^*([\pi]\phi) & \quadeq & [v^* \pi)]v^* \phi \\
& \quadequiv &  [v^* \pi)]v^* \phi 
 & \mbox{by Lemma~\ref{lemma-equivalence-actions-two}}\\
& \quadeq & [w^*(\alpha)]w^*(\phi)  \\
& \quadeq & w^*([\alpha]\phi) 
\end{array}
$$

We also must consider actions.
For actions $\sigma\vec{\psi}$, see Lemma~\ref{lemma-equivalence-actions},
part~(\ref{lemma-equivalence-actions-part-sigma}).
The induction step for actions $\alpha\then \beta$ follows
from part~(\ref{part-equivreln3}) of  Lemma~\ref{lemma-equivalence-actions}.
Here is the  step for $\pi \union \rho$.   Fix $\phi$.
The following are equivalent:
\begin{enumerate}
\item  $\proves  v^*([\pi\union\rho]\phi)$
\item $\proves  [v^* \pi\unionv^* \rho) ]v^* \phi$
\item $\proves  [v^* \pi]v^* \phi \andd  [v^*\rho] v^*\phi$
\item
\end{enumerate}
The first equivalence is the way that $v^*$ works, and the second is the Choice Axiom.
By the induction hypothesis,
Point (3) is equivalent to  $\proves  [w^* \pi]w^* (\phi) \andd  [w^*\rho] w^*(\phi)$.
And then the same equivalences as above but with $w$ instead of $v$
show the desired equivalence.
\end{proof}
}

\rem{
\subsection{Trash below?}

\footnote{We probably should drop all of this.}

\paragraph{Equivalence of updates} \label{section-equivalence-actions}

We also study 
a more semantic equivalence.
We remind the reader that
bisimulation preserving updates were defined in
 Section~\ref{section-bisim-preserve-updates}.
Also, we sometimes use the word ``state'' to refer to the entities in a state
model, and sometimes we also use it to refer to pairs of the form $(\bS,s)$.

A bisimulation preserving update may be regarded as a relation on the quotient
of states (in the second sense)  by the relation $\equiv$ of bisimulation.    
It makes sense to say that bisimulation preserving $\actionalpha$
and $\actionalpha'$ are {\em equivalent\/} if they induce the same relation on
the quotient.  More concretely, for all states $(\bS, s)$,
\begin{enumerate}
\item If $x\in \bS(\actionalpha)$ and
 $s\ \actionalpha_{\bS}\ x$, then there is some $y\in \bS(\actionalpha')$
such that  $s\ \actionalpha'_{\bS}\ y$ and
 $(\bS(\actionalpha),x)\equiv (\bS(\actionalpha'),y)$.
\item If $y\in \bS(\actionalpha')$ and $s\ \actionalpha'_{\bS}\ y$,
 then there is some $x\in\bS(\actionalpha)$  such that $s\ \actionalpha_{\bS}\ x$ and
 $(\bS(\actionalpha),x)\equiv (\bS(\actionalpha'),y)$.
\end{enumerate}

\begin{lemma}  Let $\alpha$ and $\alpha'$ be simple actions.
If $\alpha\equiv\alpha'$, then   $\alpha\approx_u\alpha'$.
\label{lemma-questionable}
\end{lemma}

\begin{proof}   Fix a syntactic bisimulation
 $R_0$ relating $\alpha$ and $\alpha'$.

Write $\actionalpha$ for $\semanticsnm{\alpha}$ and 
$\actionalpha'$ for $\semanticsnm{\alpha'}$.
Fix $\bS$. 
Consider the following relation $R$ between
$\bS(\actionalpha)$ and $\bS(\actionalpha')$: 
\begin{equation}
\mbox{
$x\ R \ y$ iff there is some $t\in S$ and some 
 $\beta \ R_0 \ \beta'$ such that $t\ \semanticsnm{\beta}_{\bS} \ x$
and $t\ \semanticsnm{\beta'}_{\bS} \ y$.
}
\label{eq-tocheck}
\end{equation}
 We check that $R$ is a bisimulation.
Note first that since we are dealing with simple actions,
the relations $\semanticsnm{\beta}$ are all partial functions.
And since $R_0$ is syntactic, if 
 $\beta \ R_0 \ \beta'$, then $\semanticsnm{\beta}$ and
 $\semanticsnm{\beta'}$ have the same domain.
Getting back to $R$, the ``atomic harmony'' condition
comes from the fact that the functions $\semanticsnm{\beta}$ never
change atomic sentences.  The ``zig-zag'' conditions
come from  Proposition~\ref{proposition-square-biggersystem}.
For example, suppose that $x \ R \ y$ via $t$ as in 
equation (\ref{eq-tocheck}), 
and also
 $x \arrowA x'$.  Then by the Proposition,
there are some $t'$ and $\gamma$ 
such that $t\arrowA t'$, $\beta\arrowA \gamma$ and 
$t'\ \semanticsnm{\gamma}_{\bS} \ x'$.
Let $\gamma'$ be such that $\gamma \ R_0\ \gamma'$ and $\beta'\arrowA \gamma'$.
Since $\semanticsnm{\gamma}$ and $\semanticsnm{\gamma'}$ have the same
domain, let $y'$ be such that 
$t'\ \semanticsnm{\gamma'}_{\bS} \ y'$.
Then $x' \ R\ y'$.  By the direction (2)$\Longleftrightarrow$(1) of  
  Proposition~\ref{proposition-square-biggersystem}, $y \arrowA y'$.
\end{proof}

\rem{
 if  $(\bS(\actionalpha),s') \equiv (\bT(\actionbeta),t')$,
then the following hold:
\begin{enumerate}
\item  If    $s\ \actionalpha_{\bS}\ s'$,
then there is some $t'$ such that
 $t\ r^{\actionbeta}_{\bT}\ t'$ and $(\bS(\actionalpha),s') \equiv (\bT(\actionbeta),t')$.
\item  If   $t\ \actionalpha_{\bT}\ t'$,
then there is some $s'$
such that  $s\ \actionalpha_{\bS}\ s'$
 and $(\bS(\actionalpha),s') \equiv (\bT(\actionbeta),t')$.
\end{enumerate}
}
So it makes sense to say that two programs $\pi$ and $\rho$
 are equivalent if their semantics 
 $\semanticsnm{\pi}$ and   $\semanticsnm{\rho}$ are equivalent
updates.  
We write $\pi\approx_u \rho$ in this case.
We shall study this notion primarily when 
$\pi$ and $\rho$ are simple actions.

\paragraph{Digression: equivalence of simple actions in specific logics}
We  noted a few general facts
in  Lemma~\ref{lemma-equivalence-actions}
 concerning the  equivalence of simple actions.
It is natural to ask whether in special cases the notion
of equivalence is axiomatizable by a finite system.
For example, in the logic of public announcements, we have the iterated
announcement law
\begin{equation} \Pub\ \phi \then \Pub\ \psi \quad\approx_u\quad
 \Pub\ (\pair{\Pub\ \phi}
\psi)
\label{pan}
\end{equation}
This leads to consequences such as the
following: for atomic $p$ and $q$,
$$\Pub\ p \then \Pub\ q \quad 
\approx_u \quad \Pub\ q\then\Pub\ p \quad \approx_u\quad \Pub 
(p\andd q),$$
and also that, e.g., $\Pub\ \necc_A p \then \Pub\ \necc_A p
\approx_u\Pub\ \necc_A p$.
But we do not know if we obtain a complete proof system for $\approx_u$ 
by taking
our logical system for sentential validities and adding the laws of
 Lemma~\ref{lemma-equivalence-actions} and  also (\ref{pan}).

\paragraph{Equivalence of programs}
We can also study equivalence of actions which are not simple actions.
We get facts such as the following:
 $\alpha \union \crash \approx_u \alpha$;
  $\alpha \union \beta \approx_u \beta \union \alpha$;
  $\alpha \union (\beta + \gamma) \approx_u  (\alpha \union  \beta) \union \gamma$;
  $\alpha \union \alpha \approx_u  \alpha$;
and if $\pi\approx_u \rho$, then $\pi^*\approx_u \rho^*$.
 }

 %% Section 5 on the Action rule

\Section{Completeness theorems}
\label{section-completeness-theorems}
\label{section-completeness}

In this section, we prove the completeness of our logical systems 
for $\lang_0(\bSigma)$ and $\lang_1(\bSigma)$. (See Figure~\ref{fig-lang} for these.)
  Recall that
the difference between the two languages is that the second
has the common-knowledge propositional operators $\necc^*_{\BB}$
while the first does not.  This difference makes the bigger
system much more expressive, and more difficult to study.
As it happens, the extension of 
$\lang_1(\bSigma)$ to the full logic $\lang(\bSigma)$,
the extension via the  action iteration operation $\pi^*$, 
leads to a logical language whose validity problem is 
complete $\Pi^1_1$.  So there cannot be a recursively
axiomatized logical system for the validities of  $\lang(\bSigma)$.
Returning to the smaller  $\lang_1(\bSigma)$, even here the  
   common knowledge operators $\necc^*_{\BB}$
give a logical system which is not compact.  So we cannot have a  
{\em strongly} complete logic for it; that is, 
we cannot axiomatize the notion 
of validity under hypotheses $T\models \phi$. The best one can hope for is 
weak completeness: $\proves\phi$ if and only if $\models \phi$.
We prove this in Theorem~\ref{theorem-completeness-Kstar}.
As a result of some preliminary results aimed toward that
result, we establish the 
strong completeness of our axiomatization
of the weaker logic $\lang_0(\bSigma)$.   The work there is
easier because it relies on a translation into modal logic.

In a later section,
we show that in contrast to our translation results
for $\langaction$, the larger  language
 $\langactionstar$ cannot
be translated into $\lang$ or even to $\langstar$
(modal  logic with extra modalities $\necc_{\BB}^*$).
So  completeness results for $\langactionstar$
cannot simply be based on translation.

\paragraph{The ideas}  Our completeness proofs are somewhat involved,
and it might help the reader to have  a preview of some of the
ideas before we get started. 
 For, $\lang_0(\bSigma)$, the leading idea is that we can
translate the logic back to ordinary modal logic.  Then we 
get completeness by taking any logical system for modal
logic and adding whatever principles are needed in order
to make the translation.  This idea is simple enough, and
after looking at a few examples one can see how
the  translation of $\lang_0(\bSigma)$ to modal logic should
go.  But the formal definition of the translation is
complicated.  Our definition goes via a {\em term rewriting
system\/} related to the axioms of the logic, and in particular
we'll need to do prove the {\em termination\/} of our system.
It would  have been nice to use some off-the-shelf results
of term rewriting theory to get the termination of our system,
but this does not seem to be possible.  In any case, we prove
termination by establishing a {\em decreasing interpretation\/}
of
the system; that is, rewriting a sentence 
leads to a decrease in an order $<$ that we study at length.
 Our interpretation is  exponential rather than
polynomial, and it was found by hand.   

When we turn to $\lang_1(\bSigma)$, our model is the filtration
proof of the completeness of PDL due to Kozen and Parikh~\cite{KP}.
We need to use the Action Rule rather than an Induction Rule,
but modulo this difference, the work is similar.  We also need
to use some of the details on the ordering $<$ mentioned above. 
    That is, even if
one were interested in the completeness theorem of 
 $\lang_0(\bSigma)$
alone, our proof would still involve the general
rewriting apparatus.

\begin{definition}
Let $\NF$ be the smallest set of ground sentences 
 and actions in $\lang_1(\bSigma)$ 
%%% it had been  $\lang_1^+(\bSigma)$ 
with the
following properties:
\begin{enumerate}
\item Each atomic $p$ belongs to $\NF$.
\item If $\phi, \psi\in \NF$, then also 
$\nott\phi$, $\phi\andd\psi$, $\necc_A\phi$,
and $\necc_{\CC}^*\phi$
belong to $\NF$.
\item If $\alpha\in \NF$, $\phi\in \NF$, and
 $\CC\subseteq \Agents$,
then $[\alpha]\necc^*_{\CC}\phi$ belongs to $\NF$.
\item If $k\geq 0$, if $\vec{\psi^1}$, $\ldots$,  $\vec{\psi^k}$
is a sequence of sequences of length $n(\bSigma)$ of elements of $\NF$,
and  if $\sigma_1,\ldots, \sigma_k\in \Sigma$,
\begin{equation}
(\cdots 
((\sigma_1\vec{\psi^1} \then \sigma_2\vec{\psi^2})\cdots  \sigma_k\vec{\psi^k})
\label{eq-nf}
\end{equation}
is an action term in $\NF$.
\end{enumerate}
\end{definition}

\begin{lemma}
There is a function $\map{\nf}{\lang_1(\bSigma)}{\NF}$ such that
for all $\phi$,  $\proves\phi\iiff \nf(\phi)$.  Moreover,
if $\phi\in\lang_0(\bSigma)$,  then $\nf(\phi)$ is a purely
modal sentence (it contains no actions).
\label{lemma-nfs-statement}
\end{lemma}

\begin{lemma}  There is a well-order $<$
on   $\lang_1(\bSigma)$ such 
that for all $\phi$ and $\alpha$:
\begin{enumerate}
\item $\nf(\phi)\leq \phi$.
\item $\Pre(\alpha) < [\alpha]\phi$.
 
\item If  $\alpha\arrowA \beta$, then $[\beta]\phi <
\nott [\beta]\phi <  [\alpha]\necc^*_{\CC}\phi$.
%\item   If  $\alpha\arrowA \beta$, then
%$\pair{\beta}\nott\psi < [\alpha]\necc_{\CC}^*\psi$.
\end{enumerate}
\label{lemma-property-nfs-statement}
\end{lemma}

\begin{lemma}
For every $\phi$ there is a finite set $f(\phi)\subseteq \NF$
such that
\begin{enumerate}
\item $f(\phi)$ is closed under subsentences.
\label{part-s}
\item
If $[\justplain{\gamma}]\necc_{\CC}^*\chi\in f(\phi)$,
$\justplain{\gamma}\rightarrow_{\CC}^{*}  \justplain{\delta}$, 
and $A\in \CC$,   then
$f(\phi)$ also contains
$\necc_A [\justplain{\delta}]\necc_{\CC}^*\chi$,
$[\justplain{\delta}]\necc_{\CC}^*\chi$,
$\nf(\Pre(\delta))$,
and $\nf([\justplain{\delta}]\chi)$.
\label{part-crucial}
\end{enumerate}
\label{lemma-FL}
\end{lemma}

The proofs of these will appear in Section~\ref{section-proofs}
just below.  The proofs are technical,
and so the reader not interested in those details might wish
to omit the next section on a first reading of this paper.
We shall use Lemmas~\ref{lemma-nfs-statement}-~\ref{lemma-FL}
in the work on completeness below, but neither
the details of the proofs nor the other results of 
Section~\ref{section-proofs} will be used in the rest of this paper.

\subsection{Proofs of the main facts on normal forms and the well-order $<$}
\label{section-proofs}

We turn to the proofs of  Lemmas~\ref{lemma-nfs-statement}-~\ref{lemma-FL}
just above.
Our proofs are complicated and circuitous, so there
might be shorter arguments.  For example, we don't know of any proofs
that avoid term rewriting theory. 
For the record, here are some of the reasons why we
feel that the study of our system is complicated:
\begin{enumerate}
\item The original statement of   axioms such as the Action-Knowledge
Axiom is in terms of actions of the form $[\sigma_i\vec{\psi}]$.    
\item  On the other hand, the Action Rule is best stated in terms
of actions which are compositions of the actions $[\sigma_i\vec{\psi}]$.    
In a term-rewriting setting, this point and the previous one
work against each other.  Our work will be to work with the 
versions of the axioms that are generalized to the case of all 
simple actions.

\item Again mentioning term rewriting, we'll need to pick an orientation
for the Composition Axiom.  This will either be
$[\alpha][\beta]\phi \leadsto [\alpha\circ\beta]\phi$, or
$[\alpha\circ\beta]\phi \leadsto [\alpha][\beta]\phi$.\footnote{Another 
possibility which we did not explore is to consider
rewriting modulo the identity 
$[\alpha\circ\beta]\phi = [\alpha][\beta]\phi$.}
Both alternatives lead to difficulties at various points.
We chose the first alternative, and for this reason, we'll
need a formulation of the Action-Knowledge Axiom as an infinite scheme.

\item  Our language has the program union operator
$\union$, but because the axioms are not in general
sound for sums, we need to reformulate things to avoid $\union$.  This 
is not difficult, but it would have been nice to avoid it.
\end{enumerate}

It is convenient to replace $\lang_1(\bSigma)$ by a slightly different
language which we call $\lang_1^+(\bSigma)$. 
This new language is shown in Figure~\ref{fig-langplus}.
For the purposes of this section, we stress the formulation
of this as an algebra for a two-sorted signature.
Starting with some fixed action signature $\bSigma$, we construct a
 two-sorted signature
$\Delta = \Delta(\bSigma)$ of terms, obtained in
the following way.
 Let $n$ be the 
number of simple actions in $\Sigma$.

\begin{enumerate}
\item $\Delta$ has two sorts: $s$ (for sentences) and $a$ (for actions).
\item  Each $p\in\AtSen$ is
  a constant symbol of sort $s$. 
\item
$\nott$, $\necc_A$, and $\necc^*_{\BB}$ are  
function  symbols of type $s\rightarrow s$.
\item $\andd$ and $\rightarrow$  are binary function symbols of
type $s\times s\rightarrow s$.
\item Each
  $\sigma\in \Sigma$ is a function symbol of sort $s^n\rightarrow a$.
%\item $\skipp$ is a constant of sort $a$.
\item $\then$ is  function symbols of sort $a\times a 
\rightarrow a$.
\item $app$ is a  binary symbol of type 
$a\times s \rightarrow s$.
\item $\Pre$ is a function symbol of sort $a\rightarrow s$.
\end{enumerate}

\begin{figure}[tbp]
\fbox{
\begin{minipage}{6.0in}
$$
\begin{array}{lcc|c|c|c|c|c|c|c}
\mbox{sentences $\phi$}  &  & p_i  & \neg\varphi &  
\varphi\wedge\psi & \varphi\rightarrow \psi
 &\necc_A\varphi  & \necc_{\BB}^*\varphi
 & [\alpha]\varphi  & \Pre(\alpha) \\  
\mbox{actions $\alpha$}  &      &        \sigma
\psi_1,\ldots,
\psi_{n}  &
\alpha \then \beta  &  &  &  & &  \\ \\
\end{array}
$$
%%%
\end{minipage}
}
\caption{The language $\lang_1^+(\bSigma)$\label{fig-langplus}}
\end{figure}

 The most important 
addition here is that we have $\Pre$ as a first-class part of the
syntax; previously it had been an abbreviation.    We also add
the implication symbol $\rightarrow$, but this is only for
convenience.  Obviously, $\rightarrow$ may
be dropped from the system.  On the other
hand, we dropped
$\union$ (as we mentioned,  some axioms are not sound as
equations in general if we have $\union$).  We might as 
well drop $\skipp$ as well, since it, too,
 can be translated away.  

Incidentally,
the fact that $\Pre$ is not a symbol of 
$\lang_1(\bSigma)$ makes the issue of translating between
$\lang_1(\bSigma)$ and $\lang_1^+(\bSigma)$ delicate.
The reader might wish to formulate a careful translation
in order to appreciate some of the features of our ordering
$<$ which we shall introduce in due course.

We adopt the usual notational conventions that $\andd$, $\iiff$,
and $\then$, are 
used as infix symbols.  We continue our practice of writing
$\sigma_i\vec{\psi}$ for what technically would be 
$\sigma_1(\psi_1,\ldots,\psi_n)$.
Also, we write 
 $[\alpha]\psi$ instead of 
 $app(\alpha,\psi)$.

When dealing with $\lang_1^+(\bSigma)$, we let $\alpha$, $\beta$, etc., range over
terms of sort $a$, and $\phi$, $\psi$, $\chi$, etc., range over
terms of sort $s$. 
Finally, we use letters like $t$ and $u$ for terms of either sort (so as
to shorten many of our statements).
 We also will need to adjoin new variables to
our signature in order to formulate 
the notion of {\em substitution\/} that
leads to a {\em term rewriting system}.  For this,
we let  
$X_{a}$ and $X_{s}$
be sets of new symbols.
In order to simplify our notation, we'll use letters like $x$, $y$, and $z$
to range over {\em both\/} of these sets.  That is, we will not notationally
distinguish the two sorts of variables.    The context will always make it
clear what the sort of any given variable is.

Let $\lang_1^+(\bSigma,X)$ be the terms built from our signature
which now may contain the variables from $X$.   Examples of
such terms may be found in Figure~\ref{fig-RR}.

\rem{
\begin{lemma}
Let $\phi$ be  a sentence  or action of 
$\lang_1(\bSigma)$.
Then there is some $\phi^*\in \lang_1(\bSigma)$
 which does not contain $\union$, $\crash$
 or $\skipp$ such that $\proves \phi\iiff \phi^*$.
Moreover, $\phi^*$ may
be regarded as a ground term of $\lang_1^+(\bSigma)$.
\label{lemma-regard}
\end{lemma}

\begin{proof}  We show by induction $k$ that
every sentence $\phi$ with at most $k$
$\then$ and $\union$ symbols is provably equivalent to
a sentence $\phi^*$ without $\then$, $\union$, $\crash$
 or $\skipp$.
For $k= 0$, the result is an easy induction using the 
Skip and Crash Axioms, Lemma~\ref{lemma-equivalence-actions}
part (\ref{lemma-equivalence-actions-part-sigma}) and the Substitution
Lemma~\ref{lemma-substitution-biggersystem},
and also necessitation.  
The induction step on $k$ is similar, using also the
Composition and Choice Axioms.
\end{proof}
}

\begin{lemma}
Let $\phi$ be  a sentence   of
$\lang_1(\bSigma)$.
Then there is some $\phi^\dag\in \lang_1(\bSigma)$
 which does not contain $\union$, $\crash$
 or $\skipp$ such that $\proves \phi\iiff \phi^\dag$.
 \label{lemma-regard}
\end{lemma}

\begin{proof} (J.~Sack)
We define $\phi\mapsto\phi^\dag$ as follows:
$$\begin{array}{lcl}
\true^\dag & \quadeq & \true\\
p_i^\dag &  \quadeq &  p_i \\
(\nott \phi)^\dag  & \quadeq & \nott \phi^\dag\\
(\phi\andd\psi)^\dag & \quadeq & \phi^\dag\andd\psi^\dag\\
(\necc_A \phi)^\dag& \quadeq & \necc_A\phi^\dag\\
(\necc^*_\BB \phi)^\dag& \quadeq & \necc^*_\BB\phi^\dag\\
\end{array}
\qquad 
\begin{array}{lcl}
([\skipp]\phi)^\dag  & \quadeq & \phi^\dag  \\
([\crash]\phi)^\dag  & \quadeq & \true
  \\
  ([\sigma_i \psi_1\cdots \psi_n]\phi)^\dag & \quadeq & 
  [\sigma_i \psi^\dag _1\cdots \psi^\dag _n]\phi^\dag   \\
([\pi\union\rho]\phi)^\dag & \quadeq & ([\pi]\phi)^\dag \andd ([\rho]\phi)^\dag 
\\
([\pi\then\rho]\phi)^\dag & \quadeq & ([\pi][\rho]\phi)^\dag \\ 
\\
\end{array}
$$
The definition is by recursion on the number  $k$ of 
composition 
 symbols ($\then$)  in $\phi$.
For a fixed $k$, we then use recursion on the total number of symbols.
The overall recursion allows us to define 
This allows us to define 
$([\pi\then\rho]\phi)^\dag$ to be $([\pi][\rho]\phi)^\dag $.  And 
the 'inside' recursion on the number of symbols allows us to define
$([\pi\union\rho]\phi)^\dag$ in terms of $([\pi]\phi)^\dag$ 
and $([\rho]\phi)^\dag$.

We indicate two of the verifications that our definition works.
Here are the details for the line
involving 
 $([\sigma_i \psi_1\cdots \psi_n]\phi)^\dag$ .
 Assuming the relevant induction hypotheses, we see that
 $\sigma_i \psi_1\cdots \psi_n \equiv \sigma_i \psi^\dag _1\cdots \psi^\dag _n$
 (see  Lemma~\ref{lemma-equivalence-actions}, part~\ref{lemma-equivalence-actions-part-sigma}).
 And then by Lemma~\ref{lemma-equivalence-actions-two}, 
 $$[\sigma_i \psi_1\cdots \psi_n]\phi
 \quadequiv  [\sigma_i \psi^\dag _1\cdots \psi^\dag _n]\phi .$$
 Using necessitation, we also have equivalence to
 $ [\sigma_i \psi^\dag _1\cdots \psi^\dag _n]\phi^\dag $.

Some of the other cases use necessitation in this way also.
In all cases, the verifications are similar.
\end{proof}

\begin{remark}
In the remainder of this paper, we shall assume that sentences
and actions of
$\lang_1(\bSigma)$ do not contain $\union$, $\crash$
 or $\skipp$.  
 We also regard them as \emph{ground terms} of $\lang_1^+(\bSigma)$:
 these  are terms without
variables.

Incidentally, the proof of Lemma~\ref{lemma-regard} shows that
$\then$ is also eliminable.  However, we shall {\em not\/} assume
that $\then$ is not found in our sentences.  In fact, the rewriting
system $\RR$ that we present shortly will introduce compositions.   
\end{remark}

\paragraph{The rewriting system $\RR$ and its interpretation}

We recall here the general notion of term rewriting as it applies
to  $\lang_1^+(\bSigma,X)$.
Consider a {\em rewrite rule\/} (r) of the form $l\leadsto r$, where
$l$ and $r$ are elements of  $\lang_1^+(\bSigma,X)$.
This (r) generates a relation of {\em immediate rewriting\/}
on  $\lang_1^+(\bSigma,X)$: we say that 
{\em $t_1$ rewrites to $t_2$ via (r)} if there is  a term $u$
with exactly one occurrence of a variable $x$ (of either sort)
and a substitution $\sigma$ 
such that $u(x\leftarrow l^\sigma) = t_1$ and $v(x\leftarrow r^\sigma) = t_2$.
That is, $t_1$ and $t_2$ result from $u$ by substituting 
$l$ and $r$ in for $x$; however, we need not use $l$ and $r$ literally, but
we could as well take a substitution instance of them.
We would write $t_1\arrowr t_2$ for this.

Given a set
$\RR$ of rewrite rules, 
we write $t_1\arrowRR t_2$ if for some $r\in \RR$,  $t_1\arrowr t_2$.
We  naturally consider the transitive closure $\arrowRRstar$ of $\arrowRR$.
We say that {\em $t_1$ rewrites to $t_2$ via $\RR$\/} if they stand in
this relation $\arrowRRstar$.

We are mostly interested in the rewriting of ground terms.
But we need the notion of variables to formulate this, and this
is the only reason for introducing variables.

\paragraph{The rewriting system $\RR$.}
Now that we have made this preliminary discussion,
we consider the term rewriting
system  $\RR$  shown in Figure~\ref{fig-RR}.
We have numbered the rules, and we use this numbering in the sequel.

\begin{figure}[tbp]
\fbox{
\begin{minipage}{6.0in}$$
\begin{array}{llcll}
\mbox{(r1)} &
x\iif y & \leadsto & \nott(x\andd\nott y)\\
\mbox{(r2)}  &
\Pre(\sigma_i y_1 \cdots y_n) & \leadsto & y_i\\
\mbox{(r3)}  &
\Pre(x\circ y) & \leadsto & \Pre(x)\andd [x]\Pre(y)\\
\mbox{(r4)}  &
{}[x]  p  &  \leadsto &  \Pre(x) \iif p  \\
\mbox{(r5)}  &
{}[x] \nott y   &  \leadsto  &  
 \Pre(x) \iif \nott[x]y\\
\mbox{(r6)}  &
{}[x] (y\andd z)  &  \leadsto &
 [x]  y \andd  [x]  z \\
\mbox{(r7$_{\alpha}$)}  &
 {} [\alpha]\necc_A x  &  \leadsto &
      \Pre(\alpha) \iif\bigwedge\set{
        \necc_A [\beta]x :
\alpha\arrowA \beta} \\
\mbox{(r8)} &
{} [x][y]z &  \leadsto &[x\circ y]z \\
%\mbox{(r9)} & {}[\skipp] x & \leadsto & x \\
%\mbox{(r10)} & {}[x \union y] z & \leadsto & [x]z \andd [y]z \\
\mbox{(r9)} & x\then(y\then z) & \leadsto & 
 (x\then y)\then z \\
\end{array}
$$
\end{minipage}
}
\caption{The rewrite system $\RR$\label{fig-RR}}
\end{figure}

\rem{
\begin{figure}[tbp]
\fbox{
\begin{minipage}{6.0in}$$
\begin{array}{llcll}
\mbox{(r1)} &
\phi\iif \psi & \leadsto & \nott(\phi\andd\nott\psi)\\
\mbox{(r2)}  &
\Pre(\sigma_i \vec{\psi}) & \leadsto & \psi_i\\
\mbox{(r3)}  &
\Pre(\alpha\circ\beta) & \leadsto & \Pre(\alpha)\andd [\alpha]\Pre(\beta)\\
\mbox{(r4)}  &
{}[\alpha]  p  &  \leadsto &  \Pre(\alpha) \iif p  \\
\mbox{(r5)}  &
{}[\alpha] \nott\phi   &  \leadsto  &  
 \Pre(\alpha) \iif \nott[\alpha] \phi\\
\mbox{(r6)}  &
{}[\alpha] (\phi\andd\chi)  &  \leadsto &
 [\alpha]  \phi  \andd  [\alpha]  \chi \\
\mbox{(r7$_{t}$)}  &
 {} [t]\necc_A\phi  &  \leadsto &
      \Pre(t) \iif\bigwedge\set{
        \necc_A [u]\phi :
t\arrowA u} \\
\mbox{(r8)} &
{} [\alpha][\beta]\phi &  \leadsto &[\alpha\circ\beta]\phi \\
\mbox{(r9)} & \alpha\then(\beta\then \gamma) & \leadsto & 
 (\alpha\then \beta)\then \gamma \\
\end{array}
$$
\end{minipage}
}
\caption{The rewrite system $\RR$.\label{fig-RR}}
\end{figure}
}

\begin{remark}
(r7$_{\alpha}$) 
is an infinite scheme, and   in it, $\alpha$ is a term of action sort.
But $\alpha$ need not be an action variable.
(Indeed, the scheme is only interesting when $\alpha$ is a composition
of actions of the form $\sigma_i\vec{\psi}$.)
The right side of (r7$_\alpha$)  depends on $\alpha$, and this is why
we use a scheme rather than a single rule.   
We need to explain what it means.  
First, the action terms of $\lang_1^+(\bSigma,X)$ carry the structure
of a frame by taking for each $A$ the smallest
relation $\arrowA$ such that the following two conditions
hold:
\begin{enumerate}
\item If $i\arrowA j$ in $\bSigma$, then
$\sigma_i\vec{\psi}\arrowA \sigma_j\vec{\psi}$.
\item if $\alpha\arrowA \alpha'$ and 
 $\beta\arrowA \beta'$, then $\alpha\then\alpha'\arrowA \beta \then\beta'$.
\end{enumerate}
Moreover,
we use the notation  
$\bigwedge\set{
        \necc_A [\beta]\phi :
\alpha\arrowA \beta}$ to stand for some fixed, but arbitrary, rendering of the
conjunction as a nested list of binary conjuncts.  The order will
not matter, but we shall need  an estimate of how many
conjuncts there are.   

We have a length function $\ell$ on
action terms: $\ell(x) = 0$ for variables $x$,
 $\ell(\sigma_i\vec{\psi}) = 1$, and $\ell(\alpha\then\beta) = 
\ell(\alpha) + \ell(\beta)$.  For each $\alpha$, the number of $\beta$ such that
$\alpha\arrowA \beta$ is at most $n^{\ell(\alpha)}$,
where $n$ is the size of $\Sigma$.
  This is easy to check by induction 
on $\ell(\alpha)$.

Also on (r7$_{\alpha}$): one is tempted to replace this infinite scheme by
the finite one
$$\begin{array}{llcl}  
\mbox{(r10)}  & [\sigma_i\vec{\psi}]\necc_A\phi  &  \leadsto &
    \psi_i \iif\bigwedge\set{
        \necc_A [\sigma_j\vec{\psi}]\phi :
i\arrowA j} \\
\end{array}
$$
If one did this, 
one would have to reverse (r8)
to read
$$
\begin{array}{llcl}
\mbox{(r11)} & [\alpha\circ\beta]\phi &  \leadsto &[\alpha][\beta]\phi
\end{array}
$$
The reason is that (r10) above only allows us to reduce
expressions $[\alpha]\necc_A\phi$ when $\alpha$ is of the form
$\sigma_i\vec{\psi}$.  The problem is that the
normal forms of  (r1)--(r6) + (r9) + (r10) + (r11) would include terms like
$$ 
[\sigma_1\vec{\psi^1}][\sigma_2\vec{\psi^2}]\cdots [\sigma_k\vec{\psi^k}]\necc^*_{\CC}\phi,
$$
and these are not suitable for use in connection with the Action Rule.
We will show about  our system (r1) -- (r9) that its normal forms includes terms like
$$ 
[(\cdots (\sigma_1\vec{\psi^1}\then \sigma_2\vec{\psi^2}) \cdots
 \sigma_k\vec{\psi^k})]\necc^*_{\CC}\phi.
$$
Again, this is mainly due to our choice in the direction of (r8). 
\end{remark}

\paragraph{Our interpretation}
An {\em interpretation} of a signature $\Delta$ is a {\em $\Delta$-algebra}.
This is a {\em carrier set\/} for the sentences, a carrier set for the 
actions, and  for each
$n$-ary function symbol $f$ of $\Delta$
a function of the appropriate sort.
   In our setting,
both carrier sets will be $N_{\geq 3}$, the set of natural numbers
which are at least 3.   We shall use $a$, $b$, etc., to range over
$N_{\geq 3}$ in this section.
Our interpretation 
shown in Figure~\ref{fig-interpretation}.  
\rem{In the figure, we have abused notation somewhat because
it matches our future use.  We should write statements such as
$$\begin{array}{lcl}
\interpretation{\nott} & \quadeq & \lambda x. x+1 \\
\interpretation{\andd} &  \quadeq & \lambda xy. x+y   \\
\interpretation{\iif} &  \quadeq & \lambda xy. x+y +3    \\
\end{array}
$$
}
By recursion on terms $t(x_1,\ldots, x_k)$, we build an 
interpretation\footnote{We use the same symbol
$\interpretation{\ }$ for our interpretation
as we used for the 
valuation in a state model  in Section~\ref{section-state-models}.
Since our use of interpretations is confined to this
section of the paper, we feel that the confusion due to
overloading the symbol $\interpretation{\ }$ 
should be minimal.}
$$\map{\interpretation{t}(a_1,\ldots, a_k)}{(N_{\geq 3})^k}{N_{\geq 3}}.$$
The interpretation of each function symbol is strictly monotone
in each argument.  
(For example, we check this for $\interpretation{\then}$.
For fixed $n$ and $m$, we consider the functions
$\lambda a. n^{a+2}$ and $\lambda a. a^{m+2}$.  
Then if $a < b$, $n^{a+2} < n^{b+2}$  and  $a^{m+2} < b^{m+2}$.)  
So by induction, each $\interpretation{t}(a_1,\ldots, a_k)$ is
strictly monotone as a function in each argument. 

\begin{figure}[tbp]
\fbox{
\begin{minipage}{6.0in}
$$\begin{array}{lcl}
\interpretation{p_i} & \quadeq & 3 \\
\interpretation{\nott}(a) & \quadeq &   a  + 1 \\
\interpretation{\andd}(a, b) &  \quadeq &  a  + b   \\
\interpretation{\iif}(a, b) &  \quadeq &  a  + b   + 3\\
\interpretation{\necc_A}(a) & \quadeq &  a  + 2 \\
\end{array}
\qquad
\begin{array}{lcl}
\interpretation{\necc_\BB^*}(a) & \quadeq &  a  + 1 \\
\interpretation{\app}(a, b) & \quadeq &   a ^b  \\
\interpretation{\Pre}(a) & \quadeq &  a  \\
\interpretation{\sigma_i}(a _1,\ldots, a _n) 
 & \quadeq &   a _1 + \cdots +  a _n +1\\
\interpretation{\then}(a, b) & \quadeq &  
 a ^{b  + 1} \\
%\interpretation{\skipp} & \quadeq & 3 \\
%\interpretation{+}(a,b) & \quadeq & a b \\
\end{array}
$$
\end{minipage}
}
\caption{The interpretation\label{fig-interpretation}}
\end{figure}

\rem{
\begin{figure}[tbp]
\fbox{
\begin{minipage}{6.0in}
$$\begin{array}{lcl}
\interpretation{p_i} & \quadeq & 3 \\
\interpretation{\nott \phi} & \quadeq & \interpretation{\phi} + 1 \\
\interpretation{\phi\andd\psi} &  \quadeq & \interpretation{\phi}  + \interpretation{\psi} \\
\interpretation{\phi\iif\psi} &  \quadeq & \interpretation{\phi}  + \interpretation{\psi} + 3 \\
\interpretation{\necc_A\phi} & \quadeq & \interpretation{\phi} + 2 \\
\interpretation{\necc_\BB^*\phi} & \quadeq & \interpretation{\phi} + 1 \\
\interpretation{[\alpha]\phi} & \quadeq &  {\interpretation{\alpha}}^{\interpretation{\phi}} \\
\interpretation{\Pre(\alpha)} & \quadeq & \interpretation{\alpha} \\
\interpretation{\sigma_i\psi_1\cdots\psi_n} 
 & \quadeq &  \interpretation{\psi_1} + \cdots + 
\interpretation{\psi_n} +1\\
\interpretation{\alpha\then\beta} & \quadeq &  
{\interpretation{\alpha}}^{\interpretation{\beta} + 1} \\
\end{array}
$$
\end{minipage}
}
\caption{The interpretation\label{fig-interpretation}}
\end{figure}
}

\begin{lemma}   Let $n  = |\Sigma|$.
For all  action terms $\alpha$, 
and all maps $\iota$ of  variables  to $N_{\geq 3}$.
 $\interpretation{\alpha}(\iota(x_1),\ldots,\iota(x_n))   >  n^{l(\alpha)}$.
\label{lemma-first-interpretation}
\end{lemma}

\begin{proof}
By induction on $\alpha$. 
If $\alpha$ is a variable $x$, then  $\ell(\alpha) =0$.  
So
 $\interpretation{\alpha}(\iota(x)) = \iota(x) \geq 3 > 1 =  
  n^{\ell(\alpha)}$.
If $\alpha$ is a term of the form $\sigma_i\vec{\psi}$, then
$\ell(\alpha) = 1$, and 
$$\begin{array}{lcl}
\interpretation{\sigma_i\psi_1\cdots\psi_n}(\iota(x_1),\ldots,\iota(x_n)) 
& \quadeq &
\interpretation{\psi_1}(\iota(x_1),\ldots,\iota(x_n))  + \cdots + 
\interpretation{\psi_n}(\iota(x_1),\ldots,\iota(x_n))  
\\ & \quad > \quad & 3n \\
&  \quad > \quad  & n^{\ell(\alpha)} 
\end{array}
$$
Now assume our lemma for $\alpha$ and $\beta$.
$$\begin{array}{lcll}
\interpretation{\alpha\then\beta}(\iota(x_1),\ldots,\iota(x_n))
& \quadeq & (\interpretation{\alpha}(\iota(x_1),\ldots,\iota(x_n)))^{
 \interpretation{\beta}(\iota(x_1),\ldots,\iota(x_n)) +1}\\
& \quad > \quad  & n^{\ell(\alpha)\then (n^{\ell(\beta) +1})}\\
& \quad > \quad  & n^{\ell(\alpha) + \ell(\beta)}   \qquad \mbox{see below}  \\
& \quadeq  & n^{\ell(\alpha\then\beta)}  \\
\end{array}
$$
In asserting that  $\ell(\alpha)\then n^{\ell(\beta)} > \ell(\alpha) + \ell(\beta)$,
we assume that $n > 1$.  If $n = 1$, then our lemma is trivial.
\end{proof}

\begin{proposition}
If $\alpha$ is an action term of $\lang_1^+(\bSigma)$ and $\alpha\arrowA\beta$,
then $\interpretation{\alpha}$ and $\interpretation{\beta}$ are the same function.
In particular, if $\alpha$ and $\beta$ are ground action terms, then
 $\interpretation{\alpha} = \interpretation{\beta}$.
\label{proposition-same-function}
\end{proposition}
 
\begin{proof}
By induction on $\alpha$.  The point is that the interpretation of
the actions $\sigma_i\vec{\psi}$ does not use the number $i$ in any way.
The rest follows by an easy induction.
\end{proof}

The following proposition is a technical  but elementary
 result on exponentiation.
 It will be used in Lemma~\ref{lemma-adequate-interpretation}
below.  One might notice  that if we replace $3$ by $2$ in 
the statement, then
some of the parts are no longer true.  This is the reason why
our overall interpretation uses $N_{\geq 3}$ rather than
$N$ or $N_{\geq 2}$.

\begin{proposition}
Let $a, b,c\geq 3$.
\begin{enumerate}
\item $ 2 a^{ b }
> a + 4$.
\item $a^{b+c} > a^b + a^c$.
\item ${a}^{2}
 > 
a + 5$.
\item $a^{ b +1} > 3a + 3$.
\item $a^b > (a+1)(b+1)$.
%\item $b^{c+1} + 1 > (b+1) c$.
\end{enumerate}
\label{prop-arithmetic}
\end{proposition}

\begin{lemma} 
  For each of (r1) - (r9), the interpretation of
the left side of the rule is 
strictly larger than the right side on all tuples of arguments
in $N_{\geq 3}$.  
\rem{
Let $\iota$ map  variables  to $N_{\geq 3}$.
Then for all rules $l\rightarrow r$ of $\RR$,
$$
\interpretation{l}(\iota(x_1),\ldots,\iota(x_n))  > 
\interpretation{r}(\iota(x_1),\ldots,\iota(x_n)).
$$
}
\label{lemma-adequate-interpretation}
\end{lemma}

\begin{proof}
We remind the reader that letters $a, b,c$, etc.,
denote elements of $N_{\geq 3}$.
 We use Proposition~\ref{prop-arithmetic}
without explicit mention.

(r1) 
%is 
%$\phi\iif \psi   \leadsto   \nott(\phi\andd\nott\psi)$.
$\interpretation{x\iif y}(a,b)  =    a + b + 3
>  a + b +2 =
 \interpretation{\nott(x\andd\nott y)}(a,b)$.

(r2) $\interpretation{\Pre(\sigma_i \vec{y})}(a_1,\ldots, a_n) =
\interpretation{\sigma_i\vec{y}}(a_1,\ldots, a_n)  
\geq
a_i + 1 > \interpretation{y_i}(a_1,\ldots, a_n)$.

(r3) 
$$
\begin{array}{lcl}
\interpretation{\Pre(x\circ y )}(a,b)  & \quadeq &   
   a^{b+1} 
\\
 & \quad > \quad &  
a + a^{b} \\
 & \quadeq &  
\interpretation{\Pre(x)}(a,b) + \interpretation{[x]\Pre(y)}(a,b)  \\
 & \quadeq &  
\interpretation{ \Pre(x)\andd [x]\Pre(y)}(a,b)\\
\end{array}
$$

(r4) 
$\interpretation{[x]  p}(a)  = 
{a}^{2}
 > 
a + 5 = 
\interpretation{\Pre(x) \iif p}(a)$.
%(We used the fact that $a \geq 3$.)

(r5)
$$
\begin{array}{lcll}
\interpretation{[x] \nott y}(a,b)   &  \quadeq  &  
 a^{ b  + 1} \\
& \quad\geq\quad & 
a^{ b } 
 + 2 a^{ b }  \\ & \quad > \quad &
a   +   a^{ b }  + 4 
\\
& \quadeq & a   +   \interpretation{[x]y}(a,b)  + 4 \\
 &  \quadeq  &  
\interpretation{\Pre(x) \iif \nott[x] y}(a,b)  \\
\end{array}
$$

(r6)  
$\interpretation{[x] (y\andd z)}(a,b,c)
= a^{ b  + c}
 >  
a^{ b }
+ 
a^c = 
\interpretation{[x] y  \andd  [x]  z}(a,b,c)$.

(r7$_{\alpha}$)  Let  $y_1,\ldots, y_n$ be
the free variables of  $\alpha$.  Fix $c_1,\ldots, c_n\in N_{\geq 3}$.
 Let $a = \interpretation{\alpha}(c_1,\ldots, c_n)$, and let $b\in N_{\geq 3}$ be 
arbitrary.
Then
$$\begin{array}{lcll}
 \interpretation{[\alpha]\necc_A x}(c_1,\ldots, c_n,b) 
& \quadeq & a^{ b +2} \\
& \quad >\quad & a^{ b +1} +
a^{ b +1} 
\\ & \quad > \quad &   
a^{ b  +1 } + 3a + 3
  \\
& \quadeq &
a +   a(2 +
a^{ b }) + 3
\\
& \quad > \quad &
a +  n^{\ell(a)}(2 + a^{ b })
+ 3
\\
&\quad \geq \quad &
\interpretation{\Pre(\alpha) \iif\bigwedge\set{
        \necc_A [\beta]x :
\alpha\arrowA \beta}}(c_1,\ldots, c_n,b) \\
\end{array}
$$
Notice that we used Proposition~\ref{proposition-same-function}
and our observation that the size of the conjunction
on the right side 
of (r7$_{\alpha}$)  is at most $n^{\ell(\alpha)}$.

(r8)  
$ \interpretation{[x][y]z}(a,b,c)
=
 a^{b^c}
  >  
a^{(b+1)c}
=
(a^{b+1})^c
=
\interpretation{[x\circ y ]z}(a,b,c)
$.

% (r9) $\interpretation{[\skipp] x})(a) = 3^a > a = \interpretation{x}(a)$.

% (r10) $\interpretation{[x + y] z})(a,b,c)
%= (ab)^c = a^c b^c > a^c + b^c =   \interpretation{[x] z + [y]z})(a,b,c)$.

(r9)  
$ \interpretation{x\then (y\then z)}(a,b,c)
 = a^{b^{c+1} + 1}
  >  
a^{(b+1)(c+1)}
 =
(a^{b+1})^{c+1}
 =
\interpretation{(x\then  y)\then z}(a,b,c)
$.
\end{proof}

\rem{
(r1) 
%is 
%$\phi\iif \psi   \leadsto   \nott(\phi\andd\nott\psi)$.
$\interpretation{\phi\iif \psi}  =   \interpretation{\phi} + \interpretation{\psi} + 3  
> \interpretation{\phi} + \interpretation{\psi} +2 =
 \interpretation{\nott(\phi\andd\nott\psi)}$.

(r2) $\interpretation{\Pre(\sigma_i \vec{\psi)}} = \interpretation{\sigma_i \vec{\psi}}  
\geq
\interpretation{\psi_i} + 1 > \interpretation{\psi_i}$.

(r3) 
$$
\begin{array}{lcl}
\interpretation{\Pre(\alpha\circ\beta)}  & \quadeq &   
\interpretation{\alpha\circ\beta} 
 \\
  & \quadeq &    \interpretation{\alpha}^{\interpretation{\beta}+1} 
\\
 & \quad > \quad &  
\interpretation{\alpha} + \interpretation{\alpha}^{\interpretation{\beta}} \\
 & \quadeq &  
\interpretation{\Pre(\alpha)} + \interpretation{[\alpha]\Pre(\beta)}  \\
 & \quadeq &  
\interpretation{ \Pre(\alpha)\andd [\alpha]\Pre(\beta)}  \\
\end{array}
$$

(r4) 
$\interpretation{[\alpha]  p}  = 
{\interpretation{\alpha}}^{2}
 > 
\interpretation{\alpha} + 5 = 
\interpretation{\Pre(\alpha) \iif p}$.
(We used the fact that $\interpretation{\alpha} \geq 3$.)

(r5)
Since $\interpretation{\phi} \geq 2$ and $\interpretation{\alpha} \geq 3$,
$ 2 \interpretation{\alpha}^{\interpretation{\phi}}
> \interpretation{\alpha} + 4$.
Using this, we calculate:
$$
\begin{array}{lcll}
\interpretation{[\alpha] \nott\phi}   &  \quadeq  &  
 \interpretation{\alpha}^{\interpretation{\phi} + 1} \\
& \quad\geq\quad & 
\interpretation{\alpha}^{\interpretation{\phi}} 
 + 2 \interpretation{\alpha}^{\interpretation{\phi}}  \\ & \quad > \quad &
\interpretation{\alpha}   +   \interpretation{\alpha}^{\interpretation{\phi}}  + 4 
\\
& \quadeq & \interpretation{\alpha}   +   \interpretation{[\alpha]\phi}  + 4 \\
 &  \quadeq  &  
\interpretation{\Pre(\alpha) \iif \nott[\alpha] \phi}\\
\end{array}
$$

(r6)  
$\interpretation{[\alpha] (\phi\andd\chi)} 
= \interpretation{\alpha}^{\interpretation{\phi} + \interpretation{\chi}}
 >  
\interpretation{\alpha}^{\interpretation{\phi}}
+ 
\interpretation{\alpha}^{\interpretation{\chi}} = 
\interpretation{[\alpha]  \phi  \andd  [\alpha]  \chi}$.

(r7$_{\alpha}$)
Since $\interpretation{\phi} \geq 2$ and $\interpretation{\alpha} \geq 3$,
$\interpretation{\alpha}^{\interpretation{\phi}+1} > 3\interpretation{\alpha} + 3$.
Using this, we calculate:
$$\begin{array}{lcll}
 \interpretation{[\alpha]\necc_A\phi} 
& \quadeq & \interpretation{\alpha}^{\interpretation{\phi}+2} \\
& \quad >\quad & \interpretation{\alpha}^{\interpretation{\phi}+1} +
\interpretation{\alpha}^{\interpretation{\phi}+1} 
\\ & \quad > \quad &   
\interpretation{\alpha}^{\interpretation{\phi} +1 } + 3\interpretation{\alpha} + 3
  \\
& \quadeq &
\interpretation{\alpha} +   \interpretation{\alpha}(2 +
\interpretation{\alpha}^{\interpretation{\phi}}) + 3
\\
& \quad > \quad &
\interpretation{\alpha} +  n^{\ell(\alpha)}(2 + \interpretation{\alpha}^{\interpretation{\phi}})
+ 3
\\
&\quad \geq \quad &
\interpretation{\Pre(\alpha) \iif\bigwedge\set{
        \necc_A [\beta]\phi :
\alpha\arrowA\beta}}\\
\end{array}
$$

(r8)  Let $a = \interpretation{\alpha}$, $b = \interpretation{\beta}$
and $p = \interpretation{\phi}$.   Since $b\geq 3$ and $p\geq 2$,
$b^p > (b+1)p$.  And then
$$\interpretation{[\alpha][\beta]\phi}
\quadeq
 a^{b^p}
\quad > \quad
a^{(b+1)p}
\quadeq 
(a^{b+1})^p
\quadeq 
\interpretation{[\alpha\circ\beta]\phi}.
$$

(r9)  
 Let $a = \interpretation{\alpha}$, $b = \interpretation{\beta}$,
and 
$c  = \interpretation{\gamma}$.  
Then $b^{c+1} + 1 > (b+1) c$, and so
$$\interpretation{\alpha\then (\beta\then\gamma)} 
\quadeq a^{b^{c+1} + 1}
\quad > \quad
a^{(b+1)c}
\quadeq
(a^{b+1})^c
\quadeq
\interpretation{(\alpha\then  \beta)\then\gamma}.
$$
}

\begin{definition}  Let $\phi$ and $\psi$ be sentences in $\lang_1(\bSigma)$.
We regard each as ground terms in  $\lang_1^+(\bSigma)$.
We write $\phi  > \psi$ if $\interpretation{\phi} > \interpretation{\psi}$,
where our interpretation is the one in 
Figure~\ref{fig-interpretation}.   We write  $\phi < \psi$ if 
    $\phi  \leq \psi$ to mean the obvious things.
\end{definition}
\label{section-order-on-lang-1}

\begin{theorem} $\RR$ is terminating: there are no
infinite  sequences
$$ t_1 \quad \arrowRR\quad  t_2
\quad  \arrowRR \quad  \cdots \quad  \arrowRR \quad  t_n
\quad  \arrowRR \quad  t_{n+1} \quad  \cdots
$$
\label{theorem-termination}
\end{theorem}

\begin{proof} We recall a standard argument.  Assume we had a counterexample,
an infinite sequence of rewrites.
We may assume without loss of generality that each $t_i$ is a ground term.
Recall Lemma~\ref{lemma-adequate-interpretation} and
our observation that each $\interpretation{t}(x_1,\ldots, x_k)$ is
strictly monotone as a function in each argument. 
From this it follows that $\interpretation{t_1} > \interpretation{t_2} > \cdots$,
and these are all numbers.  So we have a contradiction.
\end{proof}

\paragraph{Normal  forms}
\label{section-normal-forms}

We remind the reader that we have formulated the notions
of equivalences of simple actions and also equivalence of
sentences of $\lang_1(\bSigma)$ in Section~\ref{section-equivalence-of-actions}.
We consider now a translation $\map{\trans}{\lang_1^+(\bSigma)}{\lang_1(\bSigma)}$.
The definition is by recursion on the well-order $<$.
We take $\trans $ to be a homomorphism for all symbols except $\Pre$.  For this,
we require the following:
$$
\begin{array}{lcl}
%\trans(p) & \quadeq & p \\
%\trans(\nott \phi)  & \quadeq &\nott  \trans(\phi)  \\
% \trans(\phi\andd\psi)   & \quadeq & \trans(\phi)  \andd  \trans(\psi)  \\
% \trans(\necc_A \phi)  & \quadeq &\necc_A  \trans(\phi)  \\
% \trans(\necc^*_{\BB} \phi)  & \quadeq &\necc^*_{\BB} \trans(\phi)  \\
%\trans(\nott \phi)  & \quadeq &\nott  \trans(\phi)  \\
%\trans([\alpha]\phi) & \quadeq & [\trans(\alpha)]\trans(\phi)   \\
\trans(\Pre(\sigma\psi_1,\ldots,\psi_n)) & \quadeq &   \trans(\psi_i) \\
\trans(\Pre(\alpha\then  \beta))
 & \quadeq & \pair{\trans(\alpha)} \trans(\Pre(\beta)) \\
 & \quadeq &   \nott[\trans(\alpha)]\nott \trans(\Pre(\beta)) \\
%\trans(\sigma\psi_1,\ldots,\psi_n)) & \quadeq &  \sigma\trans(\psi_1),\ldots,\trans(\psi_n)  \\
%\trans(\alpha\then  \beta) & \quadeq &
%\trans(\alpha)\then\trans(\beta)   \\
\end{array}
$$
The reason we need  recursion on $<$ is that $\trans(\Pre(\beta))$ figures in
$\trans(\Pre(\alpha\then  \beta))$. 
So we need to know that $\Pre(\beta) < \Pre(\alpha\then  \beta)$. 
Of course, we don't really need
the specific $<$ 
that we constructed for this; for this minor  point we might as well
define $\phi < \psi$ iff $\phi$ has fewer symbols than $\psi$.

\begin{lemma} For all action terms $\alpha$,
$\trans(\Pre(\alpha))  = \Pre(\trans(\alpha))$.
\label{lemma-trans-pre}
\end{lemma}

\begin{proof}
By induction on $<$.  The important thing again is that when 
we write $\Pre(\trans(\alpha))$ we are using $\Pre$ here as a defined symbol.
Its recursion equations match the definition of $\trans$.
\end{proof}

\begin{lemma}
Let $\alpha$ be a ground action term of  $\lang_1^+(\bSigma)$.
Then 
$$\set{\trans(\beta) : \alpha\arrowA \beta \mbox{ in $\lang_1^+(\bSigma)$}}
\quadeq
\set{ \beta : \trans(\alpha) \arrowA \beta \mbox{ in $\lang_1(\bSigma)$}}.$$
%are set-wise equivalent.
%That is, for each element of each set, there is an element of the
%other set equivalent to it.
\label{lemma-alphabeta-equ}
\end{lemma}

\begin{proof}
By induction on $\alpha$.
\end{proof}

\begin{lemma}  If $t_1$ and $t_2$ are ground 
  terms of $\lang_1^+(\bSigma)$ such that
 $t_1\arrowRR t_2$, then $\trans(t_1)\equiv \trans(t_2)$.
\label{lemma-RR-preserves-equivalence}
\end{lemma}

\begin{proof}  
Let $u$ be a term with exactly one
free variable $x$,  and let $\sigma$
be such that for some $i$, $t_1 = u(x\leftarrow l^\sigma)$ and 
$t_2 = u(x\leftarrow r^\sigma)$.
We argue by induction on $u$.
If $u$ is just $x$, then we examine the rules of the system.
All of them pertain to sentences except for (r9), and 
in the case of (r9) we use
Lemma~\ref{lemma-equivalence-actions}, 
part~(\ref{used-lemma-towards-composition-admissibility-2}).
The arguments for the other rules use Lemma~\ref{lemma-trans-pre},
and we'll give the details for two of them (r3) and  (r7).

For (r3), 
$$\begin{array}{lcll}
\trans(\Pre(\alpha\then\beta)) & \quadequiv  & 
 \pair{\trans(\alpha)} \trans(\Pre(\beta))  & \mbox{definition of $\trans$} \\
 & \quadequiv  &
 \Pre(\trans(\alpha))\andd [\trans(\alpha)\trans(\Pre(\beta))  
 & \mbox{by Lemma~\ref{lemma-pfa-general}} \\
  & \quadequiv  &  
\trans(\Pre(\alpha))\andd [\trans(\alpha)]\trans(\Pre(\beta))  
& \mbox{using Lemma~\ref{lemma-trans-pre}} \\
 & \quadequiv  &
\trans(\Pre(\alpha)\andd [\alpha]\Pre(\beta)) \\
\end{array}
$$

The infinite scheme (r7$_{\alpha}$) is repeated below:
$$\mbox{(r7$_{\alpha}$)} \qquad\qquad\qquad
[\alpha]\necc_A x  \quad \leadsto  \quad
      \Pre(\alpha) \iif\bigwedge\set{
        \necc_A [\beta]x :\alpha\arrowA \beta}.$$
Recall that $\alpha$ might have variables.
What we need to check here is that ground instances
of  (r7$_{\alpha}$) have equivalent translations.  
So take a ground sentence of the form $[\alpha]\necc_A\phi$.
Its translation is $[\trans(\alpha)]\necc_A(\trans(\phi))$.
Using the Action-Knowledge Axiom and Lemmas~\ref{lemma-alphabeta-equ}
and \ref{lemma-trans-pre}:
$$\begin{array}{lcl}
[\trans(\alpha)]\necc_A(\trans(\phi)) & \quadequiv &
\Pre(\trans(\alpha)) \iif 
\necc_A \bigwedge\set{   [\beta]\trans(\phi): \trans(\alpha)\arrowA\beta }\\
& \quadequiv &   \trans(\Pre(\alpha)) \iif\bigwedge\set{
        \necc_A [\trans(\beta)]\trans(\phi)  :\alpha\arrowA \beta})  \\
& \quadequiv &  \trans\biggl(\Pre(\alpha) \iif\bigwedge\set{
        \necc_A [\beta]\phi :\alpha\arrowA \beta}\biggr) \\
\end{array} 
$$

This settles all of the cases in this lemma
when $u$ is a variable by itself.  In
the general case, we need a fact about substitutions.
Let $t$ be any term of $\lang^+(\bSigma,X)$.  Let $\sigma$ and $\tau$
be two ground substitutions   (these are substitutions
all of  whose values
contain no variables).
such that $\trans(\sigma(x))$ and $\trans(\tau(x))$ 
equivalent
(actions or sentences) for all variables $x$.
Then $\trans(t^\sigma)$ and $\trans(t^\tau)$ are again equivalent.
The argument here is an easy induction on terms
of  $\lang^+(\bSigma,X)$. 
It amounts to several facts which we have
seen before:  $\equiv$ is a congruence for all of the syntactic
operations on sentences, and also for those on programs
(see Lemma~\ref{lemma-equivalence-actions}); equivalent actions
have equivalent preconditions (by definition);
and if $\alpha \equiv \beta$, then $[\alpha]\phi \equiv [\beta]\phi$
(Lemma~\ref{lemma-equivalence-actions-two}).
\end{proof}

\rem{
\begin{proposition}
There are translations $\map{u}{\lang_1(\bSigma)}{\lang_1^*(\bSigma)}$
and $\map{t}{\lang_1^+(\bSigma)}{\lang_1(\bSigma)}$; these translations
are semantic equivalences.
\label{proposition-two-languages}
\end{proposition}
}

\label{section-nfs}
A {\em normal form\/} in  a rewriting system is a term
which cannot be rewritten in the system.  In a terminating
system such as our $\RR$, one may define a normal form for a term by rewriting
as much as possible in some fixed way (for example, by always
rewriting in the leftmost possible way).  We shall need some
information on the normal forms of $\RR$, and we shall turn to this
shortly.  

We defined a set $\NF\subseteq \lang_1(\bSigma)$ in
 Section~\ref{section-completeness}.
$\NF$ is the smallest set containing the atomic propositions
and closed under all the boolean and modal operators;
 with the property that if $\alpha$ and $\phi$ belong
to $\NF$, then so does $[\alpha]\necc^*_{\CC} \phi$ for
all $\CC\subseteq \Agents$; and finally that the actions
in $\NF$ are closed under composition in the specific way
set out in (\ref{eq-nf}).

\begin{lemma}
A ground
term $t\in \lang_1^+(\bSigma)$ is a normal form of $\RR$ iff $t\in \NF$.
In particular,  $\Pre$ and $\rightarrow$ do  not occur
 in   normal forms.
\label{lemma-nf}
\end{lemma}

\begin{proof} An induction on $\NF\subseteq \lang_1(\bSigma)$ 
shows that all sentences in $\NF$ are
ground terms of $\lang_1^+(\bSigma)$ and are moreover
normal forms
of the rewriting system $\RR$: 
no rules of $\RR$ can apply at any point.  Going the other
way, we show by induction on $\lang_1^+(\bSigma)$
 that if $t$ is a (ground)
 normal form, then (regarding $t$ as a sentence or action in $\lang_1(\bSigma)$)
$\phi\in \NF$. The base case and the induction steps for $\nott$, $\andd$,
$\necc_A$, and $\necc^*_{\BB}$ are all easy.  Suppose our result holds for
$\phi$, and consider a normal form $[\alpha]\phi$.  Then $\alpha$ must be
an action of the form $[\sigma_i\vec{\psi}]$ with all $\psi_j$ in normal
form; if not, some rule would apply to $\alpha$ and hence to
$[\alpha]\phi$.  We are left to consider an action $\alpha$.
By induction hypothesis, the subsentences of $\alpha$, too, are
normal forms and hence belong to $\NF$. 
 We claim that $\alpha$
 must be of the right-branching form
in equation (\ref{eq-nf}).  This is because all of the other
possibilities are reducible in the system using (r9).
 \end{proof}

\begin{lemma}
A sentence $\phi\in\lang_0(\bSigma)$ is a normal form of $\RR$ iff $\phi$ is
a modal sentence (that is, if $\phi$ contains no actions).
\label{lemma-nf-langzero}
\end{lemma}

\begin{proof} By the easy part of Lemma~\ref{lemma-nf}, the modal
sentences are normal forms, In the more significant direction, one first
checks by induction on rewrite sequences that if $\phi\in\lang_0(\bSigma)$
and $\phi'$ is reachable from $\phi$ by a finite number of rewrites in
$\RR$, then $\necc^*$ does not appear in $\phi'$.  So every normal form of
$\phi$ is a purely modal sentence.  It follows that if $\phi$ were a
normal form to begin with, then $\phi$ would be a modal sentence.
\end{proof}

\begin{corollary}
For every $\phi\in\lang_1(\bSigma)$ there is a
normal form $\nf(\phi)$ such that $\proves\phi\iiff \nf(\phi)$.
Moreover,  $\nf(\phi) \leq \phi$.  
\label{corollary-nfs}
\end{corollary}

\begin{proof} 
Regard $\phi$ as a term in $\lang_1^+(\bSigma)$.  Let
$$ \phi = \phi_1 \quad \leadsto \quad  \phi_2 \quad\quad \leadsto \quad \cdots \leadsto
\phi_n = \nf(\phi)$$
be some sequence of rewrites which leads to a normal form of $\phi$.\footnote{It does
not matter which sequence or which normal form is chosen.}
By Lemma~\ref{lemma-RR-preserves-equivalence}, 
$\proves \trans(\phi) \iiff \trans(\nf(\phi))$.  But neither $\phi$ nor
$\nf(\phi)$ contain $\Pre$ or $\arrow$, so they are literally equal
to their translations.  Thus 
 $\proves\phi\iiff \nf(\phi)$, just as desired.
\end{proof}

\begin{lemma}  For all $\phi$ and $\alpha$ of $\lang_1^+(\bSigma)$:
\begin{enumerate}
 \item 
$\Pre(\alpha) < [\alpha]\phi$.
\label{arithmeticineq}
\item If  $\alpha\arrowA \beta$, then $[\beta]\phi <
\nott [\beta]\phi <  [\alpha]\necc^*_{\CC}\phi$.
\label{anotherarithfact}
%\item   If  $\alpha\arrowA \beta$, then
%$\pair{\beta}\nott\psi < [\alpha]\necc_{\CC}^*\psi$.
\label{itemintruth}
\item 
If $[\alpha]\necc_{\CC}^*\psi$
is a normal form sentence and $\alpha\arrowA \beta$, then  
$\necc_A[\beta]\necc_{\CC}^*\psi$
and $[\beta]\necc_{\CC}^*\psi$ are again normal forms.
\item Every subterm of a normal form is a normal form.
\item If $\alpha$ is a normal form action and $\alpha\arrowA \beta$,
then $\beta$ is a normal form action as well.
\end{enumerate}
\label{lemma-property-nfs}
\end{lemma}

\begin{proof}
Part~(\ref{arithmeticineq})
 is an easy calculation.  
In part~(\ref{anotherarithfact}) we use the fact that
$\interpretation{\alpha}$ and $\interpretation{\beta}$ are the same number,
by Proposition~\ref{proposition-same-function}.
  We use the fact here (and here only) that
$\interpretation{\necc^*}(a) = a + 1$ rather than $\interpretation{\necc^*}(a) = a$.
The remaining parts are also easy using 
the characterization of normal forms of Lemma~\ref{lemma-nf}.
\end{proof}

\paragraph{The function $f(\phi)$}
The last piece of business in this section is to 
define the function $f(\phi)$ from Lemma~\ref{lemma-FL}
and then to prove the required properties.

\begin{definition}
For $t$ an action or a sentence, 
let $s(t)$ be  the set of subsentences of $t$,
including $t$ itself (if $t$ is a sentence). 
  This includes all sentences occurring in
actions which occur in $\phi$ and their subsentences.
%For future use, we note that
%\begin{equation}
%\begin{array}{lcl}
%s([\alpha]\necc_{\CC}^*\phi) & \quadeq &
%\set{[\alpha]\necc_{\CC}^*\phi, \necc_{\CC}^*\phi}
%\ \cup \ s(\phi) \ \cup\ \bigcup \set{s(\Pre(\beta)) :
%\alpha\rightarrow_{\CC}^* \beta}
%\end{array}
%\label{eq-s}
%\end{equation}
We define a function $\map{f}{\NF}{\pow(\NF)}$
by recursion on the wellfounded relation $<$ as follows:
$$\begin{array}{lcll}
f(p) & \quadeq & & \set{p} \\
f(\nott \phi) & \quadeq & &f(\phi)\cup  \set{\nott\phi}\\
f(\phi\andd\psi) & \quadeq &  &  f(\phi)\cup f(\psi)
\cup\set{\phi\andd\psi}\\
f(\necc_A \phi) & \quadeq &  &  f(\phi)\cup\set{\necc_A\phi}
 \\
f(\necc_{\BB}^* \phi) & \quadeq &  & f(\phi)  \cup \set{\necc_{\BB}^*\phi}
 \cup \set{\necc_A \necc_{\BB}^*\phi : A\in \BB}
\\
f([\alpha]\necc_{\CC}^*\phi) & \quadeq & \bigcup &
\set{ s(\necc_A [\beta]\necc_{\CC}^* \phi)
 \ : \ 
\alpha\rightarrow_{\CC}^{*}  \beta \ \&\ A\in \CC}\\
% & & \cup &
%\set{ [\beta]\necc_{\CC}^* \phi
% \ : \ 
%\alpha\rightarrow_{\CC}^{*}  \beta \ \&\ A\in \CC}\\
 & & \cup & \bigcup \set{f(\psi)\ :\ 
(\exists \beta) \alpha\rightarrow_{\CC}^* \beta \ \&\ \psi\in s(\beta)} \\
 & & \cup & \bigcup \set{f(\nf(\Pre(\beta)))\ :\ 
\alpha\rightarrow_{\CC}^* \beta } \\
 & & \cup  &
f(\necc_{\CC}^* \phi)\\
& & \cup & \bigcup \set{f(\nf([ \beta ]\phi)) \ : \
\alpha\rightarrow_{\CC}^{*}  \beta
}\\
\end{array}
$$
The definition makes sense because
the calls to $f$ on the right-hand sides are all $<$ the
arguments on the left-hand sides; see Lemma~\ref{lemma-property-nfs}.)
\end{definition}

\begin{lemma} For all $\phi\in\NF$:
\begin{enumerate}
\item $\phi\in f(\phi)$.
\label{basicform-cases}
\item $f(\phi)$ is a finite set of normal form sentences.
\label{part-nf}
\item If $\psi\in f(\phi)$, then $s(\psi)\subseteq f(\phi)$.
\label{part-s-technical}
\item
If $\psi\in f(\phi)$, then $f(\psi)\subseteq f(\phi)$.
\label{part-necc-lemma}
\item If $[\justplain{\gamma}]\necc_{\CC}^*\chi\in f(\phi)$,
$\justplain{\gamma}\rightarrow_{\CC}^{*}  \justplain{\delta}$, 
and $A\in \CC$,   then
$f(\phi)$ also contains
$\necc_A [\justplain{\delta}]\necc_{\CC}^*\chi$,
$[\justplain{\delta}]\necc_{\CC}^*\chi$,
$\nf(\Pre(\delta))$,
and $\nf([\justplain{\delta}]\chi)$.
\label{part-crucial-technical}
\end{enumerate}
\label{lemma-technical}
\end{lemma}

\begin{proof}
Part (\ref{basicform-cases}) is by cases on $\phi$. The only
interesting case is for $[\alpha]\necc_{\CC}^* \phi$,
and this is a subsentence of $\necc_A [\alpha]\necc_{\CC}^* \phi$
for any $A\in \CC$.

All of other the parts are by induction on $\phi$ in the well-order $<$.
For part (\ref{part-nf}), we use 
Lemma~\ref{lemma-property-nfs}.

In  part  (\ref{part-s-technical}),
we  argue by induction on normal forms.
 
The result is immediate when $\phi$ is an atomic sentence  $p$,
and
the induction steps for $\nott$, $\andd$, and $\necc_A$ are easy. 
For $\necc^*_{\BB}\phi$, 
note that since $\phi < \necc^*_{\BB}\phi$, our induction
hypothesis implies the result for $\phi$; we verify it
for $\necc^*_{\BB}\phi$.
When $\psi$ is $\necc^*_{\BB}\phi$, then 
$$s(\psi) \quadeq  s(\phi) \cup  \set{\necc^*_{\BB}\phi}  \quad \subseteq \quad 
 f(\phi) \cup  \set{\necc^*_{\BB}\phi} 
\quad\subseteq\quad f(\necc^*_{\BB}\phi).
$$
And then, when
$\psi$ is $\necc_A \necc^*_{\BB}\phi$ for some $A\in\BB$,
we have
$$
s(\psi)   \quadeq 
 s(\necc^*_{\BB}\phi) \cup\set{\necc_A \necc^*_{\BB}\phi}
\quad\subseteq\quad
f(\necc^*_{\BB}\phi).
$$
To complete part (\ref{part-s-technical}), we consider
$[\alpha]\necc^*_{\CC}\phi$.  
If there is some $\chi < [\alpha]\necc^*_{\CC}\phi$
such that $\psi\in f(\chi)$
and $f(\chi)\subseteq f([\alpha]\necc^*_{\CC}\phi)$, 
then we argue as follows:  by induction hypothesis,
$s(\psi)\subseteq f(\chi)$; and by
hypothesis $f(\chi)\subseteq f([\alpha]\necc^*_{\CC}\phi)$.
This covers all of the cases except for $\psi$ a subsentence
of $ \necc_A[\beta]\necc_{\CC}^*\phi$.
And here $s(\psi)\subseteq s(\necc_A[\beta]\necc_{\CC}^*\phi)\subseteq
 f([\alpha]\necc^*_{\CC}\phi)$.

In  part  (\ref{part-necc-lemma}),
we again  argue by induction on normal forms.
The result is immediate when $\phi$ is an atomic sentence  $p$.
The induction steps for $\nott$, $\andd$, and $\necc_A$ are easy. 
For $\necc^*_{\BB}\phi$, 
note that since $\phi < \necc^*_{\BB}\phi$, our induction
hypothesis implies the result for $\phi$; we verify it
for $\necc^*_{\BB}\phi$.
 The only interesting case is when
$\psi$ is $\necc_A \necc^*_{\BB}\phi$ for some $A\in\BB$.
And in this case
$$
f(\psi)   \quadeq   f(\necc^*_{\BB}\phi) \cup\set{\necc_A \necc^*_{\BB}\phi}
\quad\subseteq\quad
f(\necc^*_{\BB}\phi).
$$
To complete part (\ref{part-necc-lemma}), we consider
$[\alpha]\necc^*_{\CC}\phi$.  
If there is some $\chi < [\alpha]\necc^*_{\CC}\phi$
such that $\psi\in f(\chi)$
and $f(\chi)\subseteq f([\alpha]\necc^*_{\CC}\phi)$, 
then we are easily done by the induction hypothesis.
This covers all of the cases except for $\psi$ of the form $[\beta]\necc_{\CC}^*\phi$
or of the form $\necc_A[\beta]\necc_{\CC}^*\phi$.
(If $\psi$ is a subsentence of some $\beta$, 
  then $f(\psi)\subseteq f([\alpha]\necc^*_{\CC}\phi)$ directly.
If $\psi$ is a subsentence of $\necc^*_{\CC}\phi$, then
by induction hypothesis, $f(\psi)\subseteq f(\necc^*_{\CC}\phi)$; 
and directly, $f(\necc^*_{\CC}\phi)\subseteq f([\alpha]\necc^*_{\CC}\phi)$.)
For  $[\beta]\necc_{\CC}^*\phi$, we
use the transitivity of $\arrow^*_{\CC}$ to check that
$f([\beta]\necc_{\CC}^*\phi) \subseteq f([\alpha]\necc_{\CC}^*\phi)$.
And now the   case of $\necc_A[\beta]\necc_{\CC}^*\phi$ follows:
$$f(\necc_A[\beta]\necc_{\CC}^*\phi) \quadeq
 f([\beta]\necc_{\CC}^*\phi)\cup\set{\necc_A[\beta]\necc_{\CC}^*\phi}
\quad\subseteq \quad
 f([\alpha]\necc_{\CC}^*\phi).$$

%Part (\ref{part-s})  is similar to part (\ref{part-necc-lemma}), using
%equation (\ref{eq-s}) at the beginning of this subsection. 

For part (\ref{part-crucial-technical}), 
assume that 
$[\justplain{\gamma}]\necc_{\CC}^*\chi\in f(\phi)$.
By part (\ref{part-nf}), 
$[\justplain{\gamma}]\necc_{\CC}^*\chi$ is a normal form.
The definition of $f$ implies that
$\necc_A [\justplain{\delta}]\necc_{\CC}^*\chi$,
$[\justplain{\delta}]\necc_{\CC}^*\chi$,
$\nf(\Pre(\delta))$,
and $\nf([\justplain{\delta}]\chi)$
all belong to $f([\justplain{\gamma}]\necc_{\CC}^*\chi)$,
and then part
(\ref{part-necc-lemma}) tells us that 
$f([\justplain{\gamma}]\necc_{\CC}^*\chi) \subseteq f(\phi)$.
%The first three of these sentences are immediate by the definition
%of $f$; 
%the third one follows from part (\ref{part-s}); 
%and the last 
%comes from part (\ref{basicform-cases}),
% since $\nf([\justplain{\delta}]\chi)
%\in f([\justplain{\delta}]\chi)\subseteq 
%f([\justplain{\gamma}]\necc_{\CC}^*\chi)$.
\end{proof}

\paragraph{Summary} The purpose of this section was to prove
  Lemmas~\ref{lemma-nfs-statement}-\ref{lemma-FL} in 
Section~\ref{section-completeness}.
Lemma~\ref{lemma-nfs-statement} is Corollary~\ref{corollary-nfs}.
Lemma~\ref{lemma-property-nfs-statement} comes from
Corollary~\ref{corollary-nfs} and Lemma~\ref{lemma-property-nfs}.
Lemma~\ref{lemma-FL} is contained in Lemma~\ref{lemma-technical}.

\subsection{Strong completeness for $\lang_0(\bSigma)$}
\label{subsection-completeness-small-system}

Recall that in the languages $\lang_0(\bSigma)$, we 
do not have the common knowledge operators $\necc^*_{\BB}$ or 
the action iterations $\pi^*$.  
At this point, we can put together several results from our
previous work to obtain a completeness theorem for languages
of the form $\lang_0(\bSigma)$.  The overall ideas are:  (1)
we need a logical system which is strong enough to translate
each sentence of $\lang_0(\bSigma)$ to a normal form;
(2) this normal form will be a purely modal sentence; and (3)
the system should be at least as strong as multimodal $K$.

\begin{proposition}
Every sentence $\phi$ of $\lang_0(\bSigma)$ is provably
equivalent to a sentence $\phi^*$ in which there
are no occurrences of $\union$, $\crash$, $\circ$, or $\skipp$.
\label{prop-nocirc}
\end{proposition}

\begin{theorem} The logical system for $\langaction$
is strongly complete:  For all sets $T \subseteq \lang_0(\bSigma)$,
$T\proves\phi$ iff $T\models\phi$.
\label{theorem-completeness-langaction}
\end{theorem}

\begin{proof}
The soundness half being easy, we only need to
show that if $T\models \phi$, then $T\proves\phi$.

First, we may assume that the symbols
 $\union$, $\crash$, $\circ$, and $\skipp$
 do not occur in $T$ or
$\phi$.  Thus, we may work with
$\lang_0(\bSigma) \cap \lang_1^+(\bSigma)$.   In particular,
we have normal forms.

Next,  for each $\chi$ of $\langaction$,
$\proves \chi \iiff\nf(\chi)$.  
 As a result, $T\proves \nf(\chi)$
for all $\chi\in T$.   

Finally, write $\nf(T)$ for $\set{\nf(\chi)  :\chi\in T}$.
By soundness, $\nf(T)\models\nf(\phi)$.  Since our system extends
the standard complete proof system of modal logic,
  $\nf(T)\proves\nf(\phi)$.  
  So $T\proves\nf(\phi)$.
As we know $\proves\phi\iiff \nf(\phi)$.
So we have our desired conclusion: $T\proves\phi$.
\end{proof}

\subsection{Weak completeness for $\lang_1(\bSigma)$}
\label{subsection-completeness-big-system}

The proof of completeness 
and decidability of 
$\lang_1(\bSigma)$ is based on the filtration
argument for completeness of   PDL due to Kozen and Parikh~\cite{KP}.
We show that every consistent $\phi$ has a finite model,
and that the size of the model is recursive in $\phi$.
As in the last section, 
we depend on the results of 
 Lemmas~\ref{lemma-nfs-statement}-~\ref{lemma-FL}.

\paragraph{The set $\Delta = \Delta(\phi)$} Fix a sentence $\phi$.
We set $\Delta = f(\phi)$ (i.e., we drop $\phi$ from the notation).
This set $\Delta$ is the version for our logic of
the Fischer-Ladner closure of $\pphi$, originating 
in~\cite{FischerLadner77}.
Let  $\Delta =\set{\psi_1,\ldots, \psi_n}$.
Given a maximal consistent set $U$ of $\langactionstar$, let
$$\sem{U}\quadeq
\plusminus \psi_1 \andd\cdots \andd \plusminus \psi_n,$$
where the signs are taken in accordance with membership in  $U$.
That is, if $\psi_i\in U$, then $\psi$ is a conjunct
of $\sem{U}$; but if  $\psi_i\notin U$,
then $\nott\psi_i$ is a conjunct.

Two (standard) observations are in order.
Notice that
if 
$\sem{U}\neq \sem{V}$, then 
$\sem{U}\andd\sem{V}$
is inconsistent.
Also, for all $\psi\in \Delta$,
\begin{equation}
\proves \psi \iiff
 \bigvee\set{\sem{W} :
    W \mbox{ is maximal consistent and }  \psi\in W}.
\label{eq-standard}
\end{equation}
and
\begin{equation}
\proves \neg\psi \iiff
 \bigvee\set{\sem{W} :
    W \mbox{ is maximal consistent and }  \neg\psi\in W}.
\label{eq-standardtwo}
\end{equation}
(The reason is that
$\psi$ is equivalent to the disjunction of {\em all}
complete conjunctions which contain it.
However, some of those complete conjunctions are inconsistent
and these can be dropped
from the big disjunction. The others are consistent
and hence can be extended to maximal consistent sets.)

\begin{definition}
We consider maximal consistent sets $U$
in the logic for $\langactionstar$.
Let $U \equiv V$ iff  $\sem{U} = \sem{V}$
(iff $U\cap \Delta = V\cap \Delta$).
 The {\em filtration $\gothF$}
is the model whose worlds  are the equivalence
classes $[U]$ in this relation.
 Furthermore, we set
\begin{equation}
[U] \arrowA [V] \mbox{ in $\gothF$} \quadiff
 %%\sem{U} \andd \poss_A\sem{V}
\mbox{whenever $\necc_A\psi\in U\cap \Delta$, then also $\psi\in V$}.
%\mbox{ is consistent}.
\label{equation-pdl-filtration}
\end{equation}
We complete the specification of a state
model with the valuation:
\begin{equation}
\|p\|_{\FF} \quadeq \set{[U] : p\in U\cap \Delta}.
\label{eq-pdl-atomics}
\end{equation}
\end{definition}

The definitions in equations
 (\ref{eq-pdl-atomics}) and
(\ref{equation-pdl-filtration}) are independent of the 
choice of representatives: we use 
part (\ref{part-s}) of
Lemma~\ref{lemma-FL}
to see  that
if $\necc_A\chi\in \Delta$, then also $\chi\in \Delta$.

\begin{proposition} If $\sem{U}\andd\poss_A\sem{V}$ is consistent,
then $[U]\rightarrow_A[V]$. 
\label{proposition-claim}
\end{proposition}

\begin{proof}
 Assume $\necc_A\psi\in U\cap\Delta$ and toward a contradiction that
$\psi\not\in V$.
   Since $\psi \in\Delta$ and $\neg\psi\in V$, we have
$\proves\sem{V}\iif\neg\psi$. Thus,
$\proves\poss_A\sem{V}\iif\poss_A\neg\psi$, and so $\proves
\sem{U}\andd\poss_A\sem{V}\iif\necc_A\psi\andd\poss_A\neg\psi$.
Hence $\sem{U}\andd\poss_A\sem{V}$ is inconsistent. 
\end{proof}

\begin{definition} 
Let $\pair{\alpha}\poss_{\CC}^*\psi$ be a normal form.
A {\em good path from $[V_0]$
for $\pair{\alpha}\poss_{\CC}^*\psi$}
is a  path in $\gothF$
\begin{equation}
[V_0]\quad\rightarrow_{A_1}\quad [V_1]\quad\rightarrow_{A_2}\quad
    \cdots\quad \rightarrow_{A_{k-1}} \quad
[V_{k-1}] \quad\rightarrow_{A_k} \quad [V_k]
\label{goodpath}
\end{equation}
such that $k\geq 0$,  each $A_i\in \CC$, and such that
there exist actions
$$\alpha \ = \ \alpha_0\quad\rightarrow_{A_1}\quad
   \alpha_1\quad\rightarrow_{A_2}\quad
    \cdots\quad \rightarrow_{A_{k-1}} \quad
\alpha_{k-1} \quad\rightarrow_{A_k} \quad
\alpha_k$$
such that
  $\Pre(\alpha_i)\in V_i$
 for all $0\leq i\leq k$, %%%% was $0 < i \leq k$!!!!!
 and
$\pair{\alpha_k}\psi\in V_k$.
\end{definition}

The idea behind a good path comes from considering 
Lemma~\ref{proposition-reduction-diamond-star} in $\FF$.
Of course, the special case of that result would
require that $\pair{\FF, [V_i]} \models\pre(\alpha_i)$
rather than   $\pre(\alpha_i)\in V_i$,
and similarly for $\pair{\alpha_k}\psi$ and $V_k$.
The exact formulation above was made in order that 
the Truth Lemma will go through 
for sentences of the form $\pair{\alpha}\poss_{\CC}^*\psi$
(see the final paragraphs of the proof of Lemma~\ref{lemma-truth-lemma-pdl}).

\rem{
In connection with this notion and the work below, we note a few facts.
First, if $\pair{\alpha}\poss_{\CC}^*\psi$ is a normal form,
so is $[\alpha]\necc_{\CC}^*\psi$, and vice-versa.
Second, we remind the reader that $\proves \phi\iiff \nf(\phi)$ for
all  $\phi$.  It follows from this that a consistent
set cannot contain both $\nf(\phi)$ and $\nf(\nott\phi)$ for any
$\phi$.
}
%%The next results are from Kozen and Parikh~\cite{KP}:

\rem{
\begin{lemma} %%% [Kozen and Parikh\cite{KP}]
Let $\necc_A\psi\in\Delta$.  If $[U]\arrowA [V]$
and $\necc_A\psi\in U$, then $\psi\in V$.
\label{lemma-first-kozen-parikh}
\end{lemma}

\begin{proof}
 If not, then
 $\nott\psi\in V$.
So $\proves\poss_A\sem{V}\iif \poss_A\nott\psi$.
Also, since $\necc_A\psi\in\Delta$,
 $\proves\sem{U}\iif \necc_A\psi$.
Thus
$$\proves(\sem{U}\andd\poss_A\sem{V})
\iif  \necc_A\psi \andd \poss_A\nott\psi.$$
This would contradict the consistency of
$\sem{U}\andd\poss_A\sem{V}$.
\end{proof}
}

\begin{lemma}  Let $[\alpha]\necc_{\CC}^*\psi\in \Delta$.
If there is a good path from $[V_0]$
for  $\pair{\alpha}\poss_{\CC}^*\nott\psi$,
then  $\pair{\alpha}\poss_{\CC}^*\nott\psi\in V_0$.
\label{lemma-good-path}
\end{lemma}

\begin{proof}
By induction on the length $k$ of the path.
If $k=0$, then $\pair{\alpha}\nott\psi\in V_0$.
If $\pair{\alpha}\poss_{\CC}^*\nott\psi\notin V_0$,
then
$[\alpha]\necc_{\CC}^*\psi\in V_0$.   By
Lemma~\ref{lemma-FL}, part (\ref{part-crucial}),
we have  $\nf([\alpha]\psi)\in V_0$.
This is a contradiction.

Assume the result for $k$, and suppose that there is a good
path from $[V_0]$
for  $\pair{\alpha}\poss_{\CC}^*\nott\psi$ of length $k+1$.
We adopt the notation from (\ref{goodpath}) for this good path.
Then there is a good path of length $k$ from $[V_1]$
for $\pair{\alpha_1}\poss_{\CC}^*\nott\psi$.
Also, $[\alpha_1]\necc_{\CC}^*\psi\in \Delta$,
by Lemma~\ref{lemma-FL}, part  (\ref{part-crucial}).
By induction hypothesis,
 $\pair{\alpha_1}\poss_{\CC}^*\nott\psi\in V_1$.

If $\pair{\alpha}\poss_{\CC}^*\nott\psi\notin V_0$,
then $[\alpha]\necc_{\CC}^*\psi\in V_0$.
$V_0$ contains
$[\alpha]\necc_{\CC}^*\psi\andd\nf(\Pre(\alpha))
\iif \necc_A[\alpha_1]\necc_{\CC}^*\psi$. 
(That is, this sentence is valid
by Lemma~\ref{lemma-converse-rule}, part (\ref{part-hhhh}).
Hence it belongs to every maximal consistent set.)
%%Consider  $\necc_A [\alpha_1]\necc_{\CC}^*\psi$.
So  
$V_0$ contains $\necc_A[\alpha_1]\necc_{\CC}^*\psi$.
This sentence belongs to $\Delta$
by Lemma~\ref{lemma-FL}, part  (\ref{part-crucial}).
  Now by
%%%%Lemma~\ref{lemma-first-kozen-parikh},
definition of $\arrowA$ in $\FF$,
we see that
 $[\alpha_1]\necc_{\CC}^* \psi\in V_1$.
This is a contradiction to our observation at the end of the previous
paragraph.
\end{proof}

\begin{lemma}
If  $\sem{V_0} \andd
    \pair{\alpha}\poss_{\CC}^*\psi$
is consistent, then
there is  a good path from $[V_0]$ for
 $\pair{\alpha}\poss_{\CC}^*\psi$.
\label{lemma-crucial-newsystem}
\end{lemma}

\begin{proof}
For each $\beta$ such that
$\alpha\rightarrow_{\CC}^*\beta$,
let  $S_{\beta}$ be the (finite) set of all $[W]\in \gothF$ such
that there is {\em no} good path from $[W]$
for  $\pair{\beta}\poss_{\CC}^*\psi$.
We need to see that $[V_0]\notin S_{\alpha}$;
suppose toward a contradiction that
$[V_0]\in S_{\alpha}$.
Let $$\chi_{\beta} \quadeq \bigvee
     \set{\sem{W} : W\in S_{\beta}}.$$
Note that $\nott\chi_{\beta}$  is logically equivalent
to   $\bigvee\set{\sem{X}: [X]\in \FF \mbox{ and }
X\notin S_{\beta}}$.
Since we assumed $[V_0]\in S_{\alpha}$,
 we have $\proves\sem{V_0}\iif \chi_{\alpha}$.

We first claim for  $\beta$ such that
$\alpha\rightarrow_{\CC}^*\beta$,
 $\chi_{\beta}\andd \pair{\beta}\psi$ is inconsistent.
Otherwise, there would
be
$[W]\in S_{\beta}$ such that
 $\chi_{\beta}\andd \pair{\beta}\psi\in W$.
Note that 
by the Partial Functionality Axiom,
$\proves \pair{\beta}\psi\iif \Pre(\beta)$.
But then the one-point path $[W]$ is a good path
from $[W]$ for
$\pair{\beta}\poss_{\CC}^*\psi$.
Thus $[W]\notin S_{\beta}$, and this is a contradiction.
So indeed, $\chi_{\beta}\andd\pair{\beta}\psi$
 is inconsistent.  Therefore,
$\proves\chi_{\beta} \iif [\beta]\nott\psi$.

\medskip

%We will need the following standard claim,
%% an argument for which can be
%found in Kozen and Parikh~\cite{KP}. We will also use this claim in the
%proof of Lemma~\ref{lemma-truth-lemma-pdl}.

We next show that for all $A\in \CC$ and all
$\beta$ such that
$\beta\rightarrow_{A}\justplain{\gamma}$,
$\chi_{\beta} \andd \Pre(\beta) \andd \poss_{A}
\nott \chi_{\justplain{\gamma}}$ is inconsistent.
Otherwise, there would be $[W]\in S_{\beta}$ 
with $\chi_{\beta}$, $\Pre(\beta)$, and $\poss_{A}
\nott \chi_{\justplain{\gamma}}$ in it. Then 
$\bigvee \set{\poss_A\sem{X} : X\not\in S_{\justplain{\gamma}}}$,
being equivalent to $\poss_A\nott\chi_{\gamma}$,
would belong to $W$. It
follows that $\poss_A\sem{X}\in W$ for some $[X]\not\in 
S_{\justplain{\gamma}}$. By 
Proposition~\ref{proposition-claim}, $[W]\rightarrow_A[X]$. 
Since $[X]\notin S_{\justplain{\gamma}}$, there
is a good path from $[X]$ for
$\pair{\gamma}\poss_{\CC}^*\psi$.   But
since $\beta\rightarrow_{A}\justplain{\gamma}$
and $W$ contains $\Pre(\beta)$,
we also have a good path from $[W]$
for $\pair{\beta}\poss_{\CC}^*\psi$.
This again contradicts $[W]\in S_{\beta}$.
As a result, for all relevant $A$, $\beta$,
and $\justplain{\gamma}$,
$\proves
\chi_{\beta} \andd \Pre(\beta) \iif \necc_{A}
\chi_{\justplain{\gamma}}.$

By the Action Rule,
$\proves\chi_{\alpha} \iif
[\alpha]\necc_{\CC}^*\nott\psi$.
Now $\proves\sem{V_0}\iif
\chi_{\alpha}$.
So  $\proves\sem{V_0}\iif
 [\alpha]\necc_{\CC}^*\nott\psi$.
This contradicts the assumption with which we began
this  proof.
\end{proof}

\begin{lemma}[Truth Lemma]
Consider a sentence $\phi$, and also the set
$\Delta = f(\phi)$.
For all $\chi\in \Delta$ and $[U]\in \gothF$: $\chi\in U$
iff\/ $\pair{\gothF,[U]}\models  \chi$.
\label{lemma-truth-lemma-pdl}
\end{lemma}

\begin{proof}
  We argue by induction on the wellfounded $<$ that if
$\chi\in\Delta$, then:
$\chi\in U$
iff\/ $\pair{\gothF,[U]}\models  \chi$.
 The case of $\chi$ atomic
is trivial.  Now assume this Truth Lemma for sentences
$< \chi$.  
Recall that $\Delta \subseteq\NF$, our
set of normal forms (see Section~\ref{section-normal-forms}). 
We argue by cases on $\chi$.

The cases that $\chi$ is either 
a negation or conjunction are trivial.

Suppose next that $\chi \equiv \necc_A \psi$.
Suppose $\necc_A\psi\in U$; we show $\pair{\gothF,[U]}\models\necc_A\psi$.
Let $[V]$ be such that  $[U]\arrowA [V]$.
Then by
definition of $\arrowA$,
%% Lemma~\ref{lemma-first-kozen-parikh},
$\psi\in V$. 
The induction hypothesis applies to $\psi$, since
$\psi < \necc_A\psi$, and since $\psi\in \Delta$
by Lemma~\ref{lemma-FL},
part (\ref{part-s}).
 So by induction hypothesis,
 $\pair{\gothF,[V]}\models\psi$.  This  gives half of our equivalence.
 Conversely,  suppose that
$\pair{\gothF,[U]}\models\necc_A\psi$.
Suppose towards a contradiction that
$\poss_A\nott\psi\in U$.
So $\sem{U}\andd\poss_A\nott\psi$ is consistent.
We use equation (\ref{eq-standardtwo}) and the fact that
  $\poss_A$ distributes over
disjunctions to see that 
$\sem{U}\andd\poss_A\nott\psi$ is logically equivalent
to  $\bigvee(\sem{U}\andd\poss_A\sem{V})$,
where the disjunction is taken over all $V$ which
contain $\nott\psi$.  Since $\sem{U}\andd\poss_A\nott\psi$
is consistent, one of the disjuncts 
$\sem{U}\andd\poss_A\sem{V}$ must be consistent.
The induction hypothesis again applies, and we use it to see that
$\pair{\gothF,[V]}\models\nott\psi$.
By Proposition~\ref{proposition-claim}, 
$[U]\arrowA [V]$. We conclude that
$\pair{\gothF,[U]}\models\poss_A\nott\psi$, and
this is a contradiction.

For 
$\chi$ of the form $\necc_{\CC}^*\psi$, we
use the standard argument for PDL
(see Kozen and Parikh~\cite{KP}).
This is based on lemmas that parallel
 Lemmas~\ref{lemma-good-path} and~\ref{lemma-crucial-newsystem}.
The work is somewhat easier than
what we do below for sentences of the form
$[\alpha]\necc_{\CC}^*\psi$,
and so we omit these details.

We conclude with the case 
when $\chi$ is a normal form  sentence of the form
$[\alpha]\necc_{\CC}^*\psi\in \Delta$.
Assume that $[\alpha]\necc_{\CC}^*\psi\in \Delta$.
First, suppose that
$[\alpha]\necc_{\CC}^*\psi\notin U$.
Then  by Lemma~\ref{lemma-crucial-newsystem},
there is  a good path from  $[U]$ for
$\pair{\alpha}\poss_{\CC}^*\nott\psi$.
We want to apply
Lemma~\ref{proposition-reduction-diamond-star}
in $\gothF$  to assert that
$\pair{\gothF,[U]}\models \pair{\alpha}\poss_{\CC}^*\nott\psi$.
Let $k$ be the length of  the good path.
For $i\leq k$, $\pre(\alpha_i)\in U_i$.
Now each $\nf(\pre(\alpha_i))$ belongs to $\Delta$
 by Lemma~\ref{lemma-FL}, part  (\ref{part-crucial}),
 and
is $< [\alpha]\necc_{\CC}^*\psi$.
So by induction hypothesis,
$\pair{\gothF,[U_i]}\models\nf(\pre(\alpha_i))$.
By soundness, $\pair{\gothF,[U_i]}\models\pre(\alpha_i)$.
We also need to check  that
$\pair{\gothF,[U_k]}\models \pair{\alpha_k}\nott\psi$.
For this, recall from
 Lemma~\ref{lemma-property-nfs-statement}
 that
$\Delta$ contains $\nf(\nott[\alpha_k]\psi)\leq
\nott[\alpha_k]\psi < [\alpha]\necc_{\CC}^*\psi$.
Since the path is good, $U_k$ contains $\pair{\alpha_k}\nott\psi$;
thus it contains $\nott[\alpha_k]\psi$; and finally it
 contains $\nf(\nott[\alpha_k]\psi)$.
By induction hypothesis,  $\pair{\gothF,[U_k]}\models\nf(\nott[\alpha_k]\psi)$.
By soundness,  $\pair{\gothF,[U_k]}\models\nott[\alpha_k]\psi$.
Thus  $\pair{\gothF,[U_k]}\models\pair{\alpha_k}\nott\psi$.
Now it does follow from
Lemma~\ref{proposition-reduction-diamond-star}
that $\pair{\gothF,[U]}\models
\pair{\alpha}\poss_{\CC}^*\nott\psi$.

Going the other way,
suppose that
 $\pair{\gothF,[U]}\models \pair{\alpha}\poss_{\CC}^*\nott\psi$.
By Lemma~\ref{proposition-reduction-diamond-star},
we get a path in $\gothF$ witnessing this.
The  argument of the previous paragraph
shows that this path is a good path
from
$[U]$ for  $\pair{\alpha}\poss_{\CC}^*\nott\psi$.
By Lemma~\ref{lemma-good-path},
$U$ contains
 $\pair{\alpha}\poss_{\CC}^*\nott\psi$.
This completes the proof.
\end{proof}

\begin{theorem} [Completeness]
For all $\pphi$, $\proves\pphi$ iff $\models\pphi$.
Moreover,  this relation is decidable.
\label{theorem-completeness-Kstar}
\end{theorem}

\begin{proof}
 By Lemma~\ref{lemma-nfs-statement},
$\proves\phi\iiff\nf(\phi)$.
Let $\phi$ be consistent.
By the Truth Lemma,   $\nf(\phi)$ holds at some world
in the
filtration $\gothF$.
So $\nf(\phi)$ has a model; thus $\phi$ has one, too.
This establishes completeness.  For decidability, note
that the size of the filtration is computable in the
size of the original $\phi$.
(Another  proof  of decidability: 
we show in Section~\ref{section-upper-bounds}
that $\lang_1(\bSigma)$ can be
 translated into propositional dynamic logic
(PDL) fairly directly,
 and that logic is decidable. Neither argument
gives a good estimate of the complexity.)
\end{proof}

\rem{%%% there's nothing wrong with this section, 
but it doesn't  seem to have enough to belong in the paper at this point.
\subsection{Some other rewriting systems related to $\lang_1(\bSigma)$}
\label{section-RR-other}

In this section, we mention a few other results
concerning rewriting systems for  $\lang_1(\bSigma)$.

\paragraph{A rewriting system for the logic of public announcements}
In the logic of public announcements, 
the Action-Knowledge Axiom may be stated for simple announcements
in a way which avoids the exponential blowup on the right:
Rather than (\ref{eq-aka-general}), we have
\begin{equation}  
\proves [\alpha]\necc_A\phi
\iiff   (\Pre(\alpha) \iif
 \necc_A [\alpha]\phi)
\label{eq-aka-new}
\end{equation}
 This leads to a rewriting system for the logic, obtained by
taking (r1) - (r6) + (r8) in Figure~\ref{fig-RR} and a variant
of (r7):
$$
\begin{array}{llcll}
\mbox{(r7')}  &
 {} [x]\necc_A y  &  \leadsto &
      \Pre(x) \iif 
        \necc_A [x]y \\
\end{array}
$$
It turns out that there is no reason to adopt (r9) or its converse.
The resulting system may be shown to terminate by a 
{\em polynomial\/} interpretation:
$$\begin{array}{lcl}
\interpretation{p_i} & \quadeq & 1 \\
\interpretation{\nott}(a) & \quadeq &   a  + 1 \\
\interpretation{\andd}(a, b) &  \quadeq &  a  + b  + 1  \\
\interpretation{\iif}(a, b) &  \quadeq &  a  + b   + 4\\
\interpretation{\necc_A}(a) & \quadeq &  a  + 1 \\
\end{array}
\qquad
\begin{array}{lcl}
\interpretation{\necc_\BB^*}(a) & \quadeq &  a  + 1 \\
\interpretation{\app}(a, b) & \quadeq &   6ab + a + 6b + 1  \\
\interpretation{\Pre}(a) & \quadeq &  a  \\
\interpretation{\Pub}(a) 
 & \quadeq &   a +1\\
\interpretation{\then}(a, b) & \quadeq &  
   6ab + 6a + 6b + 5 \\
%\interpretation{\skipp} & \quadeq & 3 \\
%\interpretation{+}(a,b) & \quadeq & a b \\
\end{array}
$$
Indeed, this interpretation also has the property that 
$\andd$ and $\then$ are interpreted by functions which
are both commutative and associative.

\paragraph{An LPO}

For background on the lexicographic partial order (LPO),
see, e.g., Dershowitz~\cite{dersh82} or Plaisted~\cite{plaisted}.
We construct an LPO  using a  wellorder $<$ on the symbols in
$\Delta$.
We  first fix a wellfounded relation $<$ on the function symbols.
We set   $\app$  to be greater  in $<$ than all other symbols.
We also set $\sigma_i \equiv \sigma_j$ for all $i, j$.
In all other cases, distinct function symbols are unordered.

Let $ <$ be the LPO obtained from these choices.

We remind the reader that our rewriting system $\RR$ is
given by 
$$
\begin{array}{lcll}
{}\app(\sigma_i\vec{\psi}, p)  &  \leadsto &  \psi_i \iif p  \\
{}\app(\sigma_i\vec{\psi},\nott\phi)  &  \leadsto  &  
\psi_i \iif \nott\app(\sigma_i\vec{\psi},\phi)\\
{}\app(\sigma_i\vec{\psi},\phi\andd\chi)  &  \leadsto &
 \app(\sigma_i\vec{\psi},\phi)  \andd  \app(\sigma_i\vec{\psi},\chi) \\
 {}\app(\sigma_i\vec{\psi},\necc_A\phi) &  \leadsto &
     \psi_i \iif {}\bigwedge\set{
        \necc_A \app(\sigma_j\vec{\psi},\phi):
\sigma_i\arrowA\sigma_j} \\
{} \app(\alpha\then\beta,\phi) & \leadsto & \app(\alpha,\app(\beta,\phi)) \\
\end{array}
$$

\begin{lemma}
For all  rules $\phi\leadsto \psi$ of $\RR$, $\psi <  \phi$.  
\end{lemma}

\begin{proof}
For the first part, recall that $\rightarrow < \app$.
So we use (LPO3): we must check that both $\psi_i$ and $p$
are $< \app(\sigma_i\vec{\psi}, p) $.  This is by the subterm
property.

In the second part, we argue similarly.  We must check
this time that
$\app(\sigma_i\vec{\psi},\phi) < \app(\sigma_i\vec{\psi},\nott\phi)$.
This is by (LPO1) and the fact that
$\phi < \app(\sigma_i\vec{\psi},\nott\phi)$; this last is by the subterm property.

Turning to the third part,
we check that $ \app(\sigma_i\vec{\psi},\phi)  <  \app(\sigma_i\vec{\psi},\phi\andd\chi)$
This uses  (LPO1) and the fact that
$\phi < \phi\andd\psi$, and also
$\phi <  \app(\sigma_i\vec{\psi},\phi\andd\chi)$.

Here is the proof for the fourth part.
We first check that for all $j$,
$\necc_A \app(\sigma_j\vec{\psi},\phi) < \app(\sigma_i\vec{\psi},\necc_A\phi)
$.
Since $\necc_A < \app$, we need only show that
$\app(\sigma_j\vec{\psi},\phi) < \app(\sigma_i\vec{\psi},\necc_A\phi)$ and
use (LPO3).
For this, we use (LPO1).

Here is the proof for the last part. 
Note first that $\beta < \alpha \then \beta$, and $\phi < \app(\alpha\then\beta,\phi)$.
So by (LPO1), 
$\app(\beta,\phi) <  \app(\alpha\then\beta,\phi)$.
This with $\alpha < \alpha\then \beta$ allows us to conclude
again by (LPO1) that
$\app(\alpha,\app(\beta,\phi)) < \app(\alpha\then\beta,\phi) $.
\end{proof}
}

\subsection{Extensions to the completeness theorem}
\label{section-two-extensions}

We briefly mention   extensions of the Completeness
Theorem~\ref{theorem-completeness-Kstar}.
%These extensions come from our discussion at the end
%of Section~\ref{section-postpone}.

First, consider the case of S5 (or K45) actions.  We
change our logical system by restricting to these
S5 actions, and we    add the S5 axioms to our
logical system.  We interpret this new system on
S5  models. 
It is easy to check that applying an S5 action to an S5 model
gives another S5 model. 
Further, the S5 actions are closed under composition.
Finally, if $\alpha$ is an S5 action
and $\alpha\arrow_A \beta$, then   $\beta$ also is an S5 action.
These easily imply the soundness
of the new axioms.
For completeness, we need only check that
if we assume the S5 axioms, then the filtration $\FF$
from the previous section has the property that 
each $\arrowA$ is an equivalence relation.  This is
a standard exercise in modal logic (see, e.g., Fagin et al~\cite{fhmv},
Theorem 3.3.1).
\medskip

Second, we also have completeness not only for the
languages $\lang_1(\bSigma)$, but also for
the languages in~\cite{kra} that are constructed from
{\em families\/} $\SS$ of action signatures.
This construction is most significant in the case when
$\SS$ is an infinite set of signatures; for example, $\SS$
might contain a copy of every finite signature.  In that setting,
the language $\lang(\SS)$ would not be the language of any
finite signature.  But for the (weak) completeness result
of this section, the advantage of the extended definition is
lost.
The point is that for single sentences, we may restrict attention
to a finite subset of $\SS$.   And for finite sets $\SS$,
$\lang(\SS)$ is literally the language of the coproduct signature
$$\oplus\set{\bSigma: \bSigma\in\SS}.$$
We already have completeness for such languages.

Our final extension concerns the move from actions
as we have been working them to actions which change
the truth values of atomic sentences.
If we make this move, then the axiom of Atomic Permanence
is no longer sound.  However, it is easy to formulate 
the relevant axioms.  For example, if we have an
action $\alpha$ which  effects the change $p := p\andd \nott q$,
then we would take an  axiom
$[\alpha ]p \iiff (\Pre(\alpha) \iif p\andd \nott q)$.
Having made these changes, all of the rest of the 
work we have done goes through.  In this way,
 we get a completeness
theorem for this logic.\footnote{{\bf from larry}: I think
we need more on this point}

\paragraph{Endnotes}
Special cases of Theorem~\ref{theorem-completeness-langaction}
for some of the target logics are due
to Plaza~\cite{plaza}, Gerbrandy~\cite{gerbrandy98,gerbrandyphd},
and Gerbrandy and Groeneveld~\cite{gerbrandygroeneveld}.
The proofs in these sources also go via translation to modal logic.
 %% Section 6

\section{Results on expressive power}
\label{section-expressive-power}
\label{section-expressivepower}

In this section,
we study a number of expressive power issues related to
our logics.  The four subsections are for
the most part independent.

\subsection{Translation of $\lang_1(\bSigma)$ into PDL}
\label{section-upper-bounds}

In this section $\bSigma$ is an arbitrary action signature.
In Section~\ref{section-nfs}, we saw normal forms for
$\lang_0(\bSigma)$ and $\lang_1(\bSigma)$.  
We showed in that section that every sentence in 
$\lang_1(\bSigma)$ is provably equivalent to its normal
form, and the normal forms of sentences of $\lang_0(\bSigma)$
are exactly the purely modal sentences.  This proves
that in terms of expressive power,
$\lang_0(\bSigma)$  is equivalent to
$\lang_0$, ordinary modal logic.

Further, we can show that  $\lang_1(\bSigma)$
and indeed the full language $\lang(\bSigma)$
is a sublogic of
$\lang_0^{\omega}$,
 the extension of modal logic with
countable boolean conjunction and disjunction.
The idea is contained in the following clauses:
$$
\begin{array}{lcl}
(\necc_{\BB}^* \phi)^t & \quadeq &
    \bigwedge_{\pair{A_1,\ldots,A_n}\in \BB^*}
(\necc_{A_1}\cdots \necc_{A_n}\phi)^t \\
([\alpha]\necc_{\BB}^*\psi)^t & \quadeq &
\bigwedge_{\pair{A_1,\ldots,A_n}\in \BB^*}
  ([\alpha]\necc_{A_1}\cdots \necc_{A_n}\psi)^t \\
([\pi]^*\phi)^t & \quadeq &
\bigwedge_n ([\pi]^n\phi)^t \\
 \end{array} $$
 
\noindent
It is natural to ask whether there are any finite logics  which
have been previously studied and into which
our logics can be embedded.  One possibility is the Modal
Iteration Calculus ($MIC$) introduced
in Dawar, Gr\"{a}del, and Kreutzer~\cite{dgk}.
The full language
$\lang(\bSigmaPub)$ is a sublanguage
 of $MIC$, see~\cite{miller} for
details. It is likely that this result extends to all other finite action
signatures. In another direction, we ask whether the fragments 
$\lang_1(\bSigma)$ are sublogics of previously-studied systems.

\begin{theorem}
Every sentence of $\lang_1(\bSigma)$ is equivalent to  a sentence of PDL.
\label{theorem-eq-pdl}
\end{theorem}

\begin{proof} (Sketch)
We argue by induction on the wellorder $<$ introduced and studied 
in Section~\ref{section-order-on-lang-1}.  It is sufficient to show
that each sentence in the set $\NF$ of normal forms of $\lang_1(\bSigma)$
is equivalent to a sentence of PDL.  (The normal forms were introduced
in Section~\ref{section-normal-forms}, and the reader may wish
to look back at Lemma~\ref{lemma-nf}.)  We just give the main induction
step.  

Suppose that $[\alpha]\necc^*_{\CC}\phi$ is a normal form sentence.
Our induction hypothesis implies that each $\psi < [\alpha]\necc^*_{\CC}\phi$
is equivalent to some PDL sentence $\psi'$.  We shall show that
$\pair{\alpha}\poss^*_{\CC}$ is itself equivalent to some PDL sentence;
hence also $[\alpha]\necc^*_{\CC}\phi$ has this property.  For this,
we use the semantic equivalent given in Lemma~\ref{proposition-reduction-diamond-star}.
Recall first that there are only finitely many $\beta$ such that
$\alpha\arrowstar\beta$.  Let $X$ be the (finite) set of all such $\beta$ 
such that $\alpha\arrowCC\beta$.
That is, the actions $\beta$ reachable from $\alpha$ by a path
labeled by agents in the set $\CC$.
  We consider $X$ as a 
sub-action structure
of $\bOmega$.

We consider $\CC\cup X$, and we assume that this union is disjoint.
We consider the set $(\CC\cup X)^*$ of all finite words on this set.
We are interested in finite $\CC$-labeled paths through $X$ beginning 
at $\alpha$ and ending at an arbitrary element of $X$.  At this point,
we shall develop our proof only by example.

 Suppose that
$\CC = \set{A,B}$ and that 
 $X = \set{\alpha,\beta}$,
with   
$\alpha\arrowA\alpha$, $\alpha\arrowB\beta$, and $\beta\arrowA\alpha$.
Then one of the paths of interest would be 
$$ \alpha \quad  A \quad 
 \alpha \quad  B \quad  \beta \quad  A \quad \alpha
$$
(Note that this corresponds to $\alpha\arrowA \alpha\arrowB\beta\arrowA \alpha$.)
We are only interested in paths that respect the structure of $\bOmega$.
By Kleene's Theorem, the set $P$ of finite paths of this type
 beginning at our fixed action $\alpha$ is
a regular  language on $\CC\cup X$. 
In our example, $P$ is given by
the regular expression 
$$((\alpha(A\alpha)^*(B\beta{A}))^*(\epsilon + B\beta).
$$
With each such regular expression we associate a PDL program that represents it.
In our example, we would have
$$(?\Pre'(\alpha);((\pair{A};?\Pre'(\alpha))^*;(\pair{B};?\Pre'(\beta);\pair{A}))^*;
(?\true + \pair{B};?\Pre'(\beta)).
$$
%%%%%
%%%%%
\rem{stuff with \gamma below:
 Suppose that
$\CC = \set{A,B}$ and that 
 $X = \set{\alpha,\beta,\gamma}$,
with   $\beta\arrowA\gamma$,
$\gamma\arrowA\gamma$, and $\gamma\arrowB\beta$.
Then one of the paths of interest would be 
$$\gamma \quad  A  \quad  \gamma \quad  B \quad 
 \gamma \quad  B \quad  \beta \quad  A \quad \gamma
$$
(Note that this corresponds to $\gamma\arrowA \gamma\arrowB\beta\arrowA \gamma$.)
We are only interested in paths that respect the structure of $\bOmega$.
By Kleene's Theorem, the set $P$ of finite paths of this type
 beginning at our fixed action $\alpha$ is
a regular  language on $\CC\cup X$. 
To continue our example, suppose that our $\alpha$ is
$\gamma$.  Then $P$ is given by
the regular expression 
$$((\gamma(A\gamma)^*(B\beta{A}))^*(\epsilon + B\beta).
$$
With each such regular expression we associate a PDL program that represents it.
In our example, we would have
$$(?\Pre'(\gamma);((\pair{A};?\Pre'(\gamma))^*;(\pair{B};?\Pre'(\beta);\pair{A}))^*;
(?\true + \pair{B};?\Pre'(\beta)).
$$
}
%%%%%
(The notation $\Pre'(\gamma)$ means the PDL translation   of $\Pre(\gamma)$;
such a PDL sentence exists since $\Pre(\gamma) < [\alpha]\necc^*_{\CC}\phi$.)
Again, this is a PDL program $\pi  = \pi(X,\alpha)$ whose denotation
$\semantics{\pi}{\bS}$ in a state model $\bS$ is the set of pairs $(s,t)$ of
states such that there is  a path
$$s = s_0 \quad \arrow_{A_1} \quad
 s_1 \quad \arrow_{A_2} \quad \cdots \quad \arrow_{A_{k-1}} \quad s_{k-1}
\quad \arrow_{A_{k}} \quad s_{k}
$$
and also a sequence of actions of the same  length $k$,
$$\alpha \ = \ \alpha_0\quad\rightarrow_{A_1}\quad
   \alpha_1\quad\rightarrow_{A_2}\quad
    \cdots\quad \rightarrow_{A_{k-1}} \quad
\alpha_{k-1} \quad\rightarrow_{A_k} \quad
\alpha_k$$
such that each $A_i\in \CC$, each $s_i\in \semantics{\Pre(\alpha_i)}{\bS}$ for all
$0\leq i\leq k$, and finally such that $s_k = t$.
Consider now the PDL sentence $\pair{\pi}(\pair{\alpha_k}\phi)'$.
(As above, $(\pair{\alpha_k}\phi)'$ means the PDL translation 
of $\pair{\alpha_k}\phi$.  This exists because
$[\alpha_k]\phi < [\alpha]\necc^*_{\CC}\phi$; see  Lemma~\ref{lemma-nf}.)
A state $s$ satisfies $\pair{\pi}(\pair{\alpha_k}\phi)'$
in $\bS$ iff there is some $(s,t)\in \semantics{\pi}{\bS}$
such that $t\in \semantics{\pair{\alpha_k}\phi}{\bS}$.  Putting together our
description of  $\semantics{\pi}{\bS}$ with this, we see that
$(s,t)\in \semantics{\pi}{\bS}$ iff $s\in \semantics{\pair{\alpha}\poss^*_{\CC}}{\bS}$; see 
Lemma~\ref{proposition-reduction-diamond-star}.
\end{proof}

At this point, we know that $\lang_1(\bSigma)$ is a sublogic of
PDL.  So it is interesting to ask whether this extends to the 
full logic $\lang(\bSigma)$; recall that this last language
has the operation of {\em program\/} iteration $\pi^*$.
It turns out that $\lang(\bSigma)$ lacks the finite model
property even when $\bSigma$ is as simple as the signature
of public announcements
and indeed when the only sentenced announced publically and
repeatedly
is $\poss\true$ (see Section~\ref{nofmp} below).  
Since PDL has the finite model property,
we see that $\lang(\bSigma)$ is not a sublogic of PDL.
In fact, it also shows that 
$\lang(\bSigma)$ is not even a sublogic of the modal mu-calculus.

\subsection{$\lang_1(\bSigmaPub)$ is more expressive than
$\lang_1$}
\label{section-Cn-results}

Recall that $\lang_1$ in this paper is multi-agent modal logic
together with the common-knowledge operators $\necc^*_{\BB}$ for
sets of agents.
  Our main result here is that 
$\lang_1$  is strictly weaker than $\lang_1(\bSigmaPub)$, the logic
obtained by adding public announcements to $\lang_1$.

We  define a {\em rank} $|\phi|$ on sentences from 
$\lang_1(\bSigmaPri)$. 
Let $|p|=0$ for $p$ atomic, $|\nott \phi| = |\phi|$,
$|\phi\wedge\psi|=\max (|\phi|, |\psi |)$,
 $|\poss_A\phi|= 1+|\phi|$, for all $A\in \Agents$,
and 
$|\poss_{\BB}^*\phi|= 1+|\phi|$ for all $\BB\subseteq \Agents$.

%$|[\Pri^X\ \phi]\psi|=\max (|\phi|, |\psi|)$ for $X=A$, $X=B$, or $X=AB$. 

\paragraph{Games for $\lang_1$}

The main technique in the proof is an
adaptation of Fraisse-Ehrenfeucht games to the setting
of  modal  logic. 
Let $(\bS,s)$ and $(\bT,t)$ be states; i.e.,
 model-world pairs.  By recursion on the natural
number
$n$ we define a game 
$G_n((\bS,s), (\bT,t))$. 
For $n=0$,
$\II$ immediately wins if the following holds:
for all $p\in\AtProp$,
$(\bS,s)\models p$ iff $(\bT,t)\models p$.
And if $s$ and $t$ differ on some atomic sentence,
$\I$ immediately wins.  Continuing, here is how we define
$G_{n+1}((\bS,s), (\bT,t))$.
As in the case of the $G_0$ games, we first check if 
$s$ and $t$ differ on some atomic sentence.
If they do, then $\I$ immediately wins. Otherwise,
the play continues. Now $\I$ can make two types of moves. 

\begin{enumerate}
\item A $\poss_A$-move:
$\I$ has a choice of playing from $\bS$ or from $\bT$, and also
some agent $A$.
If $\I$ chooses $\bS$, then $\I$ continues by choosing
some $s'$ such that $s\arrowA s'$ in $\bS$.
Then $\II$ replies with some $t'\in T$ such that
$t\arrowA t'$.  
Of course, if $\I$ had chosen in $\bT$, then $\II$ would have
chosen in $S$.
Either way, points $s'$ and $t'$ 
are determined, and the two players then play
$G_{n}((\bS,s'), (\bT,t'))$.

\item A $\poss^{*}_{\BB}$-move:
$\I$ plays by selecting $\bS$ (or $\bT$, but we ignore this
symmetric case below), and some set $\BB$ of agents,
and  then $\I$ continues by playing some $s'$ (say)
reachable from $s$ in the reflexive-transitive closure
$\arrowBBstar$ of $\arrowBB$; $\II$ responds with a
point $t'$
in  the other model, $\bT$, which is similarly
related to $t$.  
\end{enumerate}

\rem{
The game consists of $n+1$ rounds
(starting with round $0$). In each round players I and II pick worlds
according to the following procedure: In round $0$, I picks a world
$w_I^0$ which can be equal to either $u$ or $v$ and
II picks $w_{II}^0$ which is the corresponding world in the other model.
Now we
describe what happens in round $i+1$. For concreteness assume that the
world picked by I in round $i$ is in $U$ and therefore the world picked  
by II is in $V$. Now, I picks first a path 
$w_I^i=u_0\rightarrow u_1\rightarrow \cdots\rightarrow u_k$ 
and II responds with a
path in $V$, $w_{II}^i=v_0\rightarrow v_1\rightarrow \cdots\rightarrow
v_m$. Now 
I selects as $w_I^{i+1}$ either one of the endpoints $u_k$ or $v_m$ in
which case  II must select the other endpoint as $w_{II}^{i+1}$, or I
selects one of the $v_i$'s with $0<i<m$ and II must then select a $u_j$
with $0<j\leq k$. 
I wins the game precisely when for some $i\leq n$, $w_I^i$ and $w_{II}^i$
differ on atomic sentences. Otherwise, II wins. 
}

We write $(\bS,s)\sim_n (\bT,t)$ if $\II$
has a winning strategy in the game
$G_n((\bS,s), (\bT,t))$.   It is easy to check that
by induction on $m$
that if $(\bS,s)\sim_n (\bT,t)$
and $m < n$, then $(\bS,s)\sim_m (\bT,t)$.

\begin{proposition}
 If $(\bS,s)\sim_n (\bT,t)$,
then for all $\phi$ with $|\phi|\leq n$, 
$(\bS,s)\models \phi$ iff $(\bT,t)\models\phi$. 
\label{proposition-game}
\end{proposition}

\noindent The proof is standard, except perhaps for 
the easy extra step for $\poss^*$ moves.

\rem{
\begin{proof}
The proof is by induction on $\phi$.
Let  $\phi$ be atomic.
Suppose $\pair{U,u}\sim_n \pair{V,v}$.
Then since $\II$ has  a winning strategy, the atomic
sentences satisfied by $u$ and $v$ must be the same.
So we are done in this case.

The induction steps for the boolean connectives are trivial.
For $\necc\phi$, suppose that $|\necc\phi|\leq n$,
$\pair{U,u}\sim_n \pair{V,v}$,
and $\pair{U,u}\models \necc\phi$.  Suppose towards a contradiction
that $\pair{V,v}\models\poss\nott\phi$.
Let $v'$ be such that
$v\rightarrow v'$ in $V$ and $\pair{V,v'}\models\nott\phi$.
Let $\I$ begin a play of 
$G_{n-1}(\pair{U,u}, \pair{V,v})$ by choosing to play $v'\in V$.
Then $\II$'s winning strategy responds with some $u'$
such that  $(U,u')\sim_{n-1} (V,v')$.
Since $|\phi| \leq n-1$, our induction hypothesis implies that
 $\pair{U,u'}\models\poss\nott\phi$.
This is a contradiction. 

The argument for $\necc^{*}\phi$ is similar and we leave it to the reader.
\end{proof} 
}

\bigskip

We have two results that show that public announcements add
expressive power to $\lang_1$.  The first is for
one agent on arbitrary models, and the second is
for  two agents on equivalence
relations.

In the first result, let $\Agents$ be a singleton $\set{A}$.
We drop the $A$ from the notation.

\begin{theorem}
The $\lang_1(\bSigmaPub)$ sentence
 $\pair{\Pub\ p}\poss^* q$ is not expressible in $\lang_1$, even
by a set of sentences.
\label{theorem-oneagent-inexpressive}
\end{theorem}

\begin{proof}
Fix a number $n$.  We first show that 
 $\pair{\Pub\ p}\poss^* q$ is not expressible by any single
sentence of   $\lang_1$ of rank $n$.
Let $A_n$ be the cycle
$$a_0\ \rightarrow \ a_1
\ \rightarrow \ \cdots\ \rightarrow \  a_{n+2} \ \rightarrow \  a_{n+3} 
\ \rightarrow \  \cdots 
a_{2n+4}\ =\ a_0.$$
We set $p$ true everywhere except $a_{n+2}$ and $q$ true only 
at $a_0$.  

Announcing $p$ means that we delete $a_{n+2}$.
So $A_n(\Pub\ p)$ splits into two disjoint pieces.
this means that in $A_n(\Pub\ p)$, $a_1$ does not satisfy
$\poss^* q$.  But $a_{n+3}$ does satisfy it.

We show that $a_1$ and $a_{n+3}$ agree 
  on
all sentences of $\lang_1$ of rank $\leq n$.
For this, we show that $\II$ has a 
winning strategy  in the $n$-round game on 
between $(A_n,a_{n+1})$ and $(A_n,a_{n+3})$.
Then we appeal to
 Proposition~\ref{proposition-game}.
$\II$'s strategy is as follows: if $\I$ ever makes a 
$\poss^*$ move, $\II$ should make a move on the other side
to the exact same point.  (Recall that $A_n$ is a cycle.)
Thereafter, $\II$ should mimic $\I$'s moves exactly.
Since the play will end with the same point in the two
structures, $\II$ wins.
But if $\I$ never makes a $\poss^*$ move, the play will
consist of $n$ $\poss$-moves. $\II$ should simply
make the same moves in the appropriate structures.
Since $a_{n+2}$ is $n+1$ 
steps from $a_1$, and $a_0 = a_{2n+4}$ is $n+1$ steps from $a_{n+3}$,
 $\II$ will win the play
in this case.

So at this point we conclude that  for all $n$,
 $\pair{\Pub\ p}\poss^* q$ is not expressible by any single
sentence of   $\lang_1$ of rank $n$. 
We conclude by extending this to show that
 $\pair{\Pub\ p}\poss^* q$ is not expressible by 
any {\em set\/} of sentences of   $\lang_1$.  Suppose towards
a contradiction
that $\pair{\Pub\ p}\poss^* q$ were equivalent to the set
$T\subseteq \lang_1$.  Consider the following models $\AA$
and $\BB$:  $\AA = \oplus_{n\geq 0}(A_n,a^n_1)$; i.e., take
the   disjoint union of the models 
$\oplus  A_n$, and then identify all points $a^n_1$.
We also rename the common $a^n_1$ point to be $a$,
and we take this as the
distinguished point.  So we consider $(\AA,a)$.
Similarly, let $\BB = \oplus_{n\geq 0}(A_n,a^n_{n+3})$.
For clarity we'll rename each point $a^m_j$ to be
$b^m_j$.  And we write $b$ for the distinguished point of $\BB$.
The construction 
insures that $(\AA,a)\models \nott\pair{\Pub\ p}\poss^*
q$ and  $(\BB,b)\models \pair{\Pub\ p}\poss^* q$.

By definition of $T$, $(\BB,b)\models \phi$ for all $\phi\in T$.
Let $\phi\in T$ be such that $(\AA,a)\models \nott \phi$.
Let $m = |\phi|$.
Let $$ \CC \quadeq \oplus
(\set{(A_n, a^n_1) : n\neq m}\cup \set{(A_m,a^m_{m+3})}).$$
Notice that we switched exactly one of the identified points.
Just as before, we'll rename the points of $\CC$ and call the overall
distinguished point $c$.  It is easy to check
that 
 $(\CC,c)\models \pair{\Pub\ p}\poss^* q$.
And since this sentence is equivalent to $T$, we also have
$(\CC,c)\models \phi$.  But we now show that
$(\AA,a)$ and $(\CC,c)$ agree on all sentences of rank $\leq m$.
This implies that $(\AA,a)\models \phi$, giving the needed contradiction.

Here is a winning strategy for $\II$ in the $m$-round game between
$(\AA,a)$ and $(\CC,c)$. 
  If $\I$ opens with anything besides $c^m_{m+3}$ or $a^m_{1}$,
$\II$ should play the corresponding point on the other side
and thereafter play in the obvious way.  If $\I$ opens with
$c^m_{m+3}$, $\II$ should play with $a^m_{1}$ and thereafter
play with basically the same strategy as in the first part
of this theorem; similarly, if $\I$ opens with $a^m_1$, $\II$
should reply $c^m_{m+3}$ and thereafter play via the strategy 
in the first part of this proof, where we dealt with single sentences.
\end{proof}

\paragraph{Partition models}
Often in epistemic  logic one is concerned with models
in which every accessibility relation $\arrowA$ is an equivalence
relation.  We can obtain a version of 
Theorem~\ref{theorem-oneagent-inexpressive} which shows that
even on this smaller class of models, public announcements add
expressive power to $\lang_1$.  However, we must use two agents:
 with one agent and an equivalence relation $\necc^*
\phi$ is equivalent to $\necc\phi$.  Thus every sentence in 
 $\lang_1(\bSigmaPub)$ on one-agent equivalence relations is
equivalent to a purely modal sentence.   

\begin{theorem}
The $\lang_1(\bSigmaPub)$ sentence  
 $\pair{\Pub\ p}\poss_{A,B}^* q$ is not expressible 
by any set of sentences of
$\lang_1$, even on the class of models in which $\arrowA$ and $\arrowB$
are equivalence relations.
\label{theorem-S5-negative}
\end{theorem}

\begin{proof}
We fix a number $N$ and first show that 
 $\pair{\Pub\ p}\poss^* q$ is not expressible by any single
sentence of   $\lang_1$ of rank $N$.
 Let $n$ be the smallest even number strictly 
larger than $N$.
Let $C_n$ be as defined in Section~\ref{section-Cn} above.
\rem{
  We consider a model $C_n$ which is a cycle of
$5n$ points 	$a_1, \ldots, a_{5n}$ arranged as follows:
$$\xymatrix{
a_1 \ar@{<->}[r]^{A} & a_2 \ar@{<->}[r]^{B} & a_3
& \cdots & a_{5n-1} \ar@{<->}[r]^{A} & a_{5n} \ar@{<->}[r]^{B} & a_1
}
$$
Since $n$ is even, for $1\leq i\leq 5$, the connection is $a_{in-1}\arrowA a_{in}\arrowB
a_{in+1}$. (We are taking subscripts modulo $5n$ here.)
 We also specify that 
$p$ is true at all points except $a_1$ and $a_{2n+1}$,
and $q$ is true {\em only\/} at $a_{4n+1}$.
}
Note that  $C_n$ does have the property that $\arrowA$
and $\arrowB$ are equivalence relations.  In addition,
$(\arrowA)^*$ is the same relation as
$\arrowA$.
So $\necc^*_A\phi$ is equivalent to 
$\necc_A\phi$ (and similarly for $B$).
Finally,  $(\arrowA\cup\arrowB)^*$ is the universal relation.
So if any point whatsoever satisfies a sentence
$\necc^*_{A,B}\phi$, then all points satisfy it.
\rem{
When $p$ is publically announced, $a_1$ and $a_{2n+1}$ 
disappear.  The cycle breaks into two disconnected
components, and in the resulting model
$C_n(\Pub\ p)$,  $a_{n+1}\models\nott\poss^*_{A,B}q$.
On the other hand, $a_{3n+1}$ and $a_{4n+1}$ are in the same
connected component after the update, so 
in $C_n(\Pub\ p)$,
 $a_{3n+1}\models\poss^*_{A,B}q$. 
Returning to $C_n$, }

The analysis of Section~\ref{section-Cn} shows that
the points $a_{n+1}$ and $a_{3n+1}$ differ on
our sentence $\pair{\Pub\ p}\poss_{A,B}^* q$.

  We continue by showing
that $\II$ has a 
winning strategy  in the $n$-round game on $C_n$  
from $a_{n+1}$ and $a_{3n+1}$.  If $\I$ ever makes a 
$\poss^*$ move, then $\II$ should move to the same
point and thereafter mimic $\I$ perfectly.   
Thus we may assume that $\I$ never makes any $\poss^*$ moves.
In this case, $\I$ will never move either point to
$a_1$, $a_{2n+1}$, or $a_{4n+1}$.  In other words, 
all points in the play will satisfy $p\andd\nott q$.
So as long as $\I$ plays $\poss$-moves, $\II$ can
follow arbitrarily.  Once again, since the game goes for $n$ rounds,
this will be a winning strategy.

At this point we know that
 $\pair{\Pub\ p}\poss^* q$ is not expressible by any single
sentence of   $\lang_1$ of rank $n > N$.  We use the same idea
as in Theorem~\ref{theorem-oneagent-inexpressive}
to show that 
 $\pair{\Pub\ p}\poss^* q$ is not expressible by any set
of sentences of  $\lang_1$.  In fact, virtually the same proof
goes through.
\end{proof}

\rem{
Further, we see what needs to be done for $\langactionstar$.
The idea is to use 
Lemma~\ref{proposition-reduction-diamond-star}.
We allow $\I$ to play $\poss$-moves, and also a new
kind of move which we
call an {\em $(\alpha,\CC)$-move}.   In a $(\alpha,\CC)$-move,
 $\I$ begins
by choosing to play either in $U$ or in $V$.
In what follows, we give the details only for  $U$.
Continuing, $\I$ declares
some action $\alpha$ and some set $\CC$ of agents.
Then $\I$ plays a 
sequence of worlds from $U$
$$u \ = \ u_0\quad\rightarrow_{A_1}\quad u_1\quad\rightarrow_{A_2}\quad
    \cdots\quad \rightarrow_{A_{k-1}} \quad
u_{k-1} \quad\rightarrow_{A_k} \quad u_k$$
where $k\geq 0$,
and also a sequence of actions of the same  length $k$,
$$\overline{\alpha} \ = \ \overline{\alpha}_0\quad\rightarrow_{A_1}\quad
   \overline{\alpha}_1\quad\rightarrow_{A_2}\quad
    \cdots\quad \rightarrow_{A_{k-1}} \quad
\overline{\alpha}_{k-1} \quad\rightarrow_{A_k} \quad
\overline{\alpha}_k$$
such that each
 $A_i\in \CC$.
Then $\II$ responds with a  sequence of worlds from $V$
$$v \ = \ v_0\quad\rightarrow_{B_1}\quad v_1\quad\rightarrow_{B_2}\quad
    \cdots\quad \rightarrow_{B_{l-1}} \quad
v_{l-1} \quad\rightarrow_{B_k} \quad v_l$$
where $l\geq 0$ (and $l$ need not be the same as $k$),
and also a sequence of actions of the same  length $l$,
$$\overline{\alpha} \ = \ \overline{\beta}_0\quad\rightarrow_{B_1}\quad
   \overline{\beta}_1\quad\rightarrow_{B_2}\quad
    \cdots\quad \rightarrow_{B_{l-1}} \quad
\overline{\beta}_{l-1} \quad\rightarrow_{B_l} \quad
\overline{\beta}_l$$
such that each
 $B_j\in \CC$.
At this point, $\I$ has two options.
$\I$ may  decide to continue with the endpoints,
and then the 
game continues with a play of
$G_n(\pair{W, w_k},\pair{V,v_l})$.  Or, $\I$ might 
decide to 
 select some 
 $v_i$ (say) with $0 <i<m$; then  $\II$ must then select a $v_j$
with $0 < j \leq k$.  At this point, the game continues with
  a play of  $G_n(\pair{W, w_i},\pair{V,v_j})$.  

In  connection with the game for $\langactionstar$, we define
a rank on normal forms for this logic.  These
are determined in Section~\ref{section-normal-forms}.
We set
$$ |[\alpha]\necc^*_{\CC}\phi|
\quadeq 1 + \max(\set{|\phi|}\cup 
 \set{ |\nf([\beta]\phi)|  :  \alpha \arrowstar_{\CC}\beta })
$$
This is a definition by recursion on the well-order $<$ of
$\langactionstar$.  We are also using
Proposition~\ref{proposition-lpo-appendix},
part~\ref{part-at-end}, and also
Lemma~\ref{lemma-normal-formal-equivalents-appendix}.
Then Lemma~\ref{lemma-game} holds when restricted
to normal forms.  The only interesting step
is the one for $[\alpha]\necc^*_{\CC}\phi$.
This follows from the argument we gave for Lemma~\ref{lemma-game},
except that we also use 
Lemma~\ref{proposition-reduction-diamond-star}.
}

\rem{
In the result below, there will be only one agent $A$, and so we omit the
letter $A$ from the notation. 
Recall that $\lang_1(\bSigmaPub)$  is modal
logic with announcements (to this $A$) and $\poss^* = \poss^*_A$ and
that $\lang_1$ is modal logic with $\necc^*$ but no announcements.

\begin{theorem}
There is a sentence of 
$\lang_1(\bSigmaPub)$ 
which cannot be expressed by
any set of sentences of $\lang_1$.
\label{theorem-negative-expressivity}
\end{theorem}

\begin{proof}
We show first that $[\Pub\ p]\poss^+q = [\Pub\ p]\poss\poss^* q $
cannot be expressed by any single sentence of $\langstar$.
(Incidentally,  the same holds for $[\Pub\ p]\poss^* q$.)
Fix a natural  number $n$.
We define  structures ${\cal A} = {\cal A}_n$ and 
${\cal B}= {\cal B}_n$ as follows.
First ${\cal B}$ has $2n + 3 $  points arranged cyclically
as 
$$ 0 \rightarrow 1 \rightarrow \cdots \rightarrow n
\rightarrow n+1 \rightarrow -n \rightarrow \cdots
 \rightarrow -1 \rightarrow 0.$$
For the atomic sentences, we set $p$ true at all points 
except $n+1$, and $q$ true only at $0$.

The structure ${\cal A}$ is a copy of ${\cal B}$ with $n$ more
points $\topofeight{1}, \ldots, \topofeight{n}$
arranged as
$$ 0\rightarrow \topofeight{1}\rightarrow \cdots 
\rightarrow \topofeight{n}\rightarrow 0
.$$
The shape of ${\cal A}$  is a figure-8.
In both structures, every point is reachable  from
every point by the transitive closure of the $\rightarrow$
relation.
At the  points $\topofeight{i}$, $p$ is true and $q$ is false.
Notice that ${1}\models [\Pub\ p]\poss^+ q$ in ${\cal A}$,
but $1\not\models[\Pub\ p]\poss^+ q$ in ${\cal B}$.

For $0\leq i\leq n$, we let 
$S_i\subseteq {{\cal A}}\times {{\cal B}}$ be the following set
$$
\begin{array}{lcll}
S_i & \quadeq & & 
   \set{(0,0), \ldots, (n,n), (n+1, n+1), (-n, -n), \ldots, (-1,-1)} \\
 & &  \ \cup\ &
\set{(\topofeight{n},-1), (\topofeight{n-1},-2)
\ldots,(\topofeight{2}, -(n-1)), (\topofeight{1},-n)  } \\
 & &  \ \cup\ &
\set{(\topofeight{1},1), \ldots, (\topofeight{n-i},n-i)  } \\
\end{array}
$$
In the case of $i = n$, then the last disjunct is empty.
Note that $S_0 \supset S_1 \supset \cdots \supset S_n$.
Also, for $0\leq i\leq n$,  every point of
one structure is related by $S_i$ to some point of the other.

\begin{claimtwo} If $0\leq i\leq n$  and   
 $(a,b)\in S_i$, then 
$\pair{\AA,a}\sim_{i} \pair{\BB,b}$.
\end{claimtwo}

\medskip
\noindent The proof is by induction on $i$.
  If $i=0$, this is due to the fact  
 that pairs in $S_0$ agree on the atomic formulas.
Assume the statement for $i$,
and that  $i+1\leq n$.
Let $(a,b)\in S_{i+1}$.    We only need to show 
that $\II$ can respond to any play and have
the resulting pair belong to $S_{i}$.
Suppose first that
 $\I$ plays a $\poss$-move.
  Suppose  also that $a = b$, so that $(a,a)$ comes
from the first subset of $S_{i+1}$.
In this case, we only need to notice that 
$(a+1, a+1)\in S_i$ if $|a| \leq n$, 
$(-n, -n)\in S_i $ if $a = n+1$, and 
$(\topofeight{1},1) \in S_i$ if $a=0$,  since $i < n$.
The case of $(a,b)$ from
 the second subset is similar.  Finally, 
if $(\topofeight{a},a)$ belongs to the third subset of $S_{i+1}$,
then $a\leq n-(i+1) = n-i-1$.  So $a+1 \leq n-i$,
and $(\topofeight{a+1}, a+1)$ belongs to the third subset
of $S_{i}$.  This tells $\II$  how to play.

\rem{
We consider the first
round of $G_{i+1}(\pair{\AA,a},\pair{\BB,b})$.

Thus, $\II$ can reply to all $\poss$-moves and preserve
$S_i$. 
By the induction  hypothesis,
  $\pair{\AA,a}\sim_{i+1} \pair{\BB,b}$ in this case.
}

We remarked above that each $S_{i}$ is a total  relation. Moreover, each
world can be reached from any other one in $\cal A$ and in $\cal B$. 
This implies that if $\I$ makes a $\poss^*$-move, $\II$ can respond.
%opens a play of
% $G_{i+1}(\pair{\AA,a},\pair{\BB,b})$
%with a  $\poss^*$-move, then
%$\II$ can respond to preserve $S_{i+1}\subseteq S_i$.
This completes the proof of the claim.  

\medskip

It follows that $\pair{\AA_n, 1} \sim_n \pair{\BB_n, 1}$.
So by Claim 1, for each sentence $\phi\in \langstar$ and all $n\geq
|\phi|$, $\pair{\AA_n, {1}} \models\phi$ iff $\pair{\BB_n, 1} \models
\phi$.  This shows that $[\Pub\ p]\poss^{+}q$ cannot be expressed by a 
single sentence in ${\cal L}(\poss^{*})$. To prove the stronger result as
stated in Theorem~\ref{theorem-negative-expressivity}, we only need
to quote Lemma~\ref{lemma-sixone}.
\end{proof}
}%% end of rem of this proof

\rem{
Let $\BB^*$ be the disjoint union of all the $\BB_n$, with
all the $0$-points identified.  Then $\pair{\BB^*,0}$ does not
satisfy $[\Pub\ p]\poss^+ q$.  However, we sketch a proof that 
if $\phi\in \langstar$ has the property that
$[\Pub\ p]\poss^+ q \models \phi$, then $\pair{\BB,0}\models \phi$.
This implies our result.

Fix a sentence $\phi\in\langstar$.  Let $n = |\phi|$.
Let $\AA^*$ be $\BB$ with a new cycle of length $n+1$ added
going through $0$.  Then $\pair{\AA^*,0}\models [\Pub\ p]\poss^+ q$,
so $\pair{\AA^*,0}\models\phi$.
However, it is not hard to see that in the $\poss^*$-game
of length $n$ between $\pair{\BB^*,0}$ and $\pair{\AA^*,0}$,
player $\II$ has a winning strategy.  The argument is essentially
the same as the one we gave above.  It follows that
$\pair{\BB^*,0}\models \phi$.
}

\rem{
\subsection{$\lang_1(\bSigmaPri)$ 
is more expressive than $\lang_1(\bSigmaPub)$}
\label{section-privatepower}
\renewcommand{\langAB}{\lang_1(\bSigmaPub)}
\renewcommand{\langprivateA}{\lang_1(\bSigmaPri)}

It is natural to assume that {\em privacy\/} adds expressive power
to logics of communication.
The next result is the first result we know that establishes this.

\begin{theorem}
The $\langprivateA$-sentence $\pair{\Pri^A\ p,\true} \poss_A^*\poss_B \nott p$
cannot be expressed by any  of sentence  of  $\langAB$. 
\label{theorem-negative-definability}
\end{theorem}

\begin{proof}
Let $n\geq 1$, and recall the model $G_n$ from Section~\ref{section-Cn}.
\rem{
Let $G_n$ be the  model described as follows:
  We begin with
a cycle in $\arrowA$:
\begin{equation}
b\ 
\arrowA \ a_n\ \arrowA\ a_{n-1}\ \arrowA\ \cdots\ \arrowA\ a_2 \ \arrowA\  a_1\
\arrowA \ a_{\infty}\ \arrowA\ b\
\label{eq:cycle}
\end{equation}
We add edges   $a_i\arrowA b$ for
all $i$  (including $i =1$ and $i=\infty$), and also $a_i\arrowA a_{\infty}$
for all $x$ (again including $i =1$ and $i=\infty$).  
The only $\arrowB$ edge is 
 $a_1\arrowB b$.  (See the figure for this frame.)
The atomic sentence $p$ is true at all points except $b$.
The first thing to note is that 
in $G_n(\Pri^A\ p)$, 
$a_i\models\poss_A^*\poss_B \nott p$ for all $i < \infty$.
The relevant path is $a_i \arrowA\cdots\arrowA a_1  \arrowB b$;
the important point is that since the announcement was private, the edge
$a_1 \arrowB b$ survives the update.
On the other  hand, in this same model $G_n(\Pri^A\ p)$,
 $a_{\infty}\models\nott\poss_A^*\poss_B \nott p$.
This is  because the only path in the original model from $a_{\infty}$ to $b$ 
is  $$a_{\infty}\ \arrowA\ b\ \arrowA\ a_n\ 
\arrowA\ \cdots\ \arrowA\ a_1 \ \arrowB\ b,$$ and 
the edge $a_{\infty}\arrowA b$ is lost in the update.

\begin{figure}[t]
\fbox{
\begin{minipage}{6.0in}
$$
\xymatrix{
b \ar[r]^{A} & a_n \ar[r]^{A} & 
% a_{n-1} \ar[r]^{A}
 & \cdots & \ar[r]^{A}& a_{i} \ar[r]^{A} \ar@/_2pc/[lllll]_{A}
\ar@/^2pc/[rrrr]^{A}
 &  \cdots & %a_{i-1} \cdots
% &  a_2  \ar[r]^{A} 
 & a_1  \ar[r]^{A} \ar@/^2pc/[llllllll]^{B} 
 & a_{\infty} % \ar@/^2pc/[lllllllll]^{A} % \ar@(dr,ur)[]_{A}
}
$$
\end{minipage}
}
%\caption{$G_n$.  The picture omits the arrow $a_1\arrowB b$ j.
 \label{figure-induced}
%}
\end{figure}
}

Let $\chi = \pair{\Pri^A\ p,\true} \poss_A^*\poss_B \nott p$.
Suppose towards a contradiction that $\chi$
 were equivalent
to
$\phi\in \langAB$.  Let  $n = 1 + |\phi|$.
As we know from our discussion of $ \poss_A^*\poss_B \nott p$,
$(G_n, a_n)\models \chi$
and $(G_n, a_\infty)\models \nott\chi$.
However, this contradicts the claim below.
In the statement, we extend our rank function from $\lang$ to $\langAB$
by adding $|[\Pub\ \phi]\psi| = \max(|\phi|,|\psi|)$.

\begin{claim}
Assume that  $1 < j \leq n$,  
$\phi\in \langAB$ and  $|\phi| < j$.
Then     $(G_n,a_j)\models \phi$
iff $(G_n,a_\infty)\models \phi$.
\label{lemma-easier}
\end{claim}

  The proof is by induction on $\phi$.
For $\phi = p$, the result is clear, as are 
the induction steps for $\nott$ and $\andd$.
For $\poss_A\phi$, suppose that 
$ a_j \models \poss_A\phi$.  Either 
$  a_{\infty} \models \phi$, in which case
$ a_{\infty} \models \poss_A\phi$, or else
 $ a_{j-1} \models \phi$.  In the latter case,
by induction hypothesis, $ a_{\infty} \models \phi$;
whence $a_{\infty} \models \poss_A\phi$. 
The converse is similar.

The case of $\poss_B\phi$ is trivial:
$a_j\models \nott\poss_B\phi$ and $a_\infty\models \nott\poss_B\phi$.

For $\poss^*_A\phi$, note that since we have a cycle
(\ref{eq:cycle}) containing
all points, the truth value of $\poss^*_A\phi$ 
does not depend on the point.
The argument for $\poss^*_{AB}$ is similar, and for $\poss^*_B$
it is the same as for $\poss_B$.

For $[\Pub\ \phi]\psi$, assume the result for $\phi$ and $\psi$,
and let $|[\Pub\ \phi] \psi|   < j$.  Then also
$|\phi| < j$ and $|\psi| < j$.
 Let $H = \set{x : x\models \phi}$ be the updated model,
and recall that 
$(G_n,x) \models [\Pub\ \phi] \psi$ iff
$x\in H$ and $(H,x) \models \psi$.
We have two cases: First, $H = G_n$.
Then
$ (G_n,x) \models [\Pub\ \phi] \psi$ iff $(G_n, x) \models \psi$.
So we are done by the induction hypothesis.  

The other case is when there is some $x\notin H$.
If  $a_k\notin H$ for some  $k \geq j$ or for $k = \infty$,
then {\em all\/} these $a_k$ do not belong to $H$.
In particular,  neither $a_j$ nor $a_{\infty}$ belong.
 And so both $a_j$ and $a_\infty$ satisfy $[\Pub\ \phi] \psi$.
If $b\notin H$, then $H$ is bisimilar to a one-point model.
This is because every $a_i\in H$ would have some $\arrowA$-successor
in $H$ (e.g., $a_{\infty}$), and there would be no $\arrowB$ edges.
So we assume $b\in H$.
Thus  $a_i\notin H$ for some $i < j$.
Let $k$ be least so that for $k \leq l \leq \infty$, $a_l\models \phi$.
Then  $1 < k \leq j$. 
Let $A_{\geq k} = \set{a_l  : k\leq l\leq \infty}$.
The submodels generated by $a_j$ and $a_{\infty}$ contain the same
worlds: all worlds in $A_{\geq k}$ and $b$.
We claim that   
$(A_{\geq k} \times A_{\geq k})\cup \set{\pair{b,b}}$
 is  a bisimulation on $H$.
The verification here is easy.
\rem{
note that each $a_l\in A_{\geq k}$
 has at least one $\arrowA$-successor (e.g., $a_\infty$).
Also, $a_1\notin A_{\geq k}$, so $\arrowB$ is irrelevant.
And since $a_{k-1}\notin H$, every $\arrowA$-successor of each
  $a_l\in A_{\geq l}$ belongs to $A_{\geq l}$.
}

So in $H$, $a_j$ and $a_\infty$ agree on 
the sentences of $\langAB$.  (They agree on 
all sentences of all our languages, by Proposition~\ref{prop-bisim-preserve-again};
indeed they agree on all sentences of  
infinitary modal logic). 
 In particular,
$(H, a_j)\models\psi$ iff
$(H, a_\infty)\models\psi$. This concludes the claim. 
\end{proof}

%So at this point we conclude that 
% $\chi$ is not expressible by any single
%sentence of  $\langAB$.
\rem{  
It remains to strengthen this to sets of sentences.
The   proof is modeled on the gluing-and-switching arguments
that we saw already in 
Theorems~\ref{theorem-oneagent-inexpressive}
and~\ref{theorem-S5-negative}.

Let $\AA = \oplus_{n\geq 1} (G_n, a_{\infty})$.
We rename the common $a_{\infty}$ point to $a$.
So the points of $\AA$ are $a^r_i$ for $1\leq i \leq r <\infty$, $a$,
and $b^r$ for $r\geq 1$.  All of these are different.

Fix   $m\geq 1$, and let 
$$\CC_m
\quadeq \oplus(\set{(G_r, a_{\infty} : 1\leq r\neq m}\cup \set{(G_m, a_m)}).$$
In this model, we rename all the $a$ points to $c$, and the $b$ points to $d$'s.
So the points of $\CC_m$ are $c^r_i$ for $1\leq i \leq r <\infty$, $c$,
and $d^r$ for $r\geq 1$.  All of these are different
except that $c^m_m = c$.
It is easy to check that
that $(\AA,a)\models\nott\chi $ and $(\CC_m,c)\models\chi$. 
We claim that $(\AA,a)$ and $(\CC_m,c)$  agree on all sentences of 
rank $\leq m$.  This follows from the following stronger statement.

\begin{claim} Let $|\phi|\leq m$.   Let $(H_n, h_n)$ be states for
all $n$.  Let $h_m^*\in H_m$ be such that $(H_m,h_m)$ and $(H_m,h^*_m)$
agree on sentences of rank $\leq m$.
Let $(\AA,a) = \oplus (H_n,h_n)$,
and let $$ (\BB,b) \quadeq \oplus
     (\set{(H_n,h_n) : n\neq m}\cup \set{(H_m,h^*_m)}).$$
Then  $(\AA,a)$ and $(\BB,b)$ agree on $\phi$.
\end{claim}
 
\begin{proof}
{\bf I need to write this in}\footnote{Something to note about all
of the ``strong'' results in this section is that they are
improvements  of our earlier  inexpressivity results.
 The earlier ones made use of a lemma that said that under hypotheses,
if a sentence  $\phi$  is not expressible by a single sentence, 
then $\poss\psi$ is not  expressible
even by a set of sentences.  The arguments here eliminate the $\poss$.
But they are not general facts; they depend a lot on the structure of the 
particular models.  And  the sharpened results are a little more difficult
to get right.  So this is why I'm still writing them up.}
}

\rem{
We present a lemma which allows us, in certain circumstances, 
to do the following: from the existence of a sentence in a language ${\cal
L}_1$ which is not equivalent to any sentence in a language ${\cal L}_0$
infer that there exists a sentence in ${\cal L}_1$ not even equivalent to any
{\em theory\/} in ${\cal L}_0$. 

\begin{lemma} Let ${\cal L}_0$ be a language included in $\lang_1(\bSigmaPri)$,
and let $\psi$ be a sentence of $\lang_1(\bSigmaPri)$. Assume that for each $n$
we have models $F_n$ and $G_n$ with some worlds $f_n\in F_n$ and $g_n\in
G_n$ such that $\langle F_n, f_n\rangle$ satisfies $\nott\psi$,
$\langle G_n, g_n\rangle$ satisfies
$\psi$, and $\langle F_n, f_n\rangle$ and $\langle G_n, g_n\rangle$ agree
on all sentences of ${\cal
L}_0$ of rank $\leq n$. Then $\pos_A\psi$ is not equivalent with any theory
in ${\cal L}_0$. 
\label{lemma-sixone}
\end{lemma}

\begin{proof} For a sequence of model-world pairs $(H_n,h_n)$,
$n\in D\subseteq\omega$, we
let 
$\bigoplus_{n\in D}(H_n, h_n)$
be a model-world pair defined as follows. Let $h$ be a new world. Take
disjoint copies of the $H_n$'s and add an $A$-arrow from $h$ to each
$h_n$. All other arrows are within the 
$H_n$'s and stay the same as in
$H_n$.
No atomic sentences are true at $h$. Atomic sentences true in the worlds
belonging to the copy of $H_n$ in $\bigoplus_{n\in D}(H_n, h_n)$
are precisely those true at the corresponding worlds of $H_n$. 

Let $F$ be $\bigoplus_{n\in\omega}(F_n, f_n)$ with the new world denoted
by $f$. Define also $F^m$, for $m\in\omega$, to be
$\bigoplus_{n\in\omega}(H_n, h_n)$ with the new world $f^m$ where
$H_m=G_m$, $h_m=g_m$ and for all $n\not= m$,
$H_n=F_n$ and $h_n=f_n$

Now assume towards a contradiction that $\pos_A\psi$ is equivalent with a
theory $\Phi$ in ${\cal L}_0$. Clearly $\pos_A\psi$ fails in $\langle F,
f\rangle$. Thus some sentence $\phi\in\Phi$ fails in $\langle F,
f\rangle$. On the other hand, each $\langle F^m, f^m\rangle$ satisfies
$\pos_A\psi$, whence $\langle F^m, f^m\rangle$ satisfies $\phi$. Let
$m_0=|\phi|$. The following claim shows that both $\langle F, f\rangle$
and $\langle F^{m_0}, f^{m_0}\rangle$ make $\phi$ true or both of them
make it false, which leads to a contradiction. 

\begin{claim} Let $\phi$ be a sentence in $\Phi$ of rank $\leq m$. Let
$H_n$, $K_n$, $n\in D$, with $h_n\in H_n$ and $k_n\in K_n$ be models
such
that $\langle H_n, h_n\rangle$ and $\langle K_n,k_n\rangle$ agree on
sentences of $\Phi$ of rank $\leq m$. Then $\langle \bigoplus_n(H_n, h_n),
h\rangle$ and $\langle\bigoplus_n(K_n, k_n), k\rangle$ agree on $\phi$. 
\end{claim}

This claim is proved by induction on complexity of $\phi$. It is clear for
atomic sentences. The induction steps for boolean connectives are trivial.
A moment of thought gives the induction step for $\poss$ and $\poss^{*}$
with various subscripts. It remains to consider the case when
$\phi=[\phi_1]_A\phi_2$. (The cases when $\phi=[\phi_1]_B\phi_2$ and
$\phi=[\phi_1]_{AB}\phi_2$ are similar.) 
Fix $H_n$, $K_n$, $h_n\in H_n$, $k_n\in K_n$, with $n\in D$, such that 
$\langle H_n, h_n\rangle$ and $\langle K_n, k_n\rangle$ agree on sentences
of $\Phi$ of rank $\leq m$. Note that, for each $n\in D$, $\langle H_n,
h_n\rangle\models\phi_1$ if and only if $\langle K_n,
k_n\rangle\models\phi_1$.
Let $D_1$ be the set
of all $n\in D$ for which $\langle H_n, h_n\rangle\models\phi_1$. Let
$H_n'$ and $K_n'$ be models obtained by updating $H_n$ and $K_n$ by
$[\phi_1]_A$. By the definition of rank and the fact that $|\phi_1|\leq
m$, we have that $
\langle H_n', h_n\rangle$ and $\langle K_n', k_n\rangle$ agree on
sentences from $\Phi$ of rank
$\leq m$. Therefore, by our inductive hypothesis 
$$
\langle\bigoplus_{n\in D_1}H_n', h\rangle\models\phi_2\quadiff
\langle\bigoplus_{n\in D_1}K_n', k\rangle\models\phi_2. 
$$
However,  
$$
\langle\bigoplus_nH_n, h\rangle\models\phi\quadiff
\langle\bigoplus_{n\in D_1}H_n', h\rangle\models\phi_2
$$ 
and 
$$
\langle\bigoplus_nK_n, k\rangle\models\phi\quadiff
\langle\bigoplus_{n\in D_1}K_n', k\rangle\models\phi_2, 
$$
and we are done. 
\end{proof}
}

\rem{
\paragraph{Returning to our theorem}
We get Theorem~\ref{theorem-negative-definability} directly from the
claim, the observation that $(G_n, a_n)\models \chi$ and $(G_n,
a_\infty)\models \nott\chi$, and Lemma~\ref{lemma-sixone}.
}

%This concludes the proof of Theorem~\ref{theorem-negative-definability}.
%\end{proof}
}

\subsection{$\lang_1(\bSigmaPri)$ 
is more expressive than $\lang_1(\bSigmaPub)$}
\label{section-privatepower}
\renewcommand{\langAB}{\lang_1(\bSigmaPub)}
\renewcommand{\langprivateA}{\lang_1(\bSigmaPri)}

It is natural to assume that {\em privacy\/} adds expressive power
to logics of communication.
The result of this section
 is the first result we know that establishes this.
In it, we assume that our set of agents is the doubleton $\set{A,B}$.

\begin{theorem}
The $\langprivateA$-sentence $\chi = \pair{\Pri^A\ p,\true} \poss_A^*\poss_B \nott p$
cannot be expressed by any set  $T$
of sentences  of  $\langAB$. 
\label{theorem-negative-definability}
\end{theorem}
 
We shall use the models $\bS_{f}$ and $\bT_{f,j}$ 
from Section~\ref{section-Cn}.
It would be good to keep the picture in mind while reading
the results to follow.

\begin{lemma} For all $x\in \bS_{f}$
other than $a$, $(\bS_{f}, x)\equiv (\bT_{f,j},x)$.
That is, there is a bisimulation between $\bS_{f}$ and $\bT_{f,j}$ relating
$x$ to itself.
\label{lemma-almost-total-bisim}
\end{lemma}

\begin{proof}
The point is that $a$ is not the target of any arrows in either
structure.  So we consider the set of all pairs
of points other than $a$.  This relation
is a bisimulation between $\bS_{f}$ and $\bT_{f,j}$ relating
each $x\neq a$ to itself.
\end{proof}

\begin{lemma} Let $\chi$ be the sentence 
$\pair{\Pri^A\ p,\true} \poss_A^*\poss_B \nott p$.
Then for all $f$ and $j$, 
$(\bS_f,a)\models \nott \chi$ but $(\bT_{f,j},b)\models
\chi$.
\label{lemma-ainfinity-sufficient}
\end{lemma}

\begin{proof}  This point was mentioned in Section~\ref{section-Cn},
but we review the matter for the reader.
In the updated structure
$\bS_f\otimes (\Pri^A\ p,\true)$, the only arrow from
$(a,\Pri)$ is the self-loop $(a,\Pri)\arrowA (a,\Pri)$.
So the updated structure does not satisfy 
$\poss_A^*\poss_B \nott p$,
and the original $\bS_f$
does  not satisfy our $\chi$.
But in $\bT_{f,j}\otimes (\Pri^A\ p,\true)$ we have
$$ (a,\Pri^A) \ \arrowA\  (c^{j}_1,\Pri) \ \arrowA\ \cdots \ \arrowA \ 
(c^{j}_{f(j)},\Pri^A) \ \arrowB\ (b, \skipp).
$$
So we see that  $(\bT_{f,j}, a)\models \chi$.
\end{proof}

\rem{ lemma not used
\begin{lemma}
The relations $\arrowAstar$, $\arrowBstar$, $\arrowABstar$
are the same in $\bS_f$  and in $\bT_{f,j}$.
These relate $a$ to $a$; and
for all  $x,y\neq a$, they relate 
  $a$ to   $x$, and $x$ to $y$.
\label{lemma-ainfinity-tc}
\end{lemma}

\begin{proof}
Note that $\arrowB$ is already transitive in both models.
We only need to check that if $x\arrowAstar y$ or
$x\arrowABstar y$ in $\bT_{f,j}$,
then already in $\bS_f$  the same relation holds.
The only interesting case
to consider is when a path in $\bT_{f,j}$ uses $a\arrowA c^{j}_1$,
 and here, note that in $\bS_f$,   $a\arrowAstar c^{j}_1$.
\end{proof}
}

We define a function $|\phi|$ on $\lang(\bSigmaPub)$:
 $|p|=0$ for $p$ atomic, $|\nott\phi| = |\phi|$,
$|\phi\wedge\psi|=\max (|\phi|, |\psi |)$,
 $|\poss_A\phi|= 1+|\phi|$, $|\poss_B\phi|= 1$,
$|\poss_{\BB}^*\phi|=  |\phi|$ for all $\BB\subseteq \Agents$,
and $|[\Pub\ \phi]\psi| = \max(|\phi|,|\psi|)$.

\begin{lemma}
Let  $J\subseteq N^+$,
 $\map{f}{J}{N^+}$,
and $j\in J$. 
Let $\phi\in \lang(\bSigmaPub)$, and let  $1\leq k \leq f(j)$.
If $f(j) - k \geq |\phi|$, then
$$(\bS_f, a) \models \phi
\quadiff
(\bS_f, c^{j}_{k}) \models \phi
\quadiff
(\bT_{f,j}, a) \models \phi.$$
\label{lemma-ainfinity-main}
\end{lemma}

\begin{proof}
By induction on $\phi$.  The base step for atomic $p$
and the induction steps for the  boolean connectives
are all trivial.

\paragraph{The induction step for  $\poss_A \phi$}
Assume that $f(j) - k \geq |\poss_A\phi|$.
We prove the needed equivalences in a cycle.
\rem{
\begin{equation}
(\bS_f, a) \models \phi
\quadiff
(\bS_f, c^{j}_{k+1}) \models \phi
\quadiff
(\bT_{f,j}, a) \models \phi.
\label{lineabove}
\end{equation}
}

Assume  first that $(\bS_{f},a)\models\poss_A\phi$.
Since $a\arrowA a$ and $a\arrowA b$ in $\bS_f$,
there are two cases. 
If $a$ itself satisfies $\phi$  in $A$, then
  the reflexive arrow on $c^{j}_{k}$
shows that $(\bS_{f},c^{j}_{k})\models\poss_A\phi$.
And if $(\bS_{f},b)\models \phi$, then
clearly $(\bS_{f}, c^{j}_{k}) \models \poss\phi$ via $c^{j}_k\arrowA b$.

Second, assume that 
$(\bS_{f},c^{j}_{k})\models\poss_A\phi$.
We shall not go into details on the subcases parallel to
what we saw in the previous paragraph.
The only subcase worth mentioning is where
$(\bS_{f},c^{j}_{k+1})\models\phi$.  Note here that
$f(j) - (k+1) \geq|\poss_A\phi| - 1  = |\phi|$.
In particular, we have $k+1  \leq f(j)$.
And so by our induction hypothesis, $(\bT_{f,j}, a) \models \phi$.
Thus $(\bT_{f,j},a)\models\poss_A\phi$.

Finally,
assume that  $(\bT_{f,j},a)\models\poss_A\phi$.
 Remembering that $a\arrowA c^{j}_1$ in $B$,
 the important subcase is where
 $(\bT_{f,j},c^{j}_1)\models \phi$. 
By Lemma~\ref{lemma-almost-total-bisim}, 
 $(\bS_{f},c^{j}_1)\models \phi$. 
Also, $f(j) -  1 \geq f(j) - (k+1) \geq |\phi|$.
 So by the induction hypothesis, $(\bS_{f},a)\models \phi$.
And then $(\bS_{f},a)\models\poss_A \phi$.

\paragraph{The induction step for $\poss_B\phi$}
Recall that $|\poss_B\phi| =1$.  Since we assume
$f(j) - k \geq 1$, we have $k < f(j)$.   
In particular, 
$(\bS_{f},c^{j}_{k})\not\models\poss_B\phi$.
All
three needed statements   are automatically false,
so all three are equivalent.

\rem{
Again by the induction hypothesis in (\ref{lineabove}),
$(\bS_{f},c^{j}_{k+1})\models \phi$.   Hence
$(\bS_{f},c^{j}_{k})\models \poss\phi$.
}

\paragraph{The induction step for $\poss_{\BB}^*\phi$}
Assume $f(j) -k > |\poss^*_A\phi| = |\phi|$.
We again have a cycle of equivalences.   
Before we turn to them, we remind the reader that
with two agents, $\BB$ can 
be $\set{A}$, $\set{B}$ or $\set{A,B}$ here.  Note that $\arrowB$ is already
transitive, so all the work for $\poss^*_B\phi$ has been done
above already.   
Also, $\arrowB$ is a subrelation of $\arrowA$ in $\bS_{f}$,
and  $\arrowB$ is a subrelation of $\arrowAstar$ in $\bS_{f}$.
Thus we   only need to
work with $\BB  = \set{A}$  in this induction step. 

\renewcommand{\BB}{A}

Assume first that $(\bS_{f},a)\models \poss_{\BB}^* \phi$.
There are a number of cases.
If $(\bS_{f},a)\models  \phi$, then by induction 
hypothesis, 
$(\bS_{f},c^{j}_k)\models   \phi$.
So $(\bS_{f},c^{j}_k)\models \poss_{\BB}^* \phi$.
The other case is where $a\arrowBB x$ and $x\models \phi$.
Note in this case that $c^{j}_k\arrowBB x$ as well:
all non-null paths from $a$ go through $b$, and $c^{j}\arrowA b$.
Thus in this case
we have $(\bS_{f},c^{j}_k)\models \poss_{\BB}^* \phi$.

Next,  assume that 
$(\bS_{f},c^{j}_k)\models \poss_{\BB}^* \phi$.
Let $x$ be such that
 $c^{j}_k\arrowAstar x$ and 
$(\bS_{f},x)\models  \phi$.  Clearly $x\neq a$.
By the bisimulation noted in Lemma~\ref{lemma-almost-total-bisim}, 
$(\bT_{f,j},x)\models  \phi$. 
Also, $b\arrowAstar c^{j}_k\arrowAstar x$.
So 
$(\bT_{f,j},a)\models \poss_{\BB}^*\phi$.

Suppose finally that $(\bT_{f,j},a)\models \poss_{A}^*\phi$.
If $(\bT_{f,j},a)\models \phi$, then
 by induction hypothesis
we have $(\bS_{f},a)\models  \phi$.  So
in this case $(\bS_{f},a)\models\poss_{\BB}^*\phi$ as well.
The other case is where 
$(\bT_{f,j},x)\models  \phi$ for some $x\neq a$.
By bisimulation again,
$(\bS_{f},x)\models  \phi$.  And in $\bS_f$, $a\arrowAstar x$.
Thus $(\bS_{f},a)\models \poss_A^* \phi$, as desired.

\rem{
and
 $(\bS_{f},a)\models \poss_A^* \phi$,  and then prove
$(\bS_{f},c^{j}_k)\models \poss_A^* \phi$.
 We have several cases as to the
witness for $(\bS_{f},a)\models \poss_A^* \phi$.  If 
$(\bS_{f},b)\models   \phi$, then since $c^{j}_k\arrowA b$, 
$(\bS_{f},c^{j}_k)\models \poss_A^* \phi$.
The same idea works for all points $x$ in the structure 
besides $a$ itself: for all these $x$, $c^{j}_k\arrowAstar x$.
We are left with the possibility that
$(\bS_{f},a)\models  \phi$.  In this case, the induction hypothesis
tells us that 
$(\bS_{f},c^{j}_k)\models \phi$ and so indeed
$(\bS_{f},c^{j}_k)\models \poss_A^* \phi$.
}

\paragraph{The induction step for $[\Pub\ \phi]\psi$}
Assume that $f(j) - k > |[\Pub\ \phi]\psi| = \max(|\phi|,|\psi|)$.
We remind the reader 
of the definitions and of some of the notation used in
Section~\ref{section-Cn}.
Let 
 $$\hat{\bS}_f
\quadeq \bS_f\otimes (\bSigmaPri,\Pri^A,\propositionp,\trueproposition),$$
and let $\hat{\bT}_{f,j}$ be defined similarly.
In this notation,
\begin{equation}
\semantics{[\Pub\ \phi]\psi}{\bS_f}
\quadeq
\set{x\in S_f : 
\mbox{if $x\in \semantics{\phi}{\bS_f}$, then 
$(x,\Pri^A)\in \semantics{\psi}{\hat{\bS}_f}$}}
\label{remind}
\end{equation}
A similar equation holds for $\bT_{f,j}$.

Case I: $(\bS_{f}, a)\not\models \phi$.
Since 
$f(j) - k 
\geq |[\Pub\ \phi]\psi| \geq  |\phi|$,
 we have by induction hypothesis that
 $(\bS_{f}, c^{j}_k)\not\models \phi$, and
$(\bT_{f,j}, a)\not\models \phi$.   So 
 automatically
all of the model-world pairs satisfy our 
sentence $[\Pub\ \phi]\psi$.

Case II: $(\bS_{f}, a)\models \phi$ and $(\bS_{f}, b)\not\models \phi$.
Then by induction hypothesis, 
 $(\bS_{f}, c^{j}_k)\models \phi$, and
$(\bT_{f,j}, a)\models \phi$, and by Lemma~\ref{lemma-almost-total-bisim},
  $(\bT_{f,j}, b)\not\models \phi$.
 $\hat{\bS}_f$ and $\hat{\bT}_{f,j}$ are 
isomorphic to submodels of $\bS_f$ and $\bT_{f,j}$ which do not contain $b$.
Each point satisfies $p$
and has no $\arrowB$-successors, and every point has a
$\arrowA$-successor (itself).  Thus both model-world pairs
are  
bisimilar
to the one-point model
satisfying $p$.  So we are easily done in this case.

Case III: 
$(\bS_{f}, a)\models \phi$, $(\bS_{f}, b)\models \phi$, and 
for some $k$ such that $1\leq k\leq f(j)$, 
$(\bS_{f},c^{j}_k)\not\models \phi$.
Let $l$ be  least with this property.  By induction hypothesis,
we have $l > f(j) - |\phi| \geq k$.   So we see that
$c^j_1$, $c^j_2$, $\ldots$, $c^j_k$, $\ldots$, $c^j_{l-1}$ all
satisfy $\phi$, but $c^j_k$ does not.
We claim that
\begin{equation} (\hat{\bS}_f,(a,\Pri^A)) \quadequiv 
(\hat{\bS}_f,(c^j_k,\Pri^A))
\quadequiv
(\hat{\bT}_{f,j},(a,\Pri^A))
\label{eq-inview}
\end{equation}
The bisimulation showing the first of these two assertions is
the identity, together with the set of pairs
$$\set{((a, \Pri^A),(c^j_l,\Pri^A)) :  (c^j_l,\Pri^A)\in \hat{\bS}_f}.
$$
The bisimulation showing $(\hat{\bS}_f,(c^j_k,\Pri^A))
\equiv
(\hat{\bT}_{f,j},(a,\Pri^A))$ adds some further pairs:
$$\set{((c^j_l, \Pri^A),(c^j_m,\Pri^A)) :  (c^j_l,\Pri^A),
c^j_m\in \hat{\bS}_f}.
$$
In view of the bisimulations in 
the equivalence (\ref{eq-inview}) and the semantics
in (\ref{remind}),
the  equivalence in this case is now immediate.

Case IV: 
$(\bS_{f}, a)\models \phi$, $(\bS_{f}, b)\models \phi$, and 
for all $k$ such that $1\leq k\leq f(j)$, 
$(\bS_{f},c^{j}_k)\models \phi$.
Let 
$$ I \quadeq \set{i\in J : (\bS_{f},c^i_{f(i)}) \models \phi}.
$$
For $i\in J$, let $$g(i) \quadeq \max\set{k :  c^i_{f(i)-k +1},\ldots, c^i_{f(i)}\in
\semantics{\phi}{\bS_f}}.$$ 
Note that $j\in J$ and $g(j) = f(j)$.  For other $i\in J$, though,
$g(i)$ might well be less than $f(i)$.
Also, the key point is that all $\arrowA$-relations  out of 
$\set{c^i_1,\ldots, c^i_{f(i)-k}}\cap \semantics{\phi}{\bS_f}$
either land in this same set, or at $b$.

We claim that $(\hat{\bS}_f,a)  \equiv (\bS_g,a)$, 
$(\hat{\bS}_f,c^j_k)  \equiv (\bS_g,c^j_k)$, and
$(\hat{\bT}_{f,j},a)  \equiv (\bT_{g,j},a)$.
In all three 
of these points, we use the  same   bisimulation relation $R$:
let  $(a,\Pub)\ R\ a$,
 $(b,\Pub)\ R\ b$,
$$(c^i_k,\Pub) \ R \ c^i_{k-f(i) + g(i)}$$ 
for $i\in J$ and ${f(i)-g(i) +1} \leq  k \leq {f(i)}$,
and  
$(c^i_k,0) \ R \ a $ otherwise.

This verifies our claim.
Recall that $f(j) - k \geq|[\Pub\ \phi]\psi| \geq  |\psi|$.
We apply our induction hypothesis to $(J,j,g)$, $\psi$ and 
 $k$.
We see that 
$$(\bS_g, a) \models \phi
\quadiff
(\bS_g, c^{j}_{k}) \models \phi
\quadiff
(\bT_{g,j}, a) \models \phi.$$
In view of the bisimulation and the semantics of $[\Pub\ \phi]\psi$
detailed in (\ref{remind}),
we see that
$$(\bS_f, a) \models [\Pub\ \phi]\psi
\quadiff
(\bS_f, c^{j}_{k}) \models [\Pub\ \phi]\psi
\quadiff
(\bT_{f,j}, a) \models [\Pub\ \phi]\psi.$$
This concludes the proof of Lemma~\ref{lemma-ainfinity-main}.
\end{proof}

\begin{proof} {\bf of Theorem~\ref{theorem-negative-definability}}.
Suppose towards
a contradiction
that $\chi$ were equivalent to the set
$T\subseteq  \langAB$.
Let $f$ be the identity function on $N^+$. 
By Lemma~\ref{lemma-ainfinity-sufficient}, $(\bS_f, a)$ does 
not satisfy all of the sentences in $T$.
Let $\phi\in T$ be such that $(\bS_f, a) \models\nott\phi$.
Let $j = |\nott\phi| + 1$, so that $f(j) -1 \geq  |\nott \phi|$.
By Lemma~\ref{lemma-ainfinity-main}, 
$(\bT_{f,j}, a)\models\nott \phi$.
But  again by  
 Lemma~\ref{lemma-ainfinity-sufficient}, $(\bT_{f,j}, a)$
satisfies all sentences in $T$.  This is a contradiction.
\end{proof}

\subsection{$\lang(\bSigmapub)$ lacks the finite model property}
\label{nofmp}

We conclude with a result on {\em iterating\/} epistemic
actions.  We gave a semantics of sentences of the form
$[\alpha^*]\phi$ earlier in the paper.  The ability to
iterate actions is useful in algorithms and even in 
informal presentations of scenarios.
As it happens,  
the iteration operation on actions makes our systems 
quite expressive, so much so that the finite model
property fails for $\lang_1(\bSigma)$ as soon as $\bSigma$
contains public announcements.

\begin{proposition}
$[\Pub\ \poss\true]^* \poss\necc\false$
is satisfiable, but not in any finite model.
\end{proposition}

\begin{proof}
A state $s$ in a model $\bS$ is called an
{\em end state\/} if $s$ has no successors.
For each model $\bS$, let $\bS'$ be the same
model, except with the end states removed.
$\bS'$ is isomorphic to what we have written
earlier as $\bS\otimes(\Pub\  \poss\true)$, that
is, the result of publically announcing that
some world is possible.  So we see that
our sentence 
$[\Pub\ \poss\true]^* \poss\necc\false$
holds of $s$ just in case the following holds
for all $n$:
\begin{enumerate}
\item $s\in \bS^{(n)}$, the $n$-fold
application of the derivative operation to $\bS$.
\item $s$ has some child which is an end node
(and hence $t$ would not belong to $\bS^{(n+1)}$).
\end{enumerate}
It is clear that any model of 
$[\Pub\ \poss\true]^*  \poss\necc\false$ 
must be infinite, since the sets
$S^{(n+1)}\setminus S^{(n)}$ are pairwise disjoint
and nonempty.

There are  well-known models of 
$[\Pub\ \poss\true]^*  \poss\necc\false$.
One would be the set of decreasing sequences
of natural numbers, with $s \rightarrow t$ iff
$t$ is a one-point extension of $s$.
The end nodes are the sequences that end in $0$.
For each $n$, $\bS^{(n)}$ is the submodel
consisting of the sequences which end in $n$.
\end{proof}

This has several dramatic
consequences for this work.  First,
it means that the logics cannot be translated into
the modal mu-calculus, since that logic is known to
have the finite model property.  More importantly,  we have the following
extension of this result:

\begin{theorem} [Miller and Moss~\cite{miller}]
Concerning $\lang(\bSigmapub)$:
\begin{enumerate}
\item $\set{\phi : \phi \mbox{ is satisfiable}}$ is $\Sigma^1_1$-complete.
\item
$\set{\phi : \phi \mbox{ is satisfiable on a finite (tree)
model}}$ is $\Sigma^0_1$-complete.
\end{enumerate}
\label{theorem-withmiller}
\end{theorem}

Indeed, these even hold when $\lang(\bSigmapub)$ is
replaced by small fragments, such as the fragment built from
$[\Pub\ \poss\true]$, $\poss^*$, $\poss$, and boolean connectives
(yet without atomic sentences); or,
the fragment built from arbitrary iterated relativizations
and modal logic (without $\poss^*$).
  These negative results go via 
reduction from domino problems.

The upshot is that   logics  which allow for arbitrary finite
iterations of  epistemic actions are not going to be  axiomatizable.

\bigskip

We feel that our  results on expressive power 
are just a sample of what could be done in this area.
We did not investigate the next natural questions:
Do announcements with suspicious outsiders extend the
expressive power of modal logic with all secure private
announcements and common knowledge operators?  And then
do announcements with common knowledge of suspicion
add further expressive power?

 %% Section 7
%\include{azero}
%\include{conclusions}
%\include{appendix}
%\include{acknowledgments}
%\addcontentsline{toc}{section}{References}

\end{document}